\documentclass{amsart}%
\usepackage{amsfonts}
\usepackage{amsmath}
\usepackage{amssymb}
\usepackage{graphicx}%
\providecommand{\U}[1]{\protect\rule{.1in}{.1in}}
\newtheorem{theorem}{Theorem}
\theoremstyle{plain}
\newtheorem{acknowledgement}{Acknowledgement}

\newtheorem{conjecture}{Conjecture}

\newtheorem{definition}{Definition}
\newtheorem{example}{Example}

\newtheorem{proposition}{Proposition}
\newtheorem{remark}{Remark}

\numberwithin{equation}{section}
\begin{document}
\title[Quantum Knots]{Quantum Knots and Mosaics}
\author{Samuel J. Lomonaco}
\address{University of Maryland Baltimore County (UMBC)\\
Baltimore, MD \ 21250 \ \ USA}
\email{lomonaco@umbc.edu}
\urladdr{http://www.csee.umbc.edu/\symbol{126}lomonaco}
\author{Louis H. Kauffman}
\address{University of Illinois at Chicago\\
Chicago, IL \ 60607-7045 \ \ USA}
\email{kauffman@uic.edu}
\urladdr{http://www.math.uic.edu/\symbol{126}kauffman}
\date{February 24, 2008}

\begin{abstract}
In this paper, we give a precise and workable definition of a \textbf{quantum
knot system}, the states of which are called \textbf{quantum knots}. This
definition can be viewed as a blueprint for the construction of an actual
physical quantum system.

Moreover, this definition of a quantum knot system is intended to represent
the \textquotedblleft quantum embodiment" of a closed knotted physical piece
of rope. A quantum knot, as a state of this system, represents the state of
such a knotted closed piece of rope, i.e., the particular spatial
configuration of the knot tied in the rope. Associated with a quantum knot
system is a group of unitary transformations, called the \textbf{ambient
group}, which represents all possible ways of moving the rope around (without
cutting the rope, and without letting the rope pass through itself.)

Of course, unlike a classical closed piece of rope, a quantum knot can exhibit
non-classical behavior, such as quantum superposition and quantum
entanglement. This raises some interesting and puzzling questions about the
relation between topological and quantum entanglement.

The \textbf{knot type} of a quantum knot is simply the orbit of the quantum
knot under the action of the ambient group. We investigate quantum observables
which are invariants of quantum knot type. We also study the Hamiltonians
associated with the generators of the ambient group, and briefly look at the
quantum tunneling of overcrossings into undercrossings.

A basic building block in this paper is a \textbf{mosaic system} which is a
formal (rewriting) system of symbol strings. We conjecture that this formal
system fully captures in an axiomatic way all of the properties of tame knot theory.

\end{abstract}
\maketitle
\tableofcontents

\section{Introduction}

\bigskip

The objective of this paper is to set the foundation for a research program on
quantum knots\footnote{A PowerPoint presentation of this paper can be found at
\par
\qquad http://www.csee.umbc.edu/\symbol{126}lomonaco/Lextures.html}.

\bigskip

\textit{For simplicity of exposition, we will throughout this paper frequently
use the term "knot" to mean either a knot or a link.}

\bigskip

In part 1 of this paper, we create a formal system $\left(  \mathbb{K}%
,\mathbb{A}\right)  $ consisting of\medskip

\begin{itemize}
\item[\textbf{1)}] A graded set $\mathbb{K}$ of symbol strings, called
\textbf{knot mosaics}, and

\item[\textbf{2)}] A graded subgroup $\mathbb{A}$, called the \textbf{knot
mosaic ambient group}, of the group of all permutations of the set of knot
mosaics $\mathbb{K}$.\bigskip
\end{itemize}

We conjecture that the formal system $\left(  \mathbb{K},\mathbb{A}\right)
$\ fully captures the entire structure of tame knot theory.

\bigskip

Three examples of knot mosaics are given below:\bigskip

\hspace{-0.7in}$%
\begin{array}
[c]{cccc}%
%TCIMACRO{\FRAME{itbpF}{0.3269in}{0.3269in}{0in}{}{}{ut00.ps}%
%width 0.3269in;  height 0.3269in;  depth 0in;  original-width 3in;
%original-height 3in;  cropleft "0";  croptop "1";  cropright "1";
%cropbottom "0";  filename 'ut00.ps';file-properties "XNPEU";}}}%
%BeginExpansion
{\includegraphics[
%natheight=3.000000in,
%natwidth=3.000000in,
height=0.3269in,
width=0.3269in
]%
{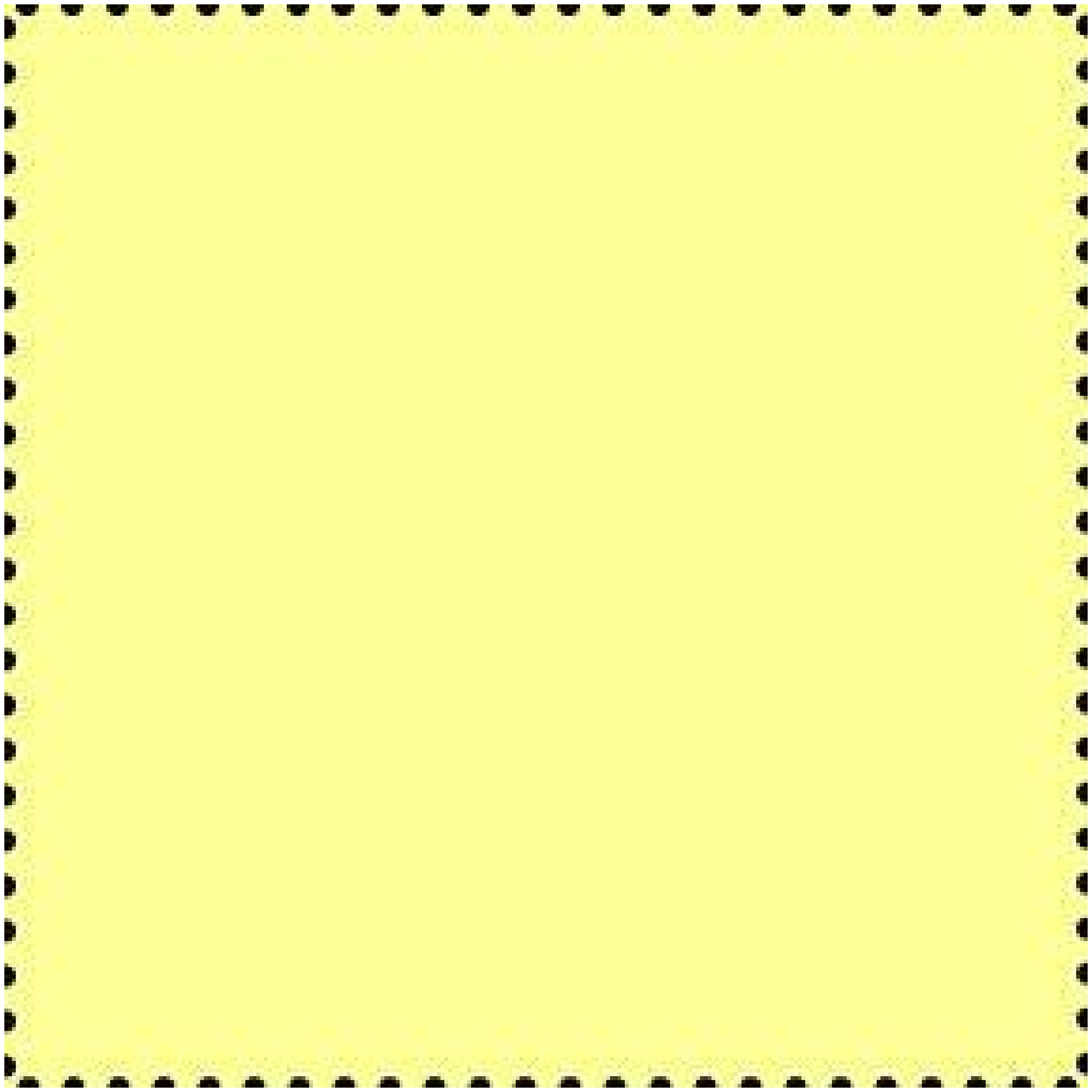}%
}%
%EndExpansion
&
%TCIMACRO{\FRAME{itbpF}{0.3269in}{0.3269in}{0in}{}{}{ut02.ps}%
%width 0.3269in;  height 0.3269in;  depth 0in;  original-width 3in;
%original-height 3in;  cropleft "0";  croptop "1";  cropright "1";
%cropbottom "0";  filename 'ut02.ps';file-properties "XNPEU";}}}%
%BeginExpansion
{\includegraphics[
%natheight=3.000000in,
%natwidth=3.000000in,
height=0.3269in,
width=0.3269in
]%
{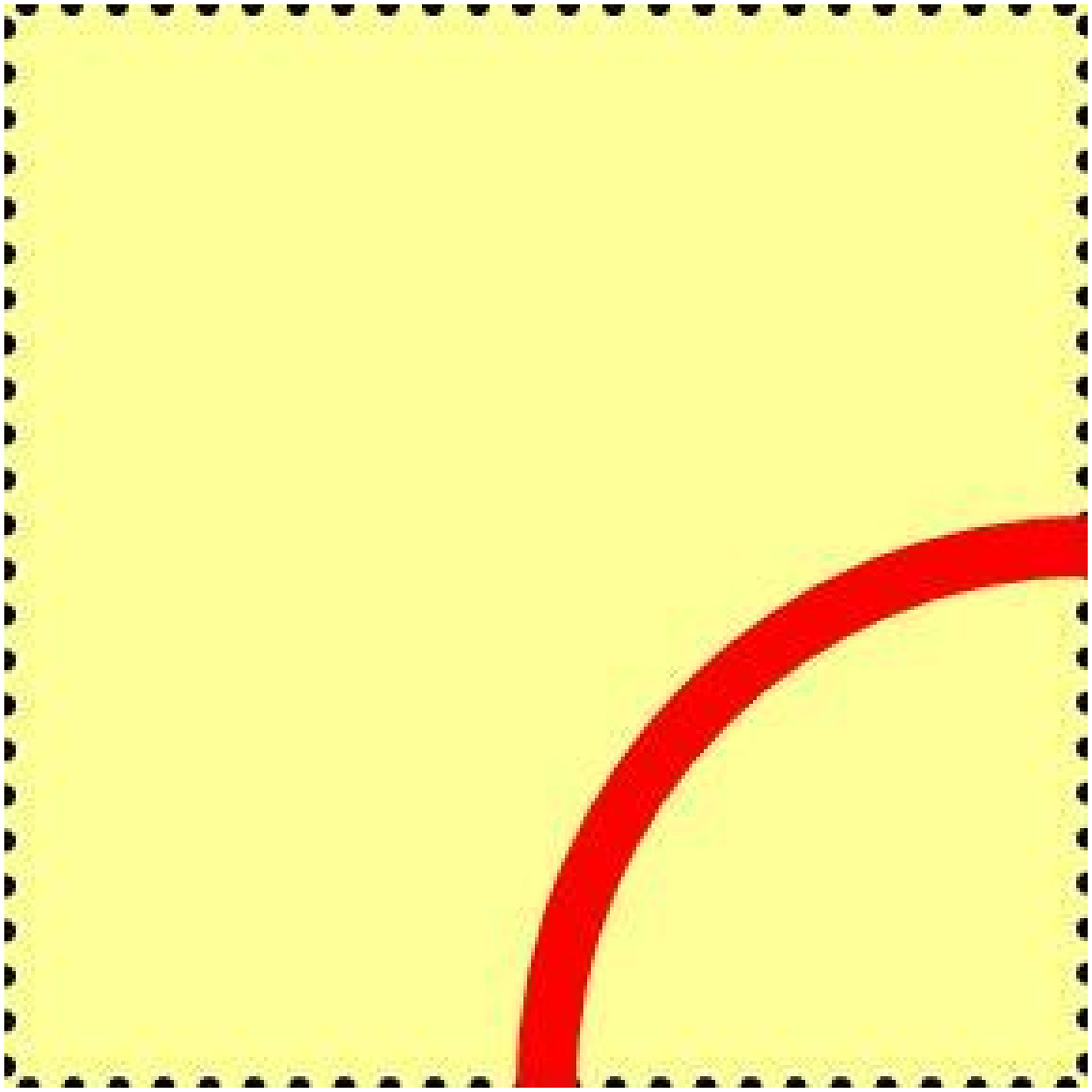}%
}%
%EndExpansion
&
%TCIMACRO{\FRAME{itbpF}{0.3269in}{0.3269in}{0in}{}{}{ut01.ps}%
%width 0.3269in;  height 0.3269in;  depth 0in;  original-width 3in;
%original-height 3in;  cropleft "0";  croptop "1";  cropright "1";
%cropbottom "0";  filename 'ut01.ps';file-properties "XNPEU";}}}%
%BeginExpansion
{\includegraphics[
%natheight=3.000000in,
%natwidth=3.000000in,
height=0.3269in,
width=0.3269in
]%
{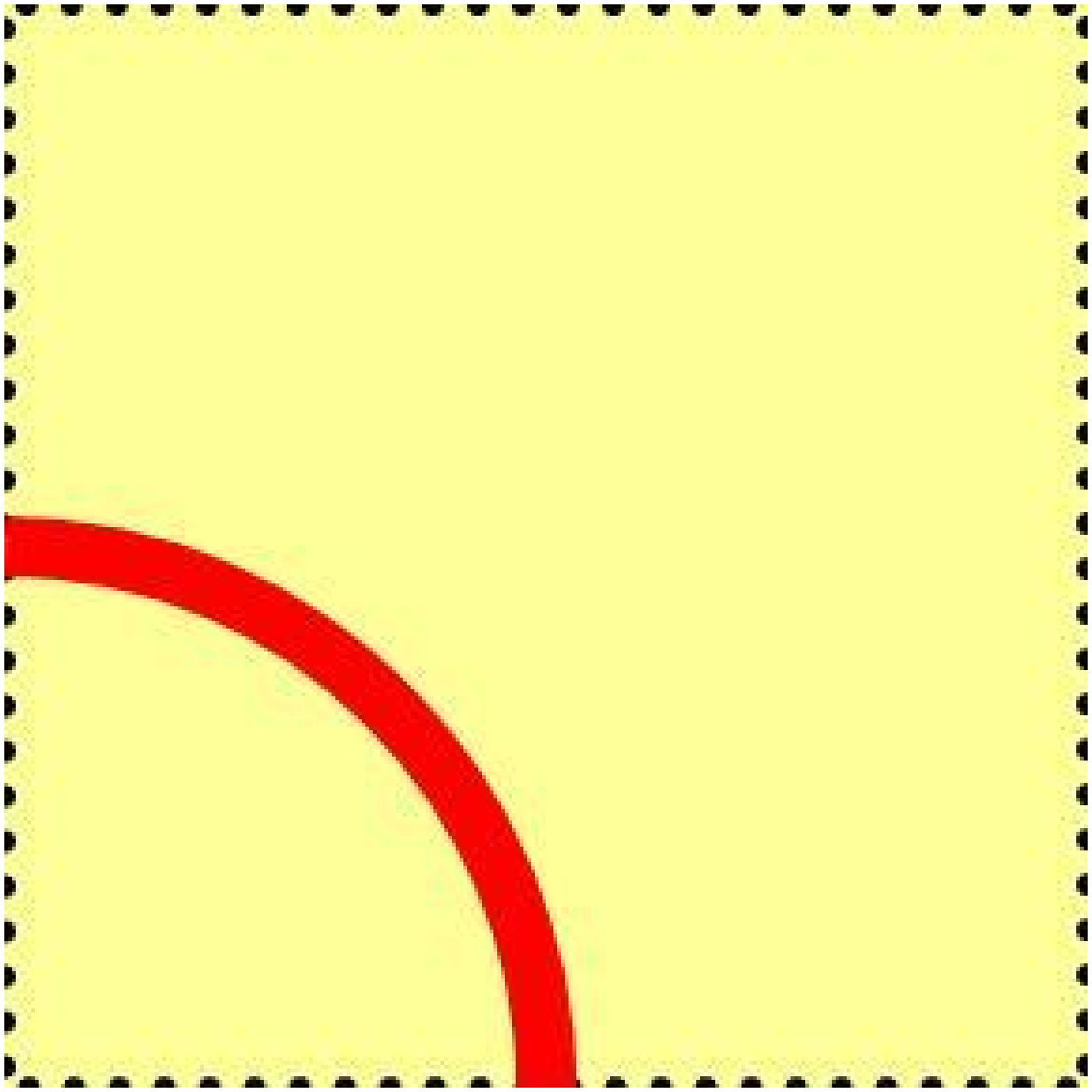}%
}%
%EndExpansion
&
%TCIMACRO{\FRAME{itbpF}{0.3269in}{0.3269in}{0in}{}{}{ut00.ps}%
%width 0.3269in;  height 0.3269in;  depth 0in;  original-width 3in;
%original-height 3in;  cropleft "0";  croptop "1";  cropright "1";
%cropbottom "0";  filename 'ut00.ps';file-properties "XNPEU";}}}%
%BeginExpansion
{\includegraphics[
%natheight=3.000000in,
%natwidth=3.000000in,
height=0.3269in,
width=0.3269in
]%
{ut00.ps}%
}%
%EndExpansion
\\
\end{array}
$, $%
\begin{array}
[c]{cccc}%
%TCIMACRO{\FRAME{itbpF}{0.3269in}{0.3269in}{0in}{}{}{ut00.ps}%
%width 0.3269in;  height 0.3269in;  depth 0in;  original-width 3in;
%original-height 3in;  cropleft "0";  croptop "1";  cropright "1";
%cropbottom "0";  filename 'ut00.ps';file-properties "XNPEU";}}}%
%BeginExpansion
{\includegraphics[
%natheight=3.000000in,
%natwidth=3.000000in,
height=0.3269in,
width=0.3269in
]%
{ut00.ps}%
}%
%EndExpansion
&
%TCIMACRO{\FRAME{itbpF}{0.3269in}{0.3269in}{0in}{}{}{ut02.ps}%
%width 0.3269in;  height 0.3269in;  depth 0in;  original-width 3in;
%original-height 3in;  cropleft "0";  croptop "1";  cropright "1";
%cropbottom "0";  filename 'ut02.ps';file-properties "XNPEU";}}}%
%BeginExpansion
{\includegraphics[
%natheight=3.000000in,
%natwidth=3.000000in,
height=0.3269in,
width=0.3269in
]%
{ut02.ps}%
}%
%EndExpansion
&
%TCIMACRO{\FRAME{itbpF}{0.3269in}{0.3269in}{0in}{}{}{ut01.ps}%
%width 0.3269in;  height 0.3269in;  depth 0in;  original-width 3in;
%original-height 3in;  cropleft "0";  croptop "1";  cropright "1";
%cropbottom "0";  filename 'ut01.ps';file-properties "XNPEU";}}}%
%BeginExpansion
{\includegraphics[
%natheight=3.000000in,
%natwidth=3.000000in,
height=0.3269in,
width=0.3269in
]%
{ut01.ps}%
}%
%EndExpansion
&
%TCIMACRO{\FRAME{itbpF}{0.3269in}{0.3269in}{0in}{}{}{ut00.ps}%
%width 0.3269in;  height 0.3269in;  depth 0in;  original-width 3in;
%original-height 3in;  cropleft "0";  croptop "1";  cropright "1";
%cropbottom "0";  filename 'ut00.ps';file-properties "XNPEU";}}}%
%BeginExpansion
{\includegraphics[
%natheight=3.000000in,
%natwidth=3.000000in,
height=0.3269in,
width=0.3269in
]%
{ut00.ps}%
}%
%EndExpansion
\\
\end{array}
$, and $%
\begin{array}
[c]{cccc}%
%TCIMACRO{\FRAME{itbpF}{0.3269in}{0.3269in}{0in}{}{}{ut00.ps}%
%width 0.3269in;  height 0.3269in;  depth 0in;  original-width 3in;
%original-height 3in;  cropleft "0";  croptop "1";  cropright "1";
%cropbottom "0";  filename 'ut00.ps';file-properties "XNPEU";}}}%
%BeginExpansion
{\includegraphics[
%natheight=3.000000in,
%natwidth=3.000000in,
height=0.3269in,
width=0.3269in
]%
{ut00.ps}%
}%
%EndExpansion
&
%TCIMACRO{\FRAME{itbpF}{0.3269in}{0.3269in}{0in}{}{}{ut02.ps}%
%width 0.3269in;  height 0.3269in;  depth 0in;  original-width 3in;
%original-height 3in;  cropleft "0";  croptop "1";  cropright "1";
%cropbottom "0";  filename 'ut02.ps';file-properties "XNPEU";}}}%
%BeginExpansion
{\includegraphics[
%natheight=3.000000in,
%natwidth=3.000000in,
height=0.3269in,
width=0.3269in
]%
{ut02.ps}%
}%
%EndExpansion
&
%TCIMACRO{\FRAME{itbpF}{0.3269in}{0.3269in}{0in}{}{}{ut01.ps}%
%width 0.3269in;  height 0.3269in;  depth 0in;  original-width 3in;
%original-height 3in;  cropleft "0";  croptop "1";  cropright "1";
%cropbottom "0";  filename 'ut01.ps';file-properties "XNPEU";}}}%
%BeginExpansion
{\includegraphics[
%natheight=3.000000in,
%natwidth=3.000000in,
height=0.3269in,
width=0.3269in
]%
{ut01.ps}%
}%
%EndExpansion
&
%TCIMACRO{\FRAME{itbpF}{0.3269in}{0.3269in}{0in}{}{}{ut00.ps}%
%width 0.3269in;  height 0.3269in;  depth 0in;  original-width 3in;
%original-height 3in;  cropleft "0";  croptop "1";  cropright "1";
%cropbottom "0";  filename 'ut00.ps';file-properties "XNPEU";}}}%
%BeginExpansion
{\includegraphics[
%natheight=3.000000in,
%natwidth=3.000000in,
height=0.3269in,
width=0.3269in
]%
{ut00.ps}%
}%
%EndExpansion
\\
\end{array}
$ \ 

Each of these knot mosaics is a string made up of the following 11 symbols
$\hspace{-0.5in}$%

\[%
%
,
\]
called \textbf{mosaic tiles}.

\bigskip

\noindent An example of an element in the mosaic ambient group $\mathbb{A}$ is
the \textbf{mosaic Reidemeister 1 move} illustrated below:
\[
N\overset{\left(  0,1\right)  }{\longleftrightarrow}N^{\prime}=%
\begin{array}
[c]{cc}%
%TCIMACRO{\FRAME{itbpF}{0.3269in}{0.3269in}{0in}{}{}{ut02.ps}%
%width 0.3269in;  height 0.3269in;  depth 0in;  original-width 3in;
%original-height 3in;  cropleft "0";  croptop "1";  cropright "1";
%cropbottom "0";  filename 'ut02.ps';file-properties "XNPEU";}}}%
%BeginExpansion
{\includegraphics[
%natheight=3.000000in,
%natwidth=3.000000in,
height=0.3269in,
width=0.3269in
]%
{ut02.ps}%
}%
%EndExpansion
&
%TCIMACRO{\FRAME{itbpF}{0.3269in}{0.3269in}{0in}{}{}{ut01.ps}%
%width 0.3269in;  height 0.3269in;  depth 0in;  original-width 3in;
%original-height 3in;  cropleft "0";  croptop "1";  cropright "1";
%cropbottom "0";  filename 'ut01.ps';file-properties "XNPEU";}}}%
%BeginExpansion
{\includegraphics[
%natheight=3.000000in,
%natwidth=3.000000in,
height=0.3269in,
width=0.3269in
]%
{ut01.ps}%
}%
%EndExpansion
\\
\end{array}
\overset{\left(  0,1\right)  }{\longleftrightarrow}%
\begin{array}
[c]{cc}%
%TCIMACRO{\FRAME{itbpF}{0.3269in}{0.3269in}{0in}{}{}{ut02.ps}%
%width 0.3269in;  height 0.3269in;  depth 0in;  original-width 3in;
%original-height 3in;  cropleft "0";  croptop "1";  cropright "1";
%cropbottom "0";  filename 'ut02.ps';file-properties "XNPEU";}}}%
%BeginExpansion
{\includegraphics[
%natheight=3.000000in,
%natwidth=3.000000in,
height=0.3269in,
width=0.3269in
]%
{ut02.ps}%
}%
%EndExpansion
&
%TCIMACRO{\FRAME{itbpF}{0.3269in}{0.3269in}{0in}{}{}{ut01.ps}%
%width 0.3269in;  height 0.3269in;  depth 0in;  original-width 3in;
%original-height 3in;  cropleft "0";  croptop "1";  cropright "1";
%cropbottom "0";  filename 'ut01.ps';file-properties "XNPEU";}}}%
%BeginExpansion
{\includegraphics[
%natheight=3.000000in,
%natwidth=3.000000in,
height=0.3269in,
width=0.3269in
]%
{ut01.ps}%
}%
%EndExpansion
\\
\end{array}
\]
This mosaic Reidemeister 1 move $N\overset{\left(  0,1\right)  }%
{\longleftrightarrow}N^{\prime}$ is a permutation which is the product of
disjoint transpositions, as illustrated by observing that the Reidemeister 1
move $N\overset{\left(  0,1\right)  }{\longleftrightarrow}N^{\prime}$
interchanges the the following two knot mosaics:

\bigskip%

\[%
%
\ \ \ \ ,
\]

\bigskip

\noindent while it leaves the following mosaic unchanged:\bigskip%

\[%
%
\ \ \
\]

\bigskip

In part 2, the formal system $\left(  \mathbb{K},\mathbb{A}\right)  $ is used
to define a \textbf{quantum knot system} $Q\left(  \mathcal{K},\mathbb{A}%
\right)  $, which is a nested sequence of quantum systems consisting of

\begin{itemize}
\item[\textbf{1)}] A graded Hilbert space $\mathcal{K}$, called the
\textbf{quantum knot state space}, defined by an orthonormal basis labelled by
and in one-to-one correspondence with the set of knot mosaics $\mathbb{K}$, and

\item[\textbf{2)}] An\textbf{ }associated graded control group, also called
the \textbf{ambient group}, and also denoted by $\mathbb{A}$ . The ambient
group $\mathbb{A}$ is a discrete subgroup of the group $U\left(
\mathcal{K}\right)  $ of all unitary transformations on $\mathcal{K}$.
\end{itemize}

\bigskip

A \textbf{quantum knot} is simply a state of the quantum knot system, i.e., an
element of the quantum knot state space $\mathcal{K}$. \ \textbf{Quantum knot
type} is defined as the orbit of the quantum knot under the action of the
ambient group $\mathbb{A}$.

\bigskip

Once having defined what is meant by a quantum knot, we then proceed to find
the Hamiltonians associated with the generators of the ambient group
$\mathbb{A}$, and to study the quantum dynamics induced by Schroedinger's
equation. \ We move on to discuss other Hamiltonians, such as for example
those associated with overcrossings quantum tunnelling into undercrossings.
\ \ We also study a class of quantum observables which are quantum knot invariants.

\bigskip

We should mention that, if one selects a fixed upper bound $n$ on knot
complexity (i.e., a fixed upper bound on the edge length $n$ of the knot
$n$-mosaics under consideration), then a quantum knot system $Q\left(
\mathcal{K}^{(n)},\mathbb{A}(n)\right)  $ is a blueprint for the construction
of an actual physical quantum system. \ Quantum knots could possibly be used
to simulate and to predict the behavior of quantum vortices that appear both
in liquid helium II and in the Bose-Einstein condensate \ They might also
possibly be a mathematical model for gaining some insight into the charge
quantification that is manifest in the fractional quantum Hall effect.

\bigskip

In the conclusion, we list a number of open questions and possible future
research directions. \ A complete table of all knot $3$-mosaics is given in
Appendix A. \ Finally, in Appendix B, we briefly outline the theory of and the
construction of \textbf{oriented knot mosaics} and\textbf{ oriented quantum
knots}. \ 

\bigskip

\textbf{The motivating intuition for the above mathematical construct
}$Q\left(  \mathcal{K},\mathbb{A}\right)  $ \textbf{is as follows:} \ A
quantum knot system is intended to represent the "quantum embodiment" of a
closed knotted physical piece of rope. \ A quantum knot is meant to represent
the state of the knotted rope, i.e., the particular spatial configuration of
the knot tied in the rope. The elements of the the ambient unitary group are
intended to represent all possible ways of moving the rope around (without
cutting the rope, and without letting it pass through itself.) \ The quantum
system is necessarily a nested set of quantum systems because one must use
longer and longer pieces of rope to tie knots of greater and greater complexity.

\bigskip

Of course, unlike classical knotted pieces of rope, quantum knots can also
represent the quantum superpositions (and also the quantum entanglements) of a
number of knotted pieces of rope. \ This raises an interesting question about
the relation between topological entanglement and quantum entanglement.\ 

\bigskip

\section{Part 1: Knot Mosaics}

\bigskip

\subsection{Unoriented knot mosaics}

\bigskip

Let $\mathbb{T}^{(u)}$ denote the set of the following 11 symbols \bigskip%

\[%
%TCIMACRO{\FRAME{itbpF}{0.3269in}{0.3269in}{0in}{}{}{ut00.ps}%
%width 0.3269in;  height 0.3269in;  depth 0in;  original-width 3in;
%original-height 3in;  cropleft "0";  croptop "1";  cropright "1";
%cropbottom "0";  filename 'ut00.ps';file-properties "XNPEU";}}}%
%BeginExpansion
{\includegraphics[
%natheight=3.000000in,
%natwidth=3.000000in,
height=0.3269in,
width=0.3269in
]%
{ut00.ps}%
}%
%EndExpansion
\qquad\ \ \ \
%TCIMACRO{\FRAME{itbpF}{0.3269in}{0.3269in}{0in}{}{}{ut01.ps}%
%width 0.3269in;  height 0.3269in;  depth 0in;  original-width 3in;
%original-height 3in;  cropleft "0";  croptop "1";  cropright "1";
%cropbottom "0";  filename 'ut01.ps';file-properties "XNPEU";}}}%
%BeginExpansion
{\includegraphics[
%natheight=3.000000in,
%natwidth=3.000000in,
height=0.3269in,
width=0.3269in
]%
{ut01.ps}%
}%
%EndExpansion
\
%TCIMACRO{\FRAME{itbpF}{0.3269in}{0.3269in}{0in}{}{}{ut02.ps}%
%width 0.3269in;  height 0.3269in;  depth 0in;  original-width 3in;
%original-height 3in;  cropleft "0";  croptop "1";  cropright "1";
%cropbottom "0";  filename 'ut02.ps';file-properties "XNPEU";}}}%
%BeginExpansion
{\includegraphics[
%natheight=3.000000in,
%natwidth=3.000000in,
height=0.3269in,
width=0.3269in
]%
{ut02.ps}%
}%
%EndExpansion
\
%TCIMACRO{\FRAME{itbpF}{0.3269in}{0.3269in}{0in}{}{}{ut03.ps}%
%width 0.3269in;  height 0.3269in;  depth 0in;  original-width 3in;
%original-height 3in;  cropleft "0";  croptop "1";  cropright "1";
%cropbottom "0";  filename 'ut03.ps';file-properties "XNPEU";}}}%
%BeginExpansion
{\includegraphics[
%natheight=3.000000in,
%natwidth=3.000000in,
height=0.3269in,
width=0.3269in
]%
{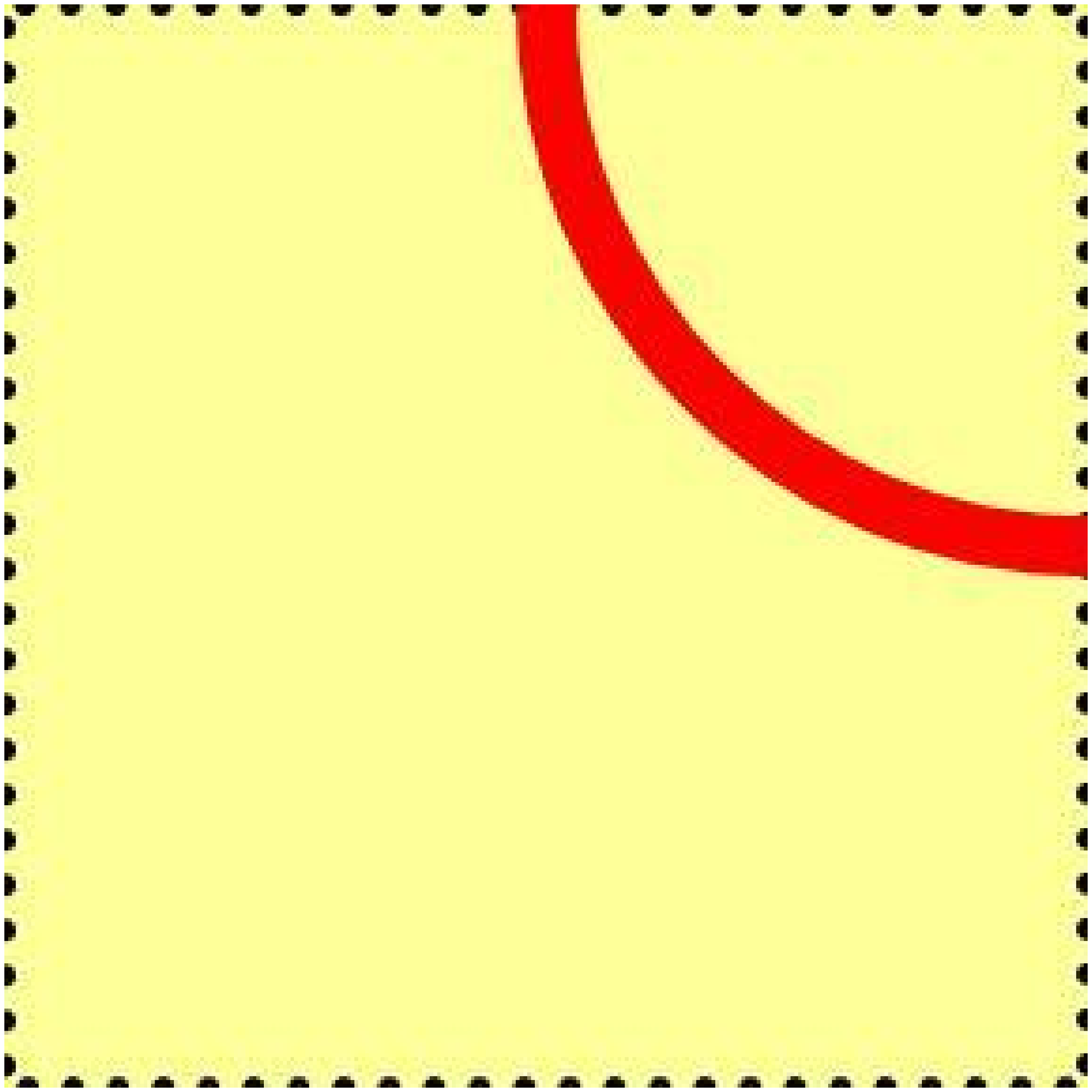}%
}%
%EndExpansion
\
%TCIMACRO{\FRAME{itbpF}{0.3269in}{0.3269in}{0in}{}{}{ut04.ps}%
%width 0.3269in;  height 0.3269in;  depth 0in;  original-width 3in;
%original-height 3in;  cropleft "0";  croptop "1";  cropright "1";
%cropbottom "0";  filename 'ut04.ps';file-properties "XNPEU";}}}%
%BeginExpansion
{\includegraphics[
%natheight=3.000000in,
%natwidth=3.000000in,
height=0.3269in,
width=0.3269in
]%
{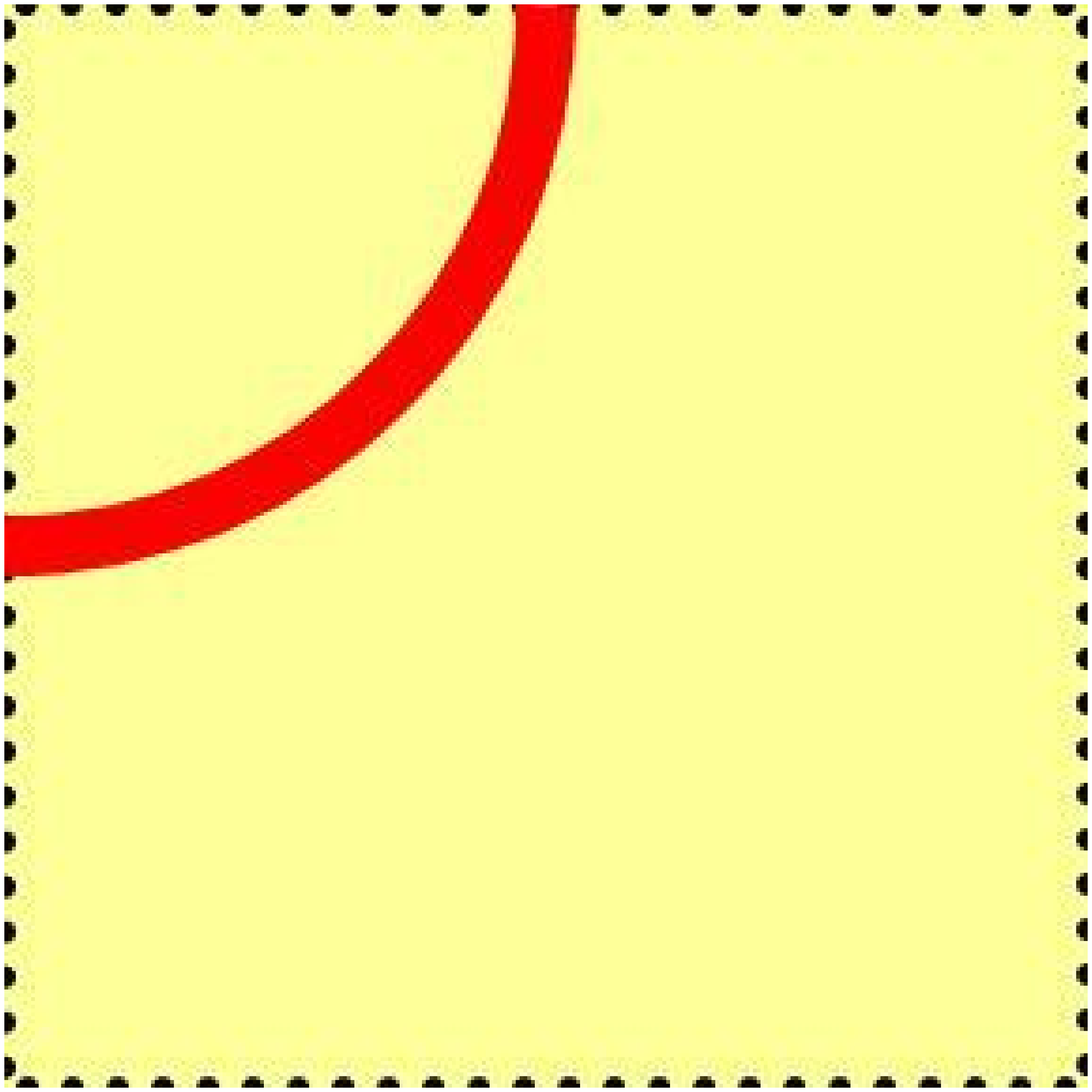}%
}%
%EndExpansion
\qquad\ \ \ \
%TCIMACRO{\FRAME{itbpF}{0.3269in}{0.3269in}{0in}{}{}{ut05.ps}%
%width 0.3269in;  height 0.3269in;  depth 0in;  original-width 3in;
%original-height 3in;  cropleft "0";  croptop "1";  cropright "1";
%cropbottom "0";  filename 'ut05.ps';file-properties "XNPEU";}}}%
%BeginExpansion
{\includegraphics[
%natheight=3.000000in,
%natwidth=3.000000in,
height=0.3269in,
width=0.3269in
]%
{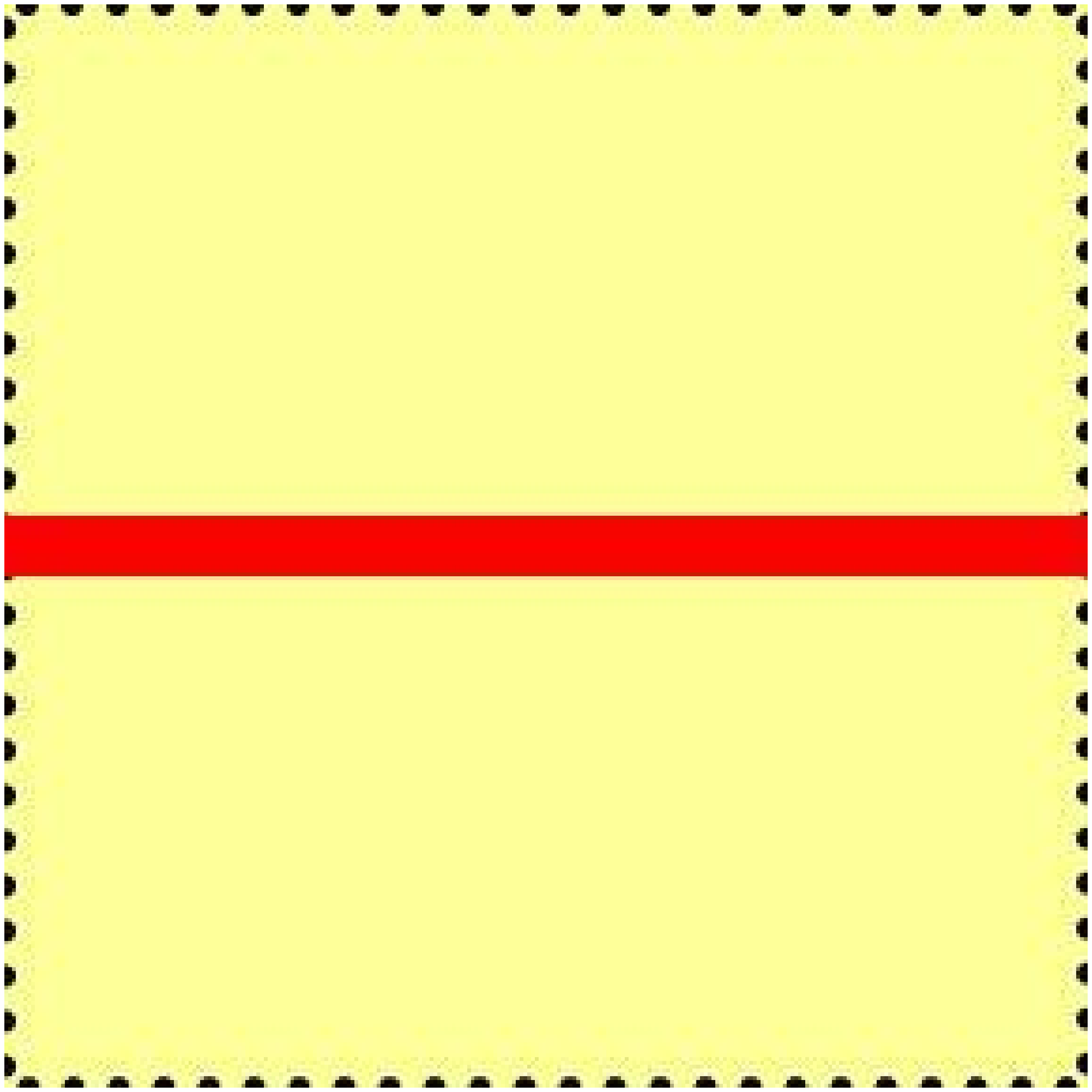}%
}%
%EndExpansion
\
%TCIMACRO{\FRAME{itbpF}{0.3269in}{0.3269in}{0in}{}{}{ut06.ps}%
%width 0.3269in;  height 0.3269in;  depth 0in;  original-width 3in;
%original-height 3in;  cropleft "0";  croptop "1";  cropright "1";
%cropbottom "0";  filename 'ut06.ps';file-properties "XNPEU";}}}%
%BeginExpansion
{\includegraphics[
%natheight=3.000000in,
%natwidth=3.000000in,
height=0.3269in,
width=0.3269in
]%
{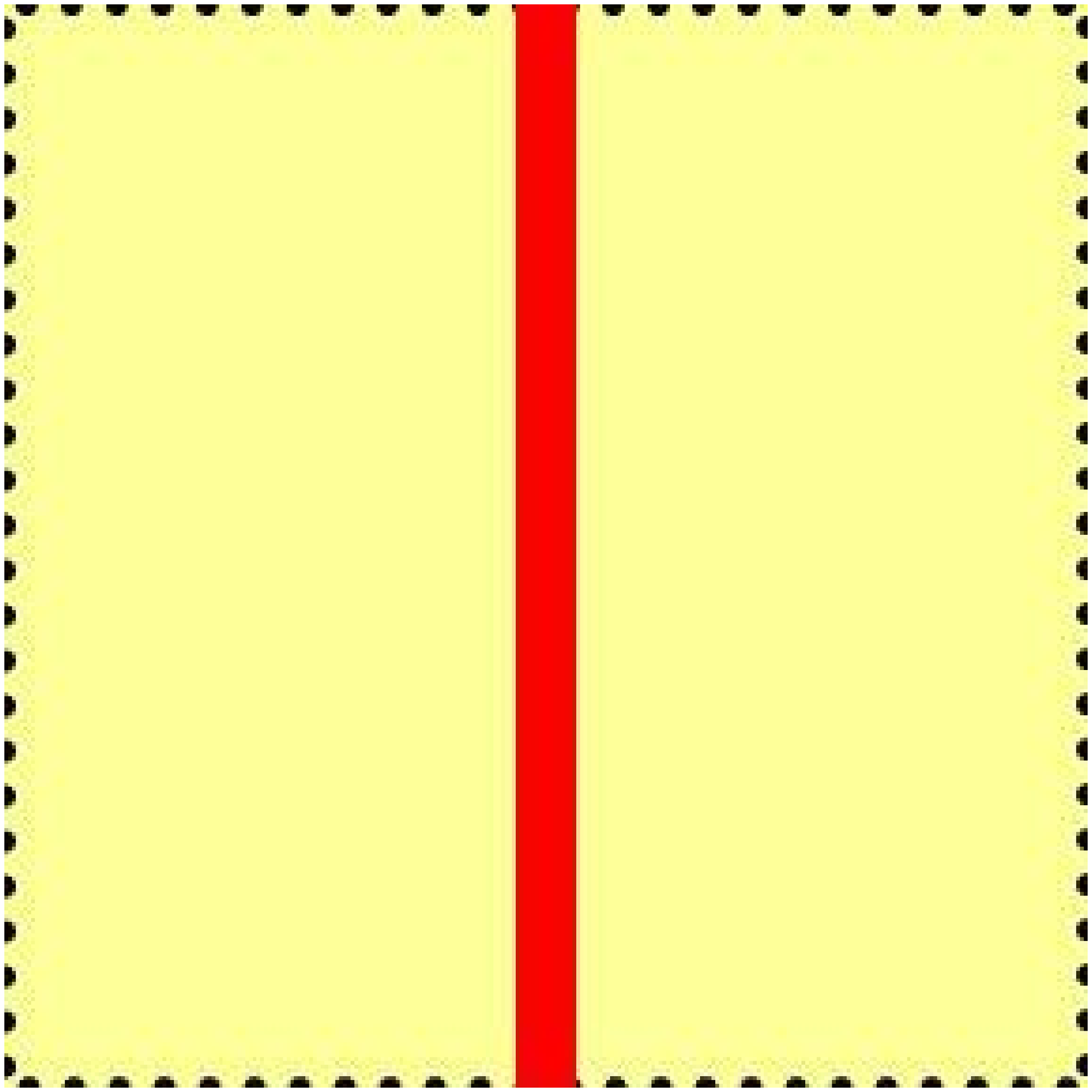}%
}%
%EndExpansion
\qquad\ \ \ \
%TCIMACRO{\FRAME{itbpF}{0.3269in}{0.3269in}{0in}{}{}{ut07.ps}%
%width 0.3269in;  height 0.3269in;  depth 0in;  original-width 3in;
%original-height 3in;  cropleft "0";  croptop "1";  cropright "1";
%cropbottom "0";  filename 'ut07.ps';file-properties "XNPEU";}}}%
%BeginExpansion
{\includegraphics[
%natheight=3.000000in,
%natwidth=3.000000in,
height=0.3269in,
width=0.3269in
]%
{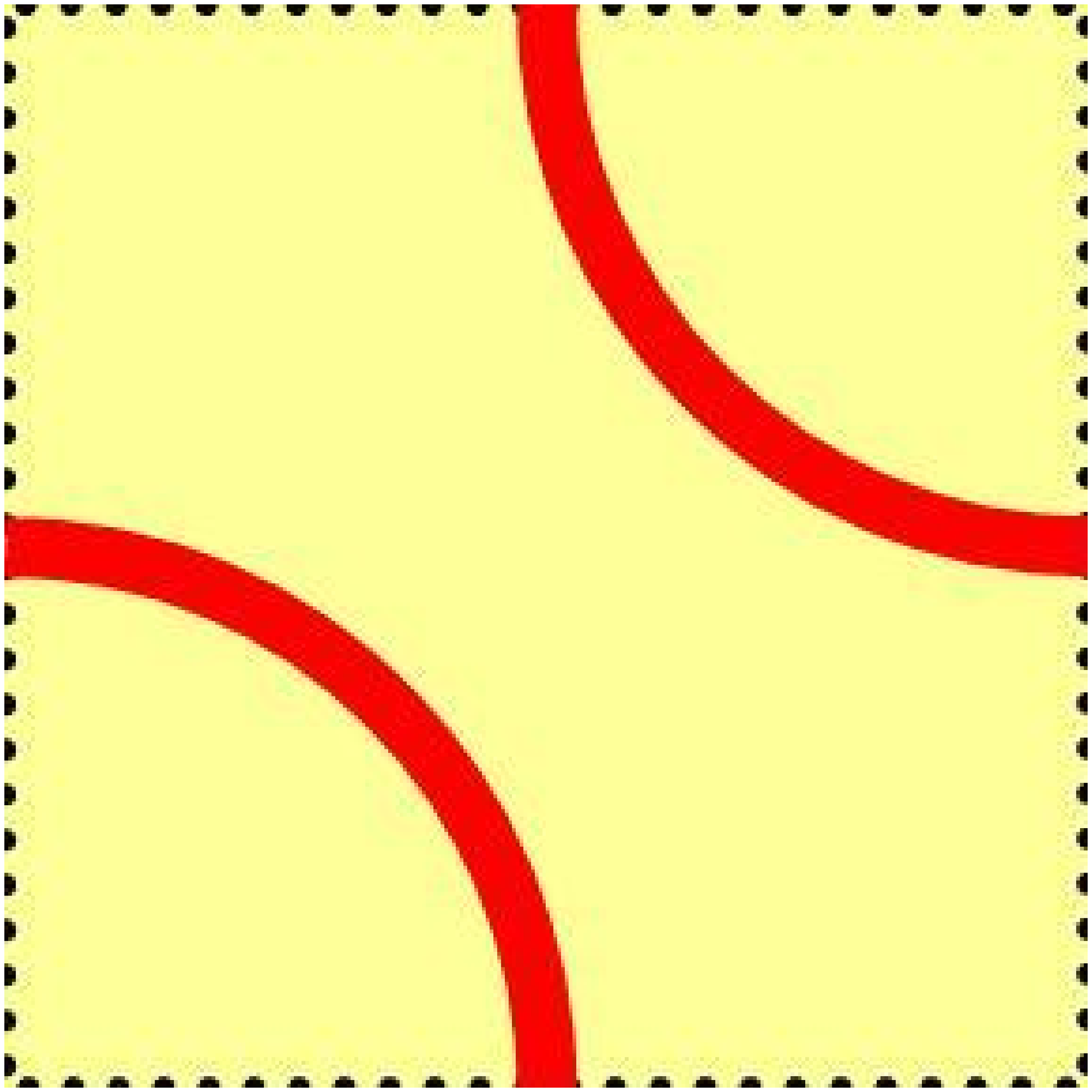}%
}%
%EndExpansion
\
%TCIMACRO{\FRAME{itbpF}{0.3269in}{0.3269in}{0in}{}{}{ut08.ps}%
%width 0.3269in;  height 0.3269in;  depth 0in;  original-width 3in;
%original-height 3in;  cropleft "0";  croptop "1";  cropright "1";
%cropbottom "0";  filename 'ut08.ps';file-properties "XNPEU";}}}%
%BeginExpansion
{\includegraphics[
%natheight=3.000000in,
%natwidth=3.000000in,
height=0.3269in,
width=0.3269in
]%
{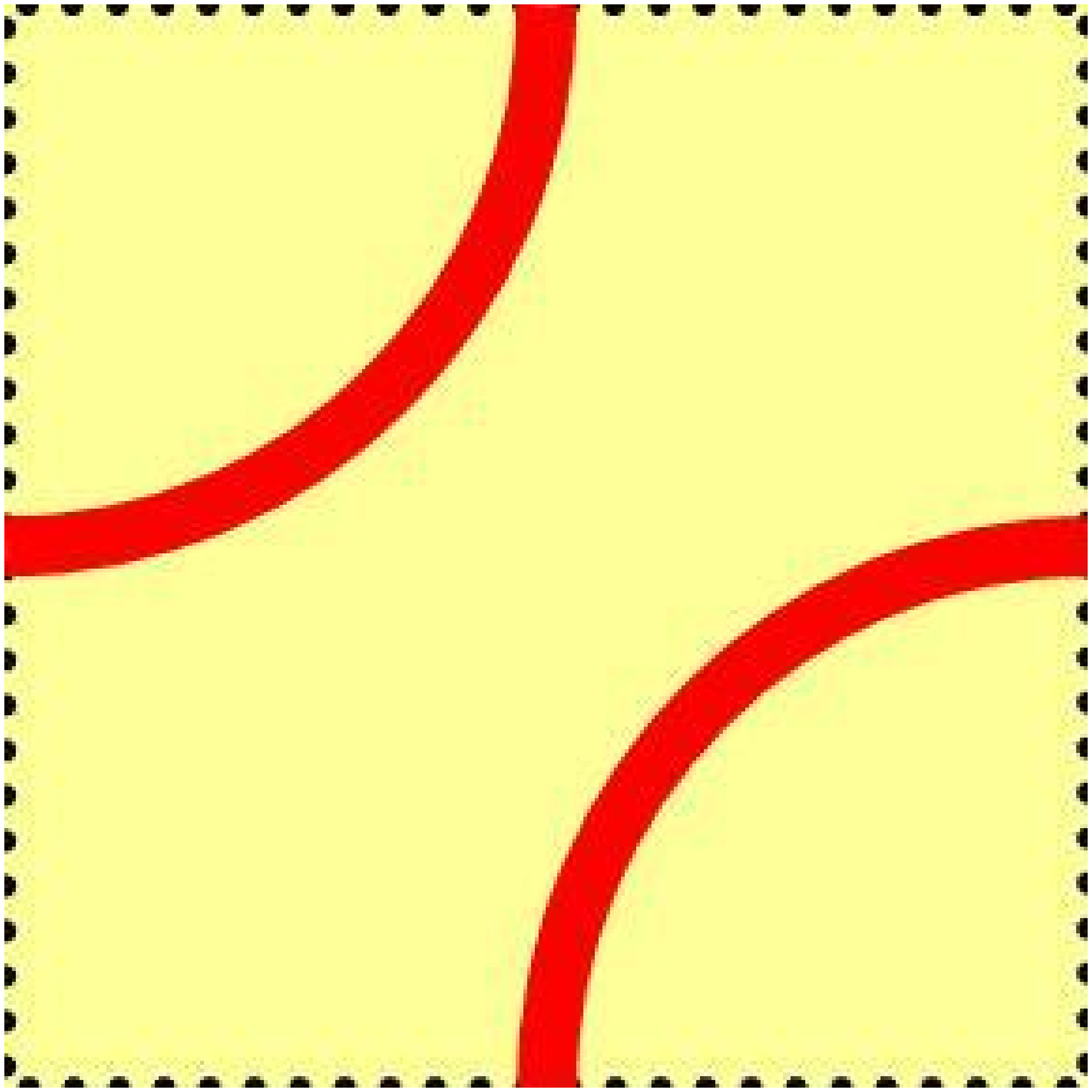}%
}%
%EndExpansion
\qquad\ \ \ \
%TCIMACRO{\FRAME{itbpF}{0.3269in}{0.3269in}{0in}{}{}{ut09.ps}%
%width 0.3269in;  height 0.3269in;  depth 0in;  original-width 3in;
%original-height 3in;  cropleft "0";  croptop "1";  cropright "1";
%cropbottom "0";  filename 'ut09.ps';file-properties "XNPEU";}}}%
%BeginExpansion
{\includegraphics[
%natheight=3.000000in,
%natwidth=3.000000in,
height=0.3269in,
width=0.3269in
]%
{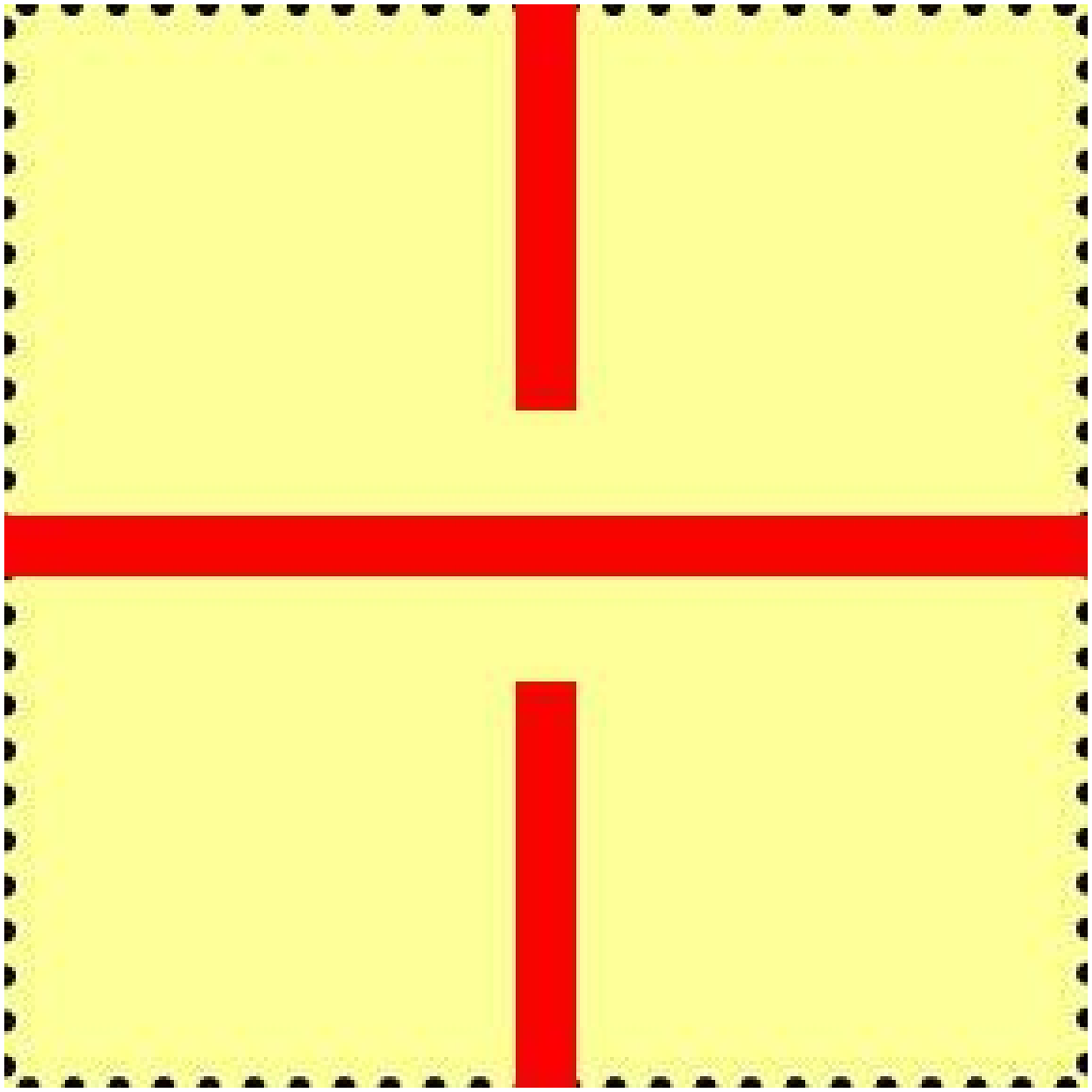}%
}%
%EndExpansion
\
%TCIMACRO{\FRAME{itbpF}{0.3269in}{0.3269in}{0in}{}{}{ut10.ps}%
%width 0.3269in;  height 0.3269in;  depth 0in;  original-width 3in;
%original-height 3in;  cropleft "0";  croptop "1";  cropright "1";
%cropbottom "0";  filename 'ut10.ps';file-properties "XNPEU";}}}%
%BeginExpansion
{\includegraphics[
%natheight=3.000000in,
%natwidth=3.000000in,
height=0.3269in,
width=0.3269in
]%
{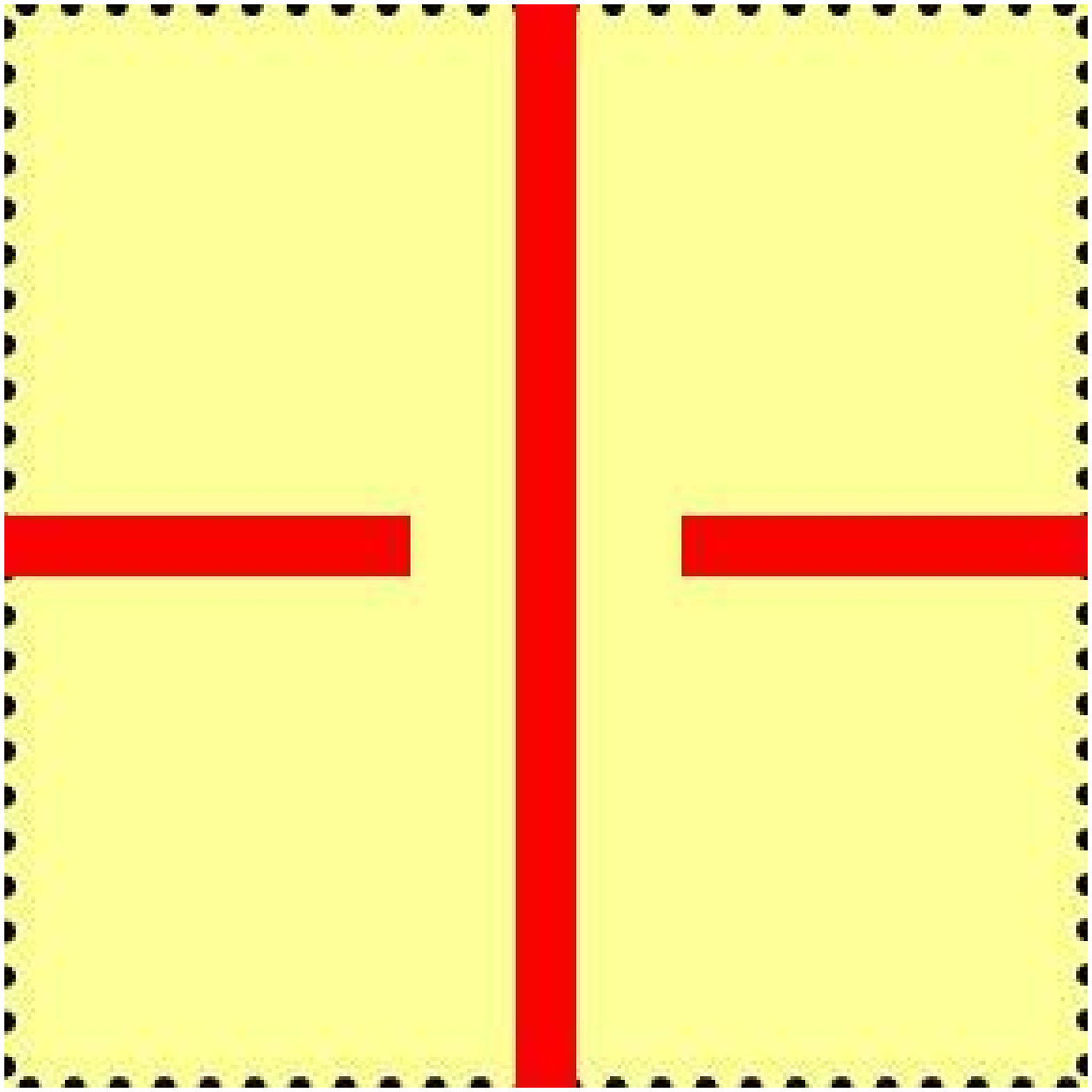}%
}%
%EndExpansion
\]

\bigskip

\noindent called (\textbf{unoriented}) \textbf{tiles}. \ We often will also
denote these tiles respectively by the following symbols%
\[%
\begin{array}
[c]{ccccccccccccccc}%
T_{0}^{(u)} & \quad & T_{1}^{(u)} & T_{2}^{(u)} & T_{3}^{(u)} & T_{4}^{(u)} &
\quad & T_{5}^{(u)} & T_{6}^{(u)} & \quad & T_{7}^{(u)} & T_{8}^{(u)} & \quad
& T_{9}^{(u)} & T_{10}^{(u)}%
\end{array}
\text{ .}%
\]
Moreover, we will frequently omit the superscript `$(u)$' (standing for
`unoriented') when it can be understood from context.

\bigskip

\begin{remark}
Please note that up to rotation there are exactly 5 distinct unoriented tiles.
\ The above unoriented tiles are grouped according to rotational equivalence.
\end{remark}

\bigskip

\begin{definition}
Let $n$ be a positive integer. We define an \textbf{(unoriented)} $\mathbf{n}%
$\textbf{-mosaic} as an $n\times n$ matrix $M=\left(  M_{ij}\right)  =\left(
T_{k\left(  i,j\right)  }\right)  $ of (unoriented) tiles \textbf{with rows
and columns indexed from} $0$ \textbf{to} $n-1$. \ We denote the \textbf{set
of }$\mathbf{n}$\textbf{-mosaics} by $\mathbb{M}^{(n)}$.\ 
\end{definition}

\bigskip

Two examples of unoriented 4-mosaics are shown below:%
\[%
\begin{array}
[c]{cccc}%
%TCIMACRO{\FRAME{itbpF}{0.3269in}{0.3269in}{0in}{}{}{ut00.ps}%
%width 0.3269in;  height 0.3269in;  depth 0in;  original-width 3in;
%original-height 3in;  cropleft "0";  croptop "1";  cropright "1";
%cropbottom "0";  filename 'ut00.ps';file-properties "XNPEU";}}}%
%BeginExpansion
{\includegraphics[
%natheight=3.000000in,
%natwidth=3.000000in,
height=0.3269in,
width=0.3269in
]%
{ut00.ps}%
}%
%EndExpansion
&
%TCIMACRO{\FRAME{itbpF}{0.3269in}{0.3269in}{0in}{}{}{ut05.ps}%
%width 0.3269in;  height 0.3269in;  depth 0in;  original-width 3in;
%original-height 3in;  cropleft "0";  croptop "1";  cropright "1";
%cropbottom "0";  filename 'ut05.ps';file-properties "XNPEU";}}}%
%BeginExpansion
{\includegraphics[
%natheight=3.000000in,
%natwidth=3.000000in,
height=0.3269in,
width=0.3269in
]%
{ut05.ps}%
}%
%EndExpansion
&
%TCIMACRO{\FRAME{itbpF}{0.3269in}{0.3269in}{0in}{}{}{ut04.ps}%
%width 0.3269in;  height 0.3269in;  depth 0in;  original-width 3in;
%original-height 3in;  cropleft "0";  croptop "1";  cropright "1";
%cropbottom "0";  filename 'ut04.ps';file-properties "XNPEU";}}}%
%BeginExpansion
{\includegraphics[
%natheight=3.000000in,
%natwidth=3.000000in,
height=0.3269in,
width=0.3269in
]%
{ut04.ps}%
}%
%EndExpansion
&
%TCIMACRO{\FRAME{itbpF}{0.3269in}{0.3269in}{0in}{}{}{ut00.ps}%
%width 0.3269in;  height 0.3269in;  depth 0in;  original-width 3in;
%original-height 3in;  cropleft "0";  croptop "1";  cropright "1";
%cropbottom "0";  filename 'ut00.ps';file-properties "XNPEU";}}}%
%BeginExpansion
{\includegraphics[
%natheight=3.000000in,
%natwidth=3.000000in,
height=0.3269in,
width=0.3269in
]%
{ut00.ps}%
}%
%EndExpansion
\\%
%TCIMACRO{\FRAME{itbpF}{0.3269in}{0.3269in}{0in}{}{}{ut04.ps}%
%width 0.3269in;  height 0.3269in;  depth 0in;  original-width 3in;
%original-height 3in;  cropleft "0";  croptop "1";  cropright "1";
%cropbottom "0";  filename 'ut04.ps';file-properties "XNPEU";}}}%
%BeginExpansion
{\includegraphics[
%natheight=3.000000in,
%natwidth=3.000000in,
height=0.3269in,
width=0.3269in
]%
{ut04.ps}%
}%
%EndExpansion
&
%TCIMACRO{\FRAME{itbpF}{0.3269in}{0.3269in}{0in}{}{}{ut09.ps}%
%width 0.3269in;  height 0.3269in;  depth 0in;  original-width 3in;
%original-height 3in;  cropleft "0";  croptop "1";  cropright "1";
%cropbottom "0";  filename 'ut09.ps';file-properties "XNPEU";}}}%
%BeginExpansion
{\includegraphics[
%natheight=3.000000in,
%natwidth=3.000000in,
height=0.3269in,
width=0.3269in
]%
{ut09.ps}%
}%
%EndExpansion
&
%TCIMACRO{\FRAME{itbpF}{0.3269in}{0.3269in}{0in}{}{}{ut04.ps}%
%width 0.3269in;  height 0.3269in;  depth 0in;  original-width 3in;
%original-height 3in;  cropleft "0";  croptop "1";  cropright "1";
%cropbottom "0";  filename 'ut04.ps';file-properties "XNPEU";}}}%
%BeginExpansion
{\includegraphics[
%natheight=3.000000in,
%natwidth=3.000000in,
height=0.3269in,
width=0.3269in
]%
{ut04.ps}%
}%
%EndExpansion
&
%TCIMACRO{\FRAME{itbpF}{0.3269in}{0.3269in}{0in}{}{}{ut01.ps}%
%width 0.3269in;  height 0.3269in;  depth 0in;  original-width 3in;
%original-height 3in;  cropleft "0";  croptop "1";  cropright "1";
%cropbottom "0";  filename 'ut01.ps';file-properties "XNPEU";}}}%
%BeginExpansion
{\includegraphics[
%natheight=3.000000in,
%natwidth=3.000000in,
height=0.3269in,
width=0.3269in
]%
{ut01.ps}%
}%
%EndExpansion
\\%
%TCIMACRO{\FRAME{itbpF}{0.3269in}{0.3269in}{0in}{}{}{ut09.ps}%
%width 0.3269in;  height 0.3269in;  depth 0in;  original-width 3in;
%original-height 3in;  cropleft "0";  croptop "1";  cropright "1";
%cropbottom "0";  filename 'ut09.ps';file-properties "XNPEU";}}}%
%BeginExpansion
{\includegraphics[
%natheight=3.000000in,
%natwidth=3.000000in,
height=0.3269in,
width=0.3269in
]%
{ut09.ps}%
}%
%EndExpansion
&
%TCIMACRO{\FRAME{itbpF}{0.3269in}{0.3269in}{0in}{}{}{ut03.ps}%
%width 0.3269in;  height 0.3269in;  depth 0in;  original-width 3in;
%original-height 3in;  cropleft "0";  croptop "1";  cropright "1";
%cropbottom "0";  filename 'ut03.ps';file-properties "XNPEU";}}}%
%BeginExpansion
{\includegraphics[
%natheight=3.000000in,
%natwidth=3.000000in,
height=0.3269in,
width=0.3269in
]%
{ut03.ps}%
}%
%EndExpansion
&
%TCIMACRO{\FRAME{itbpF}{0.3269in}{0.3269in}{0in}{}{}{ut08.ps}%
%width 0.3269in;  height 0.3269in;  depth 0in;  original-width 3in;
%original-height 3in;  cropleft "0";  croptop "1";  cropright "1";
%cropbottom "0";  filename 'ut08.ps';file-properties "XNPEU";}}}%
%BeginExpansion
{\includegraphics[
%natheight=3.000000in,
%natwidth=3.000000in,
height=0.3269in,
width=0.3269in
]%
{ut08.ps}%
}%
%EndExpansion
&
%TCIMACRO{\FRAME{itbpF}{0.3269in}{0.3269in}{0in}{}{}{ut02.ps}%
%width 0.3269in;  height 0.3269in;  depth 0in;  original-width 3in;
%original-height 3in;  cropleft "0";  croptop "1";  cropright "1";
%cropbottom "0";  filename 'ut02.ps';file-properties "XNPEU";}}}%
%BeginExpansion
{\includegraphics[
%natheight=3.000000in,
%natwidth=3.000000in,
height=0.3269in,
width=0.3269in
]%
{ut02.ps}%
}%
%EndExpansion
\\%
%TCIMACRO{\FRAME{itbpF}{0.3269in}{0.3269in}{0in}{}{}{ut03.ps}%
%width 0.3269in;  height 0.3269in;  depth 0in;  original-width 3in;
%original-height 3in;  cropleft "0";  croptop "1";  cropright "1";
%cropbottom "0";  filename 'ut03.ps';file-properties "XNPEU";}}}%
%BeginExpansion
{\includegraphics[
%natheight=3.000000in,
%natwidth=3.000000in,
height=0.3269in,
width=0.3269in
]%
{ut03.ps}%
}%
%EndExpansion
&
%TCIMACRO{\FRAME{itbpF}{0.3269in}{0.3269in}{0in}{}{}{ut06.ps}%
%width 0.3269in;  height 0.3269in;  depth 0in;  original-width 3in;
%original-height 3in;  cropleft "0";  croptop "1";  cropright "1";
%cropbottom "0";  filename 'ut06.ps';file-properties "XNPEU";}}}%
%BeginExpansion
{\includegraphics[
%natheight=3.000000in,
%natwidth=3.000000in,
height=0.3269in,
width=0.3269in
]%
{ut06.ps}%
}%
%EndExpansion
&
%TCIMACRO{\FRAME{itbpF}{0.3269in}{0.3269in}{0in}{}{}{ut04.ps}%
%width 0.3269in;  height 0.3269in;  depth 0in;  original-width 3in;
%original-height 3in;  cropleft "0";  croptop "1";  cropright "1";
%cropbottom "0";  filename 'ut04.ps';file-properties "XNPEU";}}}%
%BeginExpansion
{\includegraphics[
%natheight=3.000000in,
%natwidth=3.000000in,
height=0.3269in,
width=0.3269in
]%
{ut04.ps}%
}%
%EndExpansion
&
%TCIMACRO{\FRAME{itbpF}{0.3269in}{0.3269in}{0in}{}{}{ut00.ps}%
%width 0.3269in;  height 0.3269in;  depth 0in;  original-width 3in;
%original-height 3in;  cropleft "0";  croptop "1";  cropright "1";
%cropbottom "0";  filename 'ut00.ps';file-properties "XNPEU";}}}%
%BeginExpansion
{\includegraphics[
%natheight=3.000000in,
%natwidth=3.000000in,
height=0.3269in,
width=0.3269in
]%
{ut00.ps}%
}%
%EndExpansion
\end{array}
\qquad\qquad\qquad%
\begin{array}
[c]{cccc}%
%TCIMACRO{\FRAME{itbpF}{0.3269in}{0.3269in}{0in}{}{}{ut00.ps}%
%width 0.3269in;  height 0.3269in;  depth 0in;  original-width 3in;
%original-height 3in;  cropleft "0";  croptop "1";  cropright "1";
%cropbottom "0";  filename 'ut00.ps';file-properties "XNPEU";}}}%
%BeginExpansion
{\includegraphics[
%natheight=3.000000in,
%natwidth=3.000000in,
height=0.3269in,
width=0.3269in
]%
{ut00.ps}%
}%
%EndExpansion
&
%TCIMACRO{\FRAME{itbpF}{0.3269in}{0.3269in}{0in}{}{}{ut02.ps}%
%width 0.3269in;  height 0.3269in;  depth 0in;  original-width 3in;
%original-height 3in;  cropleft "0";  croptop "1";  cropright "1";
%cropbottom "0";  filename 'ut02.ps';file-properties "XNPEU";}}}%
%BeginExpansion
{\includegraphics[
%natheight=3.000000in,
%natwidth=3.000000in,
height=0.3269in,
width=0.3269in
]%
{ut02.ps}%
}%
%EndExpansion
&
%TCIMACRO{\FRAME{itbpF}{0.3269in}{0.3269in}{0in}{}{}{ut01.ps}%
%width 0.3269in;  height 0.3269in;  depth 0in;  original-width 3in;
%original-height 3in;  cropleft "0";  croptop "1";  cropright "1";
%cropbottom "0";  filename 'ut01.ps';file-properties "XNPEU";}}}%
%BeginExpansion
{\includegraphics[
%natheight=3.000000in,
%natwidth=3.000000in,
height=0.3269in,
width=0.3269in
]%
{ut01.ps}%
}%
%EndExpansion
&
%TCIMACRO{\FRAME{itbpF}{0.3269in}{0.3269in}{0in}{}{}{ut00.ps}%
%width 0.3269in;  height 0.3269in;  depth 0in;  original-width 3in;
%original-height 3in;  cropleft "0";  croptop "1";  cropright "1";
%cropbottom "0";  filename 'ut00.ps';file-properties "XNPEU";}}}%
%BeginExpansion
{\includegraphics[
%natheight=3.000000in,
%natwidth=3.000000in,
height=0.3269in,
width=0.3269in
]%
{ut00.ps}%
}%
%EndExpansion
\\%
%TCIMACRO{\FRAME{itbpF}{0.3269in}{0.3269in}{0in}{}{}{ut02.ps}%
%width 0.3269in;  height 0.3269in;  depth 0in;  original-width 3in;
%original-height 3in;  cropleft "0";  croptop "1";  cropright "1";
%cropbottom "0";  filename 'ut02.ps';file-properties "XNPEU";}}}%
%BeginExpansion
{\includegraphics[
%natheight=3.000000in,
%natwidth=3.000000in,
height=0.3269in,
width=0.3269in
]%
{ut02.ps}%
}%
%EndExpansion
&
%TCIMACRO{\FRAME{itbpF}{0.3269in}{0.3269in}{0in}{}{}{ut09.ps}%
%width 0.3269in;  height 0.3269in;  depth 0in;  original-width 3in;
%original-height 3in;  cropleft "0";  croptop "1";  cropright "1";
%cropbottom "0";  filename 'ut09.ps';file-properties "XNPEU";}}}%
%BeginExpansion
{\includegraphics[
%natheight=3.000000in,
%natwidth=3.000000in,
height=0.3269in,
width=0.3269in
]%
{ut09.ps}%
}%
%EndExpansion
&
%TCIMACRO{\FRAME{itbpF}{0.3269in}{0.3269in}{0in}{}{}{ut10.ps}%
%width 0.3269in;  height 0.3269in;  depth 0in;  original-width 3in;
%original-height 3in;  cropleft "0";  croptop "1";  cropright "1";
%cropbottom "0";  filename 'ut10.ps';file-properties "XNPEU";}}}%
%BeginExpansion
{\includegraphics[
%natheight=3.000000in,
%natwidth=3.000000in,
height=0.3269in,
width=0.3269in
]%
{ut10.ps}%
}%
%EndExpansion
&
%TCIMACRO{\FRAME{itbpF}{0.3269in}{0.3269in}{0in}{}{}{ut01.ps}%
%width 0.3269in;  height 0.3269in;  depth 0in;  original-width 3in;
%original-height 3in;  cropleft "0";  croptop "1";  cropright "1";
%cropbottom "0";  filename 'ut01.ps';file-properties "XNPEU";}}}%
%BeginExpansion
{\includegraphics[
%natheight=3.000000in,
%natwidth=3.000000in,
height=0.3269in,
width=0.3269in
]%
{ut01.ps}%
}%
%EndExpansion
\\%
%TCIMACRO{\FRAME{itbpF}{0.3269in}{0.3269in}{0in}{}{}{ut06.ps}%
%width 0.3269in;  height 0.3269in;  depth 0in;  original-width 3in;
%original-height 3in;  cropleft "0";  croptop "1";  cropright "1";
%cropbottom "0";  filename 'ut06.ps';file-properties "XNPEU";}}}%
%BeginExpansion
{\includegraphics[
%natheight=3.000000in,
%natwidth=3.000000in,
height=0.3269in,
width=0.3269in
]%
{ut06.ps}%
}%
%EndExpansion
&
%TCIMACRO{\FRAME{itbpF}{0.3269in}{0.3269in}{0in}{}{}{ut03.ps}%
%width 0.3269in;  height 0.3269in;  depth 0in;  original-width 3in;
%original-height 3in;  cropleft "0";  croptop "1";  cropright "1";
%cropbottom "0";  filename 'ut03.ps';file-properties "XNPEU";}}}%
%BeginExpansion
{\includegraphics[
%natheight=3.000000in,
%natwidth=3.000000in,
height=0.3269in,
width=0.3269in
]%
{ut03.ps}%
}%
%EndExpansion
&
%TCIMACRO{\FRAME{itbpF}{0.3269in}{0.3269in}{0in}{}{}{ut09.ps}%
%width 0.3269in;  height 0.3269in;  depth 0in;  original-width 3in;
%original-height 3in;  cropleft "0";  croptop "1";  cropright "1";
%cropbottom "0";  filename 'ut09.ps';file-properties "XNPEU";}}}%
%BeginExpansion
{\includegraphics[
%natheight=3.000000in,
%natwidth=3.000000in,
height=0.3269in,
width=0.3269in
]%
{ut09.ps}%
}%
%EndExpansion
&
%TCIMACRO{\FRAME{itbpF}{0.3269in}{0.3269in}{0in}{}{}{ut04.ps}%
%width 0.3269in;  height 0.3269in;  depth 0in;  original-width 3in;
%original-height 3in;  cropleft "0";  croptop "1";  cropright "1";
%cropbottom "0";  filename 'ut04.ps';file-properties "XNPEU";}}}%
%BeginExpansion
{\includegraphics[
%natheight=3.000000in,
%natwidth=3.000000in,
height=0.3269in,
width=0.3269in
]%
{ut04.ps}%
}%
%EndExpansion
\\%
%TCIMACRO{\FRAME{itbpF}{0.3269in}{0.3269in}{0in}{}{}{ut03.ps}%
%width 0.3269in;  height 0.3269in;  depth 0in;  original-width 3in;
%original-height 3in;  cropleft "0";  croptop "1";  cropright "1";
%cropbottom "0";  filename 'ut03.ps';file-properties "XNPEU";}}}%
%BeginExpansion
{\includegraphics[
%natheight=3.000000in,
%natwidth=3.000000in,
height=0.3269in,
width=0.3269in
]%
{ut03.ps}%
}%
%EndExpansion
&
%TCIMACRO{\FRAME{itbpF}{0.3269in}{0.3269in}{0in}{}{}{ut05.ps}%
%width 0.3269in;  height 0.3269in;  depth 0in;  original-width 3in;
%original-height 3in;  cropleft "0";  croptop "1";  cropright "1";
%cropbottom "0";  filename 'ut05.ps';file-properties "XNPEU";}}}%
%BeginExpansion
{\includegraphics[
%natheight=3.000000in,
%natwidth=3.000000in,
height=0.3269in,
width=0.3269in
]%
{ut05.ps}%
}%
%EndExpansion
&
%TCIMACRO{\FRAME{itbpF}{0.3269in}{0.3269in}{0in}{}{}{ut04.ps}%
%width 0.3269in;  height 0.3269in;  depth 0in;  original-width 3in;
%original-height 3in;  cropleft "0";  croptop "1";  cropright "1";
%cropbottom "0";  filename 'ut04.ps';file-properties "XNPEU";}}}%
%BeginExpansion
{\includegraphics[
%natheight=3.000000in,
%natwidth=3.000000in,
height=0.3269in,
width=0.3269in
]%
{ut04.ps}%
}%
%EndExpansion
&
%TCIMACRO{\FRAME{itbpF}{0.3269in}{0.3269in}{0in}{}{}{ut00.ps}%
%width 0.3269in;  height 0.3269in;  depth 0in;  original-width 3in;
%original-height 3in;  cropleft "0";  croptop "1";  cropright "1";
%cropbottom "0";  filename 'ut00.ps';file-properties "XNPEU";}}}%
%BeginExpansion
{\includegraphics[
%natheight=3.000000in,
%natwidth=3.000000in,
height=0.3269in,
width=0.3269in
]%
{ut00.ps}%
}%
%EndExpansion
\end{array}
\]

\bigskip

We now proceed to define what is meant by a knot mosaic:

\bigskip

A \textbf{connection point} of a tile is defined as the midpoint of a tile
edge which is also the endpoint of a curve drawn on the tile. \ 

\bigskip

Examples of tile connection points are illustrated in figure 1 below:%
%TCIMACRO{\FRAME{dtbpF}{3.3451in}{1.2834in}{0pt}{}{}{connection-points.ps}%
%width 3.3451in;  height 1.2834in;  depth 0pt;  original-width 6.6357in;
%original-height 2.5105in;  cropleft "0";  croptop "1";  cropright "1";
%cropbottom "0";  filename 'Connection-Points.ps';file-properties "XNPEU";}}}%
%BeginExpansion
\begin{center}
\includegraphics[
%natheight=2.510500in,
%natwidth=6.635700in,
height=1.2834in,
width=3.3451in
]%
{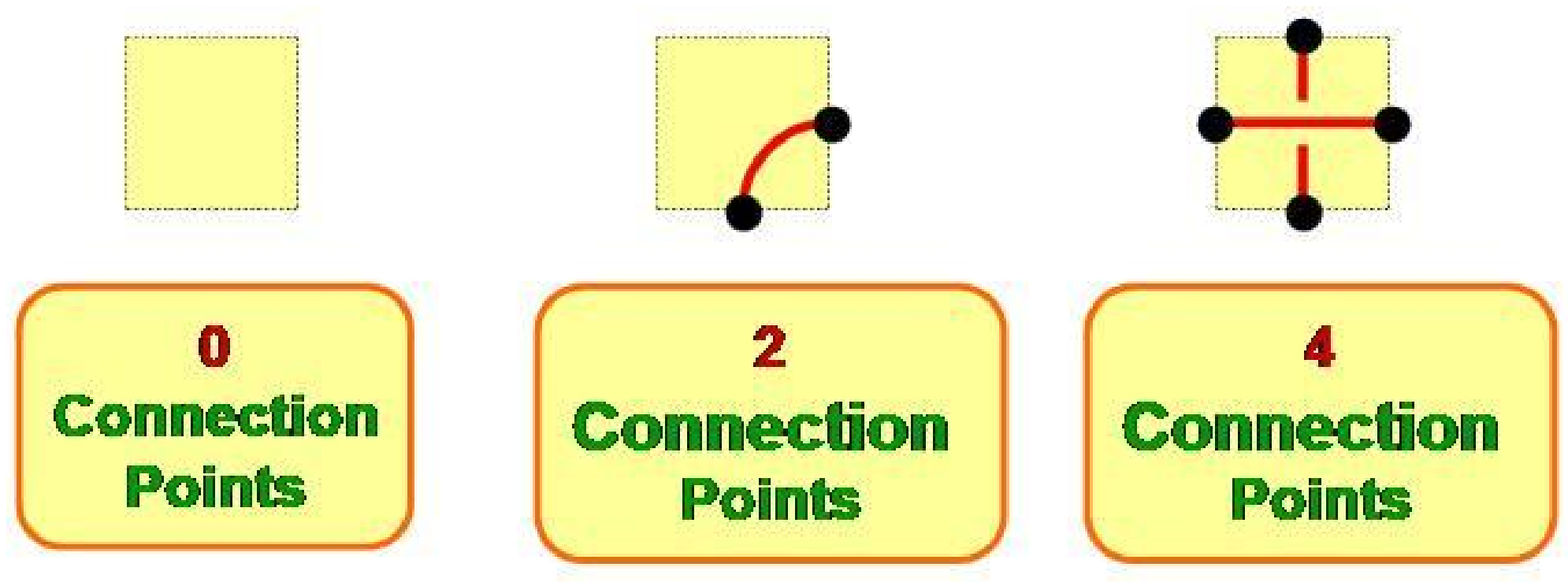}%
\end{center}
%EndExpansion

\bigskip

We say that two tiles in a mosaic are \textbf{contiguous} if they lie
immediately next to each other in either the same row or the same column. \ An
unoriented tile within a mosaic is said to be \textbf{suitably connected} if
each of its connection points touches a connection point of a contiguous tile.

\bigskip

\begin{definition}
An \textbf{(unoriented)} \textbf{knot} $\mathbf{n}$-\textbf{mosaic} is a
mosaic in which all tiles are suitably connected. \ We let $\mathbb{K}^{(n)}$
denote the subset of $\mathbb{M}^{(n)}$ of all knot $\mathbf{n}$%
-mosaics\footnote{We remind the reader of the following statement made at the
beginning of the introduction of this paper: \textit{For simplicity of
exposition, we will throughout this paper frequently use the term "knot" to
mean either a knot or a link.}}.
\end{definition}

\bigskip

The previous two $4$-mosaics shown above are examples respectively of a
non-knot $4$-mosaic and a knot $4$-mosaic. \ Other examples of knot (or links)
mosaics are the Hopf link $4$-mosaic, the figure eight knot $5$-mosaic, and
the Borromean rings $6$-mosaic, respectively illustrated below:

\bigskip%

\[%
\begin{array}
[c]{cccc}%
%TCIMACRO{\FRAME{itbpF}{0.3269in}{0.3269in}{0in}{}{}{ut02.ps}%
%width 0.3269in;  height 0.3269in;  depth 0in;  original-width 3in;
%original-height 3in;  cropleft "0";  croptop "1";  cropright "1";
%cropbottom "0";  filename 'ut02.ps';file-properties "XNPEU";}}}%
%BeginExpansion
{\includegraphics[
%natheight=3.000000in,
%natwidth=3.000000in,
height=0.3269in,
width=0.3269in
]%
{ut02.ps}%
}%
%EndExpansion
&
%TCIMACRO{\FRAME{itbpF}{0.3269in}{0.3269in}{0in}{}{}{ut05.ps}%
%width 0.3269in;  height 0.3269in;  depth 0in;  original-width 3in;
%original-height 3in;  cropleft "0";  croptop "1";  cropright "1";
%cropbottom "0";  filename 'ut05.ps';file-properties "XNPEU";}}}%
%BeginExpansion
{\includegraphics[
%natheight=3.000000in,
%natwidth=3.000000in,
height=0.3269in,
width=0.3269in
]%
{ut05.ps}%
}%
%EndExpansion
&
%TCIMACRO{\FRAME{itbpF}{0.3269in}{0.3269in}{0in}{}{}{ut01.ps}%
%width 0.3269in;  height 0.3269in;  depth 0in;  original-width 3in;
%original-height 3in;  cropleft "0";  croptop "1";  cropright "1";
%cropbottom "0";  filename 'ut01.ps';file-properties "XNPEU";}}}%
%BeginExpansion
{\includegraphics[
%natheight=3.000000in,
%natwidth=3.000000in,
height=0.3269in,
width=0.3269in
]%
{ut01.ps}%
}%
%EndExpansion
&
%TCIMACRO{\FRAME{itbpF}{0.3269in}{0.3269in}{0in}{}{}{ut00.ps}%
%width 0.3269in;  height 0.3269in;  depth 0in;  original-width 3in;
%original-height 3in;  cropleft "0";  croptop "1";  cropright "1";
%cropbottom "0";  filename 'ut00.ps';file-properties "XNPEU";}}}%
%BeginExpansion
{\includegraphics[
%natheight=3.000000in,
%natwidth=3.000000in,
height=0.3269in,
width=0.3269in
]%
{ut00.ps}%
}%
%EndExpansion
\\%
%TCIMACRO{\FRAME{itbpF}{0.3269in}{0.3269in}{0in}{}{}{ut06.ps}%
%width 0.3269in;  height 0.3269in;  depth 0in;  original-width 3in;
%original-height 3in;  cropleft "0";  croptop "1";  cropright "1";
%cropbottom "0";  filename 'ut06.ps';file-properties "XNPEU";}}}%
%BeginExpansion
{\includegraphics[
%natheight=3.000000in,
%natwidth=3.000000in,
height=0.3269in,
width=0.3269in
]%
{ut06.ps}%
}%
%EndExpansion
&
%TCIMACRO{\FRAME{itbpF}{0.3269in}{0.3269in}{0in}{}{}{ut02.ps}%
%width 0.3269in;  height 0.3269in;  depth 0in;  original-width 3in;
%original-height 3in;  cropleft "0";  croptop "1";  cropright "1";
%cropbottom "0";  filename 'ut02.ps';file-properties "XNPEU";}}}%
%BeginExpansion
{\includegraphics[
%natheight=3.000000in,
%natwidth=3.000000in,
height=0.3269in,
width=0.3269in
]%
{ut02.ps}%
}%
%EndExpansion
&
%TCIMACRO{\FRAME{itbpF}{0.3269in}{0.3269in}{0in}{}{}{ut09.ps}%
%width 0.3269in;  height 0.3269in;  depth 0in;  original-width 3in;
%original-height 3in;  cropleft "0";  croptop "1";  cropright "1";
%cropbottom "0";  filename 'ut09.ps';file-properties "XNPEU";}}}%
%BeginExpansion
{\includegraphics[
%natheight=3.000000in,
%natwidth=3.000000in,
height=0.3269in,
width=0.3269in
]%
{ut09.ps}%
}%
%EndExpansion
&
%TCIMACRO{\FRAME{itbpF}{0.3269in}{0.3269in}{0in}{}{}{ut01.ps}%
%width 0.3269in;  height 0.3269in;  depth 0in;  original-width 3in;
%original-height 3in;  cropleft "0";  croptop "1";  cropright "1";
%cropbottom "0";  filename 'ut01.ps';file-properties "XNPEU";}}}%
%BeginExpansion
{\includegraphics[
%natheight=3.000000in,
%natwidth=3.000000in,
height=0.3269in,
width=0.3269in
]%
{ut01.ps}%
}%
%EndExpansion
\\%
%TCIMACRO{\FRAME{itbpF}{0.3269in}{0.3269in}{0in}{}{}{ut03.ps}%
%width 0.3269in;  height 0.3269in;  depth 0in;  original-width 3in;
%original-height 3in;  cropleft "0";  croptop "1";  cropright "1";
%cropbottom "0";  filename 'ut03.ps';file-properties "XNPEU";}}}%
%BeginExpansion
{\includegraphics[
%natheight=3.000000in,
%natwidth=3.000000in,
height=0.3269in,
width=0.3269in
]%
{ut03.ps}%
}%
%EndExpansion
&
%TCIMACRO{\FRAME{itbpF}{0.3269in}{0.3269in}{0in}{}{}{ut09.ps}%
%width 0.3269in;  height 0.3269in;  depth 0in;  original-width 3in;
%original-height 3in;  cropleft "0";  croptop "1";  cropright "1";
%cropbottom "0";  filename 'ut09.ps';file-properties "XNPEU";}}}%
%BeginExpansion
{\includegraphics[
%natheight=3.000000in,
%natwidth=3.000000in,
height=0.3269in,
width=0.3269in
]%
{ut09.ps}%
}%
%EndExpansion
&
%TCIMACRO{\FRAME{itbpF}{0.3269in}{0.3269in}{0in}{}{}{ut04.ps}%
%width 0.3269in;  height 0.3269in;  depth 0in;  original-width 3in;
%original-height 3in;  cropleft "0";  croptop "1";  cropright "1";
%cropbottom "0";  filename 'ut04.ps';file-properties "XNPEU";}}}%
%BeginExpansion
{\includegraphics[
%natheight=3.000000in,
%natwidth=3.000000in,
height=0.3269in,
width=0.3269in
]%
{ut04.ps}%
}%
%EndExpansion
&
%TCIMACRO{\FRAME{itbpF}{0.3269in}{0.3269in}{0in}{}{}{ut06.ps}%
%width 0.3269in;  height 0.3269in;  depth 0in;  original-width 3in;
%original-height 3in;  cropleft "0";  croptop "1";  cropright "1";
%cropbottom "0";  filename 'ut06.ps';file-properties "XNPEU";}}}%
%BeginExpansion
{\includegraphics[
%natheight=3.000000in,
%natwidth=3.000000in,
height=0.3269in,
width=0.3269in
]%
{ut06.ps}%
}%
%EndExpansion
\\%
%TCIMACRO{\FRAME{itbpF}{0.3269in}{0.3269in}{0in}{}{}{ut00.ps}%
%width 0.3269in;  height 0.3269in;  depth 0in;  original-width 3in;
%original-height 3in;  cropleft "0";  croptop "1";  cropright "1";
%cropbottom "0";  filename 'ut00.ps';file-properties "XNPEU";}}}%
%BeginExpansion
{\includegraphics[
%natheight=3.000000in,
%natwidth=3.000000in,
height=0.3269in,
width=0.3269in
]%
{ut00.ps}%
}%
%EndExpansion
&
%TCIMACRO{\FRAME{itbpF}{0.3269in}{0.3269in}{0in}{}{}{ut03.ps}%
%width 0.3269in;  height 0.3269in;  depth 0in;  original-width 3in;
%original-height 3in;  cropleft "0";  croptop "1";  cropright "1";
%cropbottom "0";  filename 'ut03.ps';file-properties "XNPEU";}}}%
%BeginExpansion
{\includegraphics[
%natheight=3.000000in,
%natwidth=3.000000in,
height=0.3269in,
width=0.3269in
]%
{ut03.ps}%
}%
%EndExpansion
&
%TCIMACRO{\FRAME{itbpF}{0.3269in}{0.3269in}{0in}{}{}{ut05.ps}%
%width 0.3269in;  height 0.3269in;  depth 0in;  original-width 3in;
%original-height 3in;  cropleft "0";  croptop "1";  cropright "1";
%cropbottom "0";  filename 'ut05.ps';file-properties "XNPEU";}}}%
%BeginExpansion
{\includegraphics[
%natheight=3.000000in,
%natwidth=3.000000in,
height=0.3269in,
width=0.3269in
]%
{ut05.ps}%
}%
%EndExpansion
&
%TCIMACRO{\FRAME{itbpF}{0.3269in}{0.3269in}{0in}{}{}{ut04.ps}%
%width 0.3269in;  height 0.3269in;  depth 0in;  original-width 3in;
%original-height 3in;  cropleft "0";  croptop "1";  cropright "1";
%cropbottom "0";  filename 'ut04.ps';file-properties "XNPEU";}}}%
%BeginExpansion
{\includegraphics[
%natheight=3.000000in,
%natwidth=3.000000in,
height=0.3269in,
width=0.3269in
]%
{ut04.ps}%
}%
%EndExpansion
\end{array}
\qquad\qquad%
\begin{array}
[c]{ccccc}%
%TCIMACRO{\FRAME{itbpF}{0.3269in}{0.3269in}{0in}{}{}{ut02.ps}%
%width 0.3269in;  height 0.3269in;  depth 0in;  original-width 3in;
%original-height 3in;  cropleft "0";  croptop "1";  cropright "1";
%cropbottom "0";  filename 'ut02.ps';file-properties "XNPEU";}}}%
%BeginExpansion
{\includegraphics[
%natheight=3.000000in,
%natwidth=3.000000in,
height=0.3269in,
width=0.3269in
]%
{ut02.ps}%
}%
%EndExpansion
&
%TCIMACRO{\FRAME{itbpF}{0.3269in}{0.3269in}{0in}{}{}{ut05.ps}%
%width 0.3269in;  height 0.3269in;  depth 0in;  original-width 3in;
%original-height 3in;  cropleft "0";  croptop "1";  cropright "1";
%cropbottom "0";  filename 'ut05.ps';file-properties "XNPEU";}}}%
%BeginExpansion
{\includegraphics[
%natheight=3.000000in,
%natwidth=3.000000in,
height=0.3269in,
width=0.3269in
]%
{ut05.ps}%
}%
%EndExpansion
&
%TCIMACRO{\FRAME{itbpF}{0.3269in}{0.3269in}{0in}{}{}{ut05.ps}%
%width 0.3269in;  height 0.3269in;  depth 0in;  original-width 3in;
%original-height 3in;  cropleft "0";  croptop "1";  cropright "1";
%cropbottom "0";  filename 'ut05.ps';file-properties "XNPEU";}}}%
%BeginExpansion
{\includegraphics[
%natheight=3.000000in,
%natwidth=3.000000in,
height=0.3269in,
width=0.3269in
]%
{ut05.ps}%
}%
%EndExpansion
&
%TCIMACRO{\FRAME{itbpF}{0.3269in}{0.3269in}{0in}{}{}{ut01.ps}%
%width 0.3269in;  height 0.3269in;  depth 0in;  original-width 3in;
%original-height 3in;  cropleft "0";  croptop "1";  cropright "1";
%cropbottom "0";  filename 'ut01.ps';file-properties "XNPEU";}}}%
%BeginExpansion
{\includegraphics[
%natheight=3.000000in,
%natwidth=3.000000in,
height=0.3269in,
width=0.3269in
]%
{ut01.ps}%
}%
%EndExpansion
&
%TCIMACRO{\FRAME{itbpF}{0.3269in}{0.3269in}{0in}{}{}{ut00.ps}%
%width 0.3269in;  height 0.3269in;  depth 0in;  original-width 3in;
%original-height 3in;  cropleft "0";  croptop "1";  cropright "1";
%cropbottom "0";  filename 'ut00.ps';file-properties "XNPEU";}}}%
%BeginExpansion
{\includegraphics[
%natheight=3.000000in,
%natwidth=3.000000in,
height=0.3269in,
width=0.3269in
]%
{ut00.ps}%
}%
%EndExpansion
\\%
%TCIMACRO{\FRAME{itbpF}{0.3269in}{0.3269in}{0in}{}{}{ut06.ps}%
%width 0.3269in;  height 0.3269in;  depth 0in;  original-width 3in;
%original-height 3in;  cropleft "0";  croptop "1";  cropright "1";
%cropbottom "0";  filename 'ut06.ps';file-properties "XNPEU";}}}%
%BeginExpansion
{\includegraphics[
%natheight=3.000000in,
%natwidth=3.000000in,
height=0.3269in,
width=0.3269in
]%
{ut06.ps}%
}%
%EndExpansion
&
%TCIMACRO{\FRAME{itbpF}{0.3269in}{0.3269in}{0in}{}{}{ut00.ps}%
%width 0.3269in;  height 0.3269in;  depth 0in;  original-width 3in;
%original-height 3in;  cropleft "0";  croptop "1";  cropright "1";
%cropbottom "0";  filename 'ut00.ps';file-properties "XNPEU";}}}%
%BeginExpansion
{\includegraphics[
%natheight=3.000000in,
%natwidth=3.000000in,
height=0.3269in,
width=0.3269in
]%
{ut00.ps}%
}%
%EndExpansion
&
%TCIMACRO{\FRAME{itbpF}{0.3269in}{0.3269in}{0in}{}{}{ut02.ps}%
%width 0.3269in;  height 0.3269in;  depth 0in;  original-width 3in;
%original-height 3in;  cropleft "0";  croptop "1";  cropright "1";
%cropbottom "0";  filename 'ut02.ps';file-properties "XNPEU";}}}%
%BeginExpansion
{\includegraphics[
%natheight=3.000000in,
%natwidth=3.000000in,
height=0.3269in,
width=0.3269in
]%
{ut02.ps}%
}%
%EndExpansion
&
%TCIMACRO{\FRAME{itbpF}{0.3269in}{0.3269in}{0in}{}{}{ut09.ps}%
%width 0.3269in;  height 0.3269in;  depth 0in;  original-width 3in;
%original-height 3in;  cropleft "0";  croptop "1";  cropright "1";
%cropbottom "0";  filename 'ut09.ps';file-properties "XNPEU";}}}%
%BeginExpansion
{\includegraphics[
%natheight=3.000000in,
%natwidth=3.000000in,
height=0.3269in,
width=0.3269in
]%
{ut09.ps}%
}%
%EndExpansion
&
%TCIMACRO{\FRAME{itbpF}{0.3269in}{0.3269in}{0in}{}{}{ut01.ps}%
%width 0.3269in;  height 0.3269in;  depth 0in;  original-width 3in;
%original-height 3in;  cropleft "0";  croptop "1";  cropright "1";
%cropbottom "0";  filename 'ut01.ps';file-properties "XNPEU";}}}%
%BeginExpansion
{\includegraphics[
%natheight=3.000000in,
%natwidth=3.000000in,
height=0.3269in,
width=0.3269in
]%
{ut01.ps}%
}%
%EndExpansion
\\%
%TCIMACRO{\FRAME{itbpF}{0.3269in}{0.3269in}{0in}{}{}{ut06.ps}%
%width 0.3269in;  height 0.3269in;  depth 0in;  original-width 3in;
%original-height 3in;  cropleft "0";  croptop "1";  cropright "1";
%cropbottom "0";  filename 'ut06.ps';file-properties "XNPEU";}}}%
%BeginExpansion
{\includegraphics[
%natheight=3.000000in,
%natwidth=3.000000in,
height=0.3269in,
width=0.3269in
]%
{ut06.ps}%
}%
%EndExpansion
&
%TCIMACRO{\FRAME{itbpF}{0.3269in}{0.3269in}{0in}{}{}{ut02.ps}%
%width 0.3269in;  height 0.3269in;  depth 0in;  original-width 3in;
%original-height 3in;  cropleft "0";  croptop "1";  cropright "1";
%cropbottom "0";  filename 'ut02.ps';file-properties "XNPEU";}}}%
%BeginExpansion
{\includegraphics[
%natheight=3.000000in,
%natwidth=3.000000in,
height=0.3269in,
width=0.3269in
]%
{ut02.ps}%
}%
%EndExpansion
&
%TCIMACRO{\FRAME{itbpF}{0.3269in}{0.3269in}{0in}{}{}{ut09.ps}%
%width 0.3269in;  height 0.3269in;  depth 0in;  original-width 3in;
%original-height 3in;  cropleft "0";  croptop "1";  cropright "1";
%cropbottom "0";  filename 'ut09.ps';file-properties "XNPEU";}}}%
%BeginExpansion
{\includegraphics[
%natheight=3.000000in,
%natwidth=3.000000in,
height=0.3269in,
width=0.3269in
]%
{ut09.ps}%
}%
%EndExpansion
&
%TCIMACRO{\FRAME{itbpF}{0.3269in}{0.3269in}{0in}{}{}{ut04.ps}%
%width 0.3269in;  height 0.3269in;  depth 0in;  original-width 3in;
%original-height 3in;  cropleft "0";  croptop "1";  cropright "1";
%cropbottom "0";  filename 'ut04.ps';file-properties "XNPEU";}}}%
%BeginExpansion
{\includegraphics[
%natheight=3.000000in,
%natwidth=3.000000in,
height=0.3269in,
width=0.3269in
]%
{ut04.ps}%
}%
%EndExpansion
&
%TCIMACRO{\FRAME{itbpF}{0.3269in}{0.3269in}{0in}{}{}{ut06.ps}%
%width 0.3269in;  height 0.3269in;  depth 0in;  original-width 3in;
%original-height 3in;  cropleft "0";  croptop "1";  cropright "1";
%cropbottom "0";  filename 'ut06.ps';file-properties "XNPEU";}}}%
%BeginExpansion
{\includegraphics[
%natheight=3.000000in,
%natwidth=3.000000in,
height=0.3269in,
width=0.3269in
]%
{ut06.ps}%
}%
%EndExpansion
\\%
%TCIMACRO{\FRAME{itbpF}{0.3269in}{0.3269in}{0in}{}{}{ut03.ps}%
%width 0.3269in;  height 0.3269in;  depth 0in;  original-width 3in;
%original-height 3in;  cropleft "0";  croptop "1";  cropright "1";
%cropbottom "0";  filename 'ut03.ps';file-properties "XNPEU";}}}%
%BeginExpansion
{\includegraphics[
%natheight=3.000000in,
%natwidth=3.000000in,
height=0.3269in,
width=0.3269in
]%
{ut03.ps}%
}%
%EndExpansion
&
%TCIMACRO{\FRAME{itbpF}{0.3269in}{0.3269in}{0in}{}{}{ut09.ps}%
%width 0.3269in;  height 0.3269in;  depth 0in;  original-width 3in;
%original-height 3in;  cropleft "0";  croptop "1";  cropright "1";
%cropbottom "0";  filename 'ut09.ps';file-properties "XNPEU";}}}%
%BeginExpansion
{\includegraphics[
%natheight=3.000000in,
%natwidth=3.000000in,
height=0.3269in,
width=0.3269in
]%
{ut09.ps}%
}%
%EndExpansion
&
%TCIMACRO{\FRAME{itbpF}{0.3269in}{0.3269in}{0in}{}{}{ut10.ps}%
%width 0.3269in;  height 0.3269in;  depth 0in;  original-width 3in;
%original-height 3in;  cropleft "0";  croptop "1";  cropright "1";
%cropbottom "0";  filename 'ut10.ps';file-properties "XNPEU";}}}%
%BeginExpansion
{\includegraphics[
%natheight=3.000000in,
%natwidth=3.000000in,
height=0.3269in,
width=0.3269in
]%
{ut10.ps}%
}%
%EndExpansion
&
%TCIMACRO{\FRAME{itbpF}{0.3269in}{0.3269in}{0in}{}{}{ut05.ps}%
%width 0.3269in;  height 0.3269in;  depth 0in;  original-width 3in;
%original-height 3in;  cropleft "0";  croptop "1";  cropright "1";
%cropbottom "0";  filename 'ut05.ps';file-properties "XNPEU";}}}%
%BeginExpansion
{\includegraphics[
%natheight=3.000000in,
%natwidth=3.000000in,
height=0.3269in,
width=0.3269in
]%
{ut05.ps}%
}%
%EndExpansion
&
%TCIMACRO{\FRAME{itbpF}{0.3269in}{0.3269in}{0in}{}{}{ut04.ps}%
%width 0.3269in;  height 0.3269in;  depth 0in;  original-width 3in;
%original-height 3in;  cropleft "0";  croptop "1";  cropright "1";
%cropbottom "0";  filename 'ut04.ps';file-properties "XNPEU";}}}%
%BeginExpansion
{\includegraphics[
%natheight=3.000000in,
%natwidth=3.000000in,
height=0.3269in,
width=0.3269in
]%
{ut04.ps}%
}%
%EndExpansion
\\%
%TCIMACRO{\FRAME{itbpF}{0.3269in}{0.3269in}{0in}{}{}{ut00.ps}%
%width 0.3269in;  height 0.3269in;  depth 0in;  original-width 3in;
%original-height 3in;  cropleft "0";  croptop "1";  cropright "1";
%cropbottom "0";  filename 'ut00.ps';file-properties "XNPEU";}}}%
%BeginExpansion
{\includegraphics[
%natheight=3.000000in,
%natwidth=3.000000in,
height=0.3269in,
width=0.3269in
]%
{ut00.ps}%
}%
%EndExpansion
&
%TCIMACRO{\FRAME{itbpF}{0.3269in}{0.3269in}{0in}{}{}{ut03.ps}%
%width 0.3269in;  height 0.3269in;  depth 0in;  original-width 3in;
%original-height 3in;  cropleft "0";  croptop "1";  cropright "1";
%cropbottom "0";  filename 'ut03.ps';file-properties "XNPEU";}}}%
%BeginExpansion
{\includegraphics[
%natheight=3.000000in,
%natwidth=3.000000in,
height=0.3269in,
width=0.3269in
]%
{ut03.ps}%
}%
%EndExpansion
&
%TCIMACRO{\FRAME{itbpF}{0.3269in}{0.3269in}{0in}{}{}{ut04.ps}%
%width 0.3269in;  height 0.3269in;  depth 0in;  original-width 3in;
%original-height 3in;  cropleft "0";  croptop "1";  cropright "1";
%cropbottom "0";  filename 'ut04.ps';file-properties "XNPEU";}}}%
%BeginExpansion
{\includegraphics[
%natheight=3.000000in,
%natwidth=3.000000in,
height=0.3269in,
width=0.3269in
]%
{ut04.ps}%
}%
%EndExpansion
&
%TCIMACRO{\FRAME{itbpF}{0.3269in}{0.3269in}{0in}{}{}{ut00.ps}%
%width 0.3269in;  height 0.3269in;  depth 0in;  original-width 3in;
%original-height 3in;  cropleft "0";  croptop "1";  cropright "1";
%cropbottom "0";  filename 'ut00.ps';file-properties "XNPEU";}}}%
%BeginExpansion
{\includegraphics[
%natheight=3.000000in,
%natwidth=3.000000in,
height=0.3269in,
width=0.3269in
]%
{ut00.ps}%
}%
%EndExpansion
&
%TCIMACRO{\FRAME{itbpF}{0.3269in}{0.3269in}{0in}{}{}{ut00.ps}%
%width 0.3269in;  height 0.3269in;  depth 0in;  original-width 3in;
%original-height 3in;  cropleft "0";  croptop "1";  cropright "1";
%cropbottom "0";  filename 'ut00.ps';file-properties "XNPEU";}}}%
%BeginExpansion
{\includegraphics[
%natheight=3.000000in,
%natwidth=3.000000in,
height=0.3269in,
width=0.3269in
]%
{ut00.ps}%
}%
%EndExpansion
\end{array}
\]

\bigskip%

\[%
\begin{array}
[c]{cccccc}%
%TCIMACRO{\FRAME{itbpF}{0.3269in}{0.3269in}{0in}{}{}{ut00.ps}%
%width 0.3269in;  height 0.3269in;  depth 0in;  original-width 3in;
%original-height 3in;  cropleft "0";  croptop "1";  cropright "1";
%cropbottom "0";  filename 'ut00.ps';file-properties "XNPEU";}}}%
%BeginExpansion
{\includegraphics[
%natheight=3.000000in,
%natwidth=3.000000in,
height=0.3269in,
width=0.3269in
]%
{ut00.ps}%
}%
%EndExpansion
&
%TCIMACRO{\FRAME{itbpF}{0.3269in}{0.3269in}{0in}{}{}{ut00.ps}%
%width 0.3269in;  height 0.3269in;  depth 0in;  original-width 3in;
%original-height 3in;  cropleft "0";  croptop "1";  cropright "1";
%cropbottom "0";  filename 'ut00.ps';file-properties "XNPEU";}}}%
%BeginExpansion
{\includegraphics[
%natheight=3.000000in,
%natwidth=3.000000in,
height=0.3269in,
width=0.3269in
]%
{ut00.ps}%
}%
%EndExpansion
&
%TCIMACRO{\FRAME{itbpF}{0.3269in}{0.3269in}{0in}{}{}{ut02.ps}%
%width 0.3269in;  height 0.3269in;  depth 0in;  original-width 3in;
%original-height 3in;  cropleft "0";  croptop "1";  cropright "1";
%cropbottom "0";  filename 'ut02.ps';file-properties "XNPEU";}}}%
%BeginExpansion
{\includegraphics[
%natheight=3.000000in,
%natwidth=3.000000in,
height=0.3269in,
width=0.3269in
]%
{ut02.ps}%
}%
%EndExpansion
&
%TCIMACRO{\FRAME{itbpF}{0.3269in}{0.3269in}{0in}{}{}{ut05.ps}%
%width 0.3269in;  height 0.3269in;  depth 0in;  original-width 3in;
%original-height 3in;  cropleft "0";  croptop "1";  cropright "1";
%cropbottom "0";  filename 'ut05.ps';file-properties "XNPEU";}}}%
%BeginExpansion
{\includegraphics[
%natheight=3.000000in,
%natwidth=3.000000in,
height=0.3269in,
width=0.3269in
]%
{ut05.ps}%
}%
%EndExpansion
&
%TCIMACRO{\FRAME{itbpF}{0.3269in}{0.3269in}{0in}{}{}{ut05.ps}%
%width 0.3269in;  height 0.3269in;  depth 0in;  original-width 3in;
%original-height 3in;  cropleft "0";  croptop "1";  cropright "1";
%cropbottom "0";  filename 'ut05.ps';file-properties "XNPEU";}}}%
%BeginExpansion
{\includegraphics[
%natheight=3.000000in,
%natwidth=3.000000in,
height=0.3269in,
width=0.3269in
]%
{ut05.ps}%
}%
%EndExpansion
&
%TCIMACRO{\FRAME{itbpF}{0.3269in}{0.3269in}{0in}{}{}{ut01.ps}%
%width 0.3269in;  height 0.3269in;  depth 0in;  original-width 3in;
%original-height 3in;  cropleft "0";  croptop "1";  cropright "1";
%cropbottom "0";  filename 'ut01.ps';file-properties "XNPEU";}}}%
%BeginExpansion
{\includegraphics[
%natheight=3.000000in,
%natwidth=3.000000in,
height=0.3269in,
width=0.3269in
]%
{ut01.ps}%
}%
%EndExpansion
\\%
%TCIMACRO{\FRAME{itbpF}{0.3269in}{0.3269in}{0in}{}{}{ut02.ps}%
%width 0.3269in;  height 0.3269in;  depth 0in;  original-width 3in;
%original-height 3in;  cropleft "0";  croptop "1";  cropright "1";
%cropbottom "0";  filename 'ut02.ps';file-properties "XNPEU";}}}%
%BeginExpansion
{\includegraphics[
%natheight=3.000000in,
%natwidth=3.000000in,
height=0.3269in,
width=0.3269in
]%
{ut02.ps}%
}%
%EndExpansion
&
%TCIMACRO{\FRAME{itbpF}{0.3269in}{0.3269in}{0in}{}{}{ut05.ps}%
%width 0.3269in;  height 0.3269in;  depth 0in;  original-width 3in;
%original-height 3in;  cropleft "0";  croptop "1";  cropright "1";
%cropbottom "0";  filename 'ut05.ps';file-properties "XNPEU";}}}%
%BeginExpansion
{\includegraphics[
%natheight=3.000000in,
%natwidth=3.000000in,
height=0.3269in,
width=0.3269in
]%
{ut05.ps}%
}%
%EndExpansion
&
%TCIMACRO{\FRAME{itbpF}{0.3269in}{0.3269in}{0in}{}{}{ut10.ps}%
%width 0.3269in;  height 0.3269in;  depth 0in;  original-width 3in;
%original-height 3in;  cropleft "0";  croptop "1";  cropright "1";
%cropbottom "0";  filename 'ut10.ps';file-properties "XNPEU";}}}%
%BeginExpansion
{\includegraphics[
%natheight=3.000000in,
%natwidth=3.000000in,
height=0.3269in,
width=0.3269in
]%
{ut10.ps}%
}%
%EndExpansion
&
%TCIMACRO{\FRAME{itbpF}{0.3269in}{0.3269in}{0in}{}{}{ut01.ps}%
%width 0.3269in;  height 0.3269in;  depth 0in;  original-width 3in;
%original-height 3in;  cropleft "0";  croptop "1";  cropright "1";
%cropbottom "0";  filename 'ut01.ps';file-properties "XNPEU";}}}%
%BeginExpansion
{\includegraphics[
%natheight=3.000000in,
%natwidth=3.000000in,
height=0.3269in,
width=0.3269in
]%
{ut01.ps}%
}%
%EndExpansion
&
%TCIMACRO{\FRAME{itbpF}{0.3269in}{0.3269in}{0in}{}{}{ut00.ps}%
%width 0.3269in;  height 0.3269in;  depth 0in;  original-width 3in;
%original-height 3in;  cropleft "0";  croptop "1";  cropright "1";
%cropbottom "0";  filename 'ut00.ps';file-properties "XNPEU";}}}%
%BeginExpansion
{\includegraphics[
%natheight=3.000000in,
%natwidth=3.000000in,
height=0.3269in,
width=0.3269in
]%
{ut00.ps}%
}%
%EndExpansion
&
%TCIMACRO{\FRAME{itbpF}{0.3269in}{0.3269in}{0in}{}{}{ut06.ps}%
%width 0.3269in;  height 0.3269in;  depth 0in;  original-width 3in;
%original-height 3in;  cropleft "0";  croptop "1";  cropright "1";
%cropbottom "0";  filename 'ut06.ps';file-properties "XNPEU";}}}%
%BeginExpansion
{\includegraphics[
%natheight=3.000000in,
%natwidth=3.000000in,
height=0.3269in,
width=0.3269in
]%
{ut06.ps}%
}%
%EndExpansion
\\%
%TCIMACRO{\FRAME{itbpF}{0.3269in}{0.3269in}{0in}{}{}{ut06.ps}%
%width 0.3269in;  height 0.3269in;  depth 0in;  original-width 3in;
%original-height 3in;  cropleft "0";  croptop "1";  cropright "1";
%cropbottom "0";  filename 'ut06.ps';file-properties "XNPEU";}}}%
%BeginExpansion
{\includegraphics[
%natheight=3.000000in,
%natwidth=3.000000in,
height=0.3269in,
width=0.3269in
]%
{ut06.ps}%
}%
%EndExpansion
&
%TCIMACRO{\FRAME{itbpF}{0.3269in}{0.3269in}{0in}{}{}{ut02.ps}%
%width 0.3269in;  height 0.3269in;  depth 0in;  original-width 3in;
%original-height 3in;  cropleft "0";  croptop "1";  cropright "1";
%cropbottom "0";  filename 'ut02.ps';file-properties "XNPEU";}}}%
%BeginExpansion
{\includegraphics[
%natheight=3.000000in,
%natwidth=3.000000in,
height=0.3269in,
width=0.3269in
]%
{ut02.ps}%
}%
%EndExpansion
&
%TCIMACRO{\FRAME{itbpF}{0.3269in}{0.3269in}{0in}{}{}{ut09.ps}%
%width 0.3269in;  height 0.3269in;  depth 0in;  original-width 3in;
%original-height 3in;  cropleft "0";  croptop "1";  cropright "1";
%cropbottom "0";  filename 'ut09.ps';file-properties "XNPEU";}}}%
%BeginExpansion
{\includegraphics[
%natheight=3.000000in,
%natwidth=3.000000in,
height=0.3269in,
width=0.3269in
]%
{ut09.ps}%
}%
%EndExpansion
&
%TCIMACRO{\FRAME{itbpF}{0.3269in}{0.3269in}{0in}{}{}{ut10.ps}%
%width 0.3269in;  height 0.3269in;  depth 0in;  original-width 3in;
%original-height 3in;  cropleft "0";  croptop "1";  cropright "1";
%cropbottom "0";  filename 'ut10.ps';file-properties "XNPEU";}}}%
%BeginExpansion
{\includegraphics[
%natheight=3.000000in,
%natwidth=3.000000in,
height=0.3269in,
width=0.3269in
]%
{ut10.ps}%
}%
%EndExpansion
&
%TCIMACRO{\FRAME{itbpF}{0.3269in}{0.3269in}{0in}{}{}{ut01.ps}%
%width 0.3269in;  height 0.3269in;  depth 0in;  original-width 3in;
%original-height 3in;  cropleft "0";  croptop "1";  cropright "1";
%cropbottom "0";  filename 'ut01.ps';file-properties "XNPEU";}}}%
%BeginExpansion
{\includegraphics[
%natheight=3.000000in,
%natwidth=3.000000in,
height=0.3269in,
width=0.3269in
]%
{ut01.ps}%
}%
%EndExpansion
&
%TCIMACRO{\FRAME{itbpF}{0.3269in}{0.3269in}{0in}{}{}{ut06.ps}%
%width 0.3269in;  height 0.3269in;  depth 0in;  original-width 3in;
%original-height 3in;  cropleft "0";  croptop "1";  cropright "1";
%cropbottom "0";  filename 'ut06.ps';file-properties "XNPEU";}}}%
%BeginExpansion
{\includegraphics[
%natheight=3.000000in,
%natwidth=3.000000in,
height=0.3269in,
width=0.3269in
]%
{ut06.ps}%
}%
%EndExpansion
\\%
%TCIMACRO{\FRAME{itbpF}{0.3269in}{0.3269in}{0in}{}{}{ut06.ps}%
%width 0.3269in;  height 0.3269in;  depth 0in;  original-width 3in;
%original-height 3in;  cropleft "0";  croptop "1";  cropright "1";
%cropbottom "0";  filename 'ut06.ps';file-properties "XNPEU";}}}%
%BeginExpansion
{\includegraphics[
%natheight=3.000000in,
%natwidth=3.000000in,
height=0.3269in,
width=0.3269in
]%
{ut06.ps}%
}%
%EndExpansion
&
%TCIMACRO{\FRAME{itbpF}{0.3269in}{0.3269in}{0in}{}{}{ut06.ps}%
%width 0.3269in;  height 0.3269in;  depth 0in;  original-width 3in;
%original-height 3in;  cropleft "0";  croptop "1";  cropright "1";
%cropbottom "0";  filename 'ut06.ps';file-properties "XNPEU";}}}%
%BeginExpansion
{\includegraphics[
%natheight=3.000000in,
%natwidth=3.000000in,
height=0.3269in,
width=0.3269in
]%
{ut06.ps}%
}%
%EndExpansion
&
%TCIMACRO{\FRAME{itbpF}{0.3269in}{0.3269in}{0in}{}{}{ut03.ps}%
%width 0.3269in;  height 0.3269in;  depth 0in;  original-width 3in;
%original-height 3in;  cropleft "0";  croptop "1";  cropright "1";
%cropbottom "0";  filename 'ut03.ps';file-properties "XNPEU";}}}%
%BeginExpansion
{\includegraphics[
%natheight=3.000000in,
%natwidth=3.000000in,
height=0.3269in,
width=0.3269in
]%
{ut03.ps}%
}%
%EndExpansion
&
%TCIMACRO{\FRAME{itbpF}{0.3269in}{0.3269in}{0in}{}{}{ut09.ps}%
%width 0.3269in;  height 0.3269in;  depth 0in;  original-width 3in;
%original-height 3in;  cropleft "0";  croptop "1";  cropright "1";
%cropbottom "0";  filename 'ut09.ps';file-properties "XNPEU";}}}%
%BeginExpansion
{\includegraphics[
%natheight=3.000000in,
%natwidth=3.000000in,
height=0.3269in,
width=0.3269in
]%
{ut09.ps}%
}%
%EndExpansion
&
%TCIMACRO{\FRAME{itbpF}{0.3269in}{0.3269in}{0in}{}{}{ut10.ps}%
%width 0.3269in;  height 0.3269in;  depth 0in;  original-width 3in;
%original-height 3in;  cropleft "0";  croptop "1";  cropright "1";
%cropbottom "0";  filename 'ut10.ps';file-properties "XNPEU";}}}%
%BeginExpansion
{\includegraphics[
%natheight=3.000000in,
%natwidth=3.000000in,
height=0.3269in,
width=0.3269in
]%
{ut10.ps}%
}%
%EndExpansion
&
%TCIMACRO{\FRAME{itbpF}{0.3269in}{0.3269in}{0in}{}{}{ut04.ps}%
%width 0.3269in;  height 0.3269in;  depth 0in;  original-width 3in;
%original-height 3in;  cropleft "0";  croptop "1";  cropright "1";
%cropbottom "0";  filename 'ut04.ps';file-properties "XNPEU";}}}%
%BeginExpansion
{\includegraphics[
%natheight=3.000000in,
%natwidth=3.000000in,
height=0.3269in,
width=0.3269in
]%
{ut04.ps}%
}%
%EndExpansion
\\%
%TCIMACRO{\FRAME{itbpF}{0.3269in}{0.3269in}{0in}{}{}{ut03.ps}%
%width 0.3269in;  height 0.3269in;  depth 0in;  original-width 3in;
%original-height 3in;  cropleft "0";  croptop "1";  cropright "1";
%cropbottom "0";  filename 'ut03.ps';file-properties "XNPEU";}}}%
%BeginExpansion
{\includegraphics[
%natheight=3.000000in,
%natwidth=3.000000in,
height=0.3269in,
width=0.3269in
]%
{ut03.ps}%
}%
%EndExpansion
&
%TCIMACRO{\FRAME{itbpF}{0.3269in}{0.3269in}{0in}{}{}{ut09.ps}%
%width 0.3269in;  height 0.3269in;  depth 0in;  original-width 3in;
%original-height 3in;  cropleft "0";  croptop "1";  cropright "1";
%cropbottom "0";  filename 'ut09.ps';file-properties "XNPEU";}}}%
%BeginExpansion
{\includegraphics[
%natheight=3.000000in,
%natwidth=3.000000in,
height=0.3269in,
width=0.3269in
]%
{ut09.ps}%
}%
%EndExpansion
&
%TCIMACRO{\FRAME{itbpF}{0.3269in}{0.3269in}{0in}{}{}{ut05.ps}%
%width 0.3269in;  height 0.3269in;  depth 0in;  original-width 3in;
%original-height 3in;  cropleft "0";  croptop "1";  cropright "1";
%cropbottom "0";  filename 'ut05.ps';file-properties "XNPEU";}}}%
%BeginExpansion
{\includegraphics[
%natheight=3.000000in,
%natwidth=3.000000in,
height=0.3269in,
width=0.3269in
]%
{ut05.ps}%
}%
%EndExpansion
&
%TCIMACRO{\FRAME{itbpF}{0.3269in}{0.3269in}{0in}{}{}{ut04.ps}%
%width 0.3269in;  height 0.3269in;  depth 0in;  original-width 3in;
%original-height 3in;  cropleft "0";  croptop "1";  cropright "1";
%cropbottom "0";  filename 'ut04.ps';file-properties "XNPEU";}}}%
%BeginExpansion
{\includegraphics[
%natheight=3.000000in,
%natwidth=3.000000in,
height=0.3269in,
width=0.3269in
]%
{ut04.ps}%
}%
%EndExpansion
&
%TCIMACRO{\FRAME{itbpF}{0.3269in}{0.3269in}{0in}{}{}{ut06.ps}%
%width 0.3269in;  height 0.3269in;  depth 0in;  original-width 3in;
%original-height 3in;  cropleft "0";  croptop "1";  cropright "1";
%cropbottom "0";  filename 'ut06.ps';file-properties "XNPEU";}}}%
%BeginExpansion
{\includegraphics[
%natheight=3.000000in,
%natwidth=3.000000in,
height=0.3269in,
width=0.3269in
]%
{ut06.ps}%
}%
%EndExpansion
&
%TCIMACRO{\FRAME{itbpF}{0.3269in}{0.3269in}{0in}{}{}{ut00.ps}%
%width 0.3269in;  height 0.3269in;  depth 0in;  original-width 3in;
%original-height 3in;  cropleft "0";  croptop "1";  cropright "1";
%cropbottom "0";  filename 'ut00.ps';file-properties "XNPEU";}}}%
%BeginExpansion
{\includegraphics[
%natheight=3.000000in,
%natwidth=3.000000in,
height=0.3269in,
width=0.3269in
]%
{ut00.ps}%
}%
%EndExpansion
\\%
%TCIMACRO{\FRAME{itbpF}{0.3269in}{0.3269in}{0in}{}{}{ut00.ps}%
%width 0.3269in;  height 0.3269in;  depth 0in;  original-width 3in;
%original-height 3in;  cropleft "0";  croptop "1";  cropright "1";
%cropbottom "0";  filename 'ut00.ps';file-properties "XNPEU";}}}%
%BeginExpansion
{\includegraphics[
%natheight=3.000000in,
%natwidth=3.000000in,
height=0.3269in,
width=0.3269in
]%
{ut00.ps}%
}%
%EndExpansion
&
%TCIMACRO{\FRAME{itbpF}{0.3269in}{0.3269in}{0in}{}{}{ut03.ps}%
%width 0.3269in;  height 0.3269in;  depth 0in;  original-width 3in;
%original-height 3in;  cropleft "0";  croptop "1";  cropright "1";
%cropbottom "0";  filename 'ut03.ps';file-properties "XNPEU";}}}%
%BeginExpansion
{\includegraphics[
%natheight=3.000000in,
%natwidth=3.000000in,
height=0.3269in,
width=0.3269in
]%
{ut03.ps}%
}%
%EndExpansion
&
%TCIMACRO{\FRAME{itbpF}{0.3269in}{0.3269in}{0in}{}{}{ut05.ps}%
%width 0.3269in;  height 0.3269in;  depth 0in;  original-width 3in;
%original-height 3in;  cropleft "0";  croptop "1";  cropright "1";
%cropbottom "0";  filename 'ut05.ps';file-properties "XNPEU";}}}%
%BeginExpansion
{\includegraphics[
%natheight=3.000000in,
%natwidth=3.000000in,
height=0.3269in,
width=0.3269in
]%
{ut05.ps}%
}%
%EndExpansion
&
%TCIMACRO{\FRAME{itbpF}{0.3269in}{0.3269in}{0in}{}{}{ut05.ps}%
%width 0.3269in;  height 0.3269in;  depth 0in;  original-width 3in;
%original-height 3in;  cropleft "0";  croptop "1";  cropright "1";
%cropbottom "0";  filename 'ut05.ps';file-properties "XNPEU";}}}%
%BeginExpansion
{\includegraphics[
%natheight=3.000000in,
%natwidth=3.000000in,
height=0.3269in,
width=0.3269in
]%
{ut05.ps}%
}%
%EndExpansion
&
%TCIMACRO{\FRAME{itbpF}{0.3269in}{0.3269in}{0in}{}{}{ut04.ps}%
%width 0.3269in;  height 0.3269in;  depth 0in;  original-width 3in;
%original-height 3in;  cropleft "0";  croptop "1";  cropright "1";
%cropbottom "0";  filename 'ut04.ps';file-properties "XNPEU";}}}%
%BeginExpansion
{\includegraphics[
%natheight=3.000000in,
%natwidth=3.000000in,
height=0.3269in,
width=0.3269in
]%
{ut04.ps}%
}%
%EndExpansion
&
%TCIMACRO{\FRAME{itbpF}{0.3269in}{0.3269in}{0in}{}{}{ut00.ps}%
%width 0.3269in;  height 0.3269in;  depth 0in;  original-width 3in;
%original-height 3in;  cropleft "0";  croptop "1";  cropright "1";
%cropbottom "0";  filename 'ut00.ps';file-properties "XNPEU";}}}%
%BeginExpansion
{\includegraphics[
%natheight=3.000000in,
%natwidth=3.000000in,
height=0.3269in,
width=0.3269in
]%
{ut00.ps}%
}%
%EndExpansion
\end{array}
\]

\bigskip

\subsection{Mosaic moves}

\bigskip

We now continue with our program of using mosaics to create a formal model of
(tame) knot theory.

\bigskip

\begin{definition}
Let $k$ and $n$ be positive integers such that $k\leq n$. A $k$-mosaic $N$ is
said to be a $k$\textbf{-submosaic} of an $n$-mosaic $M$ if it is a $k\times
k$ submatrix of $M$. \ The $k$-submosaic $N$ is said to be at
\textbf{location} $\left(  i,j\right)  $ in the $n$-mosaic $M$ if the top left
entry of $N$ lies in row $i$ and column $j$ of $M$. \ Obviously, the set of
possible locations for a $k$-submosaic of an $n$-mosaic is $\left\{  \left(
i,j\right)  :0\leq i,j\leq n-k\right\}  $. Moreover, there are exactly
$\left(  n-k+1\right)  ^{2}$ different locations. \ Let $\mathbf{M}^{\left(
k:i,j\right)  }$ denote the $k$\textbf{-submosaic of }$\mathbf{M}$\textbf{ at
location }$(i,j)$.
\end{definition}

\bigskip

\bigskip

For example, the 3-mosaic
\[%
\begin{array}
[c]{ccc}%
%TCIMACRO{\FRAME{itbpF}{0.3269in}{0.3269in}{0in}{}{}{ut05.ps}%
%width 0.3269in;  height 0.3269in;  depth 0in;  original-width 3in;
%original-height 3in;  cropleft "0";  croptop "1";  cropright "1";
%cropbottom "0";  filename 'ut05.ps';file-properties "XNPEU";}}}%
%BeginExpansion
{\includegraphics[
%natheight=3.000000in,
%natwidth=3.000000in,
height=0.3269in,
width=0.3269in
]%
{ut05.ps}%
}%
%EndExpansion
&
%TCIMACRO{\FRAME{itbpF}{0.3269in}{0.3269in}{0in}{}{}{ut04.ps}%
%width 0.3269in;  height 0.3269in;  depth 0in;  original-width 3in;
%original-height 3in;  cropleft "0";  croptop "1";  cropright "1";
%cropbottom "0";  filename 'ut04.ps';file-properties "XNPEU";}}}%
%BeginExpansion
{\includegraphics[
%natheight=3.000000in,
%natwidth=3.000000in,
height=0.3269in,
width=0.3269in
]%
{ut04.ps}%
}%
%EndExpansion
&
%TCIMACRO{\FRAME{itbpF}{0.3269in}{0.3269in}{0in}{}{}{ut00.ps}%
%width 0.3269in;  height 0.3269in;  depth 0in;  original-width 3in;
%original-height 3in;  cropleft "0";  croptop "1";  cropright "1";
%cropbottom "0";  filename 'ut00.ps';file-properties "XNPEU";}}}%
%BeginExpansion
{\includegraphics[
%natheight=3.000000in,
%natwidth=3.000000in,
height=0.3269in,
width=0.3269in
]%
{ut00.ps}%
}%
%EndExpansion
\\
\end{array}
\]
is the submosaic $M^{(3:0,1)}$ of the 4-mosaic
\[
M=%
\begin{array}
[c]{cccc}%
%TCIMACRO{\FRAME{itbpF}{0.3269in}{0.3269in}{0in}{}{}{ut00.ps}%
%width 0.3269in;  height 0.3269in;  depth 0in;  original-width 3in;
%original-height 3in;  cropleft "0";  croptop "1";  cropright "1";
%cropbottom "0";  filename 'ut00.ps';file-properties "XNPEU";}}}%
%BeginExpansion
{\includegraphics[
%natheight=3.000000in,
%natwidth=3.000000in,
height=0.3269in,
width=0.3269in
]%
{ut00.ps}%
}%
%EndExpansion
&
%TCIMACRO{\FRAME{itbpF}{0.3269in}{0.3269in}{0in}{}{}{ut05.ps}%
%width 0.3269in;  height 0.3269in;  depth 0in;  original-width 3in;
%original-height 3in;  cropleft "0";  croptop "1";  cropright "1";
%cropbottom "0";  filename 'ut05.ps';file-properties "XNPEU";}}}%
%BeginExpansion
{\includegraphics[
%natheight=3.000000in,
%natwidth=3.000000in,
height=0.3269in,
width=0.3269in
]%
{ut05.ps}%
}%
%EndExpansion
&
%TCIMACRO{\FRAME{itbpF}{0.3269in}{0.3269in}{0in}{}{}{ut04.ps}%
%width 0.3269in;  height 0.3269in;  depth 0in;  original-width 3in;
%original-height 3in;  cropleft "0";  croptop "1";  cropright "1";
%cropbottom "0";  filename 'ut04.ps';file-properties "XNPEU";}}}%
%BeginExpansion
{\includegraphics[
%natheight=3.000000in,
%natwidth=3.000000in,
height=0.3269in,
width=0.3269in
]%
{ut04.ps}%
}%
%EndExpansion
&
%TCIMACRO{\FRAME{itbpF}{0.3269in}{0.3269in}{0in}{}{}{ut00.ps}%
%width 0.3269in;  height 0.3269in;  depth 0in;  original-width 3in;
%original-height 3in;  cropleft "0";  croptop "1";  cropright "1";
%cropbottom "0";  filename 'ut00.ps';file-properties "XNPEU";}}}%
%BeginExpansion
{\includegraphics[
%natheight=3.000000in,
%natwidth=3.000000in,
height=0.3269in,
width=0.3269in
]%
{ut00.ps}%
}%
%EndExpansion
\\
\end{array}
\text{ ,}%
\]
and the 2-mosaic
\[%
\begin{array}
[c]{cc}%
%TCIMACRO{\FRAME{itbpF}{0.3269in}{0.3269in}{0in}{}{}{ut10.ps}%
%width 0.3269in;  height 0.3269in;  depth 0in;  original-width 3in;
%original-height 3in;  cropleft "0";  croptop "1";  cropright "1";
%cropbottom "0";  filename 'ut10.ps';file-properties "XNPEU";}}}%
%BeginExpansion
{\includegraphics[
%natheight=3.000000in,
%natwidth=3.000000in,
height=0.3269in,
width=0.3269in
]%
{ut10.ps}%
}%
%EndExpansion
&
%TCIMACRO{\FRAME{itbpF}{0.3269in}{0.3269in}{0in}{}{}{ut01.ps}%
%width 0.3269in;  height 0.3269in;  depth 0in;  original-width 3in;
%original-height 3in;  cropleft "0";  croptop "1";  cropright "1";
%cropbottom "0";  filename 'ut01.ps';file-properties "XNPEU";}}}%
%BeginExpansion
{\includegraphics[
%natheight=3.000000in,
%natwidth=3.000000in,
height=0.3269in,
width=0.3269in
]%
{ut01.ps}%
}%
%EndExpansion
\\
\end{array}
\]
is the submosaic $M^{(2:1,2)}$ of the 4-mosaic
\[
M=%
\begin{array}
[c]{cccc}%
%TCIMACRO{\FRAME{itbpF}{0.3269in}{0.3269in}{0in}{}{}{ut00.ps}%
%width 0.3269in;  height 0.3269in;  depth 0in;  original-width 3in;
%original-height 3in;  cropleft "0";  croptop "1";  cropright "1";
%cropbottom "0";  filename 'ut00.ps';file-properties "XNPEU";}}}%
%BeginExpansion
{\includegraphics[
%natheight=3.000000in,
%natwidth=3.000000in,
height=0.3269in,
width=0.3269in
]%
{ut00.ps}%
}%
%EndExpansion
&
%TCIMACRO{\FRAME{itbpF}{0.3269in}{0.3269in}{0in}{}{}{ut02.ps}%
%width 0.3269in;  height 0.3269in;  depth 0in;  original-width 3in;
%original-height 3in;  cropleft "0";  croptop "1";  cropright "1";
%cropbottom "0";  filename 'ut02.ps';file-properties "XNPEU";}}}%
%BeginExpansion
{\includegraphics[
%natheight=3.000000in,
%natwidth=3.000000in,
height=0.3269in,
width=0.3269in
]%
{ut02.ps}%
}%
%EndExpansion
&
%TCIMACRO{\FRAME{itbpF}{0.3269in}{0.3269in}{0in}{}{}{ut01.ps}%
%width 0.3269in;  height 0.3269in;  depth 0in;  original-width 3in;
%original-height 3in;  cropleft "0";  croptop "1";  cropright "1";
%cropbottom "0";  filename 'ut01.ps';file-properties "XNPEU";}}}%
%BeginExpansion
{\includegraphics[
%natheight=3.000000in,
%natwidth=3.000000in,
height=0.3269in,
width=0.3269in
]%
{ut01.ps}%
}%
%EndExpansion
&
%TCIMACRO{\FRAME{itbpF}{0.3269in}{0.3269in}{0in}{}{}{ut00.ps}%
%width 0.3269in;  height 0.3269in;  depth 0in;  original-width 3in;
%original-height 3in;  cropleft "0";  croptop "1";  cropright "1";
%cropbottom "0";  filename 'ut00.ps';file-properties "XNPEU";}}}%
%BeginExpansion
{\includegraphics[
%natheight=3.000000in,
%natwidth=3.000000in,
height=0.3269in,
width=0.3269in
]%
{ut00.ps}%
}%
%EndExpansion
\\
\end{array}
\text{ .}%
\]

\bigskip

\begin{definition}
Let $k$ and $n$ be positive integers such that $k\leq n$. \ For any two
$k$-mosaics $N$ and $N^{\prime}$, we define a $\mathbf{k}$\textbf{-move at
location }$\left(  \mathbf{i,j}\right)  $ on the set of $n$-mosaics
$\mathbb{M}^{(n)}$, denoted by
\[
N\overset{\left(  i,j\right)  }{\longleftrightarrow}N^{\prime}\text{ ,}%
\]
as the map from $\mathbb{M}^{(n)}$ to $\mathbb{M}^{(n)}$ defined by%
\[
\left(  N\overset{\left(  i,j\right)  }{\longleftrightarrow}N^{\prime}\right)
\left(  M\right)  =\left\{
\begin{array}
[c]{ll}%
M\text{ with }M^{\left(  k:i,j\right)  }\text{ replaced by }N^{\prime} &
\text{if }M^{\left(  k:i,j\right)  }=N\\
M\text{ with }M^{\left(  k:i,j\right)  }\text{ replaced by }N & \text{if
}M^{\left(  k:i,j\right)  }=N^{\prime}\\
M & \text{otherwise}%
\end{array}
\right.
\]

\end{definition}

\bigskip

As an example, consider the $2$-move $N\overset{\left(  0,1\right)
}{\longleftrightarrow}N^{\prime}$ defined by%
\[%
\begin{array}
[c]{cc}%
%TCIMACRO{\FRAME{itbpF}{0.3269in}{0.3269in}{0in}{}{}{ut09.ps}%
%width 0.3269in;  height 0.3269in;  depth 0in;  original-width 3in;
%original-height 3in;  cropleft "0";  croptop "1";  cropright "1";
%cropbottom "0";  filename 'ut09.ps';file-properties "XNPEU";}}}%
%BeginExpansion
{\includegraphics[
%natheight=3.000000in,
%natwidth=3.000000in,
height=0.3269in,
width=0.3269in
]%
{ut09.ps}%
}%
%EndExpansion
&
%TCIMACRO{\FRAME{itbpF}{0.3269in}{0.3269in}{0in}{}{}{ut02.ps}%
%width 0.3269in;  height 0.3269in;  depth 0in;  original-width 3in;
%original-height 3in;  cropleft "0";  croptop "1";  cropright "1";
%cropbottom "0";  filename 'ut02.ps';file-properties "XNPEU";}}}%
%BeginExpansion
{\includegraphics[
%natheight=3.000000in,
%natwidth=3.000000in,
height=0.3269in,
width=0.3269in
]%
{ut02.ps}%
}%
%EndExpansion
\\
\end{array}
\overset{\left(  0,1\right)  }{\longleftrightarrow}%
\begin{array}
[c]{cc}%
%TCIMACRO{\FRAME{itbpF}{0.3269in}{0.3269in}{0in}{}{}{ut07.ps}%
%width 0.3269in;  height 0.3269in;  depth 0in;  original-width 3in;
%original-height 3in;  cropleft "0";  croptop "1";  cropright "1";
%cropbottom "0";  filename 'ut07.ps';file-properties "XNPEU";}}}%
%BeginExpansion
{\includegraphics[
%natheight=3.000000in,
%natwidth=3.000000in,
height=0.3269in,
width=0.3269in
]%
{ut07.ps}%
}%
%EndExpansion
&
%TCIMACRO{\FRAME{itbpF}{0.3269in}{0.3269in}{0in}{}{}{ut00.ps}%
%width 0.3269in;  height 0.3269in;  depth 0in;  original-width 3in;
%original-height 3in;  cropleft "0";  croptop "1";  cropright "1";
%cropbottom "0";  filename 'ut00.ps';file-properties "XNPEU";}}}%
%BeginExpansion
{\includegraphics[
%natheight=3.000000in,
%natwidth=3.000000in,
height=0.3269in,
width=0.3269in
]%
{ut00.ps}%
}%
%EndExpansion
\\
\end{array}
\text{ .}%
\]
Then,
\begin{align*}
\left(  N\overset{\left(  0,1\right)  }{\longleftrightarrow}N^{\prime}\right)
\left(
\begin{array}
[c]{ccc}%
%TCIMACRO{\FRAME{itbpF}{0.3269in}{0.3269in}{0in}{}{}{ut05.ps}%
%width 0.3269in;  height 0.3269in;  depth 0in;  original-width 3in;
%original-height 3in;  cropleft "0";  croptop "1";  cropright "1";
%cropbottom "0";  filename 'ut05.ps';file-properties "XNPEU";}}}%
%BeginExpansion
{\includegraphics[
%natheight=3.000000in,
%natwidth=3.000000in,
height=0.3269in,
width=0.3269in
]%
{ut05.ps}%
}%
%EndExpansion
&
%TCIMACRO{\FRAME{itbpF}{0.3269in}{0.3269in}{0in}{}{}{ut09.ps}%
%width 0.3269in;  height 0.3269in;  depth 0in;  original-width 3in;
%original-height 3in;  cropleft "0";  croptop "1";  cropright "1";
%cropbottom "0";  filename 'ut09.ps';file-properties "XNPEU";}}}%
%BeginExpansion
{\includegraphics[
%natheight=3.000000in,
%natwidth=3.000000in,
height=0.3269in,
width=0.3269in
]%
{ut09.ps}%
}%
%EndExpansion
&
%TCIMACRO{\FRAME{itbpF}{0.3269in}{0.3269in}{0in}{}{}{ut02.ps}%
%width 0.3269in;  height 0.3269in;  depth 0in;  original-width 3in;
%original-height 3in;  cropleft "0";  croptop "1";  cropright "1";
%cropbottom "0";  filename 'ut02.ps';file-properties "XNPEU";}}}%
%BeginExpansion
{\includegraphics[
%natheight=3.000000in,
%natwidth=3.000000in,
height=0.3269in,
width=0.3269in
]%
{ut02.ps}%
}%
%EndExpansion
\\
\end{array}
\right)   &  =%
\begin{array}
[c]{ccc}%
%TCIMACRO{\FRAME{itbpF}{0.3269in}{0.3269in}{0in}{}{}{ut05.ps}%
%width 0.3269in;  height 0.3269in;  depth 0in;  original-width 3in;
%original-height 3in;  cropleft "0";  croptop "1";  cropright "1";
%cropbottom "0";  filename 'ut05.ps';file-properties "XNPEU";}}}%
%BeginExpansion
{\includegraphics[
%natheight=3.000000in,
%natwidth=3.000000in,
height=0.3269in,
width=0.3269in
]%
{ut05.ps}%
}%
%EndExpansion
&
%TCIMACRO{\FRAME{itbpF}{0.3269in}{0.3269in}{0in}{}{}{ut07.ps}%
%width 0.3269in;  height 0.3269in;  depth 0in;  original-width 3in;
%original-height 3in;  cropleft "0";  croptop "1";  cropright "1";
%cropbottom "0";  filename 'ut07.ps';file-properties "XNPEU";}}}%
%BeginExpansion
{\includegraphics[
%natheight=3.000000in,
%natwidth=3.000000in,
height=0.3269in,
width=0.3269in
]%
{ut07.ps}%
}%
%EndExpansion
&
%TCIMACRO{\FRAME{itbpF}{0.3269in}{0.3269in}{0in}{}{}{ut00.ps}%
%width 0.3269in;  height 0.3269in;  depth 0in;  original-width 3in;
%original-height 3in;  cropleft "0";  croptop "1";  cropright "1";
%cropbottom "0";  filename 'ut00.ps';file-properties "XNPEU";}}}%
%BeginExpansion
{\includegraphics[
%natheight=3.000000in,
%natwidth=3.000000in,
height=0.3269in,
width=0.3269in
]%
{ut00.ps}%
}%
%EndExpansion
\\
\end{array}
\text{ \ (Mosaics switched)}\\
& \\
\left(  N\overset{\left(  0,1\right)  }{\longleftrightarrow}N^{\prime}\right)
\left(
\begin{array}
[c]{ccc}%
%TCIMACRO{\FRAME{itbpF}{0.3269in}{0.3269in}{0in}{}{}{ut05.ps}%
%width 0.3269in;  height 0.3269in;  depth 0in;  original-width 3in;
%original-height 3in;  cropleft "0";  croptop "1";  cropright "1";
%cropbottom "0";  filename 'ut05.ps';file-properties "XNPEU";}}}%
%BeginExpansion
{\includegraphics[
%natheight=3.000000in,
%natwidth=3.000000in,
height=0.3269in,
width=0.3269in
]%
{ut05.ps}%
}%
%EndExpansion
&
%TCIMACRO{\FRAME{itbpF}{0.3269in}{0.3269in}{0in}{}{}{ut07.ps}%
%width 0.3269in;  height 0.3269in;  depth 0in;  original-width 3in;
%original-height 3in;  cropleft "0";  croptop "1";  cropright "1";
%cropbottom "0";  filename 'ut07.ps';file-properties "XNPEU";}}}%
%BeginExpansion
{\includegraphics[
%natheight=3.000000in,
%natwidth=3.000000in,
height=0.3269in,
width=0.3269in
]%
{ut07.ps}%
}%
%EndExpansion
&
%TCIMACRO{\FRAME{itbpF}{0.3269in}{0.3269in}{0in}{}{}{ut00.ps}%
%width 0.3269in;  height 0.3269in;  depth 0in;  original-width 3in;
%original-height 3in;  cropleft "0";  croptop "1";  cropright "1";
%cropbottom "0";  filename 'ut00.ps';file-properties "XNPEU";}}}%
%BeginExpansion
{\includegraphics[
%natheight=3.000000in,
%natwidth=3.000000in,
height=0.3269in,
width=0.3269in
]%
{ut00.ps}%
}%
%EndExpansion
\\
\end{array}
\right)   &  =%
\begin{array}
[c]{ccc}%
%TCIMACRO{\FRAME{itbpF}{0.3269in}{0.3269in}{0in}{}{}{ut05.ps}%
%width 0.3269in;  height 0.3269in;  depth 0in;  original-width 3in;
%original-height 3in;  cropleft "0";  croptop "1";  cropright "1";
%cropbottom "0";  filename 'ut05.ps';file-properties "XNPEU";}}}%
%BeginExpansion
{\includegraphics[
%natheight=3.000000in,
%natwidth=3.000000in,
height=0.3269in,
width=0.3269in
]%
{ut05.ps}%
}%
%EndExpansion
&
%TCIMACRO{\FRAME{itbpF}{0.3269in}{0.3269in}{0in}{}{}{ut09.ps}%
%width 0.3269in;  height 0.3269in;  depth 0in;  original-width 3in;
%original-height 3in;  cropleft "0";  croptop "1";  cropright "1";
%cropbottom "0";  filename 'ut09.ps';file-properties "XNPEU";}}}%
%BeginExpansion
{\includegraphics[
%natheight=3.000000in,
%natwidth=3.000000in,
height=0.3269in,
width=0.3269in
]%
{ut09.ps}%
}%
%EndExpansion
&
%TCIMACRO{\FRAME{itbpF}{0.3269in}{0.3269in}{0in}{}{}{ut02.ps}%
%width 0.3269in;  height 0.3269in;  depth 0in;  original-width 3in;
%original-height 3in;  cropleft "0";  croptop "1";  cropright "1";
%cropbottom "0";  filename 'ut02.ps';file-properties "XNPEU";}}}%
%BeginExpansion
{\includegraphics[
%natheight=3.000000in,
%natwidth=3.000000in,
height=0.3269in,
width=0.3269in
]%
{ut02.ps}%
}%
%EndExpansion
\\
\end{array}
\text{ \ (Mosaics switched)}\\
& \\
\left(  N\overset{\left(  0,1\right)  }{\longleftrightarrow}N^{\prime}\right)
\left(
\begin{array}
[c]{ccc}%
%TCIMACRO{\FRAME{itbpF}{0.3269in}{0.3269in}{0in}{}{}{ut05.ps}%
%width 0.3269in;  height 0.3269in;  depth 0in;  original-width 3in;
%original-height 3in;  cropleft "0";  croptop "1";  cropright "1";
%cropbottom "0";  filename 'ut05.ps';file-properties "XNPEU";}}}%
%BeginExpansion
{\includegraphics[
%natheight=3.000000in,
%natwidth=3.000000in,
height=0.3269in,
width=0.3269in
]%
{ut05.ps}%
}%
%EndExpansion
&
%TCIMACRO{\FRAME{itbpF}{0.3269in}{0.3269in}{0in}{}{}{ut00.ps}%
%width 0.3269in;  height 0.3269in;  depth 0in;  original-width 3in;
%original-height 3in;  cropleft "0";  croptop "1";  cropright "1";
%cropbottom "0";  filename 'ut00.ps';file-properties "XNPEU";}}}%
%BeginExpansion
{\includegraphics[
%natheight=3.000000in,
%natwidth=3.000000in,
height=0.3269in,
width=0.3269in
]%
{ut00.ps}%
}%
%EndExpansion
&
%TCIMACRO{\FRAME{itbpF}{0.3269in}{0.3269in}{0in}{}{}{ut02.ps}%
%width 0.3269in;  height 0.3269in;  depth 0in;  original-width 3in;
%original-height 3in;  cropleft "0";  croptop "1";  cropright "1";
%cropbottom "0";  filename 'ut02.ps';file-properties "XNPEU";}}}%
%BeginExpansion
{\includegraphics[
%natheight=3.000000in,
%natwidth=3.000000in,
height=0.3269in,
width=0.3269in
]%
{ut02.ps}%
}%
%EndExpansion
\\
\end{array}
\right)   &  =%
\begin{array}
[c]{ccc}%
%TCIMACRO{\FRAME{itbpF}{0.3269in}{0.3269in}{0in}{}{}{ut05.ps}%
%width 0.3269in;  height 0.3269in;  depth 0in;  original-width 3in;
%original-height 3in;  cropleft "0";  croptop "1";  cropright "1";
%cropbottom "0";  filename 'ut05.ps';file-properties "XNPEU";}}}%
%BeginExpansion
{\includegraphics[
%natheight=3.000000in,
%natwidth=3.000000in,
height=0.3269in,
width=0.3269in
]%
{ut05.ps}%
}%
%EndExpansion
&
%TCIMACRO{\FRAME{itbpF}{0.3269in}{0.3269in}{0in}{}{}{ut00.ps}%
%width 0.3269in;  height 0.3269in;  depth 0in;  original-width 3in;
%original-height 3in;  cropleft "0";  croptop "1";  cropright "1";
%cropbottom "0";  filename 'ut00.ps';file-properties "XNPEU";}}}%
%BeginExpansion
{\includegraphics[
%natheight=3.000000in,
%natwidth=3.000000in,
height=0.3269in,
width=0.3269in
]%
{ut00.ps}%
}%
%EndExpansion
&
%TCIMACRO{\FRAME{itbpF}{0.3269in}{0.3269in}{0in}{}{}{ut02.ps}%
%width 0.3269in;  height 0.3269in;  depth 0in;  original-width 3in;
%original-height 3in;  cropleft "0";  croptop "1";  cropright "1";
%cropbottom "0";  filename 'ut02.ps';file-properties "XNPEU";}}}%
%BeginExpansion
{\includegraphics[
%natheight=3.000000in,
%natwidth=3.000000in,
height=0.3269in,
width=0.3269in
]%
{ut02.ps}%
}%
%EndExpansion
\\
\end{array}
\text{ \ (Mosaic unchanged)}%
\end{align*}

\bigskip

The following proposition is an almost immediate consequence of the definition
of a $k$-move:

\bigskip

\begin{proposition}
Each $k$-move $N\overset{\left(  i,j\right)  }{\longleftrightarrow}N^{\prime}$
is a permutation of $\mathbb{M}^{(n)}$. \ In fact, it is a permutation which
is the product of disjoint transpositions.
\end{proposition}

\bigskip

\subsection{Three important notational conventions}

\bigskip

For the purpose of achieving clarity of exposition and of simplifying the
exposition as much as possible, we adopt the following three nondeterministic
notational conventions which will eliminate a great deal of combinatorial clutter:

\bigskip

\noindent\textbf{Notational Convention 1.} \textit{We will use each of the
following tiles}
\[%
\begin{array}
[c]{cccccccccc}%
%TCIMACRO{\FRAME{itbpF}{0.3269in}{0.3269in}{0in}{}{}{iut01.ps}%
%width 0.3269in;  height 0.3269in;  depth 0in;  original-width 3in;
%original-height 3in;  cropleft "0";  croptop "1";  cropright "1";
%cropbottom "0";  filename 'iut01.ps';file-properties "XNPEU";}}}%
%BeginExpansion
{\includegraphics[
%natheight=3.000000in,
%natwidth=3.000000in,
height=0.3269in,
width=0.3269in
]%
{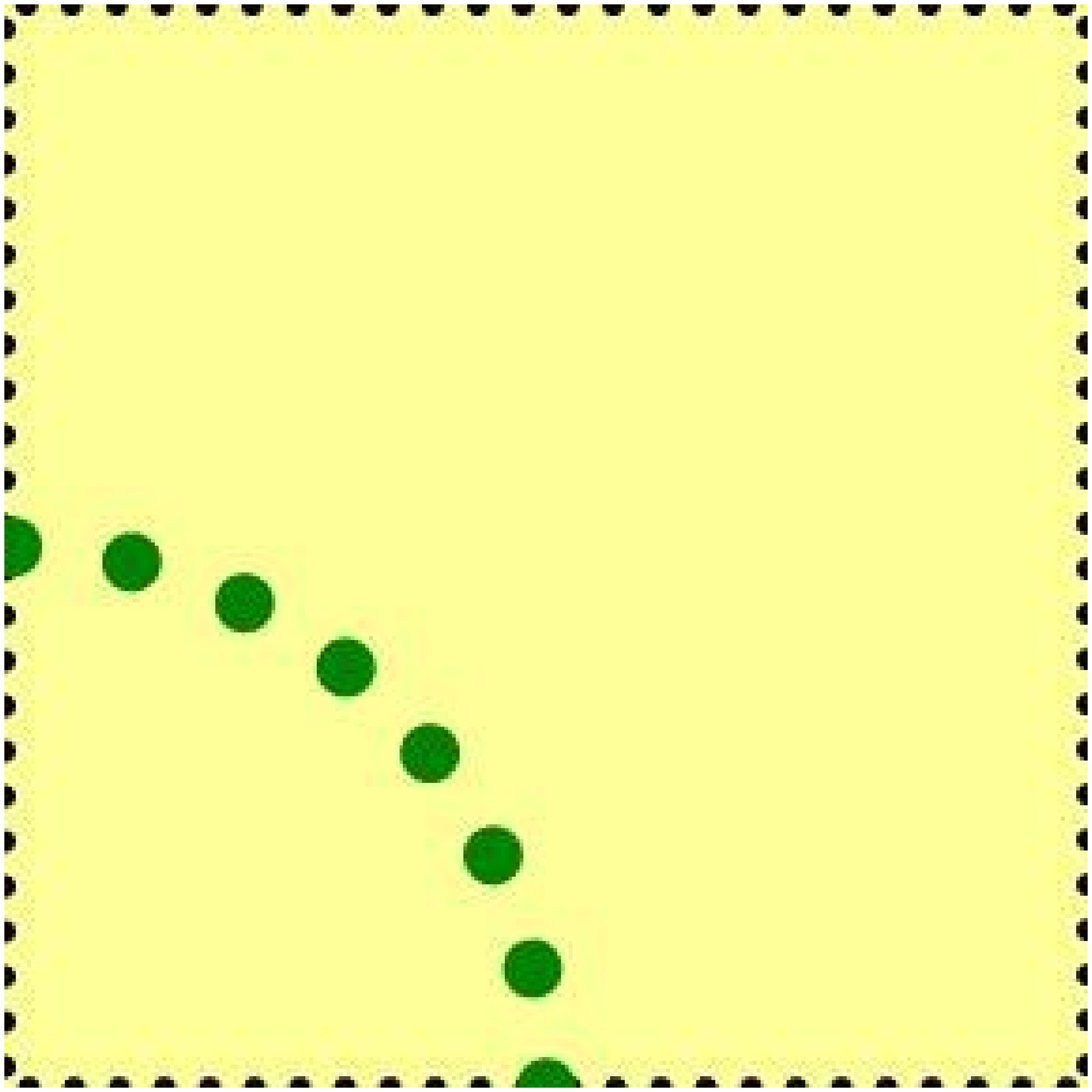}%
}%
%EndExpansion
, &
%TCIMACRO{\FRAME{itbpF}{0.3269in}{0.3269in}{0in}{}{}{iut02.ps}%
%width 0.3269in;  height 0.3269in;  depth 0in;  original-width 3in;
%original-height 3in;  cropleft "0";  croptop "1";  cropright "1";
%cropbottom "0";  filename 'iut02.ps';file-properties "XNPEU";}}}%
%BeginExpansion
{\includegraphics[
%natheight=3.000000in,
%natwidth=3.000000in,
height=0.3269in,
width=0.3269in
]%
{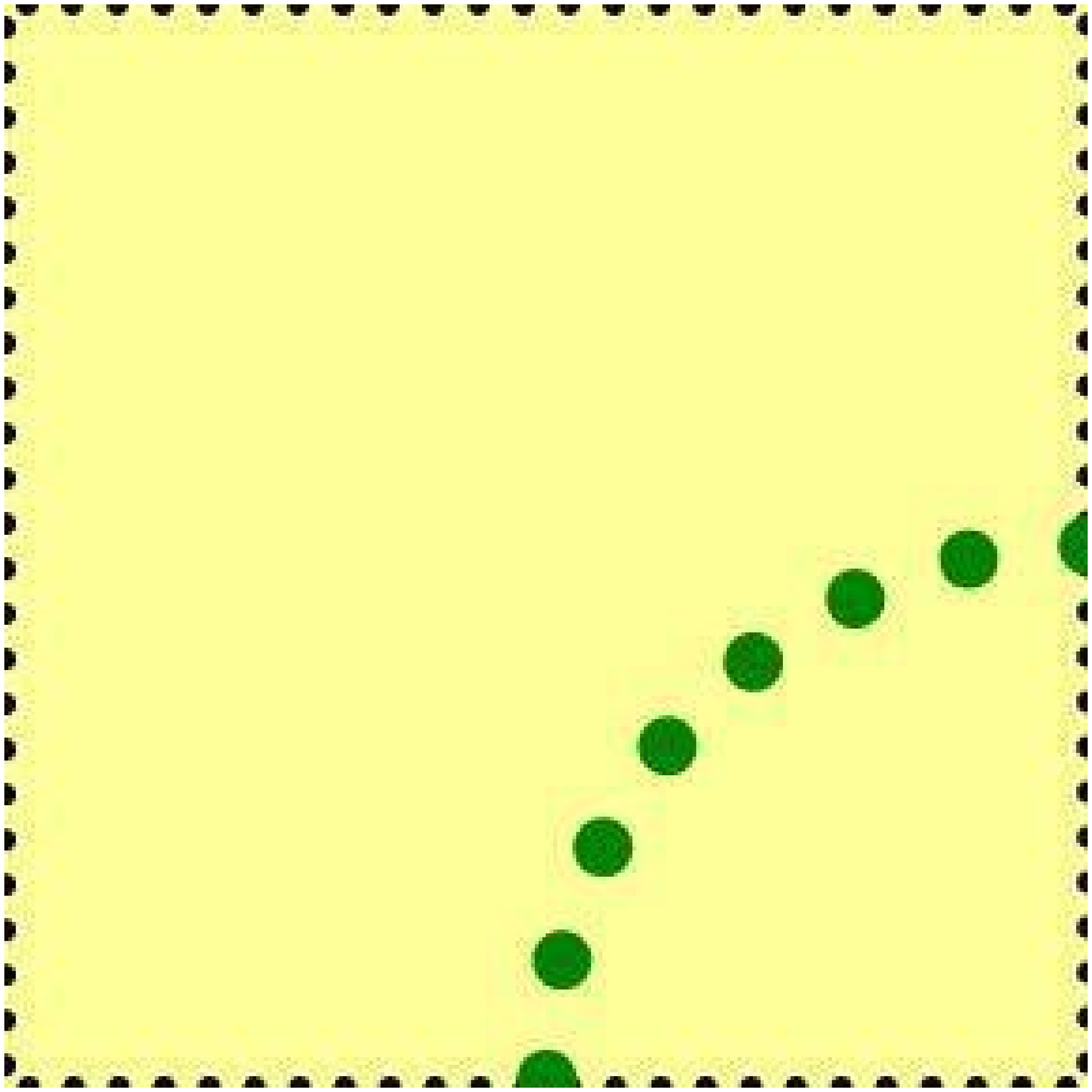}%
}%
%EndExpansion
, &
%TCIMACRO{\FRAME{itbpF}{0.3269in}{0.3269in}{0in}{}{}{iut03.ps}%
%width 0.3269in;  height 0.3269in;  depth 0in;  original-width 3in;
%original-height 3in;  cropleft "0";  croptop "1";  cropright "1";
%cropbottom "0";  filename 'iut03.ps';file-properties "XNPEU";}}}%
%BeginExpansion
{\includegraphics[
%natheight=3.000000in,
%natwidth=3.000000in,
height=0.3269in,
width=0.3269in
]%
{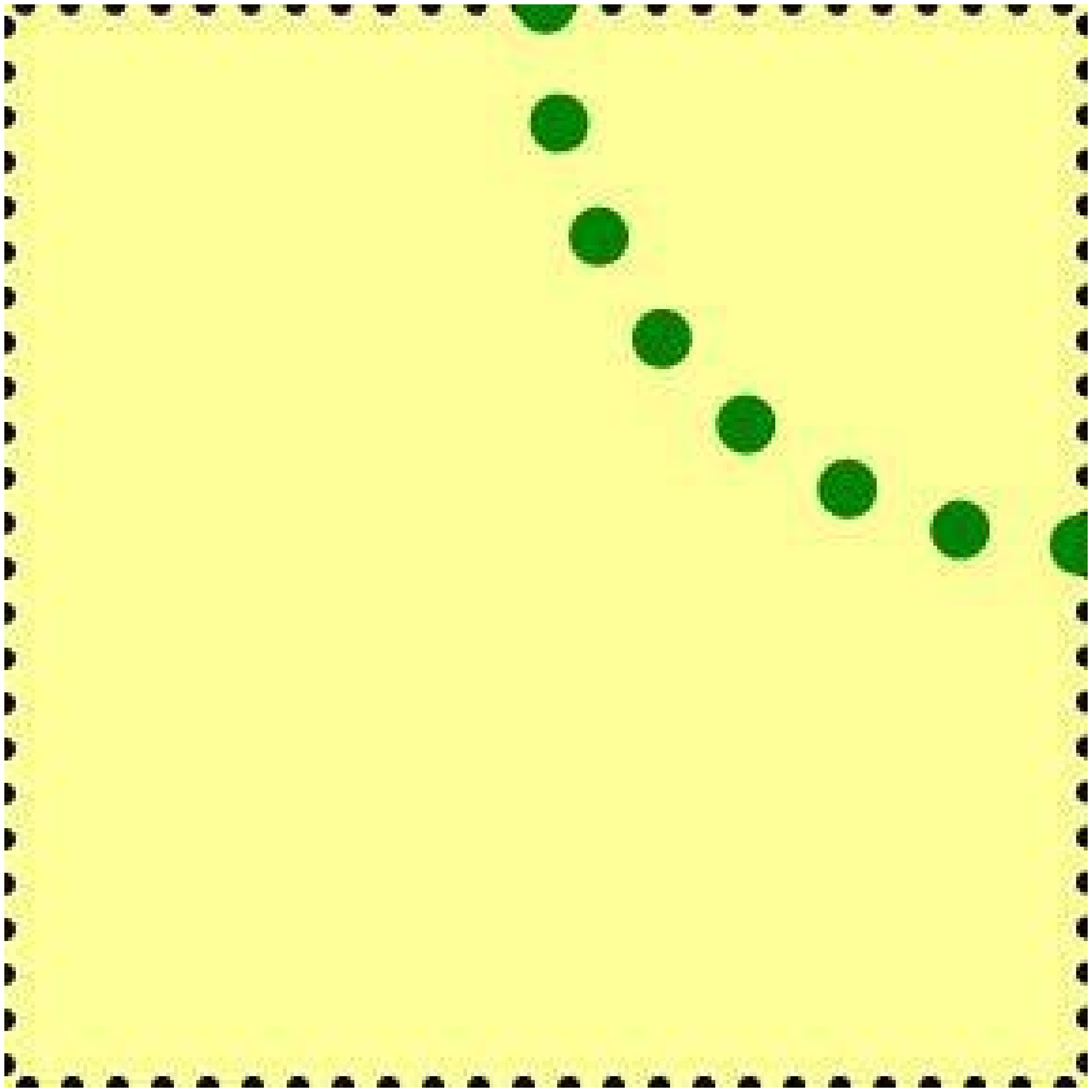}%
}%
%EndExpansion
, &
%TCIMACRO{\FRAME{itbpF}{0.3269in}{0.3269in}{0in}{}{}{iut04.ps}%
%width 0.3269in;  height 0.3269in;  depth 0in;  original-width 3in;
%original-height 3in;  cropleft "0";  croptop "1";  cropright "1";
%cropbottom "0";  filename 'iut04.ps';file-properties "XNPEU";}}}%
%BeginExpansion
{\includegraphics[
%natheight=3.000000in,
%natwidth=3.000000in,
height=0.3269in,
width=0.3269in
]%
{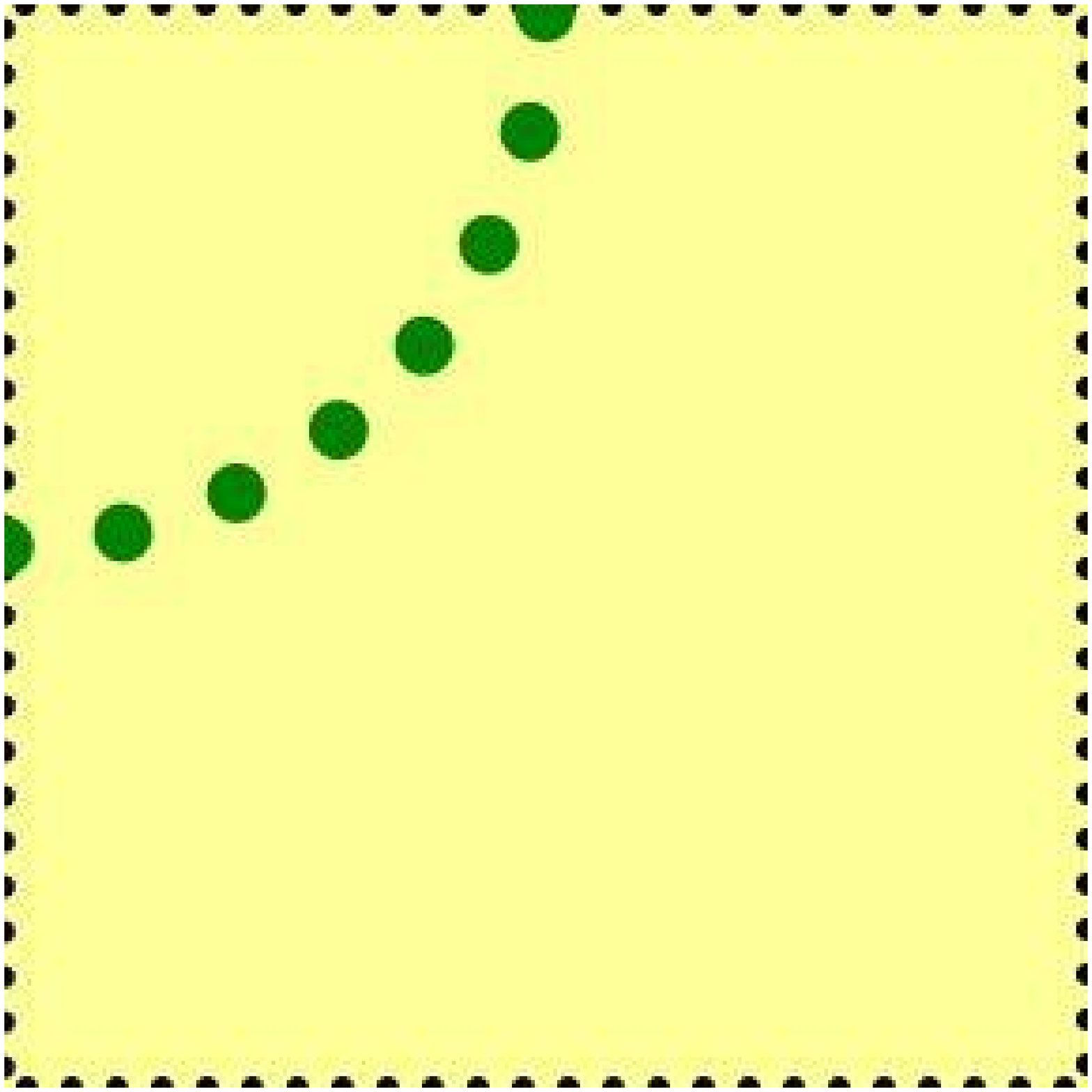}%
}%
%EndExpansion
, &
%TCIMACRO{\FRAME{itbpF}{0.3269in}{0.3269in}{0in}{}{}{iut05.ps}%
%width 0.3269in;  height 0.3269in;  depth 0in;  original-width 3in;
%original-height 3in;  cropleft "0";  croptop "1";  cropright "1";
%cropbottom "0";  filename 'iut05.ps';file-properties "XNPEU";}}}%
%BeginExpansion
{\includegraphics[
%natheight=3.000000in,
%natwidth=3.000000in,
height=0.3269in,
width=0.3269in
]%
{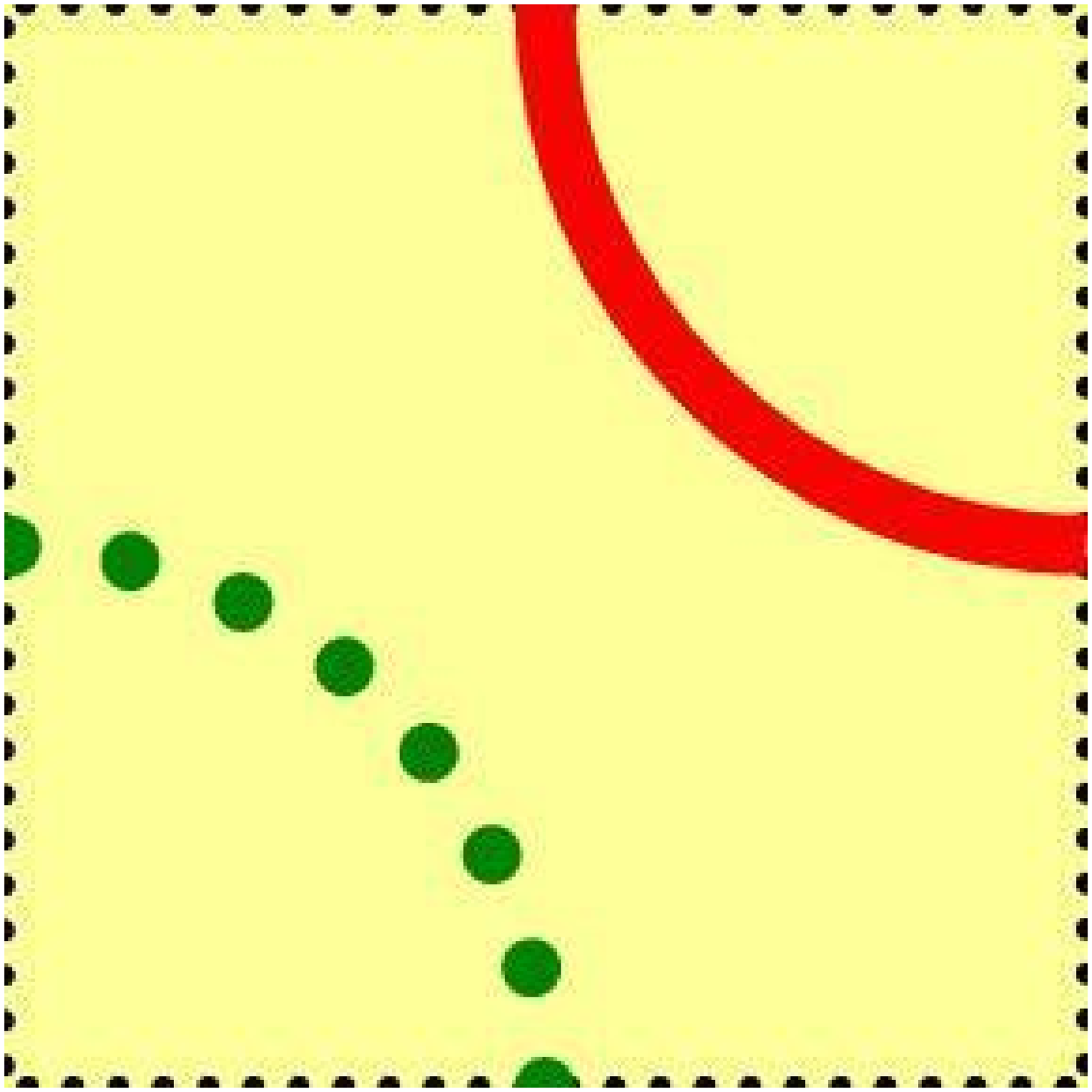}%
}%
%EndExpansion
, &
%TCIMACRO{\FRAME{itbpF}{0.3269in}{0.3269in}{0in}{}{}{iut06.ps}%
%width 0.3269in;  height 0.3269in;  depth 0in;  original-width 3in;
%original-height 3in;  cropleft "0";  croptop "1";  cropright "1";
%cropbottom "0";  filename 'iut06.ps';file-properties "XNPEU";}}}%
%BeginExpansion
{\includegraphics[
%natheight=3.000000in,
%natwidth=3.000000in,
height=0.3269in,
width=0.3269in
]%
{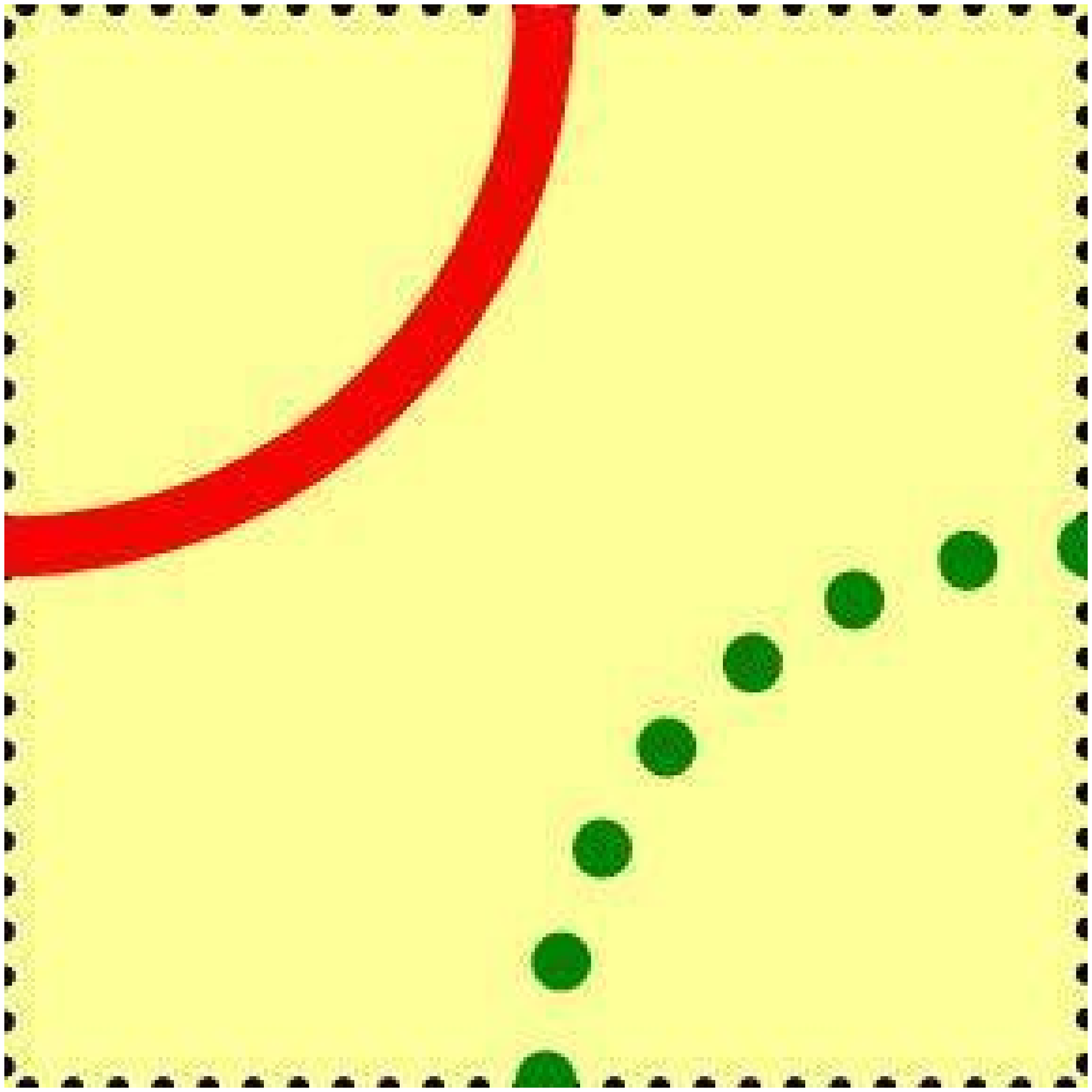}%
}%
%EndExpansion
, &
%TCIMACRO{\FRAME{itbpF}{0.3269in}{0.3269in}{0in}{}{}{iut07.ps}%
%width 0.3269in;  height 0.3269in;  depth 0in;  original-width 3in;
%original-height 3in;  cropleft "0";  croptop "1";  cropright "1";
%cropbottom "0";  filename 'iut07.ps';file-properties "XNPEU";}}}%
%BeginExpansion
{\includegraphics[
%natheight=3.000000in,
%natwidth=3.000000in,
height=0.3269in,
width=0.3269in
]%
{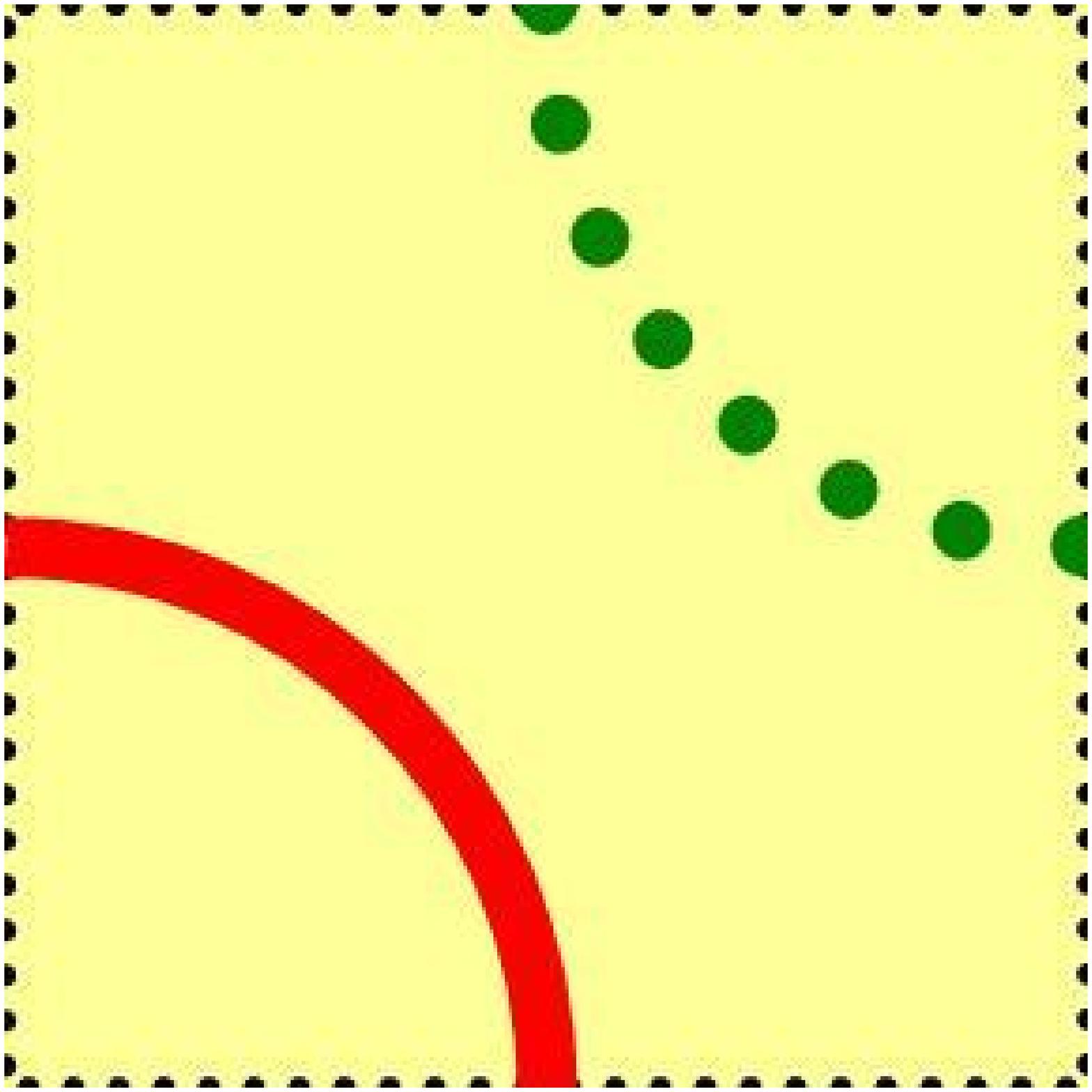}%
}%
%EndExpansion
, &
%TCIMACRO{\FRAME{itbpF}{0.3269in}{0.3269in}{0in}{}{}{iut08.ps}%
%width 0.3269in;  height 0.3269in;  depth 0in;  original-width 3in;
%original-height 3in;  cropleft "0";  croptop "1";  cropright "1";
%cropbottom "0";  filename 'iut08.ps';file-properties "XNPEU";}}}%
%BeginExpansion
{\includegraphics[
%natheight=3.000000in,
%natwidth=3.000000in,
height=0.3269in,
width=0.3269in
]%
{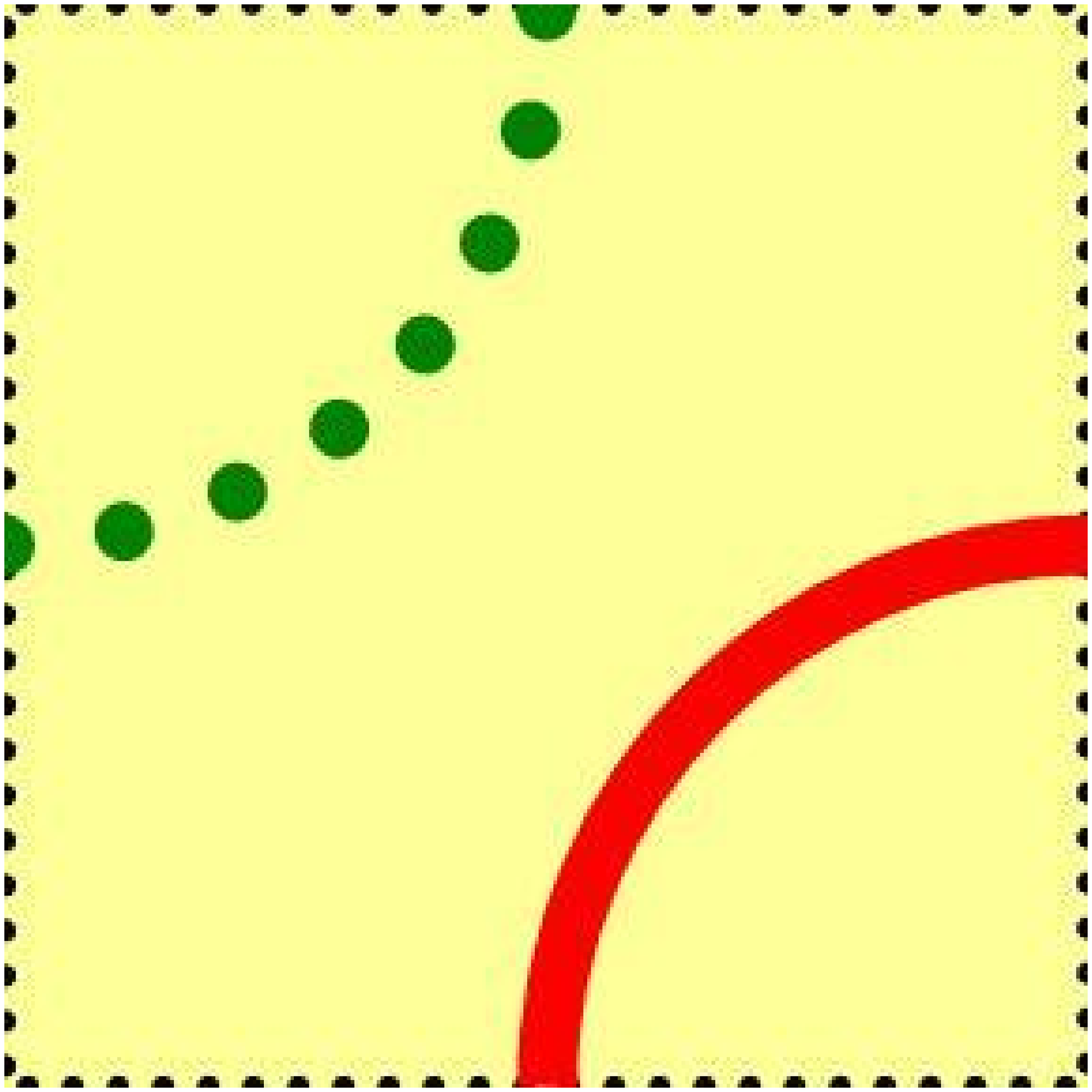}%
}%
%EndExpansion
, &
%TCIMACRO{\FRAME{itbpF}{0.3269in}{0.3269in}{0in}{}{}{iut09.ps}%
%width 0.3269in;  height 0.3269in;  depth 0in;  original-width 3in;
%original-height 3in;  cropleft "0";  croptop "1";  cropright "1";
%cropbottom "0";  filename 'iut09.ps';file-properties "XNPEU";}}}%
%BeginExpansion
{\includegraphics[
%natheight=3.000000in,
%natwidth=3.000000in,
height=0.3269in,
width=0.3269in
]%
{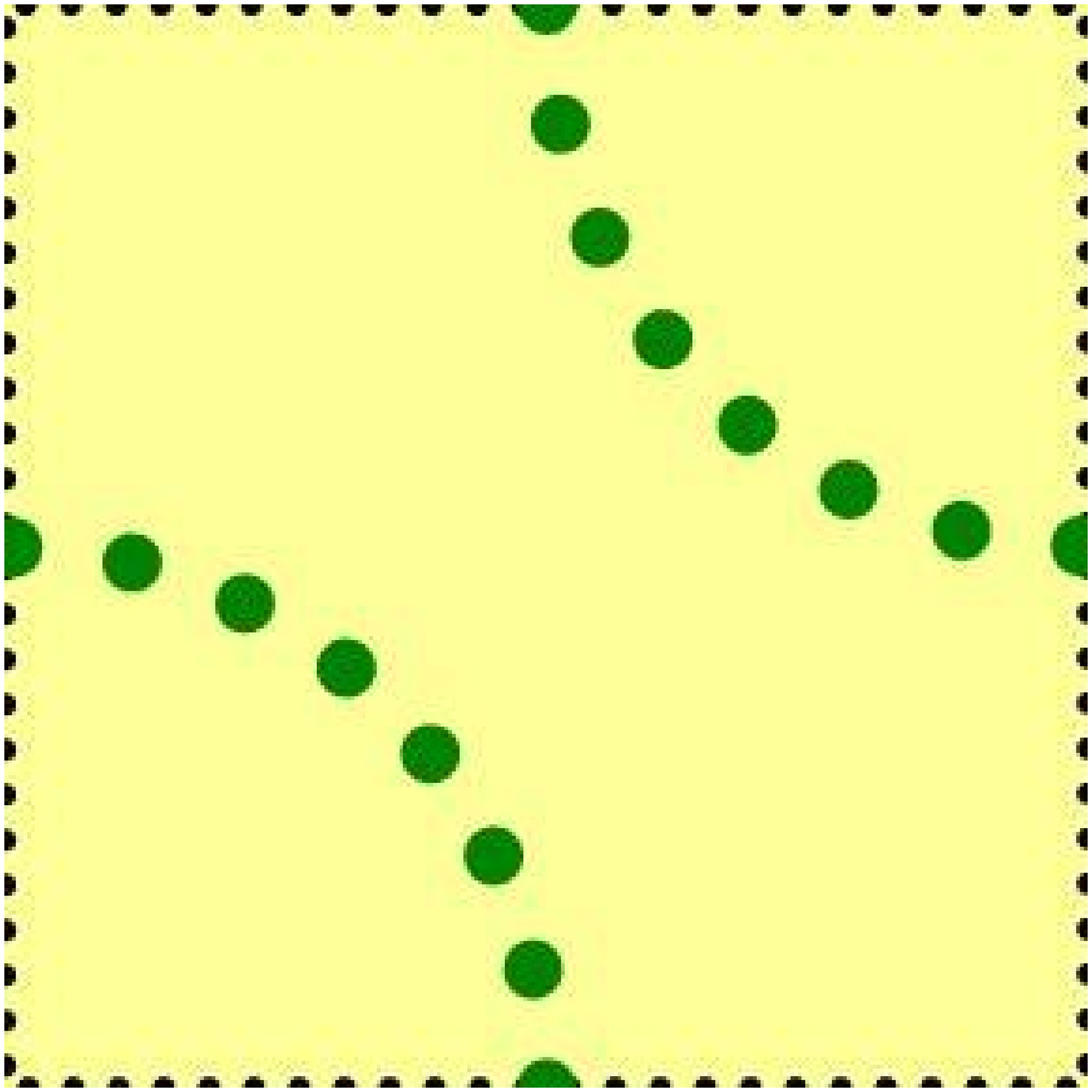}%
}%
%EndExpansion
, &
%TCIMACRO{\FRAME{itbpF}{0.3269in}{0.3269in}{0in}{}{}{iut10.ps}%
%width 0.3269in;  height 0.3269in;  depth 0in;  original-width 3in;
%original-height 3in;  cropleft "0";  croptop "1";  cropright "1";
%cropbottom "0";  filename 'iut10.ps';file-properties "XNPEU";}}}%
%BeginExpansion
{\includegraphics[
%natheight=3.000000in,
%natwidth=3.000000in,
height=0.3269in,
width=0.3269in
]%
{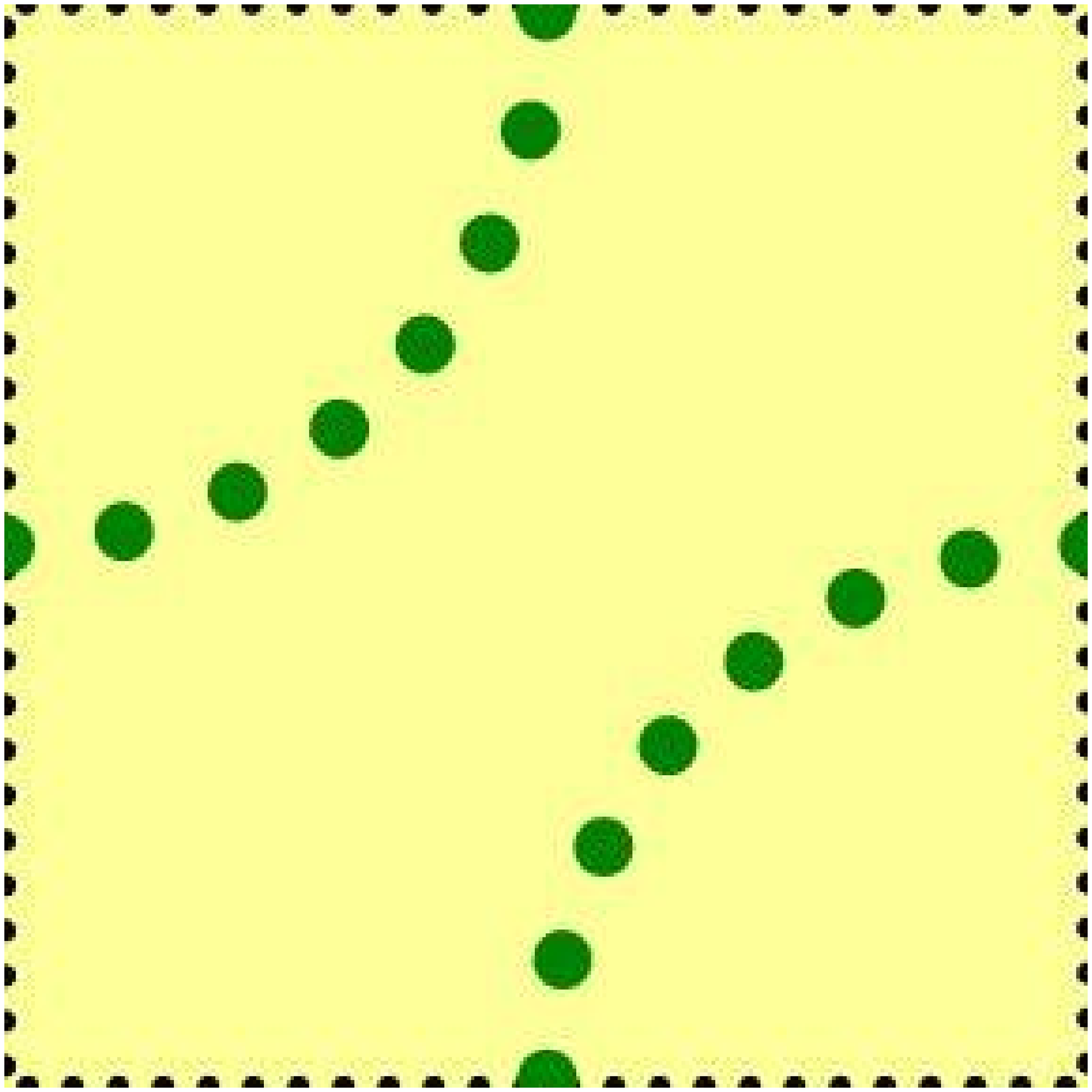}%
}%
%EndExpansion
\end{array}
\text{ ,}%
\]
\textit{called \textbf{nondeterministic tiles}, to denote either both or any
one (depending on context) of two possible tiles. }\ 

\bigskip

For example, the nondeterministic tile $%
%TCIMACRO{\FRAME{itbpF}{0.3269in}{0.3269in}{0.1003in}{}{}{iut07.ps}%
%width 0.3269in;  height 0.3269in;  depth 0.1003in;  original-width 3in;
%original-height 3in;  cropleft "0";  croptop "1";  cropright "1";
%cropbottom "0";  filename 'iut07.ps';file-properties "XNPEU";}}}%
%BeginExpansion
\raisebox{-0.1003in}{\includegraphics[
%natheight=3.000000in,
%natwidth=3.000000in,
height=0.3269in,
width=0.3269in
]%
{iut07.ps}%
}%
%EndExpansion
$ denotes either both or any one of the two tiles $%
%TCIMACRO{\FRAME{itbpF}{0.3269in}{0.3269in}{0.1003in}{}{}{ut01.ps}%
%width 0.3269in;  height 0.3269in;  depth 0.1003in;  original-width 3in;
%original-height 3in;  cropleft "0";  croptop "1";  cropright "1";
%cropbottom "0";  filename 'ut01.ps';file-properties "XNPEU";}}}%
%BeginExpansion
\raisebox{-0.1003in}{\includegraphics[
%natheight=3.000000in,
%natwidth=3.000000in,
height=0.3269in,
width=0.3269in
]%
{ut01.ps}%
}%
%EndExpansion
$ and$~%
%TCIMACRO{\FRAME{itbpF}{0.3269in}{0.3269in}{0.1003in}{}{}{ut07.ps}%
%width 0.3269in;  height 0.3269in;  depth 0.1003in;  original-width 3in;
%original-height 3in;  cropleft "0";  croptop "1";  cropright "1";
%cropbottom "0";  filename 'ut07.ps';file-properties "XNPEU";}}}%
%BeginExpansion
\raisebox{-0.1003in}{\includegraphics[
%natheight=3.000000in,
%natwidth=3.000000in,
height=0.3269in,
width=0.3269in
]%
{ut07.ps}%
}%
%EndExpansion
$.

\bigskip

\noindent\textbf{Notational Convention 2.} \textit{It is to be understood that
each mosaic move }$N\overset{\left(  i,j\right)  }{\longleftrightarrow
}N^{\prime}$ \textit{denotes either all or any one (depending on context) of
the moves obtained by simultaneously rotating }$N$ \textit{and} $N^{\prime}%
$\textit{ about their respective centers by }$0$\textit{, }$90$\textit{,
}$180$\textit{, or }$270$\textit{ degrees. }

\bigskip

For example,
\[%
\begin{array}
[c]{cc}%
%TCIMACRO{\FRAME{itbpF}{0.3269in}{0.3269in}{0in}{}{}{ut02.ps}%
%width 0.3269in;  height 0.3269in;  depth 0in;  original-width 3in;
%original-height 3in;  cropleft "0";  croptop "1";  cropright "1";
%cropbottom "0";  filename 'ut02.ps';file-properties "XNPEU";}}}%
%BeginExpansion
{\includegraphics[
%natheight=3.000000in,
%natwidth=3.000000in,
height=0.3269in,
width=0.3269in
]%
{ut02.ps}%
}%
%EndExpansion
&
%TCIMACRO{\FRAME{itbpF}{0.3269in}{0.3269in}{0in}{}{}{ut04.ps}%
%width 0.3269in;  height 0.3269in;  depth 0in;  original-width 3in;
%original-height 3in;  cropleft "0";  croptop "1";  cropright "1";
%cropbottom "0";  filename 'ut04.ps';file-properties "XNPEU";}}}%
%BeginExpansion
{\includegraphics[
%natheight=3.000000in,
%natwidth=3.000000in,
height=0.3269in,
width=0.3269in
]%
{ut04.ps}%
}%
%EndExpansion
\\
\end{array}
\overset{\left(  0,1\right)  }{\longleftrightarrow}%
\begin{array}
[c]{cc}%
%TCIMACRO{\FRAME{itbpF}{0.3269in}{0.3269in}{0in}{}{}{ut00.ps}%
%width 0.3269in;  height 0.3269in;  depth 0in;  original-width 3in;
%original-height 3in;  cropleft "0";  croptop "1";  cropright "1";
%cropbottom "0";  filename 'ut00.ps';file-properties "XNPEU";}}}%
%BeginExpansion
{\includegraphics[
%natheight=3.000000in,
%natwidth=3.000000in,
height=0.3269in,
width=0.3269in
]%
{ut00.ps}%
}%
%EndExpansion
&
%TCIMACRO{\FRAME{itbpF}{0.3269in}{0.3269in}{0in}{}{}{ut06.ps}%
%width 0.3269in;  height 0.3269in;  depth 0in;  original-width 3in;
%original-height 3in;  cropleft "0";  croptop "1";  cropright "1";
%cropbottom "0";  filename 'ut06.ps';file-properties "XNPEU";}}}%
%BeginExpansion
{\includegraphics[
%natheight=3.000000in,
%natwidth=3.000000in,
height=0.3269in,
width=0.3269in
]%
{ut06.ps}%
}%
%EndExpansion
\\
\end{array}
\]
represents either all or any one (depending on context) of the following four
2-moves:
\begin{align*}
&
\begin{array}
[c]{cc}%
%TCIMACRO{\FRAME{itbpF}{0.3269in}{0.3269in}{0in}{}{}{ut02.ps}%
%width 0.3269in;  height 0.3269in;  depth 0in;  original-width 3in;
%original-height 3in;  cropleft "0";  croptop "1";  cropright "1";
%cropbottom "0";  filename 'ut02.ps';file-properties "XNPEU";}}}%
%BeginExpansion
{\includegraphics[
%natheight=3.000000in,
%natwidth=3.000000in,
height=0.3269in,
width=0.3269in
]%
{ut02.ps}%
}%
%EndExpansion
&
%TCIMACRO{\FRAME{itbpF}{0.3269in}{0.3269in}{0in}{}{}{ut04.ps}%
%width 0.3269in;  height 0.3269in;  depth 0in;  original-width 3in;
%original-height 3in;  cropleft "0";  croptop "1";  cropright "1";
%cropbottom "0";  filename 'ut04.ps';file-properties "XNPEU";}}}%
%BeginExpansion
{\includegraphics[
%natheight=3.000000in,
%natwidth=3.000000in,
height=0.3269in,
width=0.3269in
]%
{ut04.ps}%
}%
%EndExpansion
\\
\end{array}
\overset{\left(  0,1\right)  }{\longleftrightarrow}%
\begin{array}
[c]{cc}%
%TCIMACRO{\FRAME{itbpF}{0.3269in}{0.3269in}{0in}{}{}{ut00.ps}%
%width 0.3269in;  height 0.3269in;  depth 0in;  original-width 3in;
%original-height 3in;  cropleft "0";  croptop "1";  cropright "1";
%cropbottom "0";  filename 'ut00.ps';file-properties "XNPEU";}}}%
%BeginExpansion
{\includegraphics[
%natheight=3.000000in,
%natwidth=3.000000in,
height=0.3269in,
width=0.3269in
]%
{ut00.ps}%
}%
%EndExpansion
&
%TCIMACRO{\FRAME{itbpF}{0.3269in}{0.3269in}{0in}{}{}{ut06.ps}%
%width 0.3269in;  height 0.3269in;  depth 0in;  original-width 3in;
%original-height 3in;  cropleft "0";  croptop "1";  cropright "1";
%cropbottom "0";  filename 'ut06.ps';file-properties "XNPEU";}}}%
%BeginExpansion
{\includegraphics[
%natheight=3.000000in,
%natwidth=3.000000in,
height=0.3269in,
width=0.3269in
]%
{ut06.ps}%
}%
%EndExpansion
\\
\end{array}
\qquad\qquad\qquad%
\begin{array}
[c]{cc}%
%TCIMACRO{\FRAME{itbpF}{0.3269in}{0.3269in}{0in}{}{}{ut03.ps}%
%width 0.3269in;  height 0.3269in;  depth 0in;  original-width 3in;
%original-height 3in;  cropleft "0";  croptop "1";  cropright "1";
%cropbottom "0";  filename 'ut03.ps';file-properties "XNPEU";}}}%
%BeginExpansion
{\includegraphics[
%natheight=3.000000in,
%natwidth=3.000000in,
height=0.3269in,
width=0.3269in
]%
{ut03.ps}%
}%
%EndExpansion
&
%TCIMACRO{\FRAME{itbpF}{0.3269in}{0.3269in}{0in}{}{}{ut01.ps}%
%width 0.3269in;  height 0.3269in;  depth 0in;  original-width 3in;
%original-height 3in;  cropleft "0";  croptop "1";  cropright "1";
%cropbottom "0";  filename 'ut01.ps';file-properties "XNPEU";}}}%
%BeginExpansion
{\includegraphics[
%natheight=3.000000in,
%natwidth=3.000000in,
height=0.3269in,
width=0.3269in
]%
{ut01.ps}%
}%
%EndExpansion
\\
\end{array}
\overset{\left(  0,1\right)  }{\longleftrightarrow}%
\begin{array}
[c]{cc}%
%TCIMACRO{\FRAME{itbpF}{0.3269in}{0.3269in}{0in}{}{}{ut06.ps}%
%width 0.3269in;  height 0.3269in;  depth 0in;  original-width 3in;
%original-height 3in;  cropleft "0";  croptop "1";  cropright "1";
%cropbottom "0";  filename 'ut06.ps';file-properties "XNPEU";}}}%
%BeginExpansion
{\includegraphics[
%natheight=3.000000in,
%natwidth=3.000000in,
height=0.3269in,
width=0.3269in
]%
{ut06.ps}%
}%
%EndExpansion
&
%TCIMACRO{\FRAME{itbpF}{0.3269in}{0.3269in}{0in}{}{}{ut00.ps}%
%width 0.3269in;  height 0.3269in;  depth 0in;  original-width 3in;
%original-height 3in;  cropleft "0";  croptop "1";  cropright "1";
%cropbottom "0";  filename 'ut00.ps';file-properties "XNPEU";}}}%
%BeginExpansion
{\includegraphics[
%natheight=3.000000in,
%natwidth=3.000000in,
height=0.3269in,
width=0.3269in
]%
{ut00.ps}%
}%
%EndExpansion
\\
\end{array}
\\
& \\
&
\begin{array}
[c]{cc}%
%TCIMACRO{\FRAME{itbpF}{0.3269in}{0.3269in}{0in}{}{}{ut00.ps}%
%width 0.3269in;  height 0.3269in;  depth 0in;  original-width 3in;
%original-height 3in;  cropleft "0";  croptop "1";  cropright "1";
%cropbottom "0";  filename 'ut00.ps';file-properties "XNPEU";}}}%
%BeginExpansion
{\includegraphics[
%natheight=3.000000in,
%natwidth=3.000000in,
height=0.3269in,
width=0.3269in
]%
{ut00.ps}%
}%
%EndExpansion
&
%TCIMACRO{\FRAME{itbpF}{0.3269in}{0.3269in}{0in}{}{}{ut02.ps}%
%width 0.3269in;  height 0.3269in;  depth 0in;  original-width 3in;
%original-height 3in;  cropleft "0";  croptop "1";  cropright "1";
%cropbottom "0";  filename 'ut02.ps';file-properties "XNPEU";}}}%
%BeginExpansion
{\includegraphics[
%natheight=3.000000in,
%natwidth=3.000000in,
height=0.3269in,
width=0.3269in
]%
{ut02.ps}%
}%
%EndExpansion
\\
\end{array}
\overset{\left(  0,1\right)  }{\longleftrightarrow}%
\begin{array}
[c]{cc}%
%TCIMACRO{\FRAME{itbpF}{0.3269in}{0.3269in}{0in}{}{}{ut02.ps}%
%width 0.3269in;  height 0.3269in;  depth 0in;  original-width 3in;
%original-height 3in;  cropleft "0";  croptop "1";  cropright "1";
%cropbottom "0";  filename 'ut02.ps';file-properties "XNPEU";}}}%
%BeginExpansion
{\includegraphics[
%natheight=3.000000in,
%natwidth=3.000000in,
height=0.3269in,
width=0.3269in
]%
{ut02.ps}%
}%
%EndExpansion
&
%TCIMACRO{\FRAME{itbpF}{0.3269in}{0.3269in}{0in}{}{}{ut05.ps}%
%width 0.3269in;  height 0.3269in;  depth 0in;  original-width 3in;
%original-height 3in;  cropleft "0";  croptop "1";  cropright "1";
%cropbottom "0";  filename 'ut05.ps';file-properties "XNPEU";}}}%
%BeginExpansion
{\includegraphics[
%natheight=3.000000in,
%natwidth=3.000000in,
height=0.3269in,
width=0.3269in
]%
{ut05.ps}%
}%
%EndExpansion
\\
\end{array}
\qquad\qquad\qquad%
\begin{array}
[c]{cc}%
%TCIMACRO{\FRAME{itbpF}{0.3269in}{0.3269in}{0in}{}{}{ut01.ps}%
%width 0.3269in;  height 0.3269in;  depth 0in;  original-width 3in;
%original-height 3in;  cropleft "0";  croptop "1";  cropright "1";
%cropbottom "0";  filename 'ut01.ps';file-properties "XNPEU";}}}%
%BeginExpansion
{\includegraphics[
%natheight=3.000000in,
%natwidth=3.000000in,
height=0.3269in,
width=0.3269in
]%
{ut01.ps}%
}%
%EndExpansion
&
%TCIMACRO{\FRAME{itbpF}{0.3269in}{0.3269in}{0in}{}{}{ut00.ps}%
%width 0.3269in;  height 0.3269in;  depth 0in;  original-width 3in;
%original-height 3in;  cropleft "0";  croptop "1";  cropright "1";
%cropbottom "0";  filename 'ut00.ps';file-properties "XNPEU";}}}%
%BeginExpansion
{\includegraphics[
%natheight=3.000000in,
%natwidth=3.000000in,
height=0.3269in,
width=0.3269in
]%
{ut00.ps}%
}%
%EndExpansion
\\
\end{array}
\overset{\left(  0,1\right)  }{\longleftrightarrow}%
\begin{array}
[c]{cc}%
%TCIMACRO{\FRAME{itbpF}{0.3269in}{0.3269in}{0in}{}{}{ut05.ps}%
%width 0.3269in;  height 0.3269in;  depth 0in;  original-width 3in;
%original-height 3in;  cropleft "0";  croptop "1";  cropright "1";
%cropbottom "0";  filename 'ut05.ps';file-properties "XNPEU";}}}%
%BeginExpansion
{\includegraphics[
%natheight=3.000000in,
%natwidth=3.000000in,
height=0.3269in,
width=0.3269in
]%
{ut05.ps}%
}%
%EndExpansion
&
%TCIMACRO{\FRAME{itbpF}{0.3269in}{0.3269in}{0in}{}{}{ut01.ps}%
%width 0.3269in;  height 0.3269in;  depth 0in;  original-width 3in;
%original-height 3in;  cropleft "0";  croptop "1";  cropright "1";
%cropbottom "0";  filename 'ut01.ps';file-properties "XNPEU";}}}%
%BeginExpansion
{\includegraphics[
%natheight=3.000000in,
%natwidth=3.000000in,
height=0.3269in,
width=0.3269in
]%
{ut01.ps}%
}%
%EndExpansion
\\
\end{array}
\end{align*}

\bigskip

As our final notational convention, we have:

\bigskip

\noindent\textbf{Notational Convention 3.} \textit{Finally, we omit the
location superscript }$\left(  i,j\right)  $, and write $N\longleftrightarrow
N^{\prime}$ to denote either all or any one (depending on context)\ of the
possible locations.

\bigskip

\bigskip

\noindent\textbf{Caveat:} \textit{We caution the reader that throughout the
remainder of this paper, we will be using all of the above nondeterministic
notational conventions.}

\bigskip

\subsection{The planar isotopy moves on knot mosaics}

\bigskip

As an analog to the planar isotopy moves for standard knot diagrams, we define
for mosaics the 11 \textbf{mosaic planar isotopy moves} given below:%
\[
\underset{P_{1}}{%

\ \ \ }%
\]

\bigskip

The above set of 11 planar isotopy moves was found by an exhaustive
enumeration of all 2-mosaic moves corresponding to topological planar isotopy
moves. \ The completeness of this set of moves, i.e., that every planar
isotopy moves for mosaics is a composition of a finite sequence of the above
planar isotopy moves, is addressed in section 2.7 of this paper.

\bigskip

\subsection{The Reidemeister moves on knot mosaics}

\bigskip

As an analog to the Reidemeister moves for standard knot diagrams, we create
for mosaics the \textbf{mosaic Reidemeister moves}. \ 

\bigskip

The \textbf{mosaic Reidemeister 1 moves} are the following:

\bigskip%

\[
\underset{R_{1}}{%

\ \ \ }%
\]

\bigskip

And the \textbf{mosaic Reidemeister 2 moves} are given below:

\bigskip%
\begin{align*}
&  \underset{R_{2}}{%
%
\ \ \ }%
\end{align*}

\bigskip

For describing the mosaic Reidemeister 3 moves, we will use for simplicity of
exposition the following two additional notational conventions:

\bigskip

\noindent\textbf{Notational Convention 4.} \textit{We will make use of each of
the following tiles}%
\[%
%TCIMACRO{\FRAME{itbpF}{0.3269in}{0.3269in}{0in}{}{}{ul.ps}%
%width 0.3269in;  height 0.3269in;  depth 0in;  original-width 3in;
%original-height 3in;  cropleft "0";  croptop "1";  cropright "1";
%cropbottom "0";  filename 'ul.ps';file-properties "XNPEU";}}}%
%BeginExpansion
{\includegraphics[
%natheight=3.000000in,
%natwidth=3.000000in,
height=0.3269in,
width=0.3269in
]%
{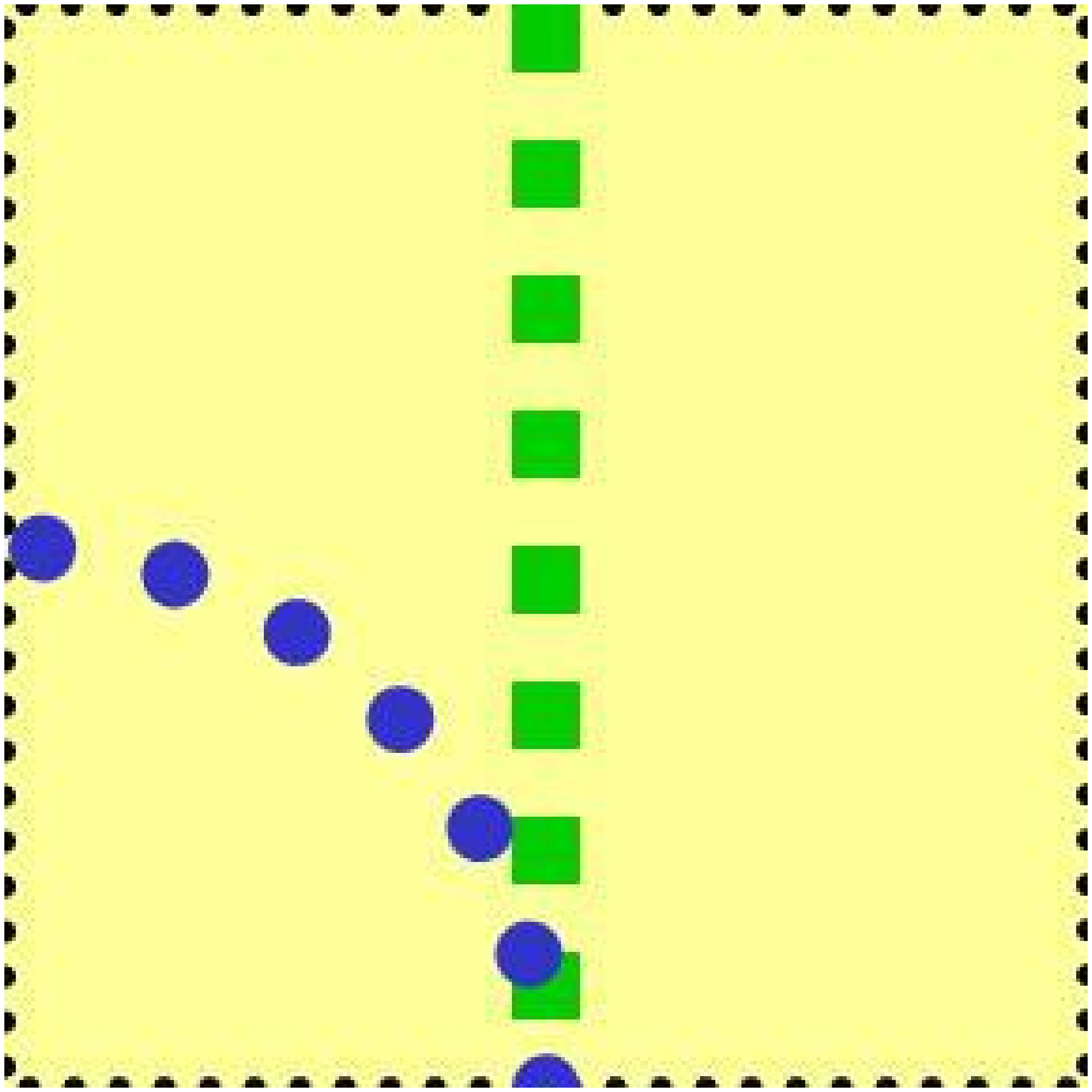}%
}%
%EndExpansion
\quad%
%TCIMACRO{\FRAME{itbpF}{0.3269in}{0.3269in}{0in}{}{}{ur.ps}%
%width 0.3269in;  height 0.3269in;  depth 0in;  original-width 3in;
%original-height 3in;  cropleft "0";  croptop "1";  cropright "1";
%cropbottom "0";  filename 'ur.ps';file-properties "XNPEU";}}}%
%BeginExpansion
{\includegraphics[
%natheight=3.000000in,
%natwidth=3.000000in,
height=0.3269in,
width=0.3269in
]%
{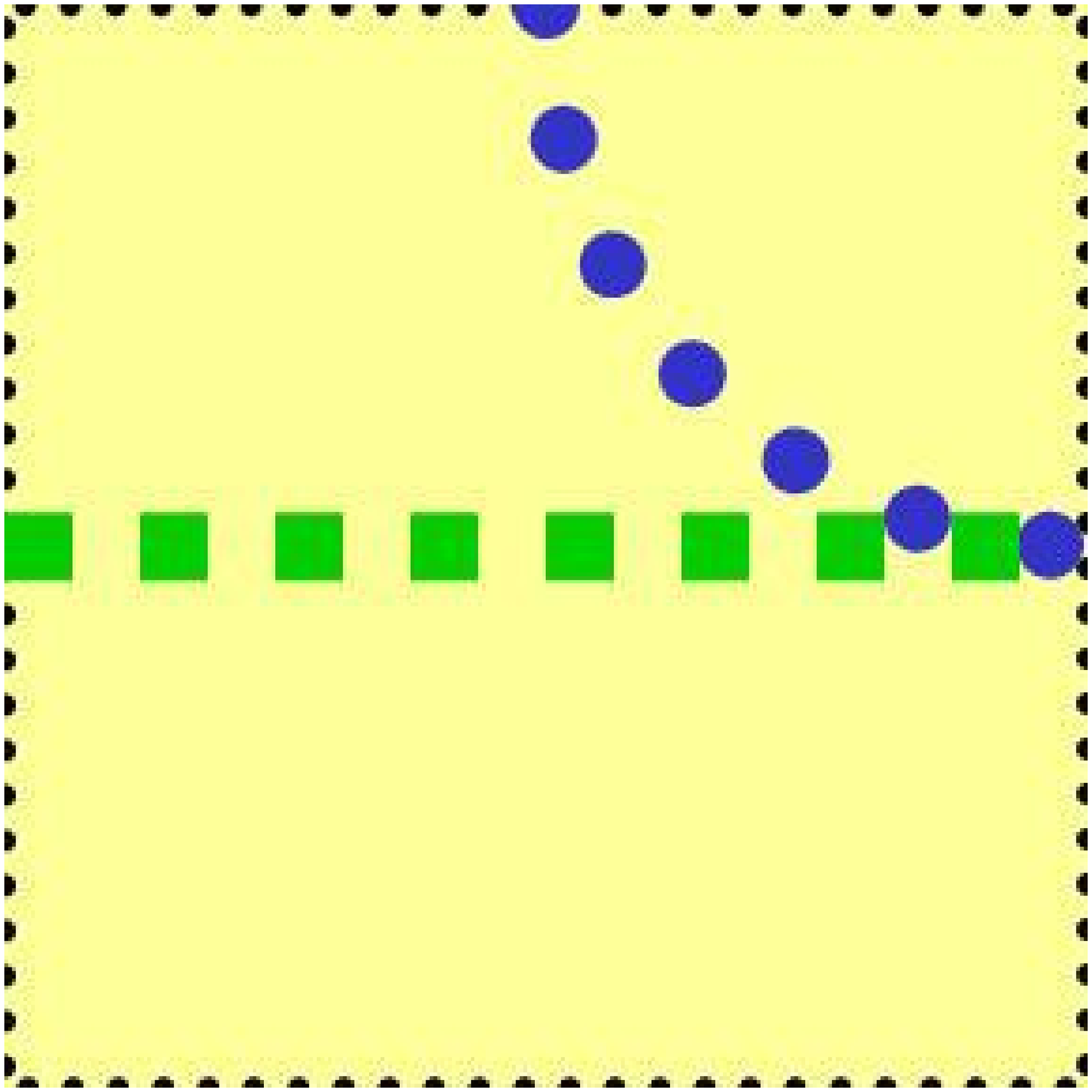}%
}%
%EndExpansion
\quad%
%TCIMACRO{\FRAME{itbpF}{0.3269in}{0.3269in}{0in}{}{}{ll.ps}%
%width 0.3269in;  height 0.3269in;  depth 0in;  original-width 3in;
%original-height 3in;  cropleft "0";  croptop "1";  cropright "1";
%cropbottom "0";  filename 'll.ps';file-properties "XNPEU";}}}%
%BeginExpansion
{\includegraphics[
%natheight=3.000000in,
%natwidth=3.000000in,
height=0.3269in,
width=0.3269in
]%
{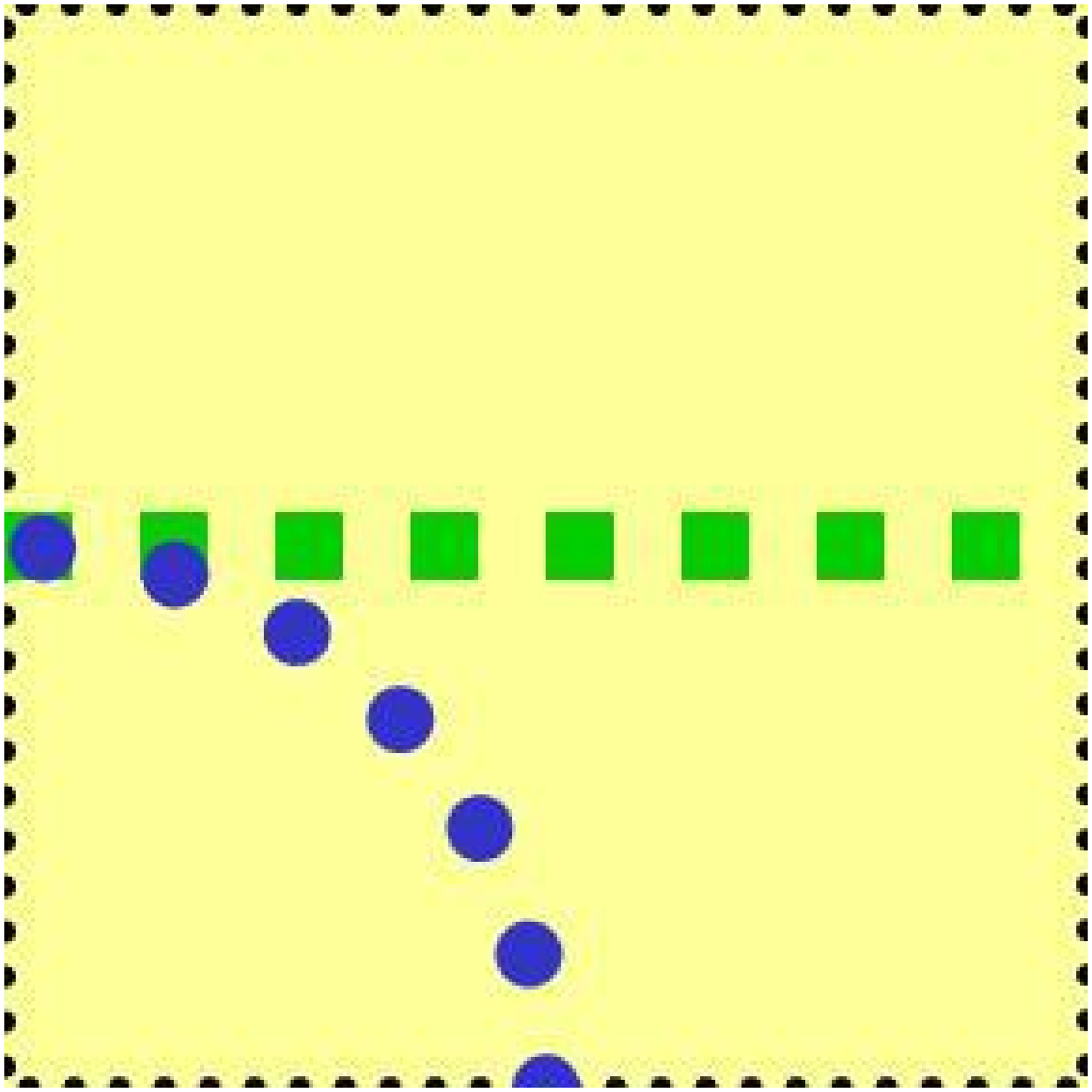}%
}%
%EndExpansion
\quad%
%TCIMACRO{\FRAME{itbpF}{0.3269in}{0.3269in}{0in}{}{}{lr.ps}%
%width 0.3269in;  height 0.3269in;  depth 0in;  original-width 3in;
%original-height 3in;  cropleft "0";  croptop "1";  cropright "1";
%cropbottom "0";  filename 'lr.ps';file-properties "XNPEU";}}}%
%BeginExpansion
{\includegraphics[
%natheight=3.000000in,
%natwidth=3.000000in,
height=0.3269in,
width=0.3269in
]%
{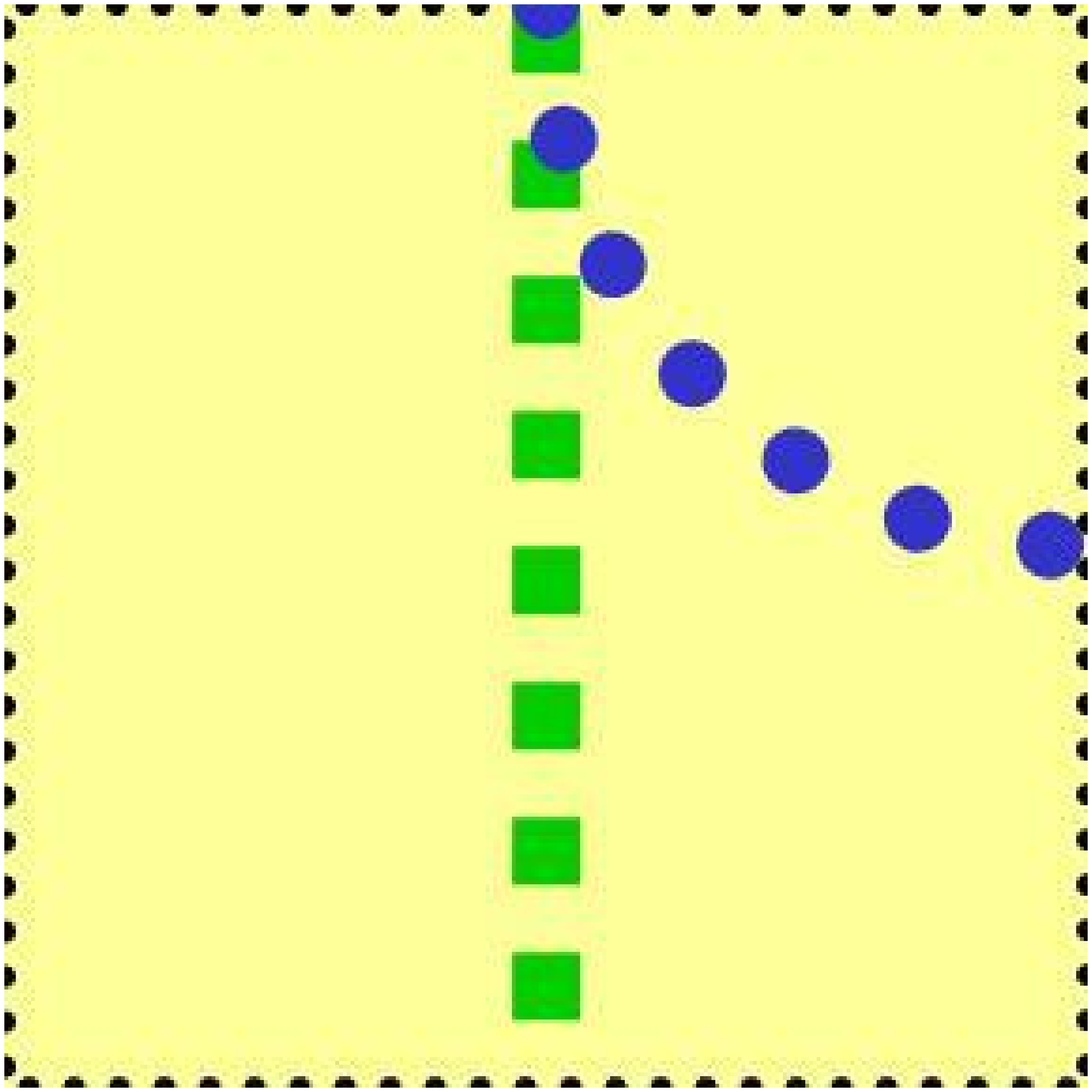}%
}%
%EndExpansion
\text{ ,}%
\]
\textit{also called \textbf{nondeterministic tiles}, to denote either one of
two possible tiles\footnote{Please note that each of these newly introduced
non-deterministic tiles denotes one of two possible deterministic tiles. On
the other hand, the non-deterministic tiles introduced in section 2.3 denote
one or all of two possible deterministic tiles, depending on context..}.} \ 

\bigskip

For example, the nondeterministic tile $%
%TCIMACRO{\FRAME{itbpF}{0.3269in}{0.3269in}{0.1003in}{}{}{ul.ps}%
%width 0.3269in;  height 0.3269in;  depth 0.1003in;  original-width 3in;
%original-height 3in;  cropleft "0";  croptop "1";  cropright "1";
%cropbottom "0";  filename 'ul.ps';file-properties "XNPEU";}}}%
%BeginExpansion
\raisebox{-0.1003in}{\includegraphics[
%natheight=3.000000in,
%natwidth=3.000000in,
height=0.3269in,
width=0.3269in
]%
{ul.ps}%
}%
%EndExpansion
$ denotes either of the following two tiles%
\[%
%TCIMACRO{\FRAME{itbpF}{0.3269in}{0.3269in}{0.1003in}{}{}{ul.ps}%
%width 0.3269in;  height 0.3269in;  depth 0.1003in;  original-width 3in;
%original-height 3in;  cropleft "0";  croptop "1";  cropright "1";
%cropbottom "0";  filename 'ul.ps';file-properties "XNPEU";}}}%
%BeginExpansion
\raisebox{-0.1003in}{\includegraphics[
%natheight=3.000000in,
%natwidth=3.000000in,
height=0.3269in,
width=0.3269in
]%
{ul.ps}%
}%
%EndExpansion
=%
%TCIMACRO{\FRAME{itbpF}{0.3269in}{0.3269in}{0.1003in}{}{}{ut01.ps}%
%width 0.3269in;  height 0.3269in;  depth 0.1003in;  original-width 3in;
%original-height 3in;  cropleft "0";  croptop "1";  cropright "1";
%cropbottom "0";  filename 'ut01.ps';file-properties "XNPEU";}}}%
%BeginExpansion
\raisebox{-0.1003in}{\includegraphics[
%natheight=3.000000in,
%natwidth=3.000000in,
height=0.3269in,
width=0.3269in
]%
{ut01.ps}%
}%
%EndExpansion
\text{ \ or \ \ }%
%TCIMACRO{\FRAME{itbpF}{0.3269in}{0.3269in}{0.1003in}{}{}{ut06.ps}%
%width 0.3269in;  height 0.3269in;  depth 0.1003in;  original-width 3in;
%original-height 3in;  cropleft "0";  croptop "1";  cropright "1";
%cropbottom "0";  filename 'ut06.ps';file-properties "XNPEU";}}}%
%BeginExpansion
\raisebox{-0.1003in}{\includegraphics[
%natheight=3.000000in,
%natwidth=3.000000in,
height=0.3269in,
width=0.3269in
]%
{ut06.ps}%
}%
%EndExpansion
\text{ .}%
\]

\bigskip

\noindent\textbf{Notational Convention 5.} \textit{Nondeterministic tiles
labeled by the same letter are synchronized as follows:}
\[
\left\{
\begin{array}
[c]{c}%
%TCIMACRO{\FRAME{itbpF}{0.3269in}{0.3269in}{0.1003in}{}{}{ula.ps}%
%width 0.3269in;  height 0.3269in;  depth 0.1003in;  original-width 3in;
%original-height 3in;  cropleft "0";  croptop "1";  cropright "1";
%cropbottom "0";  filename 'ula.ps';file-properties "XNPEU";}}}%
%BeginExpansion
\raisebox{-0.1003in}{\includegraphics[
%natheight=3.000000in,
%natwidth=3.000000in,
height=0.3269in,
width=0.3269in
]%
{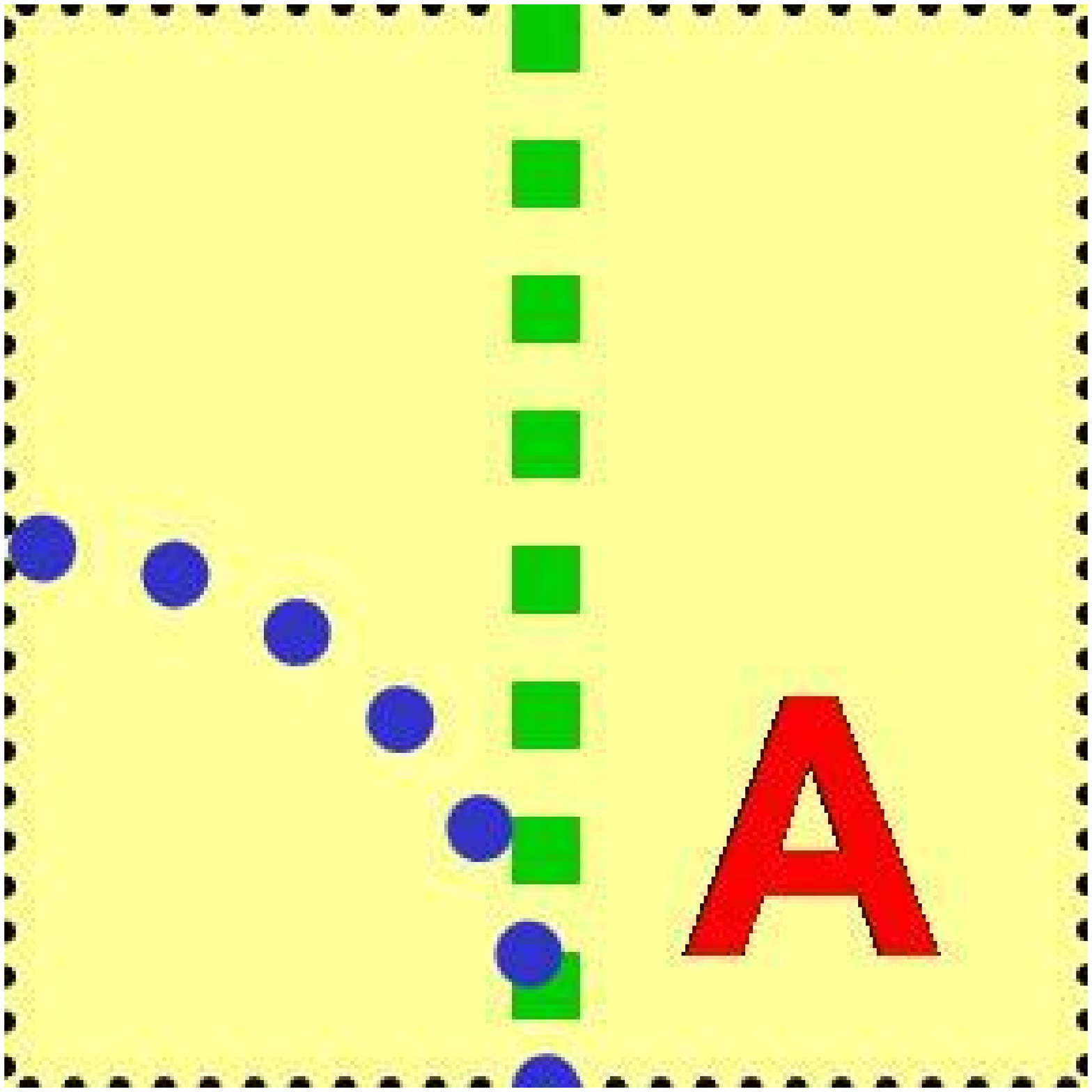}%
}%
%EndExpansion
=%
%TCIMACRO{\FRAME{itbpF}{0.3269in}{0.3269in}{0.1003in}{}{}{ut01.ps}%
%width 0.3269in;  height 0.3269in;  depth 0.1003in;  original-width 3in;
%original-height 3in;  cropleft "0";  croptop "1";  cropright "1";
%cropbottom "0";  filename 'ut01.ps';file-properties "XNPEU";}}}%
%BeginExpansion
\raisebox{-0.1003in}{\includegraphics[
%natheight=3.000000in,
%natwidth=3.000000in,
height=0.3269in,
width=0.3269in
]%
{ut01.ps}%
}%
%EndExpansion
\Longleftrightarrow%
%TCIMACRO{\FRAME{itbpF}{0.3269in}{0.3269in}{0.1003in}{}{}{ura.ps}%
%width 0.3269in;  height 0.3269in;  depth 0.1003in;  original-width 3in;
%original-height 3in;  cropleft "0";  croptop "1";  cropright "1";
%cropbottom "0";  filename 'ura.ps';file-properties "XNPEU";}}}%
%BeginExpansion
\raisebox{-0.1003in}{\includegraphics[
%natheight=3.000000in,
%natwidth=3.000000in,
height=0.3269in,
width=0.3269in
]%
{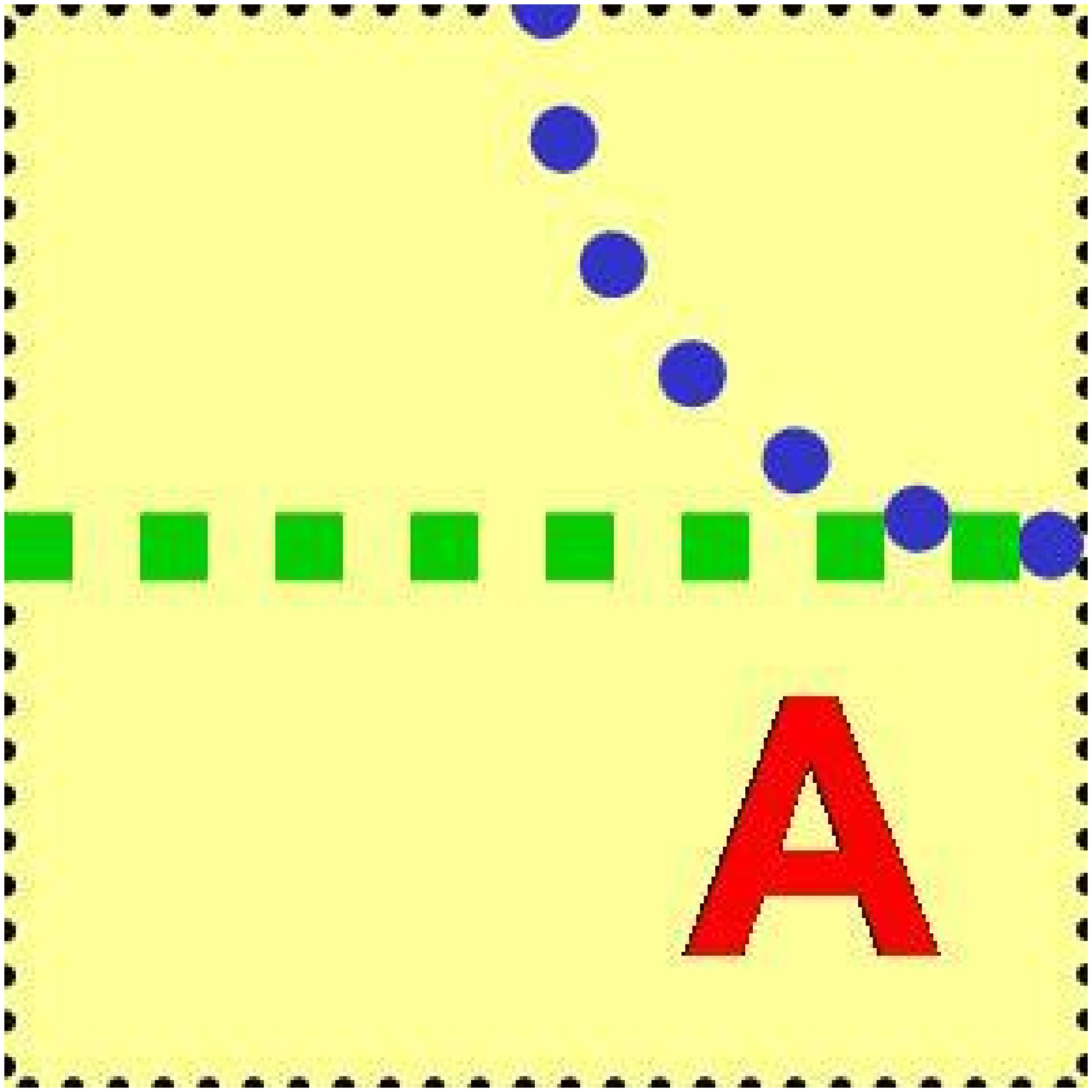}%
}%
%EndExpansion
=%
%TCIMACRO{\FRAME{itbpF}{0.3269in}{0.3269in}{0.1003in}{}{}{ut05.ps}%
%width 0.3269in;  height 0.3269in;  depth 0.1003in;  original-width 3in;
%original-height 3in;  cropleft "0";  croptop "1";  cropright "1";
%cropbottom "0";  filename 'ut05.ps';file-properties "XNPEU";}}}%
%BeginExpansion
\raisebox{-0.1003in}{\includegraphics[
%natheight=3.000000in,
%natwidth=3.000000in,
height=0.3269in,
width=0.3269in
]%
{ut05.ps}%
}%
%EndExpansion
\\
\end{array}
\right.  \text{ \ \ and \ \ }\left\{
\begin{array}
[c]{c}%
%TCIMACRO{\FRAME{itbpF}{0.3269in}{0.3269in}{0.1003in}{}{}{llb.ps}%
%width 0.3269in;  height 0.3269in;  depth 0.1003in;  original-width 3in;
%original-height 3in;  cropleft "0";  croptop "1";  cropright "1";
%cropbottom "0";  filename 'llb.ps';file-properties "XNPEU";}}}%
%BeginExpansion
\raisebox{-0.1003in}{\includegraphics[
%natheight=3.000000in,
%natwidth=3.000000in,
height=0.3269in,
width=0.3269in
]%
{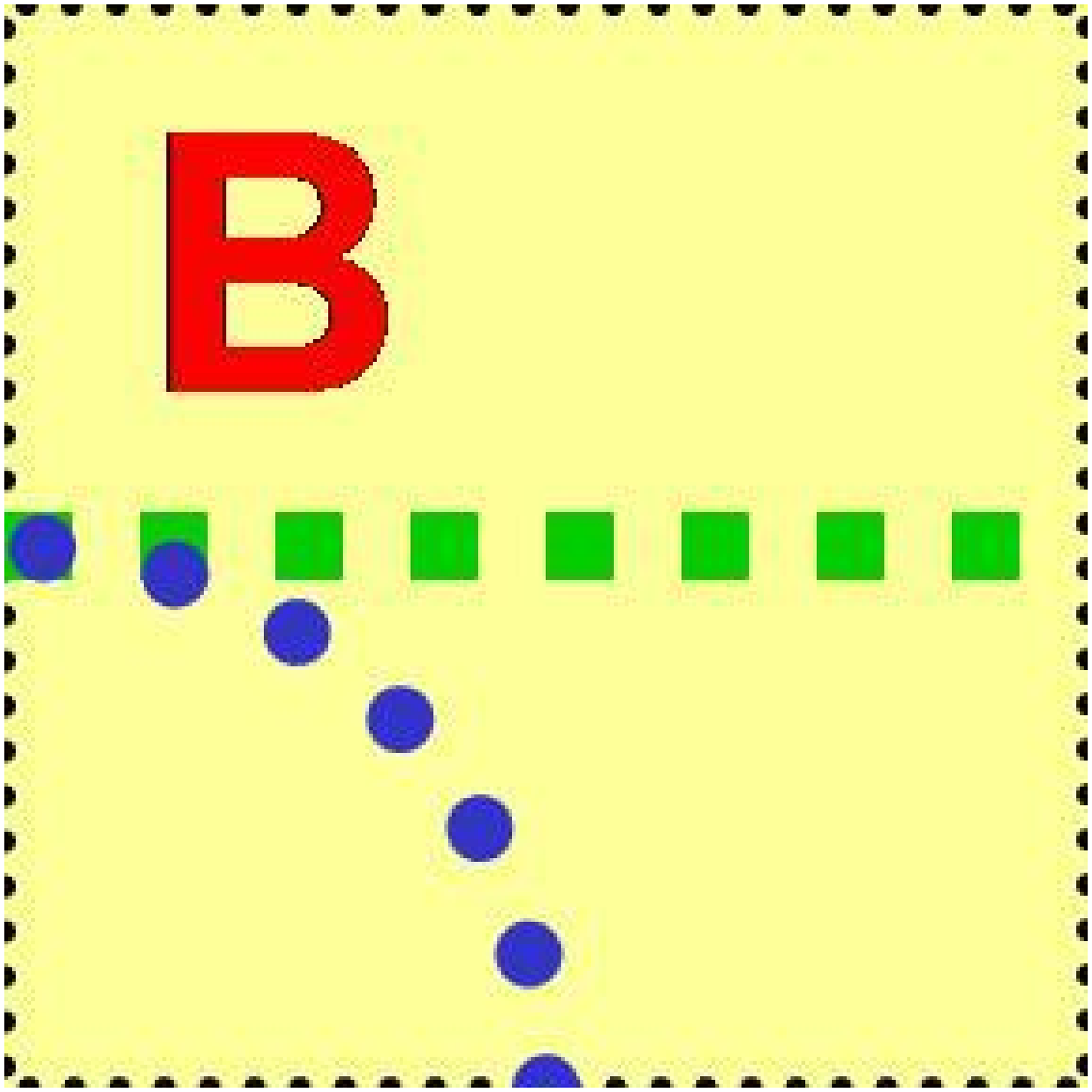}%
}%
%EndExpansion
=%
%TCIMACRO{\FRAME{itbpF}{0.3269in}{0.3269in}{0.1003in}{}{}{ut05.ps}%
%width 0.3269in;  height 0.3269in;  depth 0.1003in;  original-width 3in;
%original-height 3in;  cropleft "0";  croptop "1";  cropright "1";
%cropbottom "0";  filename 'ut05.ps';file-properties "XNPEU";}}}%
%BeginExpansion
\raisebox{-0.1003in}{\includegraphics[
%natheight=3.000000in,
%natwidth=3.000000in,
height=0.3269in,
width=0.3269in
]%
{ut05.ps}%
}%
%EndExpansion
\Longleftrightarrow%
%TCIMACRO{\FRAME{itbpF}{0.3269in}{0.3269in}{0.1003in}{}{}{lrb.ps}%
%width 0.3269in;  height 0.3269in;  depth 0.1003in;  original-width 3in;
%original-height 3in;  cropleft "0";  croptop "1";  cropright "1";
%cropbottom "0";  filename 'lrb.ps';file-properties "XNPEU";}}}%
%BeginExpansion
\raisebox{-0.1003in}{\includegraphics[
%natheight=3.000000in,
%natwidth=3.000000in,
height=0.3269in,
width=0.3269in
]%
{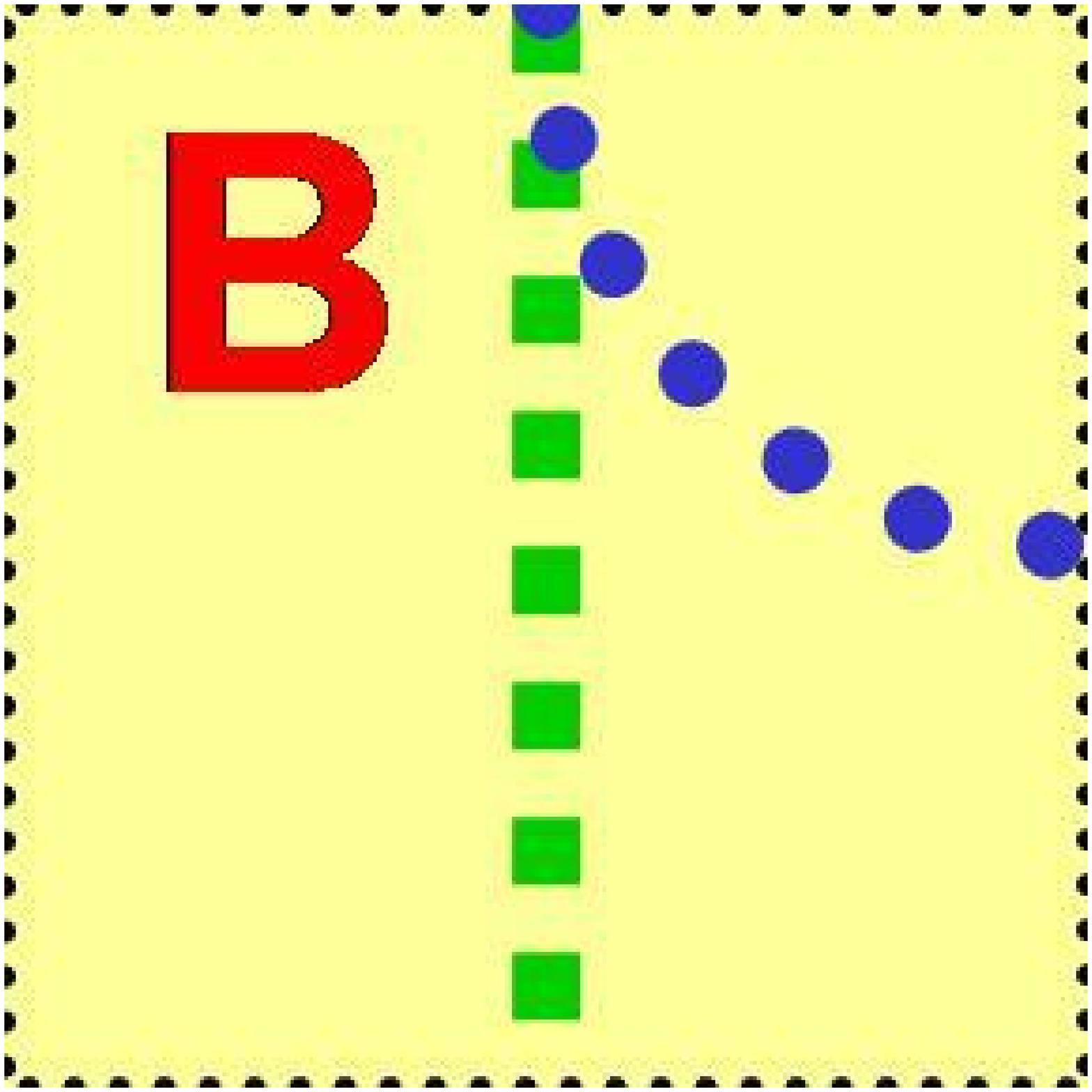}%
}%
%EndExpansion
=%
%TCIMACRO{\FRAME{itbpF}{0.3269in}{0.3269in}{0.1003in}{}{}{ut03.ps}%
%width 0.3269in;  height 0.3269in;  depth 0.1003in;  original-width 3in;
%original-height 3in;  cropleft "0";  croptop "1";  cropright "1";
%cropbottom "0";  filename 'ut03.ps';file-properties "XNPEU";}}}%
%BeginExpansion
\raisebox{-0.1003in}{\includegraphics[
%natheight=3.000000in,
%natwidth=3.000000in,
height=0.3269in,
width=0.3269in
]%
{ut03.ps}%
}%
%EndExpansion
\\
\end{array}
\right.
\]

\bigskip

With these two additional notational conventions, the \textbf{mosaic
Reidemeister 3 moves} are given below:

\bigskip

$\hspace{-1.25in}\underset{R_{3}}{%
%
}$

\bigskip

As noted in a previous section, all mosaic moves are permutations on the set
of mosaics $\mathbb{M}^{(n)}$. \ In particular, the planar isotopy moves and
the Reidemeister moves lie in the permutation group of the set of mosaics.
\ It easily follows that the planar isotopy moves and the Reidemeister moves
also lie in the group of all permutations of the set of knot mosaics
$\mathbb{K}^{(n)}$. \ Hence, we can make the following definition:

\bigskip

\begin{definition}
We define the (\textbf{knot mosaic}) \textbf{ambient group} $\mathbb{A}(n)$ as
the group of all permutations of the set of knot $n$-mosaics $\mathbb{K}%
^{(n)}$ generated by the mosaic planar isotopy and the mosaic Reidemeister
moves. \ 
\end{definition}

\bigskip

\begin{remark}
It follows from a previous proposition that the mosaic\ planar isotopy moves
and Reidemeister moves, as permutations, are each the product of disjoint transpositions.
\end{remark}

\bigskip

The completeness of the set of planar isotopy and Reidemeister moves is
addressed in section 2.7 of this paper.

\bigskip

\subsection{Knot mosaic type}

\bigskip

We now are prepared to define the analog of knot type for mosaics.

\bigskip

We define the \textbf{mosaic injection }%
\[%
\begin{array}
[c]{rrr}%
\iota:\mathbb{M}^{(n)} & \longrightarrow & \mathbb{M}^{(n+1)}\\
M^{(n)} & \longmapsto & M^{(n+1)}%
\end{array}
\]
as%
\[
M_{ij}^{(n+1)}=\left\{
\begin{array}
[c]{cl}%
M_{ij}^{(n)} & \text{if }0\leq i,j<n\\
& \\%
%TCIMACRO{\FRAME{itbpF}{0.3269in}{0.3269in}{0.1003in}{}{}{ut00.ps}%
%width 0.3269in;  height 0.3269in;  depth 0.1003in;  original-width 3in;
%original-height 3in;  cropleft "0";  croptop "1";  cropright "1";
%cropbottom "0";  filename 'ut00.ps';file-properties "XNPEU";}}}%
%BeginExpansion
\raisebox{-0.1003in}{\includegraphics[
%natheight=3.000000in,
%natwidth=3.000000in,
height=0.3269in,
width=0.3269in
]%
{ut00.ps}%
}%
%EndExpansion
& \text{otherwise}%
\end{array}
\right.
\]

\bigskip

Thus,
\[
M^{(n)}=%
\begin{array}
[c]{cccc}%
%TCIMACRO{\FRAME{itbpF}{0.3269in}{0.3269in}{0in}{}{}{ut00.ps}%
%width 0.3269in;  height 0.3269in;  depth 0in;  original-width 3in;
%original-height 3in;  cropleft "0";  croptop "1";  cropright "1";
%cropbottom "0";  filename 'ut00.ps';file-properties "XNPEU";}}}%
%BeginExpansion
{\includegraphics[
%natheight=3.000000in,
%natwidth=3.000000in,
height=0.3269in,
width=0.3269in
]%
{ut00.ps}%
}%
%EndExpansion
&
%TCIMACRO{\FRAME{itbpF}{0.3269in}{0.3269in}{0in}{}{}{ut02.ps}%
%width 0.3269in;  height 0.3269in;  depth 0in;  original-width 3in;
%original-height 3in;  cropleft "0";  croptop "1";  cropright "1";
%cropbottom "0";  filename 'ut02.ps';file-properties "XNPEU";}}}%
%BeginExpansion
{\includegraphics[
%natheight=3.000000in,
%natwidth=3.000000in,
height=0.3269in,
width=0.3269in
]%
{ut02.ps}%
}%
%EndExpansion
&
%TCIMACRO{\FRAME{itbpF}{0.3269in}{0.3269in}{0in}{}{}{ut01.ps}%
%width 0.3269in;  height 0.3269in;  depth 0in;  original-width 3in;
%original-height 3in;  cropleft "0";  croptop "1";  cropright "1";
%cropbottom "0";  filename 'ut01.ps';file-properties "XNPEU";}}}%
%BeginExpansion
{\includegraphics[
%natheight=3.000000in,
%natwidth=3.000000in,
height=0.3269in,
width=0.3269in
]%
{ut01.ps}%
}%
%EndExpansion
&
%TCIMACRO{\FRAME{itbpF}{0.3269in}{0.3269in}{0in}{}{}{ut00.ps}%
%width 0.3269in;  height 0.3269in;  depth 0in;  original-width 3in;
%original-height 3in;  cropleft "0";  croptop "1";  cropright "1";
%cropbottom "0";  filename 'ut00.ps';file-properties "XNPEU";}}}%
%BeginExpansion
{\includegraphics[
%natheight=3.000000in,
%natwidth=3.000000in,
height=0.3269in,
width=0.3269in
]%
{ut00.ps}%
}%
%EndExpansion
\\%
%TCIMACRO{\FRAME{itbpF}{0.3269in}{0.3269in}{0in}{}{}{ut02.ps}%
%width 0.3269in;  height 0.3269in;  depth 0in;  original-width 3in;
%original-height 3in;  cropleft "0";  croptop "1";  cropright "1";
%cropbottom "0";  filename 'ut02.ps';file-properties "XNPEU";}}}%
%BeginExpansion
{\includegraphics[
%natheight=3.000000in,
%natwidth=3.000000in,
height=0.3269in,
width=0.3269in
]%
{ut02.ps}%
}%
%EndExpansion
&
%TCIMACRO{\FRAME{itbpF}{0.3269in}{0.3269in}{0in}{}{}{ut09.ps}%
%width 0.3269in;  height 0.3269in;  depth 0in;  original-width 3in;
%original-height 3in;  cropleft "0";  croptop "1";  cropright "1";
%cropbottom "0";  filename 'ut09.ps';file-properties "XNPEU";}}}%
%BeginExpansion
{\includegraphics[
%natheight=3.000000in,
%natwidth=3.000000in,
height=0.3269in,
width=0.3269in
]%
{ut09.ps}%
}%
%EndExpansion
&
%TCIMACRO{\FRAME{itbpF}{0.3269in}{0.3269in}{0in}{}{}{ut10.ps}%
%width 0.3269in;  height 0.3269in;  depth 0in;  original-width 3in;
%original-height 3in;  cropleft "0";  croptop "1";  cropright "1";
%cropbottom "0";  filename 'ut10.ps';file-properties "XNPEU";}}}%
%BeginExpansion
{\includegraphics[
%natheight=3.000000in,
%natwidth=3.000000in,
height=0.3269in,
width=0.3269in
]%
{ut10.ps}%
}%
%EndExpansion
&
%TCIMACRO{\FRAME{itbpF}{0.3269in}{0.3269in}{0in}{}{}{ut01.ps}%
%width 0.3269in;  height 0.3269in;  depth 0in;  original-width 3in;
%original-height 3in;  cropleft "0";  croptop "1";  cropright "1";
%cropbottom "0";  filename 'ut01.ps';file-properties "XNPEU";}}}%
%BeginExpansion
{\includegraphics[
%natheight=3.000000in,
%natwidth=3.000000in,
height=0.3269in,
width=0.3269in
]%
{ut01.ps}%
}%
%EndExpansion
\\%
%TCIMACRO{\FRAME{itbpF}{0.3269in}{0.3269in}{0in}{}{}{ut06.ps}%
%width 0.3269in;  height 0.3269in;  depth 0in;  original-width 3in;
%original-height 3in;  cropleft "0";  croptop "1";  cropright "1";
%cropbottom "0";  filename 'ut06.ps';file-properties "XNPEU";}}}%
%BeginExpansion
{\includegraphics[
%natheight=3.000000in,
%natwidth=3.000000in,
height=0.3269in,
width=0.3269in
]%
{ut06.ps}%
}%
%EndExpansion
&
%TCIMACRO{\FRAME{itbpF}{0.3269in}{0.3269in}{0in}{}{}{ut03.ps}%
%width 0.3269in;  height 0.3269in;  depth 0in;  original-width 3in;
%original-height 3in;  cropleft "0";  croptop "1";  cropright "1";
%cropbottom "0";  filename 'ut03.ps';file-properties "XNPEU";}}}%
%BeginExpansion
{\includegraphics[
%natheight=3.000000in,
%natwidth=3.000000in,
height=0.3269in,
width=0.3269in
]%
{ut03.ps}%
}%
%EndExpansion
&
%TCIMACRO{\FRAME{itbpF}{0.3269in}{0.3269in}{0in}{}{}{ut09.ps}%
%width 0.3269in;  height 0.3269in;  depth 0in;  original-width 3in;
%original-height 3in;  cropleft "0";  croptop "1";  cropright "1";
%cropbottom "0";  filename 'ut09.ps';file-properties "XNPEU";}}}%
%BeginExpansion
{\includegraphics[
%natheight=3.000000in,
%natwidth=3.000000in,
height=0.3269in,
width=0.3269in
]%
{ut09.ps}%
}%
%EndExpansion
&
%TCIMACRO{\FRAME{itbpF}{0.3269in}{0.3269in}{0in}{}{}{ut04.ps}%
%width 0.3269in;  height 0.3269in;  depth 0in;  original-width 3in;
%original-height 3in;  cropleft "0";  croptop "1";  cropright "1";
%cropbottom "0";  filename 'ut04.ps';file-properties "XNPEU";}}}%
%BeginExpansion
{\includegraphics[
%natheight=3.000000in,
%natwidth=3.000000in,
height=0.3269in,
width=0.3269in
]%
{ut04.ps}%
}%
%EndExpansion
\\%
%TCIMACRO{\FRAME{itbpF}{0.3269in}{0.3269in}{0in}{}{}{ut03.ps}%
%width 0.3269in;  height 0.3269in;  depth 0in;  original-width 3in;
%original-height 3in;  cropleft "0";  croptop "1";  cropright "1";
%cropbottom "0";  filename 'ut03.ps';file-properties "XNPEU";}}}%
%BeginExpansion
{\includegraphics[
%natheight=3.000000in,
%natwidth=3.000000in,
height=0.3269in,
width=0.3269in
]%
{ut03.ps}%
}%
%EndExpansion
&
%TCIMACRO{\FRAME{itbpF}{0.3269in}{0.3269in}{0in}{}{}{ut05.ps}%
%width 0.3269in;  height 0.3269in;  depth 0in;  original-width 3in;
%original-height 3in;  cropleft "0";  croptop "1";  cropright "1";
%cropbottom "0";  filename 'ut05.ps';file-properties "XNPEU";}}}%
%BeginExpansion
{\includegraphics[
%natheight=3.000000in,
%natwidth=3.000000in,
height=0.3269in,
width=0.3269in
]%
{ut05.ps}%
}%
%EndExpansion
&
%TCIMACRO{\FRAME{itbpF}{0.3269in}{0.3269in}{0in}{}{}{ut04.ps}%
%width 0.3269in;  height 0.3269in;  depth 0in;  original-width 3in;
%original-height 3in;  cropleft "0";  croptop "1";  cropright "1";
%cropbottom "0";  filename 'ut04.ps';file-properties "XNPEU";}}}%
%BeginExpansion
{\includegraphics[
%natheight=3.000000in,
%natwidth=3.000000in,
height=0.3269in,
width=0.3269in
]%
{ut04.ps}%
}%
%EndExpansion
&
%TCIMACRO{\FRAME{itbpF}{0.3269in}{0.3269in}{0in}{}{}{ut00.ps}%
%width 0.3269in;  height 0.3269in;  depth 0in;  original-width 3in;
%original-height 3in;  cropleft "0";  croptop "1";  cropright "1";
%cropbottom "0";  filename 'ut00.ps';file-properties "XNPEU";}}}%
%BeginExpansion
{\includegraphics[
%natheight=3.000000in,
%natwidth=3.000000in,
height=0.3269in,
width=0.3269in
]%
{ut00.ps}%
}%
%EndExpansion
\end{array}
\overset{\iota}{\longrightarrow}M^{(n+1)}=%
\begin{array}
[c]{ccccc}%
%TCIMACRO{\FRAME{itbpF}{0.3269in}{0.3269in}{0in}{}{}{ut00.ps}%
%width 0.3269in;  height 0.3269in;  depth 0in;  original-width 3in;
%original-height 3in;  cropleft "0";  croptop "1";  cropright "1";
%cropbottom "0";  filename 'ut00.ps';file-properties "XNPEU";}}}%
%BeginExpansion
{\includegraphics[
%natheight=3.000000in,
%natwidth=3.000000in,
height=0.3269in,
width=0.3269in
]%
{ut00.ps}%
}%
%EndExpansion
&
%TCIMACRO{\FRAME{itbpF}{0.3269in}{0.3269in}{0in}{}{}{ut02.ps}%
%width 0.3269in;  height 0.3269in;  depth 0in;  original-width 3in;
%original-height 3in;  cropleft "0";  croptop "1";  cropright "1";
%cropbottom "0";  filename 'ut02.ps';file-properties "XNPEU";}}}%
%BeginExpansion
{\includegraphics[
%natheight=3.000000in,
%natwidth=3.000000in,
height=0.3269in,
width=0.3269in
]%
{ut02.ps}%
}%
%EndExpansion
&
%TCIMACRO{\FRAME{itbpF}{0.3269in}{0.3269in}{0in}{}{}{ut01.ps}%
%width 0.3269in;  height 0.3269in;  depth 0in;  original-width 3in;
%original-height 3in;  cropleft "0";  croptop "1";  cropright "1";
%cropbottom "0";  filename 'ut01.ps';file-properties "XNPEU";}}}%
%BeginExpansion
{\includegraphics[
%natheight=3.000000in,
%natwidth=3.000000in,
height=0.3269in,
width=0.3269in
]%
{ut01.ps}%
}%
%EndExpansion
&
%TCIMACRO{\FRAME{itbpF}{0.3269in}{0.3269in}{0in}{}{}{ut00.ps}%
%width 0.3269in;  height 0.3269in;  depth 0in;  original-width 3in;
%original-height 3in;  cropleft "0";  croptop "1";  cropright "1";
%cropbottom "0";  filename 'ut00.ps';file-properties "XNPEU";}}}%
%BeginExpansion
{\includegraphics[
%natheight=3.000000in,
%natwidth=3.000000in,
height=0.3269in,
width=0.3269in
]%
{ut00.ps}%
}%
%EndExpansion
&
%TCIMACRO{\FRAME{itbpF}{0.3269in}{0.3269in}{0in}{}{}{ut00.ps}%
%width 0.3269in;  height 0.3269in;  depth 0in;  original-width 3in;
%original-height 3in;  cropleft "0";  croptop "1";  cropright "1";
%cropbottom "0";  filename 'ut00.ps';file-properties "XNPEU";}}}%
%BeginExpansion
{\includegraphics[
%natheight=3.000000in,
%natwidth=3.000000in,
height=0.3269in,
width=0.3269in
]%
{ut00.ps}%
}%
%EndExpansion
\\%
%TCIMACRO{\FRAME{itbpF}{0.3269in}{0.3269in}{0in}{}{}{ut02.ps}%
%width 0.3269in;  height 0.3269in;  depth 0in;  original-width 3in;
%original-height 3in;  cropleft "0";  croptop "1";  cropright "1";
%cropbottom "0";  filename 'ut02.ps';file-properties "XNPEU";}}}%
%BeginExpansion
{\includegraphics[
%natheight=3.000000in,
%natwidth=3.000000in,
height=0.3269in,
width=0.3269in
]%
{ut02.ps}%
}%
%EndExpansion
&
%TCIMACRO{\FRAME{itbpF}{0.3269in}{0.3269in}{0in}{}{}{ut09.ps}%
%width 0.3269in;  height 0.3269in;  depth 0in;  original-width 3in;
%original-height 3in;  cropleft "0";  croptop "1";  cropright "1";
%cropbottom "0";  filename 'ut09.ps';file-properties "XNPEU";}}}%
%BeginExpansion
{\includegraphics[
%natheight=3.000000in,
%natwidth=3.000000in,
height=0.3269in,
width=0.3269in
]%
{ut09.ps}%
}%
%EndExpansion
&
%TCIMACRO{\FRAME{itbpF}{0.3269in}{0.3269in}{0in}{}{}{ut10.ps}%
%width 0.3269in;  height 0.3269in;  depth 0in;  original-width 3in;
%original-height 3in;  cropleft "0";  croptop "1";  cropright "1";
%cropbottom "0";  filename 'ut10.ps';file-properties "XNPEU";}}}%
%BeginExpansion
{\includegraphics[
%natheight=3.000000in,
%natwidth=3.000000in,
height=0.3269in,
width=0.3269in
]%
{ut10.ps}%
}%
%EndExpansion
&
%TCIMACRO{\FRAME{itbpF}{0.3269in}{0.3269in}{0in}{}{}{ut01.ps}%
%width 0.3269in;  height 0.3269in;  depth 0in;  original-width 3in;
%original-height 3in;  cropleft "0";  croptop "1";  cropright "1";
%cropbottom "0";  filename 'ut01.ps';file-properties "XNPEU";}}}%
%BeginExpansion
{\includegraphics[
%natheight=3.000000in,
%natwidth=3.000000in,
height=0.3269in,
width=0.3269in
]%
{ut01.ps}%
}%
%EndExpansion
&
%TCIMACRO{\FRAME{itbpF}{0.3269in}{0.3269in}{0in}{}{}{ut00.ps}%
%width 0.3269in;  height 0.3269in;  depth 0in;  original-width 3in;
%original-height 3in;  cropleft "0";  croptop "1";  cropright "1";
%cropbottom "0";  filename 'ut00.ps';file-properties "XNPEU";}}}%
%BeginExpansion
{\includegraphics[
%natheight=3.000000in,
%natwidth=3.000000in,
height=0.3269in,
width=0.3269in
]%
{ut00.ps}%
}%
%EndExpansion
\\%
%TCIMACRO{\FRAME{itbpF}{0.3269in}{0.3269in}{0in}{}{}{ut06.ps}%
%width 0.3269in;  height 0.3269in;  depth 0in;  original-width 3in;
%original-height 3in;  cropleft "0";  croptop "1";  cropright "1";
%cropbottom "0";  filename 'ut06.ps';file-properties "XNPEU";}}}%
%BeginExpansion
{\includegraphics[
%natheight=3.000000in,
%natwidth=3.000000in,
height=0.3269in,
width=0.3269in
]%
{ut06.ps}%
}%
%EndExpansion
&
%TCIMACRO{\FRAME{itbpF}{0.3269in}{0.3269in}{0in}{}{}{ut03.ps}%
%width 0.3269in;  height 0.3269in;  depth 0in;  original-width 3in;
%original-height 3in;  cropleft "0";  croptop "1";  cropright "1";
%cropbottom "0";  filename 'ut03.ps';file-properties "XNPEU";}}}%
%BeginExpansion
{\includegraphics[
%natheight=3.000000in,
%natwidth=3.000000in,
height=0.3269in,
width=0.3269in
]%
{ut03.ps}%
}%
%EndExpansion
&
%TCIMACRO{\FRAME{itbpF}{0.3269in}{0.3269in}{0in}{}{}{ut09.ps}%
%width 0.3269in;  height 0.3269in;  depth 0in;  original-width 3in;
%original-height 3in;  cropleft "0";  croptop "1";  cropright "1";
%cropbottom "0";  filename 'ut09.ps';file-properties "XNPEU";}}}%
%BeginExpansion
{\includegraphics[
%natheight=3.000000in,
%natwidth=3.000000in,
height=0.3269in,
width=0.3269in
]%
{ut09.ps}%
}%
%EndExpansion
&
%TCIMACRO{\FRAME{itbpF}{0.3269in}{0.3269in}{0in}{}{}{ut04.ps}%
%width 0.3269in;  height 0.3269in;  depth 0in;  original-width 3in;
%original-height 3in;  cropleft "0";  croptop "1";  cropright "1";
%cropbottom "0";  filename 'ut04.ps';file-properties "XNPEU";}}}%
%BeginExpansion
{\includegraphics[
%natheight=3.000000in,
%natwidth=3.000000in,
height=0.3269in,
width=0.3269in
]%
{ut04.ps}%
}%
%EndExpansion
&
%TCIMACRO{\FRAME{itbpF}{0.3269in}{0.3269in}{0in}{}{}{ut00.ps}%
%width 0.3269in;  height 0.3269in;  depth 0in;  original-width 3in;
%original-height 3in;  cropleft "0";  croptop "1";  cropright "1";
%cropbottom "0";  filename 'ut00.ps';file-properties "XNPEU";}}}%
%BeginExpansion
{\includegraphics[
%natheight=3.000000in,
%natwidth=3.000000in,
height=0.3269in,
width=0.3269in
]%
{ut00.ps}%
}%
%EndExpansion
\\%
%TCIMACRO{\FRAME{itbpF}{0.3269in}{0.3269in}{0in}{}{}{ut03.ps}%
%width 0.3269in;  height 0.3269in;  depth 0in;  original-width 3in;
%original-height 3in;  cropleft "0";  croptop "1";  cropright "1";
%cropbottom "0";  filename 'ut03.ps';file-properties "XNPEU";}}}%
%BeginExpansion
{\includegraphics[
%natheight=3.000000in,
%natwidth=3.000000in,
height=0.3269in,
width=0.3269in
]%
{ut03.ps}%
}%
%EndExpansion
&
%TCIMACRO{\FRAME{itbpF}{0.3269in}{0.3269in}{0in}{}{}{ut05.ps}%
%width 0.3269in;  height 0.3269in;  depth 0in;  original-width 3in;
%original-height 3in;  cropleft "0";  croptop "1";  cropright "1";
%cropbottom "0";  filename 'ut05.ps';file-properties "XNPEU";}}}%
%BeginExpansion
{\includegraphics[
%natheight=3.000000in,
%natwidth=3.000000in,
height=0.3269in,
width=0.3269in
]%
{ut05.ps}%
}%
%EndExpansion
&
%TCIMACRO{\FRAME{itbpF}{0.3269in}{0.3269in}{0in}{}{}{ut04.ps}%
%width 0.3269in;  height 0.3269in;  depth 0in;  original-width 3in;
%original-height 3in;  cropleft "0";  croptop "1";  cropright "1";
%cropbottom "0";  filename 'ut04.ps';file-properties "XNPEU";}}}%
%BeginExpansion
{\includegraphics[
%natheight=3.000000in,
%natwidth=3.000000in,
height=0.3269in,
width=0.3269in
]%
{ut04.ps}%
}%
%EndExpansion
&
%TCIMACRO{\FRAME{itbpF}{0.3269in}{0.3269in}{0in}{}{}{ut00.ps}%
%width 0.3269in;  height 0.3269in;  depth 0in;  original-width 3in;
%original-height 3in;  cropleft "0";  croptop "1";  cropright "1";
%cropbottom "0";  filename 'ut00.ps';file-properties "XNPEU";}}}%
%BeginExpansion
{\includegraphics[
%natheight=3.000000in,
%natwidth=3.000000in,
height=0.3269in,
width=0.3269in
]%
{ut00.ps}%
}%
%EndExpansion
&
%TCIMACRO{\FRAME{itbpF}{0.3269in}{0.3269in}{0in}{}{}{ut00.ps}%
%width 0.3269in;  height 0.3269in;  depth 0in;  original-width 3in;
%original-height 3in;  cropleft "0";  croptop "1";  cropright "1";
%cropbottom "0";  filename 'ut00.ps';file-properties "XNPEU";}}}%
%BeginExpansion
{\includegraphics[
%natheight=3.000000in,
%natwidth=3.000000in,
height=0.3269in,
width=0.3269in
]%
{ut00.ps}%
}%
%EndExpansion
\\%
%TCIMACRO{\FRAME{itbpF}{0.3269in}{0.3269in}{0in}{}{}{ut00.ps}%
%width 0.3269in;  height 0.3269in;  depth 0in;  original-width 3in;
%original-height 3in;  cropleft "0";  croptop "1";  cropright "1";
%cropbottom "0";  filename 'ut00.ps';file-properties "XNPEU";}}}%
%BeginExpansion
{\includegraphics[
%natheight=3.000000in,
%natwidth=3.000000in,
height=0.3269in,
width=0.3269in
]%
{ut00.ps}%
}%
%EndExpansion
&
%TCIMACRO{\FRAME{itbpF}{0.3269in}{0.3269in}{0in}{}{}{ut00.ps}%
%width 0.3269in;  height 0.3269in;  depth 0in;  original-width 3in;
%original-height 3in;  cropleft "0";  croptop "1";  cropright "1";
%cropbottom "0";  filename 'ut00.ps';file-properties "XNPEU";}}}%
%BeginExpansion
{\includegraphics[
%natheight=3.000000in,
%natwidth=3.000000in,
height=0.3269in,
width=0.3269in
]%
{ut00.ps}%
}%
%EndExpansion
&
%TCIMACRO{\FRAME{itbpF}{0.3269in}{0.3269in}{0in}{}{}{ut00.ps}%
%width 0.3269in;  height 0.3269in;  depth 0in;  original-width 3in;
%original-height 3in;  cropleft "0";  croptop "1";  cropright "1";
%cropbottom "0";  filename 'ut00.ps';file-properties "XNPEU";}}}%
%BeginExpansion
{\includegraphics[
%natheight=3.000000in,
%natwidth=3.000000in,
height=0.3269in,
width=0.3269in
]%
{ut00.ps}%
}%
%EndExpansion
&
%TCIMACRO{\FRAME{itbpF}{0.3269in}{0.3269in}{0in}{}{}{ut00.ps}%
%width 0.3269in;  height 0.3269in;  depth 0in;  original-width 3in;
%original-height 3in;  cropleft "0";  croptop "1";  cropright "1";
%cropbottom "0";  filename 'ut00.ps';file-properties "XNPEU";}}}%
%BeginExpansion
{\includegraphics[
%natheight=3.000000in,
%natwidth=3.000000in,
height=0.3269in,
width=0.3269in
]%
{ut00.ps}%
}%
%EndExpansion
&
%TCIMACRO{\FRAME{itbpF}{0.3269in}{0.3269in}{0in}{}{}{ut00.ps}%
%width 0.3269in;  height 0.3269in;  depth 0in;  original-width 3in;
%original-height 3in;  cropleft "0";  croptop "1";  cropright "1";
%cropbottom "0";  filename 'ut00.ps';file-properties "XNPEU";}}}%
%BeginExpansion
{\includegraphics[
%natheight=3.000000in,
%natwidth=3.000000in,
height=0.3269in,
width=0.3269in
]%
{ut00.ps}%
}%
%EndExpansion
\end{array}
\]

\bigskip

\begin{remark}
We now can explicitly define the \textbf{graded system} $\left(
\mathbb{K},\mathbb{A}\right)  $ that was mentioned in the introduction. \ The
symbol $\mathbb{K}$ denotes the directed system of sets $\left\{
\mathbb{K}^{(n)}\longrightarrow\mathbb{K}^{(n+1)}:n=1,2,3,\ldots\right\}  $
and $\mathbb{A}$ denotes the directed system of permutation groups $\left\{
\mathbb{A}(n)\longrightarrow\mathbb{A}(n+1):n=1,2,3,\ldots\right\}  $.
\ Thus,
\[
\left(  \mathbb{K},\mathbb{A}\right)  =\left(  \mathbb{K}^{(1)},\mathbb{A}%
\left(  1\right)  \right)  \longrightarrow\left(  \mathbb{K}^{(2)}%
,\mathbb{A}\left(  2\right)  \right)  \longrightarrow\cdots\longrightarrow
\left(  \mathbb{K}^{(n)},\mathbb{A}\left(  n\right)  \right)  \longrightarrow
\cdots
\]

\end{remark}

\bigskip

\begin{definition}
Two $n$-mosaics $M$ and $N$ are said to be of the \textbf{same knot }%
$n$\textbf{-type}, written%
\[
M\underset{n}{\sim}N\text{ ,}%
\]
provided there is an element of the ambient isotopy group $\mathbb{A}(n)$
which transforms $M$ into $N$.
\end{definition}

\bigskip

\begin{definition}
An $m$-mosaic $M$ and an $n$-mosaic $N$ are said to be of the\textbf{ same
knot mosaic type}, written%
\[
M\sim N\text{ ,}%
\]
provided there exists a non-negative integer $\ell$ such that, if $m\leq n$,
then
\[
\iota^{\ell+n-m}M\sim_{\ell+n}\iota^{\ell}N\text{ ,}%
\]
or if $m>n$, then
\[
\iota^{\ell}M\sim_{\ell+m}\iota^{\ell+m-n}N\text{ ,}%
\]
where, for each non-negative integer $p$, $\iota^{p}$ denotes the $p$-fold
composition $\underset{p}{\underbrace{\iota\circ\iota\circ\cdots\circ\iota}}$ .
\end{definition}

\bigskip

\subsection{Tame knot theory and knot mosaic theory are equivalent}

\bigskip

In the introduction of this paper, we conjecture that the formal (re-writing)
system \ $\left(  \mathbb{K},\mathbb{A}\right)  $ of knot mosaics fully
captures the entire structure of tame knot theory. \ \ We now explain in
greater detail what is meant by this conjecture.

\bigskip

Let $\mathbb{Z}$ denote the set of integers, and $\mathbb{R}^{2}$ the two
dimensional Euclidean plane. \ Let $\tau$ denote the \textbf{square tiling} of
$\mathbb{R}^{2}$ induced by the sublattice $\mathbb{Z}\times\mathbb{Z}$ of
$\mathbb{R}^{2}$, and for each $i$, $j$ in $\mathbb{Z}$, let $\tau_{ij}$
denote the subregion of $\mathbb{R}^{2}$ defined by%
\[
\tau_{ij}=\left\{  \left(  x,y\right)  \in\mathbb{R}^{2}:i\leq x\leq i+1\text{
and }j\leq y\leq j+1\right\}  \text{ .}%
\]

\bigskip

Let $k$ be an arbitrary tame knot in $3$-space $\mathbb{R}^{3}$. \ A knot
diagram of $k$, i.e., a regular projection%
\[
\pi:\left(  \mathbb{R}^{3},k\right)  \longrightarrow\left(  \mathbb{R}^{2},\pi
k\right)
\]
is said to be a \textbf{mosaic knot diagram} if

\begin{itemize}
\item[\textbf{1)}] The image under $\pi$ of $k$ lies in the first quadrant of
$\mathbb{R}^{2}$, and

\item[\textbf{2)}] For all $i$, $j$ in $\mathbb{Z}$, the pair $\left(
\tau_{ij},\left(  \pi k\right)  \cap\tau_{ij}\right)  $ is identical with the
cell pair on one of the faces of the 11 tiles $T_{0}$, $T_{1}$, $\ldots$ ,
$T_{10}$ .
\end{itemize}

\bigskip

\begin{remark}
Clearly, using standard arguments in knot theory, one can prove that every
tame knot (or link) has a mosaic knot diagram.
\end{remark}

\bigskip

Each mosaic knot diagram $\pi:\left(  \mathbb{R}^{3},k\right)  \longrightarrow
\left(  \mathbb{R}^{2},\pi k\right)  $ of a knot $k$ can naturally be
identified with a knot $n$-mosaic $K$, where $n$ is the smallest positive
integer such that $\pi k$ lies in the region%
\[
\left\{  \left(  x,y\right)  \in\mathbb{R}^{2}:0\leq x,y\leq n\right\}
\text{.}%
\]
Moreover, every knot $n$-mosaic can naturally be identified with the diagram
of a knot $k$. We call this associated knot mosaic $K$ a \textbf{(knot})
\textbf{mosaic representative} of the original knot $k$.\ 

\bigskip

This leads us to the following conjecture:

\bigskip

\begin{conjecture}
Let $k_{1}$ and $k_{2}$ be two tame knots (or links), and let $K_{1}$ and
$K_{2}$ be two arbitrary chosen mosaic representatives of $k_{1}$ and $k_{2}$,
respectively. \ Then $k_{1}$ and $k_{2}$ are of the same knot type if and only
if the representative mosaics $K_{1}$ and $K_{2}$ are of the same knot mosaic
type. \ In other words, knot mosaic type is a complete invariant of tame knots.
\end{conjecture}

\bigskip

\section{Part 2: Quantum Knots}

\bigskip

\subsection{Quantum knot systems, quantum knots, and the ambient group
$\mathbb{A}$}

\bigskip

Our sole purpose in creating the formal system $\left(  \mathbb{K}%
,\mathbb{A}\right)  $ of knot mosaics was to create a framework within which
we can explicitly define what is meant by a quantum knot. \ We are finally in
a position to do so. \ 

\bigskip

We begin by assigning a left-to-right \textbf{linear ordering, }denoted by
`$<$ ', to the 11 mosaic tiles as indicated below

\bigskip

$\hspace{-0.85in}\underset{T_{0}}{%
%TCIMACRO{\FRAME{itbpF}{0.3269in}{0.3269in}{0in}{}{}{ut00.ps}%
%width 0.3269in;  height 0.3269in;  depth 0in;  original-width 3in;
%original-height 3in;  cropleft "0";  croptop "1";  cropright "1";
%cropbottom "0";  filename 'ut00.ps';file-properties "XNPEU";}}}%
%BeginExpansion
{\includegraphics[
%natheight=3.000000in,
%natwidth=3.000000in,
height=0.3269in,
width=0.3269in
]%
{ut00.ps}%
}%
%EndExpansion
}%
\begin{array}
[c]{c}%
<\\
\mathstrut
\end{array}
\underset{T_{1}}{%
%TCIMACRO{\FRAME{itbpF}{0.3269in}{0.3269in}{0in}{}{}{ut01.ps}%
%width 0.3269in;  height 0.3269in;  depth 0in;  original-width 3in;
%original-height 3in;  cropleft "0";  croptop "1";  cropright "1";
%cropbottom "0";  filename 'ut01.ps';file-properties "XNPEU";}}}%
%BeginExpansion
{\includegraphics[
%natheight=3.000000in,
%natwidth=3.000000in,
height=0.3269in,
width=0.3269in
]%
{ut01.ps}%
}%
%EndExpansion
}%
\begin{array}
[c]{c}%
<\\
\mathstrut
\end{array}
\underset{T_{2}}{%
%TCIMACRO{\FRAME{itbpF}{0.3269in}{0.3269in}{0in}{}{}{ut02.ps}%
%width 0.3269in;  height 0.3269in;  depth 0in;  original-width 3in;
%original-height 3in;  cropleft "0";  croptop "1";  cropright "1";
%cropbottom "0";  filename 'ut02.ps';file-properties "XNPEU";}}}%
%BeginExpansion
{\includegraphics[
%natheight=3.000000in,
%natwidth=3.000000in,
height=0.3269in,
width=0.3269in
]%
{ut02.ps}%
}%
%EndExpansion
}%
\begin{array}
[c]{c}%
<\\
\mathstrut
\end{array}
\underset{T_{3}}{%
%TCIMACRO{\FRAME{itbpF}{0.3269in}{0.3269in}{0in}{}{}{ut03.ps}%
%width 0.3269in;  height 0.3269in;  depth 0in;  original-width 3in;
%original-height 3in;  cropleft "0";  croptop "1";  cropright "1";
%cropbottom "0";  filename 'ut03.ps';file-properties "XNPEU";}}}%
%BeginExpansion
{\includegraphics[
%natheight=3.000000in,
%natwidth=3.000000in,
height=0.3269in,
width=0.3269in
]%
{ut03.ps}%
}%
%EndExpansion
}%
\begin{array}
[c]{c}%
<\\
\mathstrut
\end{array}
\underset{T_{4}}{%
%TCIMACRO{\FRAME{itbpF}{0.3269in}{0.3269in}{0in}{}{}{ut04.ps}%
%width 0.3269in;  height 0.3269in;  depth 0in;  original-width 3in;
%original-height 3in;  cropleft "0";  croptop "1";  cropright "1";
%cropbottom "0";  filename 'ut04.ps';file-properties "XNPEU";}}}%
%BeginExpansion
{\includegraphics[
%natheight=3.000000in,
%natwidth=3.000000in,
height=0.3269in,
width=0.3269in
]%
{ut04.ps}%
}%
%EndExpansion
}%
\begin{array}
[c]{c}%
<\\
\mathstrut
\end{array}
\underset{T_{5}}{%
%TCIMACRO{\FRAME{itbpF}{0.3269in}{0.3269in}{0in}{}{}{ut05.ps}%
%width 0.3269in;  height 0.3269in;  depth 0in;  original-width 3in;
%original-height 3in;  cropleft "0";  croptop "1";  cropright "1";
%cropbottom "0";  filename 'ut05.ps';file-properties "XNPEU";}}}%
%BeginExpansion
{\includegraphics[
%natheight=3.000000in,
%natwidth=3.000000in,
height=0.3269in,
width=0.3269in
]%
{ut05.ps}%
}%
%EndExpansion
}%
\begin{array}
[c]{c}%
<\\
\mathstrut
\end{array}
\underset{T_{6}}{%
%TCIMACRO{\FRAME{itbpF}{0.3269in}{0.3269in}{0in}{}{}{ut06.ps}%
%width 0.3269in;  height 0.3269in;  depth 0in;  original-width 3in;
%original-height 3in;  cropleft "0";  croptop "1";  cropright "1";
%cropbottom "0";  filename 'ut06.ps';file-properties "XNPEU";}}}%
%BeginExpansion
{\includegraphics[
%natheight=3.000000in,
%natwidth=3.000000in,
height=0.3269in,
width=0.3269in
]%
{ut06.ps}%
}%
%EndExpansion
}%
\begin{array}
[c]{c}%
<\\
\mathstrut
\end{array}
\underset{T_{7}}{%
%TCIMACRO{\FRAME{itbpF}{0.3269in}{0.3269in}{0in}{}{}{ut07.ps}%
%width 0.3269in;  height 0.3269in;  depth 0in;  original-width 3in;
%original-height 3in;  cropleft "0";  croptop "1";  cropright "1";
%cropbottom "0";  filename 'ut07.ps';file-properties "XNPEU";}}}%
%BeginExpansion
{\includegraphics[
%natheight=3.000000in,
%natwidth=3.000000in,
height=0.3269in,
width=0.3269in
]%
{ut07.ps}%
}%
%EndExpansion
}%
\begin{array}
[c]{c}%
<\\
\mathstrut
\end{array}
\underset{T_{8}}{%
%TCIMACRO{\FRAME{itbpF}{0.3269in}{0.3269in}{0in}{}{}{ut08.ps}%
%width 0.3269in;  height 0.3269in;  depth 0in;  original-width 3in;
%original-height 3in;  cropleft "0";  croptop "1";  cropright "1";
%cropbottom "0";  filename 'ut08.ps';file-properties "XNPEU";}}}%
%BeginExpansion
{\includegraphics[
%natheight=3.000000in,
%natwidth=3.000000in,
height=0.3269in,
width=0.3269in
]%
{ut08.ps}%
}%
%EndExpansion
}%
\begin{array}
[c]{c}%
<\\
\mathstrut
\end{array}
\underset{T_{9}}{%
%TCIMACRO{\FRAME{itbpF}{0.3269in}{0.3269in}{0in}{}{}{ut09.ps}%
%width 0.3269in;  height 0.3269in;  depth 0in;  original-width 3in;
%original-height 3in;  cropleft "0";  croptop "1";  cropright "1";
%cropbottom "0";  filename 'ut09.ps';file-properties "XNPEU";}}}%
%BeginExpansion
{\includegraphics[
%natheight=3.000000in,
%natwidth=3.000000in,
height=0.3269in,
width=0.3269in
]%
{ut09.ps}%
}%
%EndExpansion
}%
\begin{array}
[c]{c}%
<\\
\mathstrut
\end{array}
\underset{T_{10}}{%
%TCIMACRO{\FRAME{itbpF}{0.3269in}{0.3269in}{0in}{}{}{ut10.ps}%
%width 0.3269in;  height 0.3269in;  depth 0in;  original-width 3in;
%original-height 3in;  cropleft "0";  croptop "1";  cropright "1";
%cropbottom "0";  filename 'ut10.ps';file-properties "XNPEU";}}}%
%BeginExpansion
{\includegraphics[
%natheight=3.000000in,
%natwidth=3.000000in,
height=0.3269in,
width=0.3269in
]%
{ut10.ps}%
}%
%EndExpansion
}$

\bigskip

We let $\mathcal{H}$ be the 11 dimensional Hilbert space with orthonormal
basis $\left\{  \left\vert T_{p}\right\rangle :0\leq p<n\right\}  $ labeled by
the above mosaic tiles, and we define the \textbf{Hilbert space }%
$\mathcal{M}^{(n)}$\textbf{ of }$\mathbf{n}$\textbf{-mosaics} as the tensor
product
\[
\mathcal{M}^{(n)}=%
%TCIMACRO{\dbigotimes \limits_{p=0}^{n^{2}-1}}%
%BeginExpansion
{\displaystyle\bigotimes\limits_{p=0}^{n^{2}-1}}
%EndExpansion
\mathcal{H}\text{ .}%
\]
Thus, the induced orthonormal basis of $\mathcal{M}^{(n)}$ consists of all
possible $n^{2}$-fold tensor products of the above 11 mosaic tiles, i.e., the
induced basis is
\[
\left\{
%TCIMACRO{\tbigotimes \limits_{p=0}^{n^{2}-1}}%
%BeginExpansion
{\textstyle\bigotimes\limits_{p=0}^{n^{2}-1}}
%EndExpansion
\left\vert T_{\ell\left(  p\right)  }\right\rangle \right\}  \text{ \ .}%
\]
We then use the above defined linear ordering on the set $\mathbb{T}^{(n)}$ of
mosaic tiles to \textbf{lexicographically (lex) order} all the basis elements
of $\mathcal{M}^{(n)}$. \ We also denote this linear ordering by `$<$'.
\ Finally, using \textbf{row major order}, we identify\ each basis element $%
%TCIMACRO{\tbigotimes \limits_{p=0}^{n^{2}-1}}%
%BeginExpansion
{\textstyle\bigotimes\limits_{p=0}^{n^{2}-1}}
%EndExpansion
\left\vert T_{\ell\left(  p\right)  }\right\rangle $ with the ket $\left\vert
M\right\rangle $ labeled by the $n$-mosaic $M=\left(  M_{ij}\right)  =\left(
T_{\ell\left(  ni+j\right)  }\right)  $. \ In other words, we have used row
major order to set up a one-to-one correspondence between basis elements of
$\mathcal{M}^{(n)}$ and the set of $n$-mosaics $\mathbb{M}^{(n)}$. \ 

\bigskip

For example, for $n=3$ the basis element%
\[
\left\vert T_{2}\right\rangle \otimes\left\vert T_{5}\right\rangle
\otimes\left\vert T_{4}\right\rangle \quad\otimes\quad\left\vert
T_{9}\right\rangle \otimes\left\vert T_{2}\right\rangle \otimes\left\vert
T_{1}\right\rangle \quad\otimes\quad\left\vert T_{5}\right\rangle
\otimes\left\vert T_{8}\right\rangle \otimes\left\vert T_{3}\right\rangle
\]
is identified with the $n$-mosaic labeled ket%
\[
\left\vert
\begin{array}
[c]{ccc}%
T_{2} & T_{5} & T_{4}\\
T_{9} & T_{2} & T_{1}\\
T_{5} & T_{8} & T_{3}%
\end{array}
\right\rangle \text{ .}%
\]

\bigskip

We can now define the \textbf{Hilbert space of knot }$\mathbf{n}%
$\textbf{-mosaics }$\mathcal{K}^{(n)}$ as the sub-Hilbert space $\mathcal{K}%
^{(n)}$ of $\mathcal{M}^{(n)}$ spanned by all orthonormal basis elements
labeled by knot $n$-mosaics.

\bigskip

Since
\[
\mathbb{K}^{(1)}=\left\{
%TCIMACRO{\FRAME{itbpF}{0.1773in}{0.1773in}{0.0502in}{}{}{ut00.ps}%
%width 0.1773in;  height 0.1773in;  depth 0.0502in;  original-width 3in;
%original-height 3in;  cropleft "0";  croptop "1";  cropright "1";
%cropbottom "0";  filename 'ut00.ps';file-properties "XNPEU";}}}%
%BeginExpansion
\raisebox{-0.0502in}{\includegraphics[
%natheight=3.000000in,
%natwidth=3.000000in,
height=0.1773in,
width=0.1773in
]%
{ut00.ps}%
}%
%EndExpansion
\right\}  \text{ \ \ and \ \ }\mathbb{K}^{(2)}=\left\{  \
\begin{array}
[c]{cc}%
%TCIMACRO{\FRAME{itbpF}{0.1773in}{0.1773in}{0in}{}{}{ut00.ps}%
%width 0.1773in;  height 0.1773in;  depth 0in;  original-width 3in;
%original-height 3in;  cropleft "0";  croptop "1";  cropright "1";
%cropbottom "0";  filename 'ut00.ps';file-properties "XNPEU";}}}%
%BeginExpansion
{\includegraphics[
%natheight=3.000000in,
%natwidth=3.000000in,
height=0.1773in,
width=0.1773in
]%
{ut00.ps}%
}%
%EndExpansion
&
%TCIMACRO{\FRAME{itbpF}{0.1773in}{0.1773in}{0in}{}{}{ut00.ps}%
%width 0.1773in;  height 0.1773in;  depth 0in;  original-width 3in;
%original-height 3in;  cropleft "0";  croptop "1";  cropright "1";
%cropbottom "0";  filename 'ut00.ps';file-properties "XNPEU";}}}%
%BeginExpansion
{\includegraphics[
%natheight=3.000000in,
%natwidth=3.000000in,
height=0.1773in,
width=0.1773in
]%
{ut00.ps}%
}%
%EndExpansion
\\%
%TCIMACRO{\FRAME{itbpF}{0.1773in}{0.1773in}{0in}{}{}{ut00.ps}%
%width 0.1773in;  height 0.1773in;  depth 0in;  original-width 3in;
%original-height 3in;  cropleft "0";  croptop "1";  cropright "1";
%cropbottom "0";  filename 'ut00.ps';file-properties "XNPEU";}}}%
%BeginExpansion
{\includegraphics[
%natheight=3.000000in,
%natwidth=3.000000in,
height=0.1773in,
width=0.1773in
]%
{ut00.ps}%
}%
%EndExpansion
&
%TCIMACRO{\FRAME{itbpF}{0.1773in}{0.1773in}{0in}{}{}{ut00.ps}%
%width 0.1773in;  height 0.1773in;  depth 0in;  original-width 3in;
%original-height 3in;  cropleft "0";  croptop "1";  cropright "1";
%cropbottom "0";  filename 'ut00.ps';file-properties "XNPEU";}}}%
%BeginExpansion
{\includegraphics[
%natheight=3.000000in,
%natwidth=3.000000in,
height=0.1773in,
width=0.1773in
]%
{ut00.ps}%
}%
%EndExpansion
\end{array}
,\quad%
\begin{array}
[c]{cc}%
%TCIMACRO{\FRAME{itbpF}{0.1773in}{0.1773in}{0in}{}{}{ut02.ps}%
%width 0.1773in;  height 0.1773in;  depth 0in;  original-width 3in;
%original-height 3in;  cropleft "0";  croptop "1";  cropright "1";
%cropbottom "0";  filename 'ut02.ps';file-properties "XNPEU";}}}%
%BeginExpansion
{\includegraphics[
%natheight=3.000000in,
%natwidth=3.000000in,
height=0.1773in,
width=0.1773in
]%
{ut02.ps}%
}%
%EndExpansion
&
%TCIMACRO{\FRAME{itbpF}{0.1773in}{0.1773in}{0in}{}{}{ut01.ps}%
%width 0.1773in;  height 0.1773in;  depth 0in;  original-width 3in;
%original-height 3in;  cropleft "0";  croptop "1";  cropright "1";
%cropbottom "0";  filename 'ut01.ps';file-properties "XNPEU";}}}%
%BeginExpansion
{\includegraphics[
%natheight=3.000000in,
%natwidth=3.000000in,
height=0.1773in,
width=0.1773in
]%
{ut01.ps}%
}%
%EndExpansion
\\%
%TCIMACRO{\FRAME{itbpF}{0.1773in}{0.1773in}{0in}{}{}{ut03.ps}%
%width 0.1773in;  height 0.1773in;  depth 0in;  original-width 3in;
%original-height 3in;  cropleft "0";  croptop "1";  cropright "1";
%cropbottom "0";  filename 'ut03.ps';file-properties "XNPEU";}}}%
%BeginExpansion
{\includegraphics[
%natheight=3.000000in,
%natwidth=3.000000in,
height=0.1773in,
width=0.1773in
]%
{ut03.ps}%
}%
%EndExpansion
&
%TCIMACRO{\FRAME{itbpF}{0.1773in}{0.1773in}{0in}{}{}{ut04.ps}%
%width 0.1773in;  height 0.1773in;  depth 0in;  original-width 3in;
%original-height 3in;  cropleft "0";  croptop "1";  cropright "1";
%cropbottom "0";  filename 'ut04.ps';file-properties "XNPEU";}}}%
%BeginExpansion
{\includegraphics[
%natheight=3.000000in,
%natwidth=3.000000in,
height=0.1773in,
width=0.1773in
]%
{ut04.ps}%
}%
%EndExpansion
\end{array}
\ \right\}  \text{ ,}%
\]
the first two Hilbert spaces $\mathcal{K}^{(1)}$ and $\mathcal{K}^{(2)}$ are
of dimensions $1$ and $2$, respectively. \ The third Hilbert space
$\mathcal{K}^{(3)}$ is of dimension $22$, as is demonstrated by the complete
list of all possible knot $3$-mosaics given Appendix A.

\bigskip

An example of an element of the Hilbert space $\mathcal{K}^{(4)}$ is
\[
\frac{\left\vert
\begin{array}
[c]{cccc}%
%TCIMACRO{\FRAME{itbpF}{0.1773in}{0.1773in}{0in}{}{}{ut00.ps}%
%width 0.1773in;  height 0.1773in;  depth 0in;  original-width 3in;
%original-height 3in;  cropleft "0";  croptop "1";  cropright "1";
%cropbottom "0";  filename 'ut00.ps';file-properties "XNPEU";}}}%
%BeginExpansion
{\includegraphics[
%natheight=3.000000in,
%natwidth=3.000000in,
height=0.1773in,
width=0.1773in
]%
{ut00.ps}%
}%
%EndExpansion
&
%TCIMACRO{\FRAME{itbpF}{0.1773in}{0.1773in}{0in}{}{}{ut02.ps}%
%width 0.1773in;  height 0.1773in;  depth 0in;  original-width 3in;
%original-height 3in;  cropleft "0";  croptop "1";  cropright "1";
%cropbottom "0";  filename 'ut02.ps';file-properties "XNPEU";}}}%
%BeginExpansion
{\includegraphics[
%natheight=3.000000in,
%natwidth=3.000000in,
height=0.1773in,
width=0.1773in
]%
{ut02.ps}%
}%
%EndExpansion
&
%TCIMACRO{\FRAME{itbpF}{0.1773in}{0.1773in}{0in}{}{}{ut01.ps}%
%width 0.1773in;  height 0.1773in;  depth 0in;  original-width 3in;
%original-height 3in;  cropleft "0";  croptop "1";  cropright "1";
%cropbottom "0";  filename 'ut01.ps';file-properties "XNPEU";}}}%
%BeginExpansion
{\includegraphics[
%natheight=3.000000in,
%natwidth=3.000000in,
height=0.1773in,
width=0.1773in
]%
{ut01.ps}%
}%
%EndExpansion
&
%TCIMACRO{\FRAME{itbpF}{0.1773in}{0.1773in}{0in}{}{}{ut00.ps}%
%width 0.1773in;  height 0.1773in;  depth 0in;  original-width 3in;
%original-height 3in;  cropleft "0";  croptop "1";  cropright "1";
%cropbottom "0";  filename 'ut00.ps';file-properties "XNPEU";}}}%
%BeginExpansion
{\includegraphics[
%natheight=3.000000in,
%natwidth=3.000000in,
height=0.1773in,
width=0.1773in
]%
{ut00.ps}%
}%
%EndExpansion
\\%
%TCIMACRO{\FRAME{itbpF}{0.1773in}{0.1773in}{0in}{}{}{ut02.ps}%
%width 0.1773in;  height 0.1773in;  depth 0in;  original-width 3in;
%original-height 3in;  cropleft "0";  croptop "1";  cropright "1";
%cropbottom "0";  filename 'ut02.ps';file-properties "XNPEU";}}}%
%BeginExpansion
{\includegraphics[
%natheight=3.000000in,
%natwidth=3.000000in,
height=0.1773in,
width=0.1773in
]%
{ut02.ps}%
}%
%EndExpansion
&
%TCIMACRO{\FRAME{itbpF}{0.1773in}{0.1773in}{0in}{}{}{ut09.ps}%
%width 0.1773in;  height 0.1773in;  depth 0in;  original-width 3in;
%original-height 3in;  cropleft "0";  croptop "1";  cropright "1";
%cropbottom "0";  filename 'ut09.ps';file-properties "XNPEU";}}}%
%BeginExpansion
{\includegraphics[
%natheight=3.000000in,
%natwidth=3.000000in,
height=0.1773in,
width=0.1773in
]%
{ut09.ps}%
}%
%EndExpansion
&
%TCIMACRO{\FRAME{itbpF}{0.1773in}{0.1773in}{0in}{}{}{ut10.ps}%
%width 0.1773in;  height 0.1773in;  depth 0in;  original-width 3in;
%original-height 3in;  cropleft "0";  croptop "1";  cropright "1";
%cropbottom "0";  filename 'ut10.ps';file-properties "XNPEU";}}}%
%BeginExpansion
{\includegraphics[
%natheight=3.000000in,
%natwidth=3.000000in,
height=0.1773in,
width=0.1773in
]%
{ut10.ps}%
}%
%EndExpansion
&
%TCIMACRO{\FRAME{itbpF}{0.1773in}{0.1773in}{0in}{}{}{ut01.ps}%
%width 0.1773in;  height 0.1773in;  depth 0in;  original-width 3in;
%original-height 3in;  cropleft "0";  croptop "1";  cropright "1";
%cropbottom "0";  filename 'ut01.ps';file-properties "XNPEU";}}}%
%BeginExpansion
{\includegraphics[
%natheight=3.000000in,
%natwidth=3.000000in,
height=0.1773in,
width=0.1773in
]%
{ut01.ps}%
}%
%EndExpansion
\\%
%TCIMACRO{\FRAME{itbpF}{0.1773in}{0.1773in}{0in}{}{}{ut06.ps}%
%width 0.1773in;  height 0.1773in;  depth 0in;  original-width 3in;
%original-height 3in;  cropleft "0";  croptop "1";  cropright "1";
%cropbottom "0";  filename 'ut06.ps';file-properties "XNPEU";}}}%
%BeginExpansion
{\includegraphics[
%natheight=3.000000in,
%natwidth=3.000000in,
height=0.1773in,
width=0.1773in
]%
{ut06.ps}%
}%
%EndExpansion
&
%TCIMACRO{\FRAME{itbpF}{0.1773in}{0.1773in}{0in}{}{}{ut03.ps}%
%width 0.1773in;  height 0.1773in;  depth 0in;  original-width 3in;
%original-height 3in;  cropleft "0";  croptop "1";  cropright "1";
%cropbottom "0";  filename 'ut03.ps';file-properties "XNPEU";}}}%
%BeginExpansion
{\includegraphics[
%natheight=3.000000in,
%natwidth=3.000000in,
height=0.1773in,
width=0.1773in
]%
{ut03.ps}%
}%
%EndExpansion
&
%TCIMACRO{\FRAME{itbpF}{0.1773in}{0.1773in}{0in}{}{}{ut09.ps}%
%width 0.1773in;  height 0.1773in;  depth 0in;  original-width 3in;
%original-height 3in;  cropleft "0";  croptop "1";  cropright "1";
%cropbottom "0";  filename 'ut09.ps';file-properties "XNPEU";}}}%
%BeginExpansion
{\includegraphics[
%natheight=3.000000in,
%natwidth=3.000000in,
height=0.1773in,
width=0.1773in
]%
{ut09.ps}%
}%
%EndExpansion
&
%TCIMACRO{\FRAME{itbpF}{0.1773in}{0.1773in}{0in}{}{}{ut04.ps}%
%width 0.1773in;  height 0.1773in;  depth 0in;  original-width 3in;
%original-height 3in;  cropleft "0";  croptop "1";  cropright "1";
%cropbottom "0";  filename 'ut04.ps';file-properties "XNPEU";}}}%
%BeginExpansion
{\includegraphics[
%natheight=3.000000in,
%natwidth=3.000000in,
height=0.1773in,
width=0.1773in
]%
{ut04.ps}%
}%
%EndExpansion
\\%
%TCIMACRO{\FRAME{itbpF}{0.1773in}{0.1773in}{0in}{}{}{ut03.ps}%
%width 0.1773in;  height 0.1773in;  depth 0in;  original-width 3in;
%original-height 3in;  cropleft "0";  croptop "1";  cropright "1";
%cropbottom "0";  filename 'ut03.ps';file-properties "XNPEU";}}}%
%BeginExpansion
{\includegraphics[
%natheight=3.000000in,
%natwidth=3.000000in,
height=0.1773in,
width=0.1773in
]%
{ut03.ps}%
}%
%EndExpansion
&
%TCIMACRO{\FRAME{itbpF}{0.1773in}{0.1773in}{0in}{}{}{ut05.ps}%
%width 0.1773in;  height 0.1773in;  depth 0in;  original-width 3in;
%original-height 3in;  cropleft "0";  croptop "1";  cropright "1";
%cropbottom "0";  filename 'ut05.ps';file-properties "XNPEU";}}}%
%BeginExpansion
{\includegraphics[
%natheight=3.000000in,
%natwidth=3.000000in,
height=0.1773in,
width=0.1773in
]%
{ut05.ps}%
}%
%EndExpansion
&
%TCIMACRO{\FRAME{itbpF}{0.1773in}{0.1773in}{0in}{}{}{ut04.ps}%
%width 0.1773in;  height 0.1773in;  depth 0in;  original-width 3in;
%original-height 3in;  cropleft "0";  croptop "1";  cropright "1";
%cropbottom "0";  filename 'ut04.ps';file-properties "XNPEU";}}}%
%BeginExpansion
{\includegraphics[
%natheight=3.000000in,
%natwidth=3.000000in,
height=0.1773in,
width=0.1773in
]%
{ut04.ps}%
}%
%EndExpansion
&
%TCIMACRO{\FRAME{itbpF}{0.1773in}{0.1773in}{0in}{}{}{ut00.ps}%
%width 0.1773in;  height 0.1773in;  depth 0in;  original-width 3in;
%original-height 3in;  cropleft "0";  croptop "1";  cropright "1";
%cropbottom "0";  filename 'ut00.ps';file-properties "XNPEU";}}}%
%BeginExpansion
{\includegraphics[
%natheight=3.000000in,
%natwidth=3.000000in,
height=0.1773in,
width=0.1773in
]%
{ut00.ps}%
}%
%EndExpansion
\end{array}
\right\rangle +\left\vert
\begin{array}
[c]{cccc}%
%TCIMACRO{\FRAME{itbpF}{0.1773in}{0.1773in}{0in}{}{}{ut00.ps}%
%width 0.1773in;  height 0.1773in;  depth 0in;  original-width 3in;
%original-height 3in;  cropleft "0";  croptop "1";  cropright "1";
%cropbottom "0";  filename 'ut00.ps';file-properties "XNPEU";}}}%
%BeginExpansion
{\includegraphics[
%natheight=3.000000in,
%natwidth=3.000000in,
height=0.1773in,
width=0.1773in
]%
{ut00.ps}%
}%
%EndExpansion
&
%TCIMACRO{\FRAME{itbpF}{0.1773in}{0.1773in}{0in}{}{}{ut02.ps}%
%width 0.1773in;  height 0.1773in;  depth 0in;  original-width 3in;
%original-height 3in;  cropleft "0";  croptop "1";  cropright "1";
%cropbottom "0";  filename 'ut02.ps';file-properties "XNPEU";}}}%
%BeginExpansion
{\includegraphics[
%natheight=3.000000in,
%natwidth=3.000000in,
height=0.1773in,
width=0.1773in
]%
{ut02.ps}%
}%
%EndExpansion
&
%TCIMACRO{\FRAME{itbpF}{0.1773in}{0.1773in}{0in}{}{}{ut01.ps}%
%width 0.1773in;  height 0.1773in;  depth 0in;  original-width 3in;
%original-height 3in;  cropleft "0";  croptop "1";  cropright "1";
%cropbottom "0";  filename 'ut01.ps';file-properties "XNPEU";}}}%
%BeginExpansion
{\includegraphics[
%natheight=3.000000in,
%natwidth=3.000000in,
height=0.1773in,
width=0.1773in
]%
{ut01.ps}%
}%
%EndExpansion
&
%TCIMACRO{\FRAME{itbpF}{0.1773in}{0.1773in}{0in}{}{}{ut00.ps}%
%width 0.1773in;  height 0.1773in;  depth 0in;  original-width 3in;
%original-height 3in;  cropleft "0";  croptop "1";  cropright "1";
%cropbottom "0";  filename 'ut00.ps';file-properties "XNPEU";}}}%
%BeginExpansion
{\includegraphics[
%natheight=3.000000in,
%natwidth=3.000000in,
height=0.1773in,
width=0.1773in
]%
{ut00.ps}%
}%
%EndExpansion
\\%
%TCIMACRO{\FRAME{itbpF}{0.1773in}{0.1773in}{0in}{}{}{ut02.ps}%
%width 0.1773in;  height 0.1773in;  depth 0in;  original-width 3in;
%original-height 3in;  cropleft "0";  croptop "1";  cropright "1";
%cropbottom "0";  filename 'ut02.ps';file-properties "XNPEU";}}}%
%BeginExpansion
{\includegraphics[
%natheight=3.000000in,
%natwidth=3.000000in,
height=0.1773in,
width=0.1773in
]%
{ut02.ps}%
}%
%EndExpansion
&
%TCIMACRO{\FRAME{itbpF}{0.1773in}{0.1773in}{0in}{}{}{ut09.ps}%
%width 0.1773in;  height 0.1773in;  depth 0in;  original-width 3in;
%original-height 3in;  cropleft "0";  croptop "1";  cropright "1";
%cropbottom "0";  filename 'ut09.ps';file-properties "XNPEU";}}}%
%BeginExpansion
{\includegraphics[
%natheight=3.000000in,
%natwidth=3.000000in,
height=0.1773in,
width=0.1773in
]%
{ut09.ps}%
}%
%EndExpansion
&
%TCIMACRO{\FRAME{itbpF}{0.1773in}{0.1773in}{0in}{}{}{ut07.ps}%
%width 0.1773in;  height 0.1773in;  depth 0in;  original-width 3in;
%original-height 3in;  cropleft "0";  croptop "1";  cropright "1";
%cropbottom "0";  filename 'ut07.ps';file-properties "XNPEU";}}}%
%BeginExpansion
{\includegraphics[
%natheight=3.000000in,
%natwidth=3.000000in,
height=0.1773in,
width=0.1773in
]%
{ut07.ps}%
}%
%EndExpansion
&
%TCIMACRO{\FRAME{itbpF}{0.1773in}{0.1773in}{0in}{}{}{ut01.ps}%
%width 0.1773in;  height 0.1773in;  depth 0in;  original-width 3in;
%original-height 3in;  cropleft "0";  croptop "1";  cropright "1";
%cropbottom "0";  filename 'ut01.ps';file-properties "XNPEU";}}}%
%BeginExpansion
{\includegraphics[
%natheight=3.000000in,
%natwidth=3.000000in,
height=0.1773in,
width=0.1773in
]%
{ut01.ps}%
}%
%EndExpansion
\\%
%TCIMACRO{\FRAME{itbpF}{0.1773in}{0.1773in}{0in}{}{}{ut03.ps}%
%width 0.1773in;  height 0.1773in;  depth 0in;  original-width 3in;
%original-height 3in;  cropleft "0";  croptop "1";  cropright "1";
%cropbottom "0";  filename 'ut03.ps';file-properties "XNPEU";}}}%
%BeginExpansion
{\includegraphics[
%natheight=3.000000in,
%natwidth=3.000000in,
height=0.1773in,
width=0.1773in
]%
{ut03.ps}%
}%
%EndExpansion
&
%TCIMACRO{\FRAME{itbpF}{0.1773in}{0.1773in}{0in}{}{}{ut07.ps}%
%width 0.1773in;  height 0.1773in;  depth 0in;  original-width 3in;
%original-height 3in;  cropleft "0";  croptop "1";  cropright "1";
%cropbottom "0";  filename 'ut07.ps';file-properties "XNPEU";}}}%
%BeginExpansion
{\includegraphics[
%natheight=3.000000in,
%natwidth=3.000000in,
height=0.1773in,
width=0.1773in
]%
{ut07.ps}%
}%
%EndExpansion
&
%TCIMACRO{\FRAME{itbpF}{0.1773in}{0.1773in}{0in}{}{}{ut09.ps}%
%width 0.1773in;  height 0.1773in;  depth 0in;  original-width 3in;
%original-height 3in;  cropleft "0";  croptop "1";  cropright "1";
%cropbottom "0";  filename 'ut09.ps';file-properties "XNPEU";}}}%
%BeginExpansion
{\includegraphics[
%natheight=3.000000in,
%natwidth=3.000000in,
height=0.1773in,
width=0.1773in
]%
{ut09.ps}%
}%
%EndExpansion
&
%TCIMACRO{\FRAME{itbpF}{0.1773in}{0.1773in}{0in}{}{}{ut04.ps}%
%width 0.1773in;  height 0.1773in;  depth 0in;  original-width 3in;
%original-height 3in;  cropleft "0";  croptop "1";  cropright "1";
%cropbottom "0";  filename 'ut04.ps';file-properties "XNPEU";}}}%
%BeginExpansion
{\includegraphics[
%natheight=3.000000in,
%natwidth=3.000000in,
height=0.1773in,
width=0.1773in
]%
{ut04.ps}%
}%
%EndExpansion
\\%
%TCIMACRO{\FRAME{itbpF}{0.1773in}{0.1773in}{0in}{}{}{ut00.ps}%
%width 0.1773in;  height 0.1773in;  depth 0in;  original-width 3in;
%original-height 3in;  cropleft "0";  croptop "1";  cropright "1";
%cropbottom "0";  filename 'ut00.ps';file-properties "XNPEU";}}}%
%BeginExpansion
{\includegraphics[
%natheight=3.000000in,
%natwidth=3.000000in,
height=0.1773in,
width=0.1773in
]%
{ut00.ps}%
}%
%EndExpansion
&
%TCIMACRO{\FRAME{itbpF}{0.1773in}{0.1773in}{0in}{}{}{ut03.ps}%
%width 0.1773in;  height 0.1773in;  depth 0in;  original-width 3in;
%original-height 3in;  cropleft "0";  croptop "1";  cropright "1";
%cropbottom "0";  filename 'ut03.ps';file-properties "XNPEU";}}}%
%BeginExpansion
{\includegraphics[
%natheight=3.000000in,
%natwidth=3.000000in,
height=0.1773in,
width=0.1773in
]%
{ut03.ps}%
}%
%EndExpansion
&
%TCIMACRO{\FRAME{itbpF}{0.1773in}{0.1773in}{0in}{}{}{ut04.ps}%
%width 0.1773in;  height 0.1773in;  depth 0in;  original-width 3in;
%original-height 3in;  cropleft "0";  croptop "1";  cropright "1";
%cropbottom "0";  filename 'ut04.ps';file-properties "XNPEU";}}}%
%BeginExpansion
{\includegraphics[
%natheight=3.000000in,
%natwidth=3.000000in,
height=0.1773in,
width=0.1773in
]%
{ut04.ps}%
}%
%EndExpansion
&
%TCIMACRO{\FRAME{itbpF}{0.1773in}{0.1773in}{0in}{}{}{ut00.ps}%
%width 0.1773in;  height 0.1773in;  depth 0in;  original-width 3in;
%original-height 3in;  cropleft "0";  croptop "1";  cropright "1";
%cropbottom "0";  filename 'ut00.ps';file-properties "XNPEU";}}}%
%BeginExpansion
{\includegraphics[
%natheight=3.000000in,
%natwidth=3.000000in,
height=0.1773in,
width=0.1773in
]%
{ut00.ps}%
}%
%EndExpansion
\end{array}
\right\rangle }{\sqrt{2}}%
\]

\bigskip

Our next step is to identify each element $g$ of the ambient group
$\mathbb{A}(n)$ with the linear transformation%
\[%
\begin{array}
[c]{ccc}%
\mathcal{K}^{(n)} & \overset{g}{\longrightarrow} & \mathcal{K}^{(n)}\\
\left\vert K\right\rangle  & \longmapsto & \left\vert gK\right\rangle
\end{array}
\text{ .}%
\]
This is a unitary transformation, since the element $g$ simply permutes the
basis elements of $\mathcal{K}^{(n)}$. \ In this way, the ambient group
$\mathbb{A}(n)$ is identified with a discrete unitary subgroup (also denoted
by $\mathbb{A}(n)$ ) of the group $U\left(  \mathcal{K}^{(n)}\right)  $, where
$U\left(  \mathcal{K}^{(n)}\right)  $ denotes the \textbf{group of all unitary
transformations on the Hilbert space} $\mathcal{K}^{(n)}$. \ The unitary
subgroup $\mathbb{A}(n)$ will also be called the\textbf{ ambient group}.

\bigskip

Finally, everything comes together with the following definition:

\bigskip

\begin{definition}
Let $n$ be a positive integer. \ A \textbf{quantum knot system} $Q\left(
\mathcal{K}^{(n)},\mathbb{A}(n)\right)  $ \textbf{of order }$n$ is a quantum
system with the Hilbert space $\mathcal{K}^{(n)}$ of knot $n$-mosaics as its
state space, and having the ambient group $\mathbb{A}(n)$ as an accessible
unitary control group. \ The states of the quantum system $Q\left(
\mathcal{K}^{(n)},\mathbb{A}(n)\right)  $ are called \textbf{quantum knots of
order }$n$, and the elements of the ambient group $\mathbb{A}(n)$ are called
\textbf{unitary knot moves}. \ Moreover, the quantum knot system $Q\left(
\mathcal{K}^{(n)},\mathbb{A}(n)\right)  $ of\textbf{ }order\textbf{ }$n$ is a
subsystem of the quantum knot system $Q\left(  \mathcal{K}^{(n+1)}%
,\mathbb{A}(n+1)\right)  $ of order\textbf{ }$n+1$. \ Thus, the quantum knot
systems $Q\left(  \mathcal{K}^{(n)},\mathbb{A}(n)\right)  $ collectively
become a \textbf{nested} \textbf{sequence of quantum knot systems} which we
will denote simply by $Q\left(  \mathcal{K},\mathbb{A}\right)  $. \ In other
words,
\[
Q\left(  \mathcal{K},\mathbb{A}\right)  =Q\left(  \mathcal{K}^{(1)}%
,\mathbb{A}(1)\right)  \longrightarrow Q\left(  \mathcal{K}^{(2)}%
,\mathbb{A}(2)\right)  \longrightarrow\cdots\longrightarrow Q\left(
\mathcal{K}^{(n)},\mathbb{A}(n)\right)  \longrightarrow\cdots
\]

\end{definition}

\bigskip

\begin{remark}
The nested quantum knot system $Q\left(  \mathcal{K},\mathbb{A}\right)  $ is
probably not a physically realizable system. \ However, each quantum knot
system $Q\left(  \mathcal{K}^{(n)},\mathbb{A}(n)\right)  $ of order $n$ is
physically realizable\footnote{It should be mentioned that, although the
quantum knot system $Q\left(  \mathcal{K}^{(n)},\mathbb{A}(n)\right)  $ is
physically realizable, it may or may not be implentable within today's
existing technology.}. \ 
\end{remark}

\bigskip

\begin{example}
As an example, in the quantum system $Q\left(  \mathcal{K}^{(5)}%
,\mathbb{A}(5)\right)  $, we see that the action of the unitary Reidemeister 2
move
\[%
\begin{array}
[c]{cc}%
%TCIMACRO{\FRAME{itbpF}{0.1773in}{0.1773in}{0in}{}{}{ut10.ps}%
%width 0.1773in;  height 0.1773in;  depth 0in;  original-width 3in;
%original-height 3in;  cropleft "0";  croptop "1";  cropright "1";
%cropbottom "0";  filename 'ut10.ps';file-properties "XNPEU";}}}%
%BeginExpansion
{\includegraphics[
%natheight=3.000000in,
%natwidth=3.000000in,
height=0.1773in,
width=0.1773in
]%
{ut10.ps}%
}%
%EndExpansion
&
%TCIMACRO{\FRAME{itbpF}{0.1773in}{0.1773in}{0in}{}{}{ut01.ps}%
%width 0.1773in;  height 0.1773in;  depth 0in;  original-width 3in;
%original-height 3in;  cropleft "0";  croptop "1";  cropright "1";
%cropbottom "0";  filename 'ut01.ps';file-properties "XNPEU";}}}%
%BeginExpansion
{\includegraphics[
%natheight=3.000000in,
%natwidth=3.000000in,
height=0.1773in,
width=0.1773in
]%
{ut01.ps}%
}%
%EndExpansion
\\
\end{array}
\overset{(2,1)}{\longleftrightarrow}%
\begin{array}
[c]{cc}%
%TCIMACRO{\FRAME{itbpF}{0.1773in}{0.1773in}{0in}{}{}{ut07.ps}%
%width 0.1773in;  height 0.1773in;  depth 0in;  original-width 3in;
%original-height 3in;  cropleft "0";  croptop "1";  cropright "1";
%cropbottom "0";  filename 'ut07.ps';file-properties "XNPEU";}}}%
%BeginExpansion
{\includegraphics[
%natheight=3.000000in,
%natwidth=3.000000in,
height=0.1773in,
width=0.1773in
]%
{ut07.ps}%
}%
%EndExpansion
&
%TCIMACRO{\FRAME{itbpF}{0.1773in}{0.1773in}{0in}{}{}{ut01.ps}%
%width 0.1773in;  height 0.1773in;  depth 0in;  original-width 3in;
%original-height 3in;  cropleft "0";  croptop "1";  cropright "1";
%cropbottom "0";  filename 'ut01.ps';file-properties "XNPEU";}}}%
%BeginExpansion
{\includegraphics[
%natheight=3.000000in,
%natwidth=3.000000in,
height=0.1773in,
width=0.1773in
]%
{ut01.ps}%
}%
%EndExpansion
\\
\end{array}
\in\mathbb{A}(5)
\]
on the quantum knot
\[
\frac{\left\vert
\begin{array}
[c]{ccccc}%
%TCIMACRO{\FRAME{itbpF}{0.1773in}{0.1773in}{0in}{}{}{ut00.ps}%
%width 0.1773in;  height 0.1773in;  depth 0in;  original-width 3in;
%original-height 3in;  cropleft "0";  croptop "1";  cropright "1";
%cropbottom "0";  filename 'ut00.ps';file-properties "XNPEU";}}}%
%BeginExpansion
{\includegraphics[
%natheight=3.000000in,
%natwidth=3.000000in,
height=0.1773in,
width=0.1773in
]%
{ut00.ps}%
}%
%EndExpansion
&
%TCIMACRO{\FRAME{itbpF}{0.1773in}{0.1773in}{0in}{}{}{ut02.ps}%
%width 0.1773in;  height 0.1773in;  depth 0in;  original-width 3in;
%original-height 3in;  cropleft "0";  croptop "1";  cropright "1";
%cropbottom "0";  filename 'ut02.ps';file-properties "XNPEU";}}}%
%BeginExpansion
{\includegraphics[
%natheight=3.000000in,
%natwidth=3.000000in,
height=0.1773in,
width=0.1773in
]%
{ut02.ps}%
}%
%EndExpansion
&
%TCIMACRO{\FRAME{itbpF}{0.1773in}{0.1773in}{0in}{}{}{ut05.ps}%
%width 0.1773in;  height 0.1773in;  depth 0in;  original-width 3in;
%original-height 3in;  cropleft "0";  croptop "1";  cropright "1";
%cropbottom "0";  filename 'ut05.ps';file-properties "XNPEU";}}}%
%BeginExpansion
{\includegraphics[
%natheight=3.000000in,
%natwidth=3.000000in,
height=0.1773in,
width=0.1773in
]%
{ut05.ps}%
}%
%EndExpansion
&
%TCIMACRO{\FRAME{itbpF}{0.1773in}{0.1773in}{0in}{}{}{ut01.ps}%
%width 0.1773in;  height 0.1773in;  depth 0in;  original-width 3in;
%original-height 3in;  cropleft "0";  croptop "1";  cropright "1";
%cropbottom "0";  filename 'ut01.ps';file-properties "XNPEU";}}}%
%BeginExpansion
{\includegraphics[
%natheight=3.000000in,
%natwidth=3.000000in,
height=0.1773in,
width=0.1773in
]%
{ut01.ps}%
}%
%EndExpansion
&
%TCIMACRO{\FRAME{itbpF}{0.1773in}{0.1773in}{0in}{}{}{ut00.ps}%
%width 0.1773in;  height 0.1773in;  depth 0in;  original-width 3in;
%original-height 3in;  cropleft "0";  croptop "1";  cropright "1";
%cropbottom "0";  filename 'ut00.ps';file-properties "XNPEU";}}}%
%BeginExpansion
{\includegraphics[
%natheight=3.000000in,
%natwidth=3.000000in,
height=0.1773in,
width=0.1773in
]%
{ut00.ps}%
}%
%EndExpansion
\\
\end{array}%
\right\rangle }{\sqrt{2}}%
\]

produces the quantum knot
\[
\frac{\left\vert
%
\right\rangle }{\sqrt{2}}\text{ .}%
\]

\end{example}

\bigskip

\subsection{Quantum knot type}

\bigskip

When are two quantum knots the same?

\bigskip

\begin{definition}
Let $\left\vert \psi_{1}\right\rangle $ and $\left\vert \psi_{2}\right\rangle
$ be two quantum knots of a quantum system $Q\left(  \mathcal{K}%
^{(n)},\mathbb{A}(n)\right)  $ of order $n$. \ Then $\left\vert \psi
_{1}\right\rangle $ and $\left\vert \psi_{2}\right\rangle $ are said to be of
the \textbf{same quantum knot }$\mathbf{n}$\textbf{-type}, written%
\[
\left\vert \psi_{1}\right\rangle \underset{n}{\thicksim}\left\vert \psi
_{2}\right\rangle \text{ ,}%
\]
provided there exists a unitary transformation $g$ in the ambient group
$\mathbb{A}(n)$ which transforms $\left\vert \psi_{1}\right\rangle $ into
$\left\vert \psi_{2}\right\rangle $, i.e., such that%
\[
g\left\vert \psi_{1}\right\rangle =\left\vert \psi_{2}\right\rangle \text{ .}%
\]
They are said to be of the \textbf{same quantum knot type}, written
\[
\left\vert \psi_{1}\right\rangle \thicksim\left\vert \psi_{2}\right\rangle
\text{ ,}%
\]
provided that for some non-negative integer $\ell$,
\[
\iota^{\ell}\left\vert \psi_{1}\right\rangle \underset{n+\ell}{\thicksim}%
\iota^{\ell}\left\vert \psi_{2}\right\rangle \text{ ,}%
\]
where $\iota^{\ell}\left\vert \psi_{1}\right\rangle $ and $\iota^{\ell
}\left\vert \psi_{2}\right\rangle $ are states of the quantum system $Q\left(
\mathcal{K}^{(n+\ell)},\mathbb{A}(n+\ell)\right)  $, where $\iota
:\mathcal{K}^{(m)}\longrightarrow\mathcal{K}^{(m+1)}$ is the monomorphism
induced by the previously defined injection $\iota:\mathbb{K}^{(m)}%
\longrightarrow\mathbb{K}^{(m+1)}$.
\end{definition}

\bigskip

Thus, the two quantum knots found in the last example of the previous section
are of the same quantum knot $5$-type, and also of the same quantum knot type.
\ Surprisingly, the following two quantum knots $\left\vert \psi
_{1}\right\rangle $ and $\left\vert \psi_{2}\right\rangle $ are neither of the
same quantum knot $3$-type nor knot type:
\[
\left\vert \psi_{1}\right\rangle =\left\vert
\begin{array}
[c]{ccc}%
%TCIMACRO{\FRAME{itbpF}{0.1773in}{0.1773in}{0in}{}{}{ut02.ps}%
%width 0.1773in;  height 0.1773in;  depth 0in;  original-width 3in;
%original-height 3in;  cropleft "0";  croptop "1";  cropright "1";
%cropbottom "0";  filename 'ut02.ps';file-properties "XNPEU";}}}%
%BeginExpansion
{\includegraphics[
%natheight=3.000000in,
%natwidth=3.000000in,
height=0.1773in,
width=0.1773in
]%
{ut02.ps}%
}%
%EndExpansion
&
%TCIMACRO{\FRAME{itbpF}{0.1773in}{0.1773in}{0in}{}{}{ut05.ps}%
%width 0.1773in;  height 0.1773in;  depth 0in;  original-width 3in;
%original-height 3in;  cropleft "0";  croptop "1";  cropright "1";
%cropbottom "0";  filename 'ut05.ps';file-properties "XNPEU";}}}%
%BeginExpansion
{\includegraphics[
%natheight=3.000000in,
%natwidth=3.000000in,
height=0.1773in,
width=0.1773in
]%
{ut05.ps}%
}%
%EndExpansion
&
%TCIMACRO{\FRAME{itbpF}{0.1773in}{0.1773in}{0in}{}{}{ut01.ps}%
%width 0.1773in;  height 0.1773in;  depth 0in;  original-width 3in;
%original-height 3in;  cropleft "0";  croptop "1";  cropright "1";
%cropbottom "0";  filename 'ut01.ps';file-properties "XNPEU";}}}%
%BeginExpansion
{\includegraphics[
%natheight=3.000000in,
%natwidth=3.000000in,
height=0.1773in,
width=0.1773in
]%
{ut01.ps}%
}%
%EndExpansion
\\%
%TCIMACRO{\FRAME{itbpF}{0.1773in}{0.1773in}{0in}{}{}{ut06.ps}%
%width 0.1773in;  height 0.1773in;  depth 0in;  original-width 3in;
%original-height 3in;  cropleft "0";  croptop "1";  cropright "1";
%cropbottom "0";  filename 'ut06.ps';file-properties "XNPEU";}}}%
%BeginExpansion
{\includegraphics[
%natheight=3.000000in,
%natwidth=3.000000in,
height=0.1773in,
width=0.1773in
]%
{ut06.ps}%
}%
%EndExpansion
&
%TCIMACRO{\FRAME{itbpF}{0.1773in}{0.1773in}{0in}{}{}{ut00.ps}%
%width 0.1773in;  height 0.1773in;  depth 0in;  original-width 3in;
%original-height 3in;  cropleft "0";  croptop "1";  cropright "1";
%cropbottom "0";  filename 'ut00.ps';file-properties "XNPEU";}}}%
%BeginExpansion
{\includegraphics[
%natheight=3.000000in,
%natwidth=3.000000in,
height=0.1773in,
width=0.1773in
]%
{ut00.ps}%
}%
%EndExpansion
&
%TCIMACRO{\FRAME{itbpF}{0.1773in}{0.1773in}{0in}{}{}{ut06.ps}%
%width 0.1773in;  height 0.1773in;  depth 0in;  original-width 3in;
%original-height 3in;  cropleft "0";  croptop "1";  cropright "1";
%cropbottom "0";  filename 'ut06.ps';file-properties "XNPEU";}}}%
%BeginExpansion
{\includegraphics[
%natheight=3.000000in,
%natwidth=3.000000in,
height=0.1773in,
width=0.1773in
]%
{ut06.ps}%
}%
%EndExpansion
\\%
%TCIMACRO{\FRAME{itbpF}{0.1773in}{0.1773in}{0in}{}{}{ut03.ps}%
%width 0.1773in;  height 0.1773in;  depth 0in;  original-width 3in;
%original-height 3in;  cropleft "0";  croptop "1";  cropright "1";
%cropbottom "0";  filename 'ut03.ps';file-properties "XNPEU";}}}%
%BeginExpansion
{\includegraphics[
%natheight=3.000000in,
%natwidth=3.000000in,
height=0.1773in,
width=0.1773in
]%
{ut03.ps}%
}%
%EndExpansion
&
%TCIMACRO{\FRAME{itbpF}{0.1773in}{0.1773in}{0in}{}{}{ut05.ps}%
%width 0.1773in;  height 0.1773in;  depth 0in;  original-width 3in;
%original-height 3in;  cropleft "0";  croptop "1";  cropright "1";
%cropbottom "0";  filename 'ut05.ps';file-properties "XNPEU";}}}%
%BeginExpansion
{\includegraphics[
%natheight=3.000000in,
%natwidth=3.000000in,
height=0.1773in,
width=0.1773in
]%
{ut05.ps}%
}%
%EndExpansion
&
%TCIMACRO{\FRAME{itbpF}{0.1773in}{0.1773in}{0in}{}{}{ut04.ps}%
%width 0.1773in;  height 0.1773in;  depth 0in;  original-width 3in;
%original-height 3in;  cropleft "0";  croptop "1";  cropright "1";
%cropbottom "0";  filename 'ut04.ps';file-properties "XNPEU";}}}%
%BeginExpansion
{\includegraphics[
%natheight=3.000000in,
%natwidth=3.000000in,
height=0.1773in,
width=0.1773in
]%
{ut04.ps}%
}%
%EndExpansion
\end{array}
\right\rangle \text{ \ \ and \ \ }\left\vert \psi_{2}\right\rangle =\frac
{1}{\sqrt{2}}\left(  \left\vert
\begin{array}
[c]{ccc}%
%TCIMACRO{\FRAME{itbpF}{0.1773in}{0.1773in}{0in}{}{}{ut02.ps}%
%width 0.1773in;  height 0.1773in;  depth 0in;  original-width 3in;
%original-height 3in;  cropleft "0";  croptop "1";  cropright "1";
%cropbottom "0";  filename 'ut02.ps';file-properties "XNPEU";}}}%
%BeginExpansion
{\includegraphics[
%natheight=3.000000in,
%natwidth=3.000000in,
height=0.1773in,
width=0.1773in
]%
{ut02.ps}%
}%
%EndExpansion
&
%TCIMACRO{\FRAME{itbpF}{0.1773in}{0.1773in}{0in}{}{}{ut01.ps}%
%width 0.1773in;  height 0.1773in;  depth 0in;  original-width 3in;
%original-height 3in;  cropleft "0";  croptop "1";  cropright "1";
%cropbottom "0";  filename 'ut01.ps';file-properties "XNPEU";}}}%
%BeginExpansion
{\includegraphics[
%natheight=3.000000in,
%natwidth=3.000000in,
height=0.1773in,
width=0.1773in
]%
{ut01.ps}%
}%
%EndExpansion
&
%TCIMACRO{\FRAME{itbpF}{0.1773in}{0.1773in}{0in}{}{}{ut00.ps}%
%width 0.1773in;  height 0.1773in;  depth 0in;  original-width 3in;
%original-height 3in;  cropleft "0";  croptop "1";  cropright "1";
%cropbottom "0";  filename 'ut00.ps';file-properties "XNPEU";}}}%
%BeginExpansion
{\includegraphics[
%natheight=3.000000in,
%natwidth=3.000000in,
height=0.1773in,
width=0.1773in
]%
{ut00.ps}%
}%
%EndExpansion
\\%
%TCIMACRO{\FRAME{itbpF}{0.1773in}{0.1773in}{0in}{}{}{ut03.ps}%
%width 0.1773in;  height 0.1773in;  depth 0in;  original-width 3in;
%original-height 3in;  cropleft "0";  croptop "1";  cropright "1";
%cropbottom "0";  filename 'ut03.ps';file-properties "XNPEU";}}}%
%BeginExpansion
{\includegraphics[
%natheight=3.000000in,
%natwidth=3.000000in,
height=0.1773in,
width=0.1773in
]%
{ut03.ps}%
}%
%EndExpansion
&
%TCIMACRO{\FRAME{itbpF}{0.1773in}{0.1773in}{0in}{}{}{ut04.ps}%
%width 0.1773in;  height 0.1773in;  depth 0in;  original-width 3in;
%original-height 3in;  cropleft "0";  croptop "1";  cropright "1";
%cropbottom "0";  filename 'ut04.ps';file-properties "XNPEU";}}}%
%BeginExpansion
{\includegraphics[
%natheight=3.000000in,
%natwidth=3.000000in,
height=0.1773in,
width=0.1773in
]%
{ut04.ps}%
}%
%EndExpansion
&
%TCIMACRO{\FRAME{itbpF}{0.1773in}{0.1773in}{0in}{}{}{ut00.ps}%
%width 0.1773in;  height 0.1773in;  depth 0in;  original-width 3in;
%original-height 3in;  cropleft "0";  croptop "1";  cropright "1";
%cropbottom "0";  filename 'ut00.ps';file-properties "XNPEU";}}}%
%BeginExpansion
{\includegraphics[
%natheight=3.000000in,
%natwidth=3.000000in,
height=0.1773in,
width=0.1773in
]%
{ut00.ps}%
}%
%EndExpansion
\\%
%TCIMACRO{\FRAME{itbpF}{0.1773in}{0.1773in}{0in}{}{}{ut00.ps}%
%width 0.1773in;  height 0.1773in;  depth 0in;  original-width 3in;
%original-height 3in;  cropleft "0";  croptop "1";  cropright "1";
%cropbottom "0";  filename 'ut00.ps';file-properties "XNPEU";}}}%
%BeginExpansion
{\includegraphics[
%natheight=3.000000in,
%natwidth=3.000000in,
height=0.1773in,
width=0.1773in
]%
{ut00.ps}%
}%
%EndExpansion
&
%TCIMACRO{\FRAME{itbpF}{0.1773in}{0.1773in}{0in}{}{}{ut00.ps}%
%width 0.1773in;  height 0.1773in;  depth 0in;  original-width 3in;
%original-height 3in;  cropleft "0";  croptop "1";  cropright "1";
%cropbottom "0";  filename 'ut00.ps';file-properties "XNPEU";}}}%
%BeginExpansion
{\includegraphics[
%natheight=3.000000in,
%natwidth=3.000000in,
height=0.1773in,
width=0.1773in
]%
{ut00.ps}%
}%
%EndExpansion
&
%TCIMACRO{\FRAME{itbpF}{0.1773in}{0.1773in}{0in}{}{}{ut00.ps}%
%width 0.1773in;  height 0.1773in;  depth 0in;  original-width 3in;
%original-height 3in;  cropleft "0";  croptop "1";  cropright "1";
%cropbottom "0";  filename 'ut00.ps';file-properties "XNPEU";}}}%
%BeginExpansion
{\includegraphics[
%natheight=3.000000in,
%natwidth=3.000000in,
height=0.1773in,
width=0.1773in
]%
{ut00.ps}%
}%
%EndExpansion
\end{array}
\right\rangle +\left\vert
\begin{array}
[c]{ccc}%
%TCIMACRO{\FRAME{itbpF}{0.1773in}{0.1773in}{0in}{}{}{ut00.ps}%
%width 0.1773in;  height 0.1773in;  depth 0in;  original-width 3in;
%original-height 3in;  cropleft "0";  croptop "1";  cropright "1";
%cropbottom "0";  filename 'ut00.ps';file-properties "XNPEU";}}}%
%BeginExpansion
{\includegraphics[
%natheight=3.000000in,
%natwidth=3.000000in,
height=0.1773in,
width=0.1773in
]%
{ut00.ps}%
}%
%EndExpansion
&
%TCIMACRO{\FRAME{itbpF}{0.1773in}{0.1773in}{0in}{}{}{ut00.ps}%
%width 0.1773in;  height 0.1773in;  depth 0in;  original-width 3in;
%original-height 3in;  cropleft "0";  croptop "1";  cropright "1";
%cropbottom "0";  filename 'ut00.ps';file-properties "XNPEU";}}}%
%BeginExpansion
{\includegraphics[
%natheight=3.000000in,
%natwidth=3.000000in,
height=0.1773in,
width=0.1773in
]%
{ut00.ps}%
}%
%EndExpansion
&
%TCIMACRO{\FRAME{itbpF}{0.1773in}{0.1773in}{0in}{}{}{ut00.ps}%
%width 0.1773in;  height 0.1773in;  depth 0in;  original-width 3in;
%original-height 3in;  cropleft "0";  croptop "1";  cropright "1";
%cropbottom "0";  filename 'ut00.ps';file-properties "XNPEU";}}}%
%BeginExpansion
{\includegraphics[
%natheight=3.000000in,
%natwidth=3.000000in,
height=0.1773in,
width=0.1773in
]%
{ut00.ps}%
}%
%EndExpansion
\\%
%TCIMACRO{\FRAME{itbpF}{0.1773in}{0.1773in}{0in}{}{}{ut00.ps}%
%width 0.1773in;  height 0.1773in;  depth 0in;  original-width 3in;
%original-height 3in;  cropleft "0";  croptop "1";  cropright "1";
%cropbottom "0";  filename 'ut00.ps';file-properties "XNPEU";}}}%
%BeginExpansion
{\includegraphics[
%natheight=3.000000in,
%natwidth=3.000000in,
height=0.1773in,
width=0.1773in
]%
{ut00.ps}%
}%
%EndExpansion
&
%TCIMACRO{\FRAME{itbpF}{0.1773in}{0.1773in}{0in}{}{}{ut02.ps}%
%width 0.1773in;  height 0.1773in;  depth 0in;  original-width 3in;
%original-height 3in;  cropleft "0";  croptop "1";  cropright "1";
%cropbottom "0";  filename 'ut02.ps';file-properties "XNPEU";}}}%
%BeginExpansion
{\includegraphics[
%natheight=3.000000in,
%natwidth=3.000000in,
height=0.1773in,
width=0.1773in
]%
{ut02.ps}%
}%
%EndExpansion
&
%TCIMACRO{\FRAME{itbpF}{0.1773in}{0.1773in}{0in}{}{}{ut01.ps}%
%width 0.1773in;  height 0.1773in;  depth 0in;  original-width 3in;
%original-height 3in;  cropleft "0";  croptop "1";  cropright "1";
%cropbottom "0";  filename 'ut01.ps';file-properties "XNPEU";}}}%
%BeginExpansion
{\includegraphics[
%natheight=3.000000in,
%natwidth=3.000000in,
height=0.1773in,
width=0.1773in
]%
{ut01.ps}%
}%
%EndExpansion
\\%
%TCIMACRO{\FRAME{itbpF}{0.1773in}{0.1773in}{0in}{}{}{ut00.ps}%
%width 0.1773in;  height 0.1773in;  depth 0in;  original-width 3in;
%original-height 3in;  cropleft "0";  croptop "1";  cropright "1";
%cropbottom "0";  filename 'ut00.ps';file-properties "XNPEU";}}}%
%BeginExpansion
{\includegraphics[
%natheight=3.000000in,
%natwidth=3.000000in,
height=0.1773in,
width=0.1773in
]%
{ut00.ps}%
}%
%EndExpansion
&
%TCIMACRO{\FRAME{itbpF}{0.1773in}{0.1773in}{0in}{}{}{ut03.ps}%
%width 0.1773in;  height 0.1773in;  depth 0in;  original-width 3in;
%original-height 3in;  cropleft "0";  croptop "1";  cropright "1";
%cropbottom "0";  filename 'ut03.ps';file-properties "XNPEU";}}}%
%BeginExpansion
{\includegraphics[
%natheight=3.000000in,
%natwidth=3.000000in,
height=0.1773in,
width=0.1773in
]%
{ut03.ps}%
}%
%EndExpansion
&
%TCIMACRO{\FRAME{itbpF}{0.1773in}{0.1773in}{0in}{}{}{ut04.ps}%
%width 0.1773in;  height 0.1773in;  depth 0in;  original-width 3in;
%original-height 3in;  cropleft "0";  croptop "1";  cropright "1";
%cropbottom "0";  filename 'ut04.ps';file-properties "XNPEU";}}}%
%BeginExpansion
{\includegraphics[
%natheight=3.000000in,
%natwidth=3.000000in,
height=0.1773in,
width=0.1773in
]%
{ut04.ps}%
}%
%EndExpansion
\end{array}
\right\rangle \right)
\]

\bigskip

This follows from the fact that the ambient group $\mathbb{A}(n)$ is generated
by a finite set of involutions. \ 

\bigskip

\subsection{Hamiltonians of the generators of the ambient group $\mathbb{A}$}

\bigskip

In this section, we show how to find the Hamiltonians associated with the
generators of the ambient group $\mathbb{A}\left(  n\right)  $, i.e., the
planar isotopy and Reidemeister moves for quantum knots.

\bigskip

Let $g$ be an arbitrary planar isotopy move or Reidemeister move in the
ambient group $\mathbb{A}(n)$. \ From proposition 1, we know that $g$, as a
permutation, is the product of disjoint transpositions of knot $n$-mosaics,
i.e., of the form%
\[
g=\left(  K_{\alpha_{1}},K_{\beta_{1}}\right)  \left(  K_{\alpha_{2}}%
,K_{\beta_{2}}\right)  \cdots\left(  K_{\alpha_{\ell}},K_{\beta_{\ell}%
}\right)
\]
Without loss of generality, we can assume that $\alpha_{j}<\beta_{j}$ for
$1\leq j\leq\ell$, and that $\alpha_{j}<\alpha_{j+1}$ for $1\leq j<\ell$,
where `$<$' denotes the previously defined lexicographic (lex) order on the
set of $n$-mosaics $\mathbb{M}^{(n)}$. \ For each permutation $\eta$ of
$\mathbb{M}^{(n)}$, let `$<_{\eta}$' denote the new linear ordering created by
the application of the permutation $\eta$. \ 

\bigskip

Choose a permutation $\eta$ such that
\[
K_{\alpha_{1}}<_{\eta}K_{\beta_{1}}<_{\eta}K_{\alpha_{2}}<_{\eta}K_{\beta_{2}%
}<_{\eta}\cdots<_{\eta}K_{\alpha_{\ell}}<_{\eta}K_{\beta_{\ell}}%
\]
with $K_{\beta_{\ell}}$ $<_{\eta}$ all other $n$-mosaics, and let $\sigma_{0}$
and $\sigma_{1}$ denote respectively the identity matrix and the first Pauli
spin matrix given below%
\[
\sigma_{0}=\left(
\begin{array}
[c]{cc}%
1 & 0\\
0 & 1
\end{array}
\right)  \text{ \ \ and \ \ }\sigma_{1}=\left(
\begin{array}
[c]{cc}%
0 & 1\\
1 & 0
\end{array}
\right)  \text{ .}%
\]

\bigskip

Then in the $\eta$-reordered basis of the Hilbert space $\mathcal{K}^{(n)}$,
the element $g$, as a unitary transformation, is of the form%
\[
\eta^{-1}g\eta=\left(
\begin{array}
[c]{ccccc}%
\sigma_{1} & O & \ldots & O & O\\
O & \sigma_{1} & \ldots & O & O\\
\vdots & \vdots & \ddots &  & \vdots\\
O & O & \ldots & \sigma_{1} & O\\
O & O & \ldots & O & I_{n-2\ell}%
\end{array}
\right)  =\left(  I_{\ell}\otimes\sigma_{1}\right)  \oplus I_{n-2\ell}\text{
,}%
\]
where `$O$' denotes an all zero matrix of appropriate size, where $I_{n-2\ell
}$ denotes the $\left(  n-2\ell\right)  \times\left(  n-2\ell\right)  $
identity matrix, and where `$\oplus$' denotes the direct sum of matrices,
i.e., $A\oplus B=\left(
\begin{array}
[c]{cc}%
A & O\\
O & B
\end{array}
\right)  $.

\bigskip

The natural log of $\sigma_{1}$ is%
\[
\ln\sigma_{1}=\frac{i\pi}{2}\left(  2s+1\right)  \left(  \sigma_{0}-\sigma
_{1}\right)
\]
where $s$ denotes an arbitrary integer. \ Hence, the natural log, $\ln\left(
\eta^{-1}g\eta\right)  $, of the unitary transformation $\eta^{-1}g\eta$ is

\bigskip

$\hspace{-0.85in}\frac{i\pi}{2}\left(
\begin{array}
[c]{cc}%
\left(
\begin{array}
[c]{cccc}%
\left(  2s_{1}+1\right)  \left(  \sigma_{0}-\sigma_{1}\right)  & O & \ldots &
O\\
O & \left(  2s_{2}+1\right)  \left(  \sigma_{0}-\sigma_{1}\right)  & \ldots &
O\\
\vdots & \vdots & \ddots & \vdots\\
O & O & \ldots & \left(  2s_{\ell}+1\right)  \left(  \sigma_{0}-\sigma
_{1}\right)
\end{array}
\right)  & O\\
O & O_{\left(  n-2\ell\right)  \times\left(  n-2\ell\right)  }%
\end{array}
\right)  $

\bigskip

\noindent where $s_{1},s_{2},\ldots,s_{\ell}$ are arbitrary
integers.\footnote{Let $U$ be an arbitrary finite $r\times r$ unitary matrix,
and let $W$ be a unitary matrix that diagonalizes $U$, i.e., a unitary matrix
$W$ such that $WUW^{-1}=\Delta\left(  \ \lambda(1),\lambda(2)\ldots
,\lambda(r)\ \right)  $. \ Then the natural log of $A$ is $\ln A=W^{-1}%
\Delta\left(  \ \ln\lambda(1),\ln\lambda(2)\ldots,\ln\lambda(r)\ \right)  W$.}

\bigskip

Since we are interested only in the simplest Hamiltonian, we choose the
\textbf{principal branch} $\ln_{P}$ of the natural log, i.e., the branch for
which $s_{1}=s_{2}=\cdots s_{\ell}=0$, and obtain for our Hamiltonian%
\[
H_{g}=-i\eta\left[  \ln_{P}\left(  \eta^{-1}g\eta\right)  \right]  \eta
^{-1}=\frac{\pi}{2}\eta\left(
\begin{array}
[c]{cc}%
I_{\ell}\otimes\left(  \sigma_{0}-\sigma_{1}\right)  & O\\
O & O_{\left(  n-2\ell\right)  \times\left(  n-2\ell\right)  }%
\end{array}
\right)  \eta^{-1}%
\]

\bigskip

Thus, if the initial quantum knot is%
\[
\left\vert
%
\right\rangle \text{ ,}%
\]
and if we use the Hamiltonian $H_{g}$ for for the Reidemeister 2 move
\[
g=%
%
\text{ ,}%
\]
then the solution to Schroedinger's equation is
\[
e^{i\pi t/\left(  2\hslash\right)  }\left(  \cos\left(  \frac{\pi t}{2\hslash
}\right)  \left\vert
%
\right\rangle \right)
\]
where $t$ denotes time, and where $\hslash$ denotes Planck's constant divided
by $2\pi$.

\bigskip

\subsection{Knot crossing tunnelling and other unitary transformations}

\bigskip

We should also mention a number of other miscellaneous unitary transformations
that do not lie in the ambient group $\mathbb{A}\left(  n\right)  $, but are
nonetheless of interest.

\bigskip

There is the unitary transformation $\tau_{ij}$, called the \textbf{tunnelling
transformation}, given by
\[
\tau_{ij}=%
%TCIMACRO{\FRAME{itbpF}{0.3269in}{0.3269in}{0.1003in}{}{}{ut09.ps}%
%width 0.3269in;  height 0.3269in;  depth 0.1003in;  original-width 3in;
%original-height 3in;  cropleft "0";  croptop "1";  cropright "1";
%cropbottom "0";  filename 'ut09.ps';file-properties "XNPEU";}}}%
%BeginExpansion
\raisebox{-0.1003in}{\includegraphics[
%natheight=3.000000in,
%natwidth=3.000000in,
height=0.3269in,
width=0.3269in
]%
{ut09.ps}%
}%
%EndExpansion
\overset{(i,j)}{\longleftrightarrow}%
%TCIMACRO{\FRAME{itbpF}{0.3269in}{0.3269in}{0.1003in}{}{}{ut10.ps}%
%width 0.3269in;  height 0.3269in;  depth 0.1003in;  original-width 3in;
%original-height 3in;  cropleft "0";  croptop "1";  cropright "1";
%cropbottom "0";  filename 'ut10.ps';file-properties "XNPEU";}}}%
%BeginExpansion
\raisebox{-0.1003in}{\includegraphics[
%natheight=3.000000in,
%natwidth=3.000000in,
height=0.3269in,
width=0.3269in
]%
{ut10.ps}%
}%
%EndExpansion
\text{ ,}%
\]
which enables a quantum knot overcrossing to tunnel into an undercrossing, or
vice versa.

\bigskip

From this, we can construct the \textbf{mirror image transformation} given by%
\[
\mu=%
%TCIMACRO{\dprod \limits_{i,j=0}^{n-1}}%
%BeginExpansion
{\displaystyle\prod\limits_{i,j=0}^{n-1}}
%EndExpansion
\left(
%TCIMACRO{\FRAME{itbpF}{0.3269in}{0.3269in}{0.1003in}{}{}{ut09.ps}%
%width 0.3269in;  height 0.3269in;  depth 0.1003in;  original-width 3in;
%original-height 3in;  cropleft "0";  croptop "1";  cropright "1";
%cropbottom "0";  filename 'ut09.ps';file-properties "XNPEU";}}}%
%BeginExpansion
\raisebox{-0.1003in}{\includegraphics[
%natheight=3.000000in,
%natwidth=3.000000in,
height=0.3269in,
width=0.3269in
]%
{ut09.ps}%
}%
%EndExpansion
\overset{(i,j)}{\longleftrightarrow}%
%TCIMACRO{\FRAME{itbpF}{0.3269in}{0.3269in}{0.1003in}{}{}{ut10.ps}%
%width 0.3269in;  height 0.3269in;  depth 0.1003in;  original-width 3in;
%original-height 3in;  cropleft "0";  croptop "1";  cropright "1";
%cropbottom "0";  filename 'ut10.ps';file-properties "XNPEU";}}}%
%BeginExpansion
\raisebox{-0.1003in}{\includegraphics[
%natheight=3.000000in,
%natwidth=3.000000in,
height=0.3269in,
width=0.3269in
]%
{ut10.ps}%
}%
%EndExpansion
\right)
\]
that transforms a quantum knot into its mirror image.

\bigskip

The following two unitary transformations can be used to create four
dimensional quantum knots\footnote{A four dimensional classical knot is a
knotted 2-sphere in 4-space. \ For readers interested in learning more about
higher dimensional knot theory, we refer them to \cite{Lomonaco5} and
\cite{Lomonaco6}.}. \ The first is the \textbf{hyperbolic transformation}
$\eta_{ij}$ given by%
\[
\eta_{ij}=%
%TCIMACRO{\FRAME{itbpF}{0.3269in}{0.3269in}{0.1003in}{}{}{ut07.ps}%
%width 0.3269in;  height 0.3269in;  depth 0.1003in;  original-width 3in;
%original-height 3in;  cropleft "0";  croptop "1";  cropright "1";
%cropbottom "0";  filename 'ut07.ps';file-properties "XNPEU";}}}%
%BeginExpansion
\raisebox{-0.1003in}{\includegraphics[
%natheight=3.000000in,
%natwidth=3.000000in,
height=0.3269in,
width=0.3269in
]%
{ut07.ps}%
}%
%EndExpansion
\overset{(i,j)}{\longleftrightarrow}%
%TCIMACRO{\FRAME{itbpF}{0.3269in}{0.3269in}{0.1003in}{}{}{ut08.ps}%
%width 0.3269in;  height 0.3269in;  depth 0.1003in;  original-width 3in;
%original-height 3in;  cropleft "0";  croptop "1";  cropright "1";
%cropbottom "0";  filename 'ut08.ps';file-properties "XNPEU";}}}%
%BeginExpansion
\raisebox{-0.1003in}{\includegraphics[
%natheight=3.000000in,
%natwidth=3.000000in,
height=0.3269in,
width=0.3269in
]%
{ut08.ps}%
}%
%EndExpansion
\text{ ,}%
\]
and the second is the \textbf{elliptic transformation} $\varepsilon_{ij}$
given by%
\[
\varepsilon_{ij}=%
\begin{array}
[c]{cc}%
%TCIMACRO{\FRAME{itbpF}{0.3269in}{0.3269in}{0in}{}{}{ut02.ps}%
%width 0.3269in;  height 0.3269in;  depth 0in;  original-width 3in;
%original-height 3in;  cropleft "0";  croptop "1";  cropright "1";
%cropbottom "0";  filename 'ut02.ps';file-properties "XNPEU";}}}%
%BeginExpansion
{\includegraphics[
%natheight=3.000000in,
%natwidth=3.000000in,
height=0.3269in,
width=0.3269in
]%
{ut02.ps}%
}%
%EndExpansion
&
%TCIMACRO{\FRAME{itbpF}{0.3269in}{0.3269in}{0in}{}{}{ut01.ps}%
%width 0.3269in;  height 0.3269in;  depth 0in;  original-width 3in;
%original-height 3in;  cropleft "0";  croptop "1";  cropright "1";
%cropbottom "0";  filename 'ut01.ps';file-properties "XNPEU";}}}%
%BeginExpansion
{\includegraphics[
%natheight=3.000000in,
%natwidth=3.000000in,
height=0.3269in,
width=0.3269in
]%
{ut01.ps}%
}%
%EndExpansion
\\%
%TCIMACRO{\FRAME{itbpF}{0.3269in}{0.3269in}{0in}{}{}{ut03.ps}%
%width 0.3269in;  height 0.3269in;  depth 0in;  original-width 3in;
%original-height 3in;  cropleft "0";  croptop "1";  cropright "1";
%cropbottom "0";  filename 'ut03.ps';file-properties "XNPEU";}}}%
%BeginExpansion
{\includegraphics[
%natheight=3.000000in,
%natwidth=3.000000in,
height=0.3269in,
width=0.3269in
]%
{ut03.ps}%
}%
%EndExpansion
&
%TCIMACRO{\FRAME{itbpF}{0.3269in}{0.3269in}{0in}{}{}{ut04.ps}%
%width 0.3269in;  height 0.3269in;  depth 0in;  original-width 3in;
%original-height 3in;  cropleft "0";  croptop "1";  cropright "1";
%cropbottom "0";  filename 'ut04.ps';file-properties "XNPEU";}}}%
%BeginExpansion
{\includegraphics[
%natheight=3.000000in,
%natwidth=3.000000in,
height=0.3269in,
width=0.3269in
]%
{ut04.ps}%
}%
%EndExpansion
\end{array}
\overset{(i,j)}{\longleftrightarrow}%
\begin{array}
[c]{cc}%
%TCIMACRO{\FRAME{itbpF}{0.3269in}{0.3269in}{0in}{}{}{ut00.ps}%
%width 0.3269in;  height 0.3269in;  depth 0in;  original-width 3in;
%original-height 3in;  cropleft "0";  croptop "1";  cropright "1";
%cropbottom "0";  filename 'ut00.ps';file-properties "XNPEU";}}}%
%BeginExpansion
{\includegraphics[
%natheight=3.000000in,
%natwidth=3.000000in,
height=0.3269in,
width=0.3269in
]%
{ut00.ps}%
}%
%EndExpansion
&
%TCIMACRO{\FRAME{itbpF}{0.3269in}{0.3269in}{0in}{}{}{ut00.ps}%
%width 0.3269in;  height 0.3269in;  depth 0in;  original-width 3in;
%original-height 3in;  cropleft "0";  croptop "1";  cropright "1";
%cropbottom "0";  filename 'ut00.ps';file-properties "XNPEU";}}}%
%BeginExpansion
{\includegraphics[
%natheight=3.000000in,
%natwidth=3.000000in,
height=0.3269in,
width=0.3269in
]%
{ut00.ps}%
}%
%EndExpansion
\\%
%TCIMACRO{\FRAME{itbpF}{0.3269in}{0.3269in}{0in}{}{}{ut00.ps}%
%width 0.3269in;  height 0.3269in;  depth 0in;  original-width 3in;
%original-height 3in;  cropleft "0";  croptop "1";  cropright "1";
%cropbottom "0";  filename 'ut00.ps';file-properties "XNPEU";}}}%
%BeginExpansion
{\includegraphics[
%natheight=3.000000in,
%natwidth=3.000000in,
height=0.3269in,
width=0.3269in
]%
{ut00.ps}%
}%
%EndExpansion
&
%TCIMACRO{\FRAME{itbpF}{0.3269in}{0.3269in}{0in}{}{}{ut00.ps}%
%width 0.3269in;  height 0.3269in;  depth 0in;  original-width 3in;
%original-height 3in;  cropleft "0";  croptop "1";  cropright "1";
%cropbottom "0";  filename 'ut00.ps';file-properties "XNPEU";}}}%
%BeginExpansion
{\includegraphics[
%natheight=3.000000in,
%natwidth=3.000000in,
height=0.3269in,
width=0.3269in
]%
{ut00.ps}%
}%
%EndExpansion
\end{array}
\text{ .}%
\]

\bigskip

More will be said about these transformations in future papers.

\bigskip

\subsection{Quantum observables as invariants of quantum knots}

\bigskip

We now consider the following question:

\bigskip

\noindent\textbf{Question.} \ \textit{What is a quantum knot invariant? \ How
do we define it?}

\bigskip

The objective of the first half of this section is to give a discursive
argument that justifies a definition which will be found to be equivalent to
the following:

\bigskip

\textit{A }\textbf{quantum knot }$n$\textbf{-invariant}\textit{ for a quantum
system }$Q\left(  \mathcal{K}^{(n)},\mathbb{A}(n)\right)  $\textit{ is an
observable }$\Omega$\textit{ on the Hilbert space }$\mathcal{K}^{(n)}$\textit{
of quantum knots which is invariant under the action of the ambient group
}$\mathbb{A}(n)$\textit{, i.e., such that }$U\Omega U^{-1}=\Omega$\textit{ for
all }$U$ in $\mathbb{A}(n)$\textit{. \ }

\bigskip

\noindent\textbf{Caveat:} \ \textit{We emphasize to the reader that the above
definition of quantum knot invariants is not the one currently used in quantum
topology. Quantum topology uses analogies with quantum mechanics to create
significant mathematical structures that do not necessarily correspond
directly to quantum mechanical observables. The invariants of quantum topology
have been investigated for their relevance to quantum computing and they can
be regarded, in our context, as possible secondary calculations made on the
basis of an observable. Here we are concerned with observables that are
themselves topological invariants.}

\bigskip

To justify our new use of the term `quantum knot invariant,' we will use the
following \textsc{yardstick}:

\bigskip

\noindent\textsc{Yardstick:} \textit{Quantum knot invariants are to be
physically meaningful invariants of quantum knot type. \ By "physically
meaningful," we mean that the quantum knot invariants can be directly obtained
from experimental data produced by an implementable physical
experiment.\footnote{Once again we remind the reader that, although the
quantum knot system $Q\left(  \mathcal{K}^{(n)},\mathbb{A}(n)\right)  $ is
physically implementable, it may or may not be implentable within today's
existing technology.}}

\bigskip

Let $Q\left(  \mathcal{K}^{(n)},\mathbb{A}(n)\right)  $ be a quantum knot
system, where $\mathcal{K}^{(n)}$ is the Hilbert space of quantum knots, and
where $\mathbb{A}(n)$ is the underlying ambient group on $\mathcal{K}^{(n)}$.
\ Moreover, let $\mathcal{P}^{(n)}$ denote some yet-to-be-chosen mathematical
domain. \ By an $n$\textbf{-invariant} $I^{(n)}$ of quantum knots, we mean a
map%
\[
I^{(n)}:\mathcal{K}^{(n)}\longrightarrow\mathcal{P}^{(n)}\text{ ,}%
\]
such that, when two quantum knots $\left\vert \psi_{1}\right\rangle $ and
$\left\vert \psi_{2}\right\rangle $ are of the same knot $n$-type, i.e., when%
\[
\left\vert \psi_{1}\right\rangle \underset{n}{\sim}\left\vert \psi
_{2}\right\rangle \text{ ,}%
\]
then their respective invariants must be equal, i.e.,
\[
I^{(n)}\left(  \left\vert \psi_{1}\right\rangle \right)  =I^{(n)}\left(
\left\vert \psi_{2}\right\rangle \right)
\]
In other words, $I^{(n)}:\mathcal{K}^{(n)}\longrightarrow\mathcal{P}^{(n)}$ is
a map which is invariant under the action of the ambient group $\mathbb{A}%
(n)$, i.e.,%
\[
I^{(n)}\left(  \left\vert \psi\right\rangle \right)  =I^{(n)}\left(
g\left\vert \psi\right\rangle \right)
\]
for all elements $g$ in $\mathbb{A}(n)$.

\bigskip

\noindent\textbf{Question:} \textit{But which such invariants are physically
meaningful?}\ 

\bigskip

We begin to try to answer this question by noting that the only way to extract
information from a quantum system is through quantum measurement. \ Thus, if
we wish to extract information about quantum knot type from a quantum knot
system $Q\left(  \mathcal{K}^{(n)},\mathbb{A}(n)\right)  $\textit{, }we of
necessity must make a measurement with respect to some observable. \ But what
kind of observable? \ 

\bigskip

With this in mind, we will now describe quantum measurement from a different,
but nonetheless equivalent perspective, than that which is usually given in
standard texts on quantum mechanics.\footnote{In this paper, we will focus
only on von Neumann quantum measurement. \ We will discuss more general POVM
approach to quantum knot invariants in a later paper.} \ For knot theorists
who might not be familiar with standard quantum measurement, we have included
in the figure below a brief summary of quantum measurement.\footnote{For
readers unfamiliar with quantum measurement, there are many references, for
example, \cite{Helstrom1}, \cite{Lomonaco3} \cite{Nielsen1}, \cite{Peres1},
\cite{Sakuri1}, and \cite{Shankar1}.} \ 

\bigskip%

%TCIMACRO{\FRAME{dtbpFUX}{4.0274in}{3.0277in}{0pt}{\Qcb{\QTR{bf}{Von Neumann
%measurement.}}}{}{measurement.ps}{\special{ language "Scientific Word";
%type "GRAPHIC";  maintain-aspect-ratio TRUE;  display "USEDEF";
%valid_file "F";  width 4.0274in;  height 3.0277in;  depth 0pt;
%original-width 9.9998in;  original-height 7.4996in;  cropleft "0";
%filename 'measurement.ps';file-properties "XNPEU";}}}%
%BeginExpansion
\begin{center}
\fbox{\includegraphics[
%natheight=7.499600in,
%natwidth=9.999800in,
height=3.0277in,
width=4.0274in
]%
{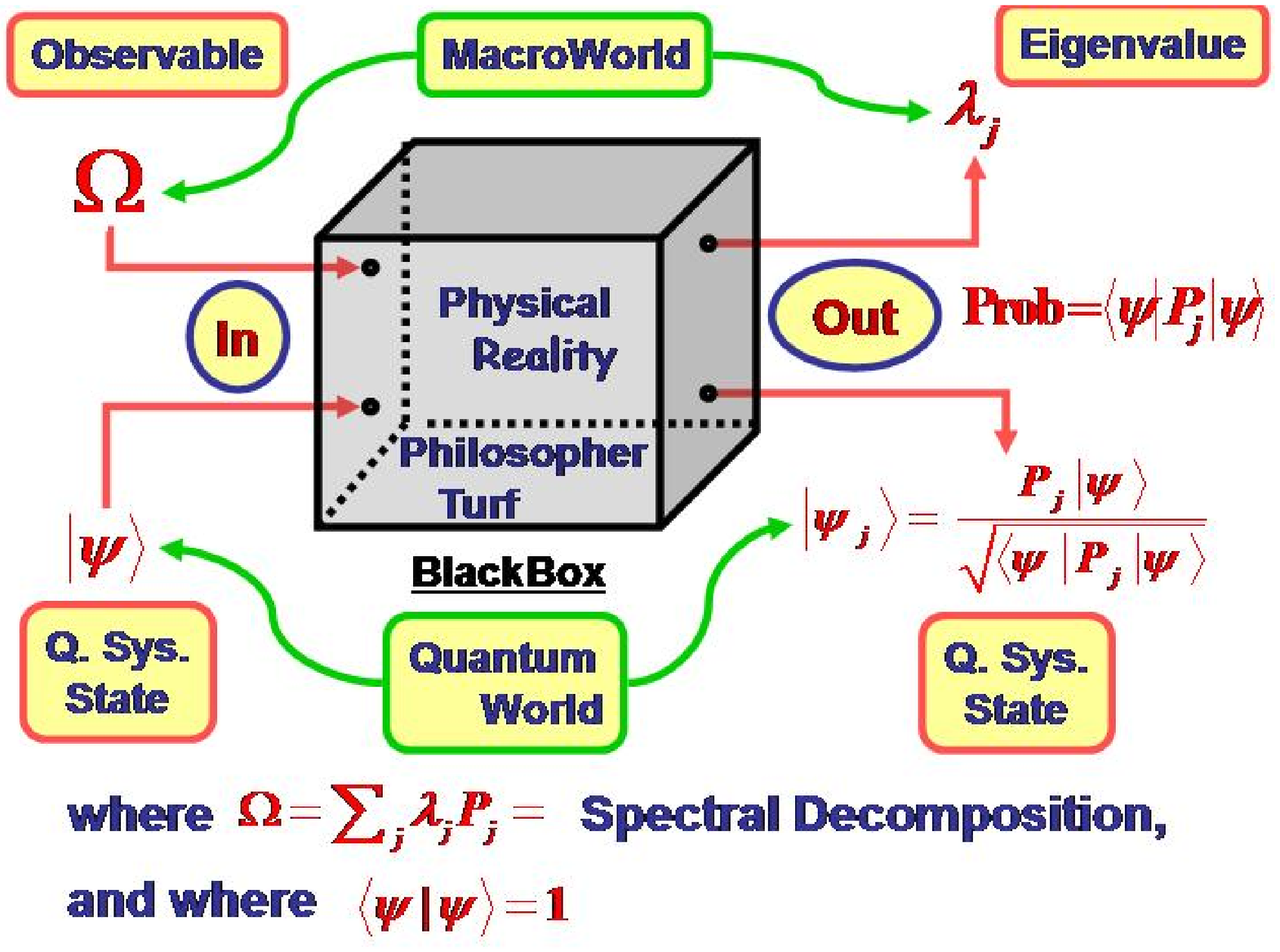}%
}\\
\textbf{Von Neumann measurement.}%
\end{center}
%EndExpansion

\bigskip

Let $\Omega$ be an observable for a quantum system $Q\left(  \mathcal{K}%
^{(n)},\mathbb{A}(n)\right)  $, i.e., a Hermitian (self-adjoint) linear
operator on the Hilbert space $\mathcal{K}^{(n)}$. \ Moreover, let
\[
\Omega=\sum_{j=1}^{m}\lambda_{j}P_{j}%
\]
be the spectral decomposition of the observable $\Omega$, where $\lambda_{j}$
is the $j$-th eigenvalue of $\Omega$, and where $P_{j}$ is the corresponding
projection operator for the associated eigenspace $V_{j}$.

\bigskip

Let $\mathcal{P}_{\Omega}^{(n)}$ denote the \textbf{set of all probability
distributions on the spectrum} $\left\{  \lambda_{1},\lambda_{2}%
,\ldots,\lambda_{m}\right\}  $ of $\Omega$, i.e.,%
\[
\mathcal{P}_{\Omega}^{(n)}=\left\{  p:\left\{  \lambda_{1},\lambda_{2}%
,\ldots,\lambda_{m}\right\}  \longrightarrow\left[  0,1\right]  :\sum
_{j=1}^{m}p\left(  \lambda_{j}\right)  =1\right\}
\]

We will call the probability distributions $p$ of $\mathcal{P}_{\Omega}^{(n)}$
\textbf{stochastic sources}.

\bigskip\bigskip

Then each observable $\Omega$ uniquely determines a map%
\[%
\begin{array}
[c]{cccc}%
\widetilde{\Omega}: & \mathcal{K}^{(n)} & \longrightarrow & \mathcal{P}%
_{\Omega}^{(n)}\\
& \left\vert \psi\right\rangle  & \longmapsto & p
\end{array}
\]
from quantum knots to stochastic sources on the spectrum of $\Omega$ given by%
\[
p_{j}\left(  \left\vert \psi\right\rangle \right)  =\frac{\left\langle
\psi\left\vert P_{j}\right\vert \psi\right\rangle }{\sqrt{\left\langle
\psi|\psi\right\rangle }}\text{.}%
\]

\bigskip

Thus, what is seen, when a quantum system $Q$ in state $\left\vert
\psi\right\rangle $ is measured with respect to an observable $\Omega$, is a
random sample from the stochastic source $\widetilde{\Omega}\left(  \left\vert
\psi\right\rangle \right)  $. \ But under what circumstances is such a random
sample a quantum knot invariant?

\bigskip

Our answer to this question is that quantum knots $\left\vert \psi
_{1}\right\rangle $ and $\left\vert \psi_{2}\right\rangle $ of the same knot
$n$-type must produce random samples from the same stochastic source when
measured with respect to the observable $\Omega$. \ This answer is captured by
the following definition:

\bigskip

\begin{definition}
Let $Q\left(  \mathcal{K}^{(n)},\mathbb{A}(n)\right)  $ be a quantum knot
system, and let $\Omega$ be an observable on $\mathcal{K}^{(n)}$ with spectral
decomposition
\[
\Omega=\sum_{j=1}^{m}\lambda_{j}P_{j}\text{ .}%
\]
Then the observable $\Omega$ is said to be a \textbf{quantum knot }%
$n$\textbf{-invariant} provided%
\[
\left\langle \psi\left\vert UP_{j}U^{-1}\right\vert \psi\right\rangle
=\left\langle \psi\left\vert P_{j}\right\vert \psi\right\rangle
\]
for all quantum knots $\left\vert \psi\right\rangle \in\mathcal{K}^{(n)}$, for
all $U\in\mathbb{A}(n)$, and for all projectors $P_{j}$.
\end{definition}

\bigskip

\begin{theorem}
Let $Q\left(  \mathcal{K}^{(n)},\mathbb{A}(n)\right)  $ and $\Omega$ be as
given in the above definition. Then the following statements are equivalent:

\begin{itemize}
\item[\textbf{1)}] The observable $\Omega$ is a quantum knot $n$-invariant

\item[\textbf{2)}] $\left[  U,P_{j}\right]  =0$ for all $U\in\mathbb{A}(n)$
and for all $P_{j}$.

\item[\textbf{3)}] $\left[  U,\Omega\right]  =0$ for all $U\in\mathbb{A}(n)$,
\end{itemize}

\noindent where $\left[  A,B\right]  $ denotes the commutator $AB-BA$ of
operators $A$ and $B$.
\end{theorem}

\bigskip

The remaining half of this section is devoted to finding an answer to the
following question:

\bigskip

\noindent\textbf{Question:} \ \textit{How do we find observables which are
quantum knot invariants?}

\bigskip

One answer to this question is the following theorem, which is an almost
immediate consequence of the definition of a minimum invariant subspace of
$\mathcal{K}^{(n)}$:

\bigskip

\begin{theorem}
Let $Q\left(  \mathcal{K}^{(n)},\mathbb{A}(n)\right)  $ be a quantum knot
system, and let
\[
\mathcal{K}^{(n)}=\bigoplus_{\ell}W_{\ell}%
\]
be a decomposition of the representation%
\[
\mathbb{A}(n)\times\mathcal{K}^{(n)}\longrightarrow\mathcal{K}^{(n)}%
\]
into irreducible representations of the ambient group $\mathbb{A}(n)$. \ Then,
for each $\ell$, the projection operator $P_{\ell}$ for the subspace $W_{\ell
}$ is an observable which is a quantum knot $n$-invariant.
\end{theorem}

\bigskip

Here is yet another way of finding quantum knot invariants:

\bigskip

\begin{theorem}
Let $Q\left(  \mathcal{K}^{(n)},\mathbb{A}(n)\right)  $ be a quantum knot
system, and let $\Omega$ be an observable on the Hilbert space $\mathcal{K}%
^{(n)}$. \ Let $St\left(  \Omega\right)  $ be the stabilizer subgroup for
$\Omega$, i.e.,
\[
St\left(  \Omega\right)  =\left\{  U\in\mathbb{A}(n):U\Omega U^{-1}%
=\Omega\right\}  \text{ .}%
\]
Then the observable%
\[
\sum_{U\in\mathbb{A}(n)/St\left(  \Omega\right)  }U\Omega U^{-1}%
\]
is a quantum knot $n$-invariant, where $\sum\limits_{U\in\mathbb{A}%
(n)/St\left(  \Omega\right)  }U\Omega U^{-1}$ denotes a sum over a complete
set of coset representatives for the stabilizer subgroup $St\left(
\Omega\right)  $ of the ambient group $\mathbb{A}(n)$.
\end{theorem}

\begin{proof}
The observable $\sum_{g\in\mathbb{A}(n)}g\Omega g^{-1}$is obviously an quantum
knot $n$-invariant, since $g^{\prime}\left(  \sum_{g\in\mathbb{A}(n)}g\Omega
g^{-1}\right)  g^{\prime-1}=\sum_{g\in\mathbb{A}(n)}g\Omega g^{-1}$ for all
$g^{\prime}\in\mathbb{A}(n)$. \ If we let $\left\vert St\left(  \Omega\right)
\right\vert $ denote the order of $\left\vert St\left(  \Omega\right)
\right\vert $, and if we let $c_{1},c_{2},\ldots,c_{p}$ denote a complete set
of coset representatives of the stabilizer subgroup $St\left(  \Omega\right)
$, then $\sum_{j=1}^{p}c_{j}\Omega c_{j}^{-1}=\frac{1}{\left\vert St\left(
\Omega\right)  \right\vert }\sum_{g\in\mathbb{A}(n)}g\Omega g^{-1}$ is also a
quantum knot invariant. \ 
\end{proof}

\bigskip

We end this section with an example of a quantum knot invariant:

\bigskip

\begin{example}
The following observable $\Omega$ is an example of a quantum knot $4$-invariant:

\bigskip

$\hspace{-0.8in}\Omega=\left\vert
\begin{array}
[c]{cccc}%
%TCIMACRO{\FRAME{itbpF}{0.1773in}{0.1773in}{0in}{}{}{ut00.ps}%
%width 0.1773in;  height 0.1773in;  depth 0in;  original-width 3in;
%original-height 3in;  cropleft "0";  croptop "1";  cropright "1";
%cropbottom "0";  filename 'ut00.ps';file-properties "XNPEU";}}}%
%BeginExpansion
{\includegraphics[
%natheight=3.000000in,
%natwidth=3.000000in,
height=0.1773in,
width=0.1773in
]%
{ut00.ps}%
}%
%EndExpansion
&
%TCIMACRO{\FRAME{itbpF}{0.1773in}{0.1773in}{0in}{}{}{ut02.ps}%
%width 0.1773in;  height 0.1773in;  depth 0in;  original-width 3in;
%original-height 3in;  cropleft "0";  croptop "1";  cropright "1";
%cropbottom "0";  filename 'ut02.ps';file-properties "XNPEU";}}}%
%BeginExpansion
{\includegraphics[
%natheight=3.000000in,
%natwidth=3.000000in,
height=0.1773in,
width=0.1773in
]%
{ut02.ps}%
}%
%EndExpansion
&
%TCIMACRO{\FRAME{itbpF}{0.1773in}{0.1773in}{0in}{}{}{ut01.ps}%
%width 0.1773in;  height 0.1773in;  depth 0in;  original-width 3in;
%original-height 3in;  cropleft "0";  croptop "1";  cropright "1";
%cropbottom "0";  filename 'ut01.ps';file-properties "XNPEU";}}}%
%BeginExpansion
{\includegraphics[
%natheight=3.000000in,
%natwidth=3.000000in,
height=0.1773in,
width=0.1773in
]%
{ut01.ps}%
}%
%EndExpansion
&
%TCIMACRO{\FRAME{itbpF}{0.1773in}{0.1773in}{0in}{}{}{ut00.ps}%
%width 0.1773in;  height 0.1773in;  depth 0in;  original-width 3in;
%original-height 3in;  cropleft "0";  croptop "1";  cropright "1";
%cropbottom "0";  filename 'ut00.ps';file-properties "XNPEU";}}}%
%BeginExpansion
{\includegraphics[
%natheight=3.000000in,
%natwidth=3.000000in,
height=0.1773in,
width=0.1773in
]%
{ut00.ps}%
}%
%EndExpansion
\\%
%TCIMACRO{\FRAME{itbpF}{0.1773in}{0.1773in}{0in}{}{}{ut02.ps}%
%width 0.1773in;  height 0.1773in;  depth 0in;  original-width 3in;
%original-height 3in;  cropleft "0";  croptop "1";  cropright "1";
%cropbottom "0";  filename 'ut02.ps';file-properties "XNPEU";}}}%
%BeginExpansion
{\includegraphics[
%natheight=3.000000in,
%natwidth=3.000000in,
height=0.1773in,
width=0.1773in
]%
{ut02.ps}%
}%
%EndExpansion
&
%TCIMACRO{\FRAME{itbpF}{0.1773in}{0.1773in}{0in}{}{}{ut09.ps}%
%width 0.1773in;  height 0.1773in;  depth 0in;  original-width 3in;
%original-height 3in;  cropleft "0";  croptop "1";  cropright "1";
%cropbottom "0";  filename 'ut09.ps';file-properties "XNPEU";}}}%
%BeginExpansion
{\includegraphics[
%natheight=3.000000in,
%natwidth=3.000000in,
height=0.1773in,
width=0.1773in
]%
{ut09.ps}%
}%
%EndExpansion
&
%TCIMACRO{\FRAME{itbpF}{0.1773in}{0.1773in}{0in}{}{}{ut10.ps}%
%width 0.1773in;  height 0.1773in;  depth 0in;  original-width 3in;
%original-height 3in;  cropleft "0";  croptop "1";  cropright "1";
%cropbottom "0";  filename 'ut10.ps';file-properties "XNPEU";}}}%
%BeginExpansion
{\includegraphics[
%natheight=3.000000in,
%natwidth=3.000000in,
height=0.1773in,
width=0.1773in
]%
{ut10.ps}%
}%
%EndExpansion
&
%TCIMACRO{\FRAME{itbpF}{0.1773in}{0.1773in}{0in}{}{}{ut01.ps}%
%width 0.1773in;  height 0.1773in;  depth 0in;  original-width 3in;
%original-height 3in;  cropleft "0";  croptop "1";  cropright "1";
%cropbottom "0";  filename 'ut01.ps';file-properties "XNPEU";}}}%
%BeginExpansion
{\includegraphics[
%natheight=3.000000in,
%natwidth=3.000000in,
height=0.1773in,
width=0.1773in
]%
{ut01.ps}%
}%
%EndExpansion
\\%
%TCIMACRO{\FRAME{itbpF}{0.1773in}{0.1773in}{0in}{}{}{ut06.ps}%
%width 0.1773in;  height 0.1773in;  depth 0in;  original-width 3in;
%original-height 3in;  cropleft "0";  croptop "1";  cropright "1";
%cropbottom "0";  filename 'ut06.ps';file-properties "XNPEU";}}}%
%BeginExpansion
{\includegraphics[
%natheight=3.000000in,
%natwidth=3.000000in,
height=0.1773in,
width=0.1773in
]%
{ut06.ps}%
}%
%EndExpansion
&
%TCIMACRO{\FRAME{itbpF}{0.1773in}{0.1773in}{0in}{}{}{ut03.ps}%
%width 0.1773in;  height 0.1773in;  depth 0in;  original-width 3in;
%original-height 3in;  cropleft "0";  croptop "1";  cropright "1";
%cropbottom "0";  filename 'ut03.ps';file-properties "XNPEU";}}}%
%BeginExpansion
{\includegraphics[
%natheight=3.000000in,
%natwidth=3.000000in,
height=0.1773in,
width=0.1773in
]%
{ut03.ps}%
}%
%EndExpansion
&
%TCIMACRO{\FRAME{itbpF}{0.1773in}{0.1773in}{0in}{}{}{ut09.ps}%
%width 0.1773in;  height 0.1773in;  depth 0in;  original-width 3in;
%original-height 3in;  cropleft "0";  croptop "1";  cropright "1";
%cropbottom "0";  filename 'ut09.ps';file-properties "XNPEU";}}}%
%BeginExpansion
{\includegraphics[
%natheight=3.000000in,
%natwidth=3.000000in,
height=0.1773in,
width=0.1773in
]%
{ut09.ps}%
}%
%EndExpansion
&
%TCIMACRO{\FRAME{itbpF}{0.1773in}{0.1773in}{0in}{}{}{ut04.ps}%
%width 0.1773in;  height 0.1773in;  depth 0in;  original-width 3in;
%original-height 3in;  cropleft "0";  croptop "1";  cropright "1";
%cropbottom "0";  filename 'ut04.ps';file-properties "XNPEU";}}}%
%BeginExpansion
{\includegraphics[
%natheight=3.000000in,
%natwidth=3.000000in,
height=0.1773in,
width=0.1773in
]%
{ut04.ps}%
}%
%EndExpansion
\\%
%TCIMACRO{\FRAME{itbpF}{0.1773in}{0.1773in}{0in}{}{}{ut03.ps}%
%width 0.1773in;  height 0.1773in;  depth 0in;  original-width 3in;
%original-height 3in;  cropleft "0";  croptop "1";  cropright "1";
%cropbottom "0";  filename 'ut03.ps';file-properties "XNPEU";}}}%
%BeginExpansion
{\includegraphics[
%natheight=3.000000in,
%natwidth=3.000000in,
height=0.1773in,
width=0.1773in
]%
{ut03.ps}%
}%
%EndExpansion
&
%TCIMACRO{\FRAME{itbpF}{0.1773in}{0.1773in}{0in}{}{}{ut05.ps}%
%width 0.1773in;  height 0.1773in;  depth 0in;  original-width 3in;
%original-height 3in;  cropleft "0";  croptop "1";  cropright "1";
%cropbottom "0";  filename 'ut05.ps';file-properties "XNPEU";}}}%
%BeginExpansion
{\includegraphics[
%natheight=3.000000in,
%natwidth=3.000000in,
height=0.1773in,
width=0.1773in
]%
{ut05.ps}%
}%
%EndExpansion
&
%TCIMACRO{\FRAME{itbpF}{0.1773in}{0.1773in}{0in}{}{}{ut04.ps}%
%width 0.1773in;  height 0.1773in;  depth 0in;  original-width 3in;
%original-height 3in;  cropleft "0";  croptop "1";  cropright "1";
%cropbottom "0";  filename 'ut04.ps';file-properties "XNPEU";}}}%
%BeginExpansion
{\includegraphics[
%natheight=3.000000in,
%natwidth=3.000000in,
height=0.1773in,
width=0.1773in
]%
{ut04.ps}%
}%
%EndExpansion
&
%TCIMACRO{\FRAME{itbpF}{0.1773in}{0.1773in}{0in}{}{}{ut00.ps}%
%width 0.1773in;  height 0.1773in;  depth 0in;  original-width 3in;
%original-height 3in;  cropleft "0";  croptop "1";  cropright "1";
%cropbottom "0";  filename 'ut00.ps';file-properties "XNPEU";}}}%
%BeginExpansion
{\includegraphics[
%natheight=3.000000in,
%natwidth=3.000000in,
height=0.1773in,
width=0.1773in
]%
{ut00.ps}%
}%
%EndExpansion
\end{array}
\right\rangle \left\langle
\begin{array}
[c]{cccc}%
%TCIMACRO{\FRAME{itbpF}{0.1773in}{0.1773in}{0in}{}{}{ut00.ps}%
%width 0.1773in;  height 0.1773in;  depth 0in;  original-width 3in;
%original-height 3in;  cropleft "0";  croptop "1";  cropright "1";
%cropbottom "0";  filename 'ut00.ps';file-properties "XNPEU";}}}%
%BeginExpansion
{\includegraphics[
%natheight=3.000000in,
%natwidth=3.000000in,
height=0.1773in,
width=0.1773in
]%
{ut00.ps}%
}%
%EndExpansion
&
%TCIMACRO{\FRAME{itbpF}{0.1773in}{0.1773in}{0in}{}{}{ut02.ps}%
%width 0.1773in;  height 0.1773in;  depth 0in;  original-width 3in;
%original-height 3in;  cropleft "0";  croptop "1";  cropright "1";
%cropbottom "0";  filename 'ut02.ps';file-properties "XNPEU";}}}%
%BeginExpansion
{\includegraphics[
%natheight=3.000000in,
%natwidth=3.000000in,
height=0.1773in,
width=0.1773in
]%
{ut02.ps}%
}%
%EndExpansion
&
%TCIMACRO{\FRAME{itbpF}{0.1773in}{0.1773in}{0in}{}{}{ut01.ps}%
%width 0.1773in;  height 0.1773in;  depth 0in;  original-width 3in;
%original-height 3in;  cropleft "0";  croptop "1";  cropright "1";
%cropbottom "0";  filename 'ut01.ps';file-properties "XNPEU";}}}%
%BeginExpansion
{\includegraphics[
%natheight=3.000000in,
%natwidth=3.000000in,
height=0.1773in,
width=0.1773in
]%
{ut01.ps}%
}%
%EndExpansion
&
%TCIMACRO{\FRAME{itbpF}{0.1773in}{0.1773in}{0in}{}{}{ut00.ps}%
%width 0.1773in;  height 0.1773in;  depth 0in;  original-width 3in;
%original-height 3in;  cropleft "0";  croptop "1";  cropright "1";
%cropbottom "0";  filename 'ut00.ps';file-properties "XNPEU";}}}%
%BeginExpansion
{\includegraphics[
%natheight=3.000000in,
%natwidth=3.000000in,
height=0.1773in,
width=0.1773in
]%
{ut00.ps}%
}%
%EndExpansion
\\%
%TCIMACRO{\FRAME{itbpF}{0.1773in}{0.1773in}{0in}{}{}{ut02.ps}%
%width 0.1773in;  height 0.1773in;  depth 0in;  original-width 3in;
%original-height 3in;  cropleft "0";  croptop "1";  cropright "1";
%cropbottom "0";  filename 'ut02.ps';file-properties "XNPEU";}}}%
%BeginExpansion
{\includegraphics[
%natheight=3.000000in,
%natwidth=3.000000in,
height=0.1773in,
width=0.1773in
]%
{ut02.ps}%
}%
%EndExpansion
&
%TCIMACRO{\FRAME{itbpF}{0.1773in}{0.1773in}{0in}{}{}{ut09.ps}%
%width 0.1773in;  height 0.1773in;  depth 0in;  original-width 3in;
%original-height 3in;  cropleft "0";  croptop "1";  cropright "1";
%cropbottom "0";  filename 'ut09.ps';file-properties "XNPEU";}}}%
%BeginExpansion
{\includegraphics[
%natheight=3.000000in,
%natwidth=3.000000in,
height=0.1773in,
width=0.1773in
]%
{ut09.ps}%
}%
%EndExpansion
&
%TCIMACRO{\FRAME{itbpF}{0.1773in}{0.1773in}{0in}{}{}{ut10.ps}%
%width 0.1773in;  height 0.1773in;  depth 0in;  original-width 3in;
%original-height 3in;  cropleft "0";  croptop "1";  cropright "1";
%cropbottom "0";  filename 'ut10.ps';file-properties "XNPEU";}}}%
%BeginExpansion
{\includegraphics[
%natheight=3.000000in,
%natwidth=3.000000in,
height=0.1773in,
width=0.1773in
]%
{ut10.ps}%
}%
%EndExpansion
&
%TCIMACRO{\FRAME{itbpF}{0.1773in}{0.1773in}{0in}{}{}{ut01.ps}%
%width 0.1773in;  height 0.1773in;  depth 0in;  original-width 3in;
%original-height 3in;  cropleft "0";  croptop "1";  cropright "1";
%cropbottom "0";  filename 'ut01.ps';file-properties "XNPEU";}}}%
%BeginExpansion
{\includegraphics[
%natheight=3.000000in,
%natwidth=3.000000in,
height=0.1773in,
width=0.1773in
]%
{ut01.ps}%
}%
%EndExpansion
\\%
%TCIMACRO{\FRAME{itbpF}{0.1773in}{0.1773in}{0in}{}{}{ut06.ps}%
%width 0.1773in;  height 0.1773in;  depth 0in;  original-width 3in;
%original-height 3in;  cropleft "0";  croptop "1";  cropright "1";
%cropbottom "0";  filename 'ut06.ps';file-properties "XNPEU";}}}%
%BeginExpansion
{\includegraphics[
%natheight=3.000000in,
%natwidth=3.000000in,
height=0.1773in,
width=0.1773in
]%
{ut06.ps}%
}%
%EndExpansion
&
%TCIMACRO{\FRAME{itbpF}{0.1773in}{0.1773in}{0in}{}{}{ut03.ps}%
%width 0.1773in;  height 0.1773in;  depth 0in;  original-width 3in;
%original-height 3in;  cropleft "0";  croptop "1";  cropright "1";
%cropbottom "0";  filename 'ut03.ps';file-properties "XNPEU";}}}%
%BeginExpansion
{\includegraphics[
%natheight=3.000000in,
%natwidth=3.000000in,
height=0.1773in,
width=0.1773in
]%
{ut03.ps}%
}%
%EndExpansion
&
%TCIMACRO{\FRAME{itbpF}{0.1773in}{0.1773in}{0in}{}{}{ut09.ps}%
%width 0.1773in;  height 0.1773in;  depth 0in;  original-width 3in;
%original-height 3in;  cropleft "0";  croptop "1";  cropright "1";
%cropbottom "0";  filename 'ut09.ps';file-properties "XNPEU";}}}%
%BeginExpansion
{\includegraphics[
%natheight=3.000000in,
%natwidth=3.000000in,
height=0.1773in,
width=0.1773in
]%
{ut09.ps}%
}%
%EndExpansion
&
%TCIMACRO{\FRAME{itbpF}{0.1773in}{0.1773in}{0in}{}{}{ut04.ps}%
%width 0.1773in;  height 0.1773in;  depth 0in;  original-width 3in;
%original-height 3in;  cropleft "0";  croptop "1";  cropright "1";
%cropbottom "0";  filename 'ut04.ps';file-properties "XNPEU";}}}%
%BeginExpansion
{\includegraphics[
%natheight=3.000000in,
%natwidth=3.000000in,
height=0.1773in,
width=0.1773in
]%
{ut04.ps}%
}%
%EndExpansion
\\%
%TCIMACRO{\FRAME{itbpF}{0.1773in}{0.1773in}{0in}{}{}{ut03.ps}%
%width 0.1773in;  height 0.1773in;  depth 0in;  original-width 3in;
%original-height 3in;  cropleft "0";  croptop "1";  cropright "1";
%cropbottom "0";  filename 'ut03.ps';file-properties "XNPEU";}}}%
%BeginExpansion
{\includegraphics[
%natheight=3.000000in,
%natwidth=3.000000in,
height=0.1773in,
width=0.1773in
]%
{ut03.ps}%
}%
%EndExpansion
&
%TCIMACRO{\FRAME{itbpF}{0.1773in}{0.1773in}{0in}{}{}{ut05.ps}%
%width 0.1773in;  height 0.1773in;  depth 0in;  original-width 3in;
%original-height 3in;  cropleft "0";  croptop "1";  cropright "1";
%cropbottom "0";  filename 'ut05.ps';file-properties "XNPEU";}}}%
%BeginExpansion
{\includegraphics[
%natheight=3.000000in,
%natwidth=3.000000in,
height=0.1773in,
width=0.1773in
]%
{ut05.ps}%
}%
%EndExpansion
&
%TCIMACRO{\FRAME{itbpF}{0.1773in}{0.1773in}{0in}{}{}{ut04.ps}%
%width 0.1773in;  height 0.1773in;  depth 0in;  original-width 3in;
%original-height 3in;  cropleft "0";  croptop "1";  cropright "1";
%cropbottom "0";  filename 'ut04.ps';file-properties "XNPEU";}}}%
%BeginExpansion
{\includegraphics[
%natheight=3.000000in,
%natwidth=3.000000in,
height=0.1773in,
width=0.1773in
]%
{ut04.ps}%
}%
%EndExpansion
&
%TCIMACRO{\FRAME{itbpF}{0.1773in}{0.1773in}{0in}{}{}{ut00.ps}%
%width 0.1773in;  height 0.1773in;  depth 0in;  original-width 3in;
%original-height 3in;  cropleft "0";  croptop "1";  cropright "1";
%cropbottom "0";  filename 'ut00.ps';file-properties "XNPEU";}}}%
%BeginExpansion
{\includegraphics[
%natheight=3.000000in,
%natwidth=3.000000in,
height=0.1773in,
width=0.1773in
]%
{ut00.ps}%
}%
%EndExpansion
\end{array}
\right\vert +\left\vert
\begin{array}
[c]{cccc}%
%TCIMACRO{\FRAME{itbpF}{0.1773in}{0.1773in}{0in}{}{}{ut00.ps}%
%width 0.1773in;  height 0.1773in;  depth 0in;  original-width 3in;
%original-height 3in;  cropleft "0";  croptop "1";  cropright "1";
%cropbottom "0";  filename 'ut00.ps';file-properties "XNPEU";}}}%
%BeginExpansion
{\includegraphics[
%natheight=3.000000in,
%natwidth=3.000000in,
height=0.1773in,
width=0.1773in
]%
{ut00.ps}%
}%
%EndExpansion
&
%TCIMACRO{\FRAME{itbpF}{0.1773in}{0.1773in}{0in}{}{}{ut02.ps}%
%width 0.1773in;  height 0.1773in;  depth 0in;  original-width 3in;
%original-height 3in;  cropleft "0";  croptop "1";  cropright "1";
%cropbottom "0";  filename 'ut02.ps';file-properties "XNPEU";}}}%
%BeginExpansion
{\includegraphics[
%natheight=3.000000in,
%natwidth=3.000000in,
height=0.1773in,
width=0.1773in
]%
{ut02.ps}%
}%
%EndExpansion
&
%TCIMACRO{\FRAME{itbpF}{0.1773in}{0.1773in}{0in}{}{}{ut01.ps}%
%width 0.1773in;  height 0.1773in;  depth 0in;  original-width 3in;
%original-height 3in;  cropleft "0";  croptop "1";  cropright "1";
%cropbottom "0";  filename 'ut01.ps';file-properties "XNPEU";}}}%
%BeginExpansion
{\includegraphics[
%natheight=3.000000in,
%natwidth=3.000000in,
height=0.1773in,
width=0.1773in
]%
{ut01.ps}%
}%
%EndExpansion
&
%TCIMACRO{\FRAME{itbpF}{0.1773in}{0.1773in}{0in}{}{}{ut00.ps}%
%width 0.1773in;  height 0.1773in;  depth 0in;  original-width 3in;
%original-height 3in;  cropleft "0";  croptop "1";  cropright "1";
%cropbottom "0";  filename 'ut00.ps';file-properties "XNPEU";}}}%
%BeginExpansion
{\includegraphics[
%natheight=3.000000in,
%natwidth=3.000000in,
height=0.1773in,
width=0.1773in
]%
{ut00.ps}%
}%
%EndExpansion
\\%
%TCIMACRO{\FRAME{itbpF}{0.1773in}{0.1773in}{0in}{}{}{ut02.ps}%
%width 0.1773in;  height 0.1773in;  depth 0in;  original-width 3in;
%original-height 3in;  cropleft "0";  croptop "1";  cropright "1";
%cropbottom "0";  filename 'ut02.ps';file-properties "XNPEU";}}}%
%BeginExpansion
{\includegraphics[
%natheight=3.000000in,
%natwidth=3.000000in,
height=0.1773in,
width=0.1773in
]%
{ut02.ps}%
}%
%EndExpansion
&
%TCIMACRO{\FRAME{itbpF}{0.1773in}{0.1773in}{0in}{}{}{ut09.ps}%
%width 0.1773in;  height 0.1773in;  depth 0in;  original-width 3in;
%original-height 3in;  cropleft "0";  croptop "1";  cropright "1";
%cropbottom "0";  filename 'ut09.ps';file-properties "XNPEU";}}}%
%BeginExpansion
{\includegraphics[
%natheight=3.000000in,
%natwidth=3.000000in,
height=0.1773in,
width=0.1773in
]%
{ut09.ps}%
}%
%EndExpansion
&
%TCIMACRO{\FRAME{itbpF}{0.1773in}{0.1773in}{0in}{}{}{ut10.ps}%
%width 0.1773in;  height 0.1773in;  depth 0in;  original-width 3in;
%original-height 3in;  cropleft "0";  croptop "1";  cropright "1";
%cropbottom "0";  filename 'ut10.ps';file-properties "XNPEU";}}}%
%BeginExpansion
{\includegraphics[
%natheight=3.000000in,
%natwidth=3.000000in,
height=0.1773in,
width=0.1773in
]%
{ut10.ps}%
}%
%EndExpansion
&
%TCIMACRO{\FRAME{itbpF}{0.1773in}{0.1773in}{0in}{}{}{ut01.ps}%
%width 0.1773in;  height 0.1773in;  depth 0in;  original-width 3in;
%original-height 3in;  cropleft "0";  croptop "1";  cropright "1";
%cropbottom "0";  filename 'ut01.ps';file-properties "XNPEU";}}}%
%BeginExpansion
{\includegraphics[
%natheight=3.000000in,
%natwidth=3.000000in,
height=0.1773in,
width=0.1773in
]%
{ut01.ps}%
}%
%EndExpansion
\\%
%TCIMACRO{\FRAME{itbpF}{0.1773in}{0.1773in}{0in}{}{}{ut03.ps}%
%width 0.1773in;  height 0.1773in;  depth 0in;  original-width 3in;
%original-height 3in;  cropleft "0";  croptop "1";  cropright "1";
%cropbottom "0";  filename 'ut03.ps';file-properties "XNPEU";}}}%
%BeginExpansion
{\includegraphics[
%natheight=3.000000in,
%natwidth=3.000000in,
height=0.1773in,
width=0.1773in
]%
{ut03.ps}%
}%
%EndExpansion
&
%TCIMACRO{\FRAME{itbpF}{0.1773in}{0.1773in}{0in}{}{}{ut07.ps}%
%width 0.1773in;  height 0.1773in;  depth 0in;  original-width 3in;
%original-height 3in;  cropleft "0";  croptop "1";  cropright "1";
%cropbottom "0";  filename 'ut07.ps';file-properties "XNPEU";}}}%
%BeginExpansion
{\includegraphics[
%natheight=3.000000in,
%natwidth=3.000000in,
height=0.1773in,
width=0.1773in
]%
{ut07.ps}%
}%
%EndExpansion
&
%TCIMACRO{\FRAME{itbpF}{0.1773in}{0.1773in}{0in}{}{}{ut09.ps}%
%width 0.1773in;  height 0.1773in;  depth 0in;  original-width 3in;
%original-height 3in;  cropleft "0";  croptop "1";  cropright "1";
%cropbottom "0";  filename 'ut09.ps';file-properties "XNPEU";}}}%
%BeginExpansion
{\includegraphics[
%natheight=3.000000in,
%natwidth=3.000000in,
height=0.1773in,
width=0.1773in
]%
{ut09.ps}%
}%
%EndExpansion
&
%TCIMACRO{\FRAME{itbpF}{0.1773in}{0.1773in}{0in}{}{}{ut04.ps}%
%width 0.1773in;  height 0.1773in;  depth 0in;  original-width 3in;
%original-height 3in;  cropleft "0";  croptop "1";  cropright "1";
%cropbottom "0";  filename 'ut04.ps';file-properties "XNPEU";}}}%
%BeginExpansion
{\includegraphics[
%natheight=3.000000in,
%natwidth=3.000000in,
height=0.1773in,
width=0.1773in
]%
{ut04.ps}%
}%
%EndExpansion
\\%
%TCIMACRO{\FRAME{itbpF}{0.1773in}{0.1773in}{0in}{}{}{ut00.ps}%
%width 0.1773in;  height 0.1773in;  depth 0in;  original-width 3in;
%original-height 3in;  cropleft "0";  croptop "1";  cropright "1";
%cropbottom "0";  filename 'ut00.ps';file-properties "XNPEU";}}}%
%BeginExpansion
{\includegraphics[
%natheight=3.000000in,
%natwidth=3.000000in,
height=0.1773in,
width=0.1773in
]%
{ut00.ps}%
}%
%EndExpansion
&
%TCIMACRO{\FRAME{itbpF}{0.1773in}{0.1773in}{0in}{}{}{ut03.ps}%
%width 0.1773in;  height 0.1773in;  depth 0in;  original-width 3in;
%original-height 3in;  cropleft "0";  croptop "1";  cropright "1";
%cropbottom "0";  filename 'ut03.ps';file-properties "XNPEU";}}}%
%BeginExpansion
{\includegraphics[
%natheight=3.000000in,
%natwidth=3.000000in,
height=0.1773in,
width=0.1773in
]%
{ut03.ps}%
}%
%EndExpansion
&
%TCIMACRO{\FRAME{itbpF}{0.1773in}{0.1773in}{0in}{}{}{ut04.ps}%
%width 0.1773in;  height 0.1773in;  depth 0in;  original-width 3in;
%original-height 3in;  cropleft "0";  croptop "1";  cropright "1";
%cropbottom "0";  filename 'ut04.ps';file-properties "XNPEU";}}}%
%BeginExpansion
{\includegraphics[
%natheight=3.000000in,
%natwidth=3.000000in,
height=0.1773in,
width=0.1773in
]%
{ut04.ps}%
}%
%EndExpansion
&
%TCIMACRO{\FRAME{itbpF}{0.1773in}{0.1773in}{0in}{}{}{ut00.ps}%
%width 0.1773in;  height 0.1773in;  depth 0in;  original-width 3in;
%original-height 3in;  cropleft "0";  croptop "1";  cropright "1";
%cropbottom "0";  filename 'ut00.ps';file-properties "XNPEU";}}}%
%BeginExpansion
{\includegraphics[
%natheight=3.000000in,
%natwidth=3.000000in,
height=0.1773in,
width=0.1773in
]%
{ut00.ps}%
}%
%EndExpansion
\end{array}
\right\rangle \left\langle
\begin{array}
[c]{cccc}%
%TCIMACRO{\FRAME{itbpF}{0.1773in}{0.1773in}{0in}{}{}{ut00.ps}%
%width 0.1773in;  height 0.1773in;  depth 0in;  original-width 3in;
%original-height 3in;  cropleft "0";  croptop "1";  cropright "1";
%cropbottom "0";  filename 'ut00.ps';file-properties "XNPEU";}}}%
%BeginExpansion
{\includegraphics[
%natheight=3.000000in,
%natwidth=3.000000in,
height=0.1773in,
width=0.1773in
]%
{ut00.ps}%
}%
%EndExpansion
&
%TCIMACRO{\FRAME{itbpF}{0.1773in}{0.1773in}{0in}{}{}{ut02.ps}%
%width 0.1773in;  height 0.1773in;  depth 0in;  original-width 3in;
%original-height 3in;  cropleft "0";  croptop "1";  cropright "1";
%cropbottom "0";  filename 'ut02.ps';file-properties "XNPEU";}}}%
%BeginExpansion
{\includegraphics[
%natheight=3.000000in,
%natwidth=3.000000in,
height=0.1773in,
width=0.1773in
]%
{ut02.ps}%
}%
%EndExpansion
&
%TCIMACRO{\FRAME{itbpF}{0.1773in}{0.1773in}{0in}{}{}{ut01.ps}%
%width 0.1773in;  height 0.1773in;  depth 0in;  original-width 3in;
%original-height 3in;  cropleft "0";  croptop "1";  cropright "1";
%cropbottom "0";  filename 'ut01.ps';file-properties "XNPEU";}}}%
%BeginExpansion
{\includegraphics[
%natheight=3.000000in,
%natwidth=3.000000in,
height=0.1773in,
width=0.1773in
]%
{ut01.ps}%
}%
%EndExpansion
&
%TCIMACRO{\FRAME{itbpF}{0.1773in}{0.1773in}{0in}{}{}{ut00.ps}%
%width 0.1773in;  height 0.1773in;  depth 0in;  original-width 3in;
%original-height 3in;  cropleft "0";  croptop "1";  cropright "1";
%cropbottom "0";  filename 'ut00.ps';file-properties "XNPEU";}}}%
%BeginExpansion
{\includegraphics[
%natheight=3.000000in,
%natwidth=3.000000in,
height=0.1773in,
width=0.1773in
]%
{ut00.ps}%
}%
%EndExpansion
\\%
%TCIMACRO{\FRAME{itbpF}{0.1773in}{0.1773in}{0in}{}{}{ut02.ps}%
%width 0.1773in;  height 0.1773in;  depth 0in;  original-width 3in;
%original-height 3in;  cropleft "0";  croptop "1";  cropright "1";
%cropbottom "0";  filename 'ut02.ps';file-properties "XNPEU";}}}%
%BeginExpansion
{\includegraphics[
%natheight=3.000000in,
%natwidth=3.000000in,
height=0.1773in,
width=0.1773in
]%
{ut02.ps}%
}%
%EndExpansion
&
%TCIMACRO{\FRAME{itbpF}{0.1773in}{0.1773in}{0in}{}{}{ut09.ps}%
%width 0.1773in;  height 0.1773in;  depth 0in;  original-width 3in;
%original-height 3in;  cropleft "0";  croptop "1";  cropright "1";
%cropbottom "0";  filename 'ut09.ps';file-properties "XNPEU";}}}%
%BeginExpansion
{\includegraphics[
%natheight=3.000000in,
%natwidth=3.000000in,
height=0.1773in,
width=0.1773in
]%
{ut09.ps}%
}%
%EndExpansion
&
%TCIMACRO{\FRAME{itbpF}{0.1773in}{0.1773in}{0in}{}{}{ut10.ps}%
%width 0.1773in;  height 0.1773in;  depth 0in;  original-width 3in;
%original-height 3in;  cropleft "0";  croptop "1";  cropright "1";
%cropbottom "0";  filename 'ut10.ps';file-properties "XNPEU";}}}%
%BeginExpansion
{\includegraphics[
%natheight=3.000000in,
%natwidth=3.000000in,
height=0.1773in,
width=0.1773in
]%
{ut10.ps}%
}%
%EndExpansion
&
%TCIMACRO{\FRAME{itbpF}{0.1773in}{0.1773in}{0in}{}{}{ut01.ps}%
%width 0.1773in;  height 0.1773in;  depth 0in;  original-width 3in;
%original-height 3in;  cropleft "0";  croptop "1";  cropright "1";
%cropbottom "0";  filename 'ut01.ps';file-properties "XNPEU";}}}%
%BeginExpansion
{\includegraphics[
%natheight=3.000000in,
%natwidth=3.000000in,
height=0.1773in,
width=0.1773in
]%
{ut01.ps}%
}%
%EndExpansion
\\%
%TCIMACRO{\FRAME{itbpF}{0.1773in}{0.1773in}{0in}{}{}{ut03.ps}%
%width 0.1773in;  height 0.1773in;  depth 0in;  original-width 3in;
%original-height 3in;  cropleft "0";  croptop "1";  cropright "1";
%cropbottom "0";  filename 'ut03.ps';file-properties "XNPEU";}}}%
%BeginExpansion
{\includegraphics[
%natheight=3.000000in,
%natwidth=3.000000in,
height=0.1773in,
width=0.1773in
]%
{ut03.ps}%
}%
%EndExpansion
&
%TCIMACRO{\FRAME{itbpF}{0.1773in}{0.1773in}{0in}{}{}{ut07.ps}%
%width 0.1773in;  height 0.1773in;  depth 0in;  original-width 3in;
%original-height 3in;  cropleft "0";  croptop "1";  cropright "1";
%cropbottom "0";  filename 'ut07.ps';file-properties "XNPEU";}}}%
%BeginExpansion
{\includegraphics[
%natheight=3.000000in,
%natwidth=3.000000in,
height=0.1773in,
width=0.1773in
]%
{ut07.ps}%
}%
%EndExpansion
&
%TCIMACRO{\FRAME{itbpF}{0.1773in}{0.1773in}{0in}{}{}{ut09.ps}%
%width 0.1773in;  height 0.1773in;  depth 0in;  original-width 3in;
%original-height 3in;  cropleft "0";  croptop "1";  cropright "1";
%cropbottom "0";  filename 'ut09.ps';file-properties "XNPEU";}}}%
%BeginExpansion
{\includegraphics[
%natheight=3.000000in,
%natwidth=3.000000in,
height=0.1773in,
width=0.1773in
]%
{ut09.ps}%
}%
%EndExpansion
&
%TCIMACRO{\FRAME{itbpF}{0.1773in}{0.1773in}{0in}{}{}{ut04.ps}%
%width 0.1773in;  height 0.1773in;  depth 0in;  original-width 3in;
%original-height 3in;  cropleft "0";  croptop "1";  cropright "1";
%cropbottom "0";  filename 'ut04.ps';file-properties "XNPEU";}}}%
%BeginExpansion
{\includegraphics[
%natheight=3.000000in,
%natwidth=3.000000in,
height=0.1773in,
width=0.1773in
]%
{ut04.ps}%
}%
%EndExpansion
\\%
%TCIMACRO{\FRAME{itbpF}{0.1773in}{0.1773in}{0in}{}{}{ut00.ps}%
%width 0.1773in;  height 0.1773in;  depth 0in;  original-width 3in;
%original-height 3in;  cropleft "0";  croptop "1";  cropright "1";
%cropbottom "0";  filename 'ut00.ps';file-properties "XNPEU";}}}%
%BeginExpansion
{\includegraphics[
%natheight=3.000000in,
%natwidth=3.000000in,
height=0.1773in,
width=0.1773in
]%
{ut00.ps}%
}%
%EndExpansion
&
%TCIMACRO{\FRAME{itbpF}{0.1773in}{0.1773in}{0in}{}{}{ut03.ps}%
%width 0.1773in;  height 0.1773in;  depth 0in;  original-width 3in;
%original-height 3in;  cropleft "0";  croptop "1";  cropright "1";
%cropbottom "0";  filename 'ut03.ps';file-properties "XNPEU";}}}%
%BeginExpansion
{\includegraphics[
%natheight=3.000000in,
%natwidth=3.000000in,
height=0.1773in,
width=0.1773in
]%
{ut03.ps}%
}%
%EndExpansion
&
%TCIMACRO{\FRAME{itbpF}{0.1773in}{0.1773in}{0in}{}{}{ut04.ps}%
%width 0.1773in;  height 0.1773in;  depth 0in;  original-width 3in;
%original-height 3in;  cropleft "0";  croptop "1";  cropright "1";
%cropbottom "0";  filename 'ut04.ps';file-properties "XNPEU";}}}%
%BeginExpansion
{\includegraphics[
%natheight=3.000000in,
%natwidth=3.000000in,
height=0.1773in,
width=0.1773in
]%
{ut04.ps}%
}%
%EndExpansion
&
%TCIMACRO{\FRAME{itbpF}{0.1773in}{0.1773in}{0in}{}{}{ut00.ps}%
%width 0.1773in;  height 0.1773in;  depth 0in;  original-width 3in;
%original-height 3in;  cropleft "0";  croptop "1";  cropright "1";
%cropbottom "0";  filename 'ut00.ps';file-properties "XNPEU";}}}%
%BeginExpansion
{\includegraphics[
%natheight=3.000000in,
%natwidth=3.000000in,
height=0.1773in,
width=0.1773in
]%
{ut00.ps}%
}%
%EndExpansion
\end{array}
\right\vert $
\end{example}

\bigskip

\begin{remark}
For yet another approach to quantum knot measurement, we refer the reader to
the brief discussion on quantum knot tomography found in item 11) in the
conclusion of this paper. \ 
\end{remark}

\bigskip

\section{Conclusion: Open questions and future directions}

\bigskip

There are many possible open questions and future directions for research.
\ We mention only a few.

\bigskip

\begin{itemize}
\item[\textbf{1)}] What is the exact structure of the ambient group
$\mathbb{A}(n)$ and its direct limit
\[
\mathbb{A}=\lim_{\longrightarrow}\mathbb{A}(n)\text{ .}%
\]
Can one write down an explicit presentation for $\mathbb{A}(n)$? for
$\mathbb{A}$? \ The fact that the ambient group $\mathbb{A}(n)$ is generated
by involutions suggests that it may be a Coxeter group. \ Is it a Coxeter
group?\bigskip

\item[\textbf{2)}] Unlike classical knots, quantum knots can exhibit the
non-classical behavior of quantum superposition and quantum entanglement.
\ Are topological entanglement and quantum entanglement related to one
another? \ If so, how? \ \bigskip

\item[\textbf{3)}] What other ways are there to distinguish quantum knots from
classical knots?\bigskip

\item[\textbf{4)}] How does one find a quantum observable for the Jones
polynomial? \ This would be a family of observables parameterized by points on
the unit circle in the complex plane. \ Does this approach lead to an
algorithmic improvement to the quantum algorithm given by Aharonov, Jones, and
Landau in \cite{Aharonov1}? \ (See also \cite{Lomonaco2}, \cite{Shor1}%
.)\bigskip

\item[\textbf{5)}] How does one create quantum knot observables that represent
other knot invariants such as, for example, the Vassiliev invariants?\bigskip

\item[\textbf{6)}] What is gained by extending the definition of quantum knot
observables to POVMs?\bigskip

\item[\textbf{7)}] What is gained by extending the definition of quantum knots
to mixed ensembles?\bigskip

\item[\textbf{8)}] Define the \textbf{mosaic number} of a knot $k$ as the
smallest integer $n$ for which $k$ is representable as a knot $n$-mosaic.
\ For example, the mosaic number of the trefoil is $4$, as is illustrated by
the following knot $n$-mosaic:%
\[%
\begin{array}
[c]{cccc}%
%TCIMACRO{\FRAME{itbpF}{0.1773in}{0.1773in}{0in}{}{}{ut00.ps}%
%width 0.1773in;  height 0.1773in;  depth 0in;  original-width 3in;
%original-height 3in;  cropleft "0";  croptop "1";  cropright "1";
%cropbottom "0";  filename 'ut00.ps';file-properties "XNPEU";}}}%
%BeginExpansion
{\includegraphics[
%natheight=3.000000in,
%natwidth=3.000000in,
height=0.1773in,
width=0.1773in
]%
{ut00.ps}%
}%
%EndExpansion
&
%TCIMACRO{\FRAME{itbpF}{0.1773in}{0.1773in}{0in}{}{}{ut02.ps}%
%width 0.1773in;  height 0.1773in;  depth 0in;  original-width 3in;
%original-height 3in;  cropleft "0";  croptop "1";  cropright "1";
%cropbottom "0";  filename 'ut02.ps';file-properties "XNPEU";}}}%
%BeginExpansion
{\includegraphics[
%natheight=3.000000in,
%natwidth=3.000000in,
height=0.1773in,
width=0.1773in
]%
{ut02.ps}%
}%
%EndExpansion
&
%TCIMACRO{\FRAME{itbpF}{0.1773in}{0.1773in}{0in}{}{}{ut01.ps}%
%width 0.1773in;  height 0.1773in;  depth 0in;  original-width 3in;
%original-height 3in;  cropleft "0";  croptop "1";  cropright "1";
%cropbottom "0";  filename 'ut01.ps';file-properties "XNPEU";}}}%
%BeginExpansion
{\includegraphics[
%natheight=3.000000in,
%natwidth=3.000000in,
height=0.1773in,
width=0.1773in
]%
{ut01.ps}%
}%
%EndExpansion
&
%TCIMACRO{\FRAME{itbpF}{0.1773in}{0.1773in}{0in}{}{}{ut00.ps}%
%width 0.1773in;  height 0.1773in;  depth 0in;  original-width 3in;
%original-height 3in;  cropleft "0";  croptop "1";  cropright "1";
%cropbottom "0";  filename 'ut00.ps';file-properties "XNPEU";}}}%
%BeginExpansion
{\includegraphics[
%natheight=3.000000in,
%natwidth=3.000000in,
height=0.1773in,
width=0.1773in
]%
{ut00.ps}%
}%
%EndExpansion
\\%
%TCIMACRO{\FRAME{itbpF}{0.1773in}{0.1773in}{0in}{}{}{ut02.ps}%
%width 0.1773in;  height 0.1773in;  depth 0in;  original-width 3in;
%original-height 3in;  cropleft "0";  croptop "1";  cropright "1";
%cropbottom "0";  filename 'ut02.ps';file-properties "XNPEU";}}}%
%BeginExpansion
{\includegraphics[
%natheight=3.000000in,
%natwidth=3.000000in,
height=0.1773in,
width=0.1773in
]%
{ut02.ps}%
}%
%EndExpansion
&
%TCIMACRO{\FRAME{itbpF}{0.1773in}{0.1773in}{0in}{}{}{ut09.ps}%
%width 0.1773in;  height 0.1773in;  depth 0in;  original-width 3in;
%original-height 3in;  cropleft "0";  croptop "1";  cropright "1";
%cropbottom "0";  filename 'ut09.ps';file-properties "XNPEU";}}}%
%BeginExpansion
{\includegraphics[
%natheight=3.000000in,
%natwidth=3.000000in,
height=0.1773in,
width=0.1773in
]%
{ut09.ps}%
}%
%EndExpansion
&
%TCIMACRO{\FRAME{itbpF}{0.1773in}{0.1773in}{0in}{}{}{ut10.ps}%
%width 0.1773in;  height 0.1773in;  depth 0in;  original-width 3in;
%original-height 3in;  cropleft "0";  croptop "1";  cropright "1";
%cropbottom "0";  filename 'ut10.ps';file-properties "XNPEU";}}}%
%BeginExpansion
{\includegraphics[
%natheight=3.000000in,
%natwidth=3.000000in,
height=0.1773in,
width=0.1773in
]%
{ut10.ps}%
}%
%EndExpansion
&
%TCIMACRO{\FRAME{itbpF}{0.1773in}{0.1773in}{0in}{}{}{ut01.ps}%
%width 0.1773in;  height 0.1773in;  depth 0in;  original-width 3in;
%original-height 3in;  cropleft "0";  croptop "1";  cropright "1";
%cropbottom "0";  filename 'ut01.ps';file-properties "XNPEU";}}}%
%BeginExpansion
{\includegraphics[
%natheight=3.000000in,
%natwidth=3.000000in,
height=0.1773in,
width=0.1773in
]%
{ut01.ps}%
}%
%EndExpansion
\\%
%TCIMACRO{\FRAME{itbpF}{0.1773in}{0.1773in}{0in}{}{}{ut06.ps}%
%width 0.1773in;  height 0.1773in;  depth 0in;  original-width 3in;
%original-height 3in;  cropleft "0";  croptop "1";  cropright "1";
%cropbottom "0";  filename 'ut06.ps';file-properties "XNPEU";}}}%
%BeginExpansion
{\includegraphics[
%natheight=3.000000in,
%natwidth=3.000000in,
height=0.1773in,
width=0.1773in
]%
{ut06.ps}%
}%
%EndExpansion
&
%TCIMACRO{\FRAME{itbpF}{0.1773in}{0.1773in}{0in}{}{}{ut03.ps}%
%width 0.1773in;  height 0.1773in;  depth 0in;  original-width 3in;
%original-height 3in;  cropleft "0";  croptop "1";  cropright "1";
%cropbottom "0";  filename 'ut03.ps';file-properties "XNPEU";}}}%
%BeginExpansion
{\includegraphics[
%natheight=3.000000in,
%natwidth=3.000000in,
height=0.1773in,
width=0.1773in
]%
{ut03.ps}%
}%
%EndExpansion
&
%TCIMACRO{\FRAME{itbpF}{0.1773in}{0.1773in}{0in}{}{}{ut09.ps}%
%width 0.1773in;  height 0.1773in;  depth 0in;  original-width 3in;
%original-height 3in;  cropleft "0";  croptop "1";  cropright "1";
%cropbottom "0";  filename 'ut09.ps';file-properties "XNPEU";}}}%
%BeginExpansion
{\includegraphics[
%natheight=3.000000in,
%natwidth=3.000000in,
height=0.1773in,
width=0.1773in
]%
{ut09.ps}%
}%
%EndExpansion
&
%TCIMACRO{\FRAME{itbpF}{0.1773in}{0.1773in}{0in}{}{}{ut04.ps}%
%width 0.1773in;  height 0.1773in;  depth 0in;  original-width 3in;
%original-height 3in;  cropleft "0";  croptop "1";  cropright "1";
%cropbottom "0";  filename 'ut04.ps';file-properties "XNPEU";}}}%
%BeginExpansion
{\includegraphics[
%natheight=3.000000in,
%natwidth=3.000000in,
height=0.1773in,
width=0.1773in
]%
{ut04.ps}%
}%
%EndExpansion
\\%
%TCIMACRO{\FRAME{itbpF}{0.1773in}{0.1773in}{0in}{}{}{ut03.ps}%
%width 0.1773in;  height 0.1773in;  depth 0in;  original-width 3in;
%original-height 3in;  cropleft "0";  croptop "1";  cropright "1";
%cropbottom "0";  filename 'ut03.ps';file-properties "XNPEU";}}}%
%BeginExpansion
{\includegraphics[
%natheight=3.000000in,
%natwidth=3.000000in,
height=0.1773in,
width=0.1773in
]%
{ut03.ps}%
}%
%EndExpansion
&
%TCIMACRO{\FRAME{itbpF}{0.1773in}{0.1773in}{0in}{}{}{ut05.ps}%
%width 0.1773in;  height 0.1773in;  depth 0in;  original-width 3in;
%original-height 3in;  cropleft "0";  croptop "1";  cropright "1";
%cropbottom "0";  filename 'ut05.ps';file-properties "XNPEU";}}}%
%BeginExpansion
{\includegraphics[
%natheight=3.000000in,
%natwidth=3.000000in,
height=0.1773in,
width=0.1773in
]%
{ut05.ps}%
}%
%EndExpansion
&
%TCIMACRO{\FRAME{itbpF}{0.1773in}{0.1773in}{0in}{}{}{ut04.ps}%
%width 0.1773in;  height 0.1773in;  depth 0in;  original-width 3in;
%original-height 3in;  cropleft "0";  croptop "1";  cropright "1";
%cropbottom "0";  filename 'ut04.ps';file-properties "XNPEU";}}}%
%BeginExpansion
{\includegraphics[
%natheight=3.000000in,
%natwidth=3.000000in,
height=0.1773in,
width=0.1773in
]%
{ut04.ps}%
}%
%EndExpansion
&
%TCIMACRO{\FRAME{itbpF}{0.1773in}{0.1773in}{0in}{}{}{ut00.ps}%
%width 0.1773in;  height 0.1773in;  depth 0in;  original-width 3in;
%original-height 3in;  cropleft "0";  croptop "1";  cropright "1";
%cropbottom "0";  filename 'ut00.ps';file-properties "XNPEU";}}}%
%BeginExpansion
{\includegraphics[
%natheight=3.000000in,
%natwidth=3.000000in,
height=0.1773in,
width=0.1773in
]%
{ut00.ps}%
}%
%EndExpansion
\end{array}
\]
\ In general, how does one compute the mosaic number? \ Is the mosaic number
related to the crossing number of a knot? \ How does one find an observable
for the mosaic number?\bigskip

\item[\textbf{9)}] Let $D_{n}$ denote the dimension of the Hilbert space
$\mathcal{K}^{(n)}$ of quantum knot $n$-mosaics. \ We have shown that
$D_{1}=1$, $D_{2}=2$, and $D_{3}=22$. \ Find $D_{n}$ for other values of $n$.
\ \ A very loose upper bound for $D_{n}$ is obviously $11^{n^{2}}$.\bigskip

\item[\textbf{10)}] Consider the following alternate stronger definitions of
quantum knot $n$-type and quantum knot type: \bigskip\newline Let $Q\left(
\mathcal{K}^{(n)},\mathbb{A}(n)\right)  $ be a quantum knot system, and let
$\mathcal{U}\left(  \mathcal{K}^{(n)}\right)  $ denote the Lie group of all
unitary transformations on the Hilbert space $\mathcal{K}^{(n)}$. \ Define the
\textbf{continuous ambient group} $\widetilde{\mathbb{A}}(n)$ as the smallest
connected Lie subgroup of $\mathcal{U}\left(  \mathcal{K}^{(n)}\right)  $
containing the discrete ambient group $\mathbb{A}(n)$.\bigskip

\begin{proposition}
Let $\mathcal{G}$ denote the set of planar isotopy and Reidemeister generators
of the discrete ambient group $\mathbb{A}(n)$, and let $a(n)$ be the Lie
algebra generated by the elements of the set%
\[
\left\{  \ln_{P}\left(  g\right)  :g\in\mathcal{G}\right\}  \text{ ,}%
\]
where $\ln_{P}$ denotes the principal branch of the natural log on
$\mathcal{U}\left(  \mathcal{K}^{(n)}\right)  $. \ Then the continuous ambient
group is given by%
\[
\widetilde{\mathbb{A}}(n)=\exp\left(  a(n)\right)  \text{ .}%
\]

\end{proposition}

We define two quantum $n$-knots $\left\vert \psi_{1}\right\rangle $ and
$\left\vert \psi_{2}\right\rangle $ to be of the \textbf{same continuous knot
}$n$\textbf{-type}, written
\[
\left\vert \psi_{2}\right\rangle \underset{n}{\approx}\left\vert \psi
_{2}\right\rangle \text{ \ ,}%
\]
provided there exists an element $g$ of the continuous ambient group
$\widetilde{\mathbb{A}}(n)$ which transforms $\left\vert \psi_{1}\right\rangle
$ into $\left\vert \psi_{2}\right\rangle $, i.e., such that $g\left\vert
\psi_{1}\right\rangle =\left\vert \psi_{2}\right\rangle $. \ They are of the
\textbf{same continuous knot type}, written $\left\vert \psi_{1}\right\rangle
\approx\left\vert \psi_{2}\right\rangle $, if there exists an integer $\ell$
such that $\iota^{\ell}\left\vert \psi_{1}\right\rangle \underset{n+\ell
}{\approx}\iota^{\ell}\left\vert \psi_{2}\right\rangle $.\bigskip

\begin{conjecture}
Let $K_{1}$ and $K_{2}$ denote two knot $n$-mosaics, and let $\left\vert
K_{1}\right\rangle $ and $\left\vert K_{2}\right\rangle $ denote the
corresponding quantum knots. \ Then
\[
\left\vert K_{1}\right\rangle \underset{n}{\approx}\left\vert K_{2}%
\right\rangle \Longleftrightarrow K_{1}\underset{n}{\sim}K_{2}\text{ \ \ and
\ \ }\left\vert K_{1}\right\rangle \approx\left\vert K_{2}\right\rangle
\Longleftrightarrow K_{1}\sim K_{2}%
\]

\end{conjecture}

Thus, if this conjectures 1 and 2 are true, these two stronger definitions of
quantum knot $n$-type and quantum knot type fully capture all of classical
tame knot theory. \ Moreover, these two stronger definitions have a number of
advantages over the weaker definitions, two of which are the following:
\ \bigskip

\item Under the Hamiltonians associated with the generators $\mathcal{G}$, the
Schroedinger equation determines a connected continuous path in $\mathcal{K}%
^{(n)}$ consisting of quantum $n$-knots all of the same quantum continuous
knot $n$-type. \ \bigskip

\item Although the following two quantum knots are not of the same discrete
knot $n$-type, they are however of the same continuous knot $n$-type%
\[
\left\vert \psi_{1}\right\rangle =\left\vert
\begin{array}
[c]{ccc}%
%TCIMACRO{\FRAME{itbpF}{0.1773in}{0.1773in}{0in}{}{}{ut02.ps}%
%width 0.1773in;  height 0.1773in;  depth 0in;  original-width 3in;
%original-height 3in;  cropleft "0";  croptop "1";  cropright "1";
%cropbottom "0";  filename 'ut02.ps';file-properties "XNPEU";}}}%
%BeginExpansion
{\includegraphics[
%natheight=3.000000in,
%natwidth=3.000000in,
height=0.1773in,
width=0.1773in
]%
{ut02.ps}%
}%
%EndExpansion
&
%TCIMACRO{\FRAME{itbpF}{0.1773in}{0.1773in}{0in}{}{}{ut05.ps}%
%width 0.1773in;  height 0.1773in;  depth 0in;  original-width 3in;
%original-height 3in;  cropleft "0";  croptop "1";  cropright "1";
%cropbottom "0";  filename 'ut05.ps';file-properties "XNPEU";}}}%
%BeginExpansion
{\includegraphics[
%natheight=3.000000in,
%natwidth=3.000000in,
height=0.1773in,
width=0.1773in
]%
{ut05.ps}%
}%
%EndExpansion
&
%TCIMACRO{\FRAME{itbpF}{0.1773in}{0.1773in}{0in}{}{}{ut01.ps}%
%width 0.1773in;  height 0.1773in;  depth 0in;  original-width 3in;
%original-height 3in;  cropleft "0";  croptop "1";  cropright "1";
%cropbottom "0";  filename 'ut01.ps';file-properties "XNPEU";}}}%
%BeginExpansion
{\includegraphics[
%natheight=3.000000in,
%natwidth=3.000000in,
height=0.1773in,
width=0.1773in
]%
{ut01.ps}%
}%
%EndExpansion
\\%
%TCIMACRO{\FRAME{itbpF}{0.1773in}{0.1773in}{0in}{}{}{ut06.ps}%
%width 0.1773in;  height 0.1773in;  depth 0in;  original-width 3in;
%original-height 3in;  cropleft "0";  croptop "1";  cropright "1";
%cropbottom "0";  filename 'ut06.ps';file-properties "XNPEU";}}}%
%BeginExpansion
{\includegraphics[
%natheight=3.000000in,
%natwidth=3.000000in,
height=0.1773in,
width=0.1773in
]%
{ut06.ps}%
}%
%EndExpansion
&
%TCIMACRO{\FRAME{itbpF}{0.1773in}{0.1773in}{0in}{}{}{ut00.ps}%
%width 0.1773in;  height 0.1773in;  depth 0in;  original-width 3in;
%original-height 3in;  cropleft "0";  croptop "1";  cropright "1";
%cropbottom "0";  filename 'ut00.ps';file-properties "XNPEU";}}}%
%BeginExpansion
{\includegraphics[
%natheight=3.000000in,
%natwidth=3.000000in,
height=0.1773in,
width=0.1773in
]%
{ut00.ps}%
}%
%EndExpansion
&
%TCIMACRO{\FRAME{itbpF}{0.1773in}{0.1773in}{0in}{}{}{ut06.ps}%
%width 0.1773in;  height 0.1773in;  depth 0in;  original-width 3in;
%original-height 3in;  cropleft "0";  croptop "1";  cropright "1";
%cropbottom "0";  filename 'ut06.ps';file-properties "XNPEU";}}}%
%BeginExpansion
{\includegraphics[
%natheight=3.000000in,
%natwidth=3.000000in,
height=0.1773in,
width=0.1773in
]%
{ut06.ps}%
}%
%EndExpansion
\\%
%TCIMACRO{\FRAME{itbpF}{0.1773in}{0.1773in}{0in}{}{}{ut03.ps}%
%width 0.1773in;  height 0.1773in;  depth 0in;  original-width 3in;
%original-height 3in;  cropleft "0";  croptop "1";  cropright "1";
%cropbottom "0";  filename 'ut03.ps';file-properties "XNPEU";}}}%
%BeginExpansion
{\includegraphics[
%natheight=3.000000in,
%natwidth=3.000000in,
height=0.1773in,
width=0.1773in
]%
{ut03.ps}%
}%
%EndExpansion
&
%TCIMACRO{\FRAME{itbpF}{0.1773in}{0.1773in}{0in}{}{}{ut05.ps}%
%width 0.1773in;  height 0.1773in;  depth 0in;  original-width 3in;
%original-height 3in;  cropleft "0";  croptop "1";  cropright "1";
%cropbottom "0";  filename 'ut05.ps';file-properties "XNPEU";}}}%
%BeginExpansion
{\includegraphics[
%natheight=3.000000in,
%natwidth=3.000000in,
height=0.1773in,
width=0.1773in
]%
{ut05.ps}%
}%
%EndExpansion
&
%TCIMACRO{\FRAME{itbpF}{0.1773in}{0.1773in}{0in}{}{}{ut04.ps}%
%width 0.1773in;  height 0.1773in;  depth 0in;  original-width 3in;
%original-height 3in;  cropleft "0";  croptop "1";  cropright "1";
%cropbottom "0";  filename 'ut04.ps';file-properties "XNPEU";}}}%
%BeginExpansion
{\includegraphics[
%natheight=3.000000in,
%natwidth=3.000000in,
height=0.1773in,
width=0.1773in
]%
{ut04.ps}%
}%
%EndExpansion
\end{array}
\right\rangle \text{ \ \ and \ \ }\left\vert \psi_{2}\right\rangle =\frac
{1}{\sqrt{2}}\left(  \left\vert
\begin{array}
[c]{ccc}%
%TCIMACRO{\FRAME{itbpF}{0.1773in}{0.1773in}{0in}{}{}{ut02.ps}%
%width 0.1773in;  height 0.1773in;  depth 0in;  original-width 3in;
%original-height 3in;  cropleft "0";  croptop "1";  cropright "1";
%cropbottom "0";  filename 'ut02.ps';file-properties "XNPEU";}}}%
%BeginExpansion
{\includegraphics[
%natheight=3.000000in,
%natwidth=3.000000in,
height=0.1773in,
width=0.1773in
]%
{ut02.ps}%
}%
%EndExpansion
&
%TCIMACRO{\FRAME{itbpF}{0.1773in}{0.1773in}{0in}{}{}{ut01.ps}%
%width 0.1773in;  height 0.1773in;  depth 0in;  original-width 3in;
%original-height 3in;  cropleft "0";  croptop "1";  cropright "1";
%cropbottom "0";  filename 'ut01.ps';file-properties "XNPEU";}}}%
%BeginExpansion
{\includegraphics[
%natheight=3.000000in,
%natwidth=3.000000in,
height=0.1773in,
width=0.1773in
]%
{ut01.ps}%
}%
%EndExpansion
&
%TCIMACRO{\FRAME{itbpF}{0.1773in}{0.1773in}{0in}{}{}{ut00.ps}%
%width 0.1773in;  height 0.1773in;  depth 0in;  original-width 3in;
%original-height 3in;  cropleft "0";  croptop "1";  cropright "1";
%cropbottom "0";  filename 'ut00.ps';file-properties "XNPEU";}}}%
%BeginExpansion
{\includegraphics[
%natheight=3.000000in,
%natwidth=3.000000in,
height=0.1773in,
width=0.1773in
]%
{ut00.ps}%
}%
%EndExpansion
\\%
%TCIMACRO{\FRAME{itbpF}{0.1773in}{0.1773in}{0in}{}{}{ut03.ps}%
%width 0.1773in;  height 0.1773in;  depth 0in;  original-width 3in;
%original-height 3in;  cropleft "0";  croptop "1";  cropright "1";
%cropbottom "0";  filename 'ut03.ps';file-properties "XNPEU";}}}%
%BeginExpansion
{\includegraphics[
%natheight=3.000000in,
%natwidth=3.000000in,
height=0.1773in,
width=0.1773in
]%
{ut03.ps}%
}%
%EndExpansion
&
%TCIMACRO{\FRAME{itbpF}{0.1773in}{0.1773in}{0in}{}{}{ut04.ps}%
%width 0.1773in;  height 0.1773in;  depth 0in;  original-width 3in;
%original-height 3in;  cropleft "0";  croptop "1";  cropright "1";
%cropbottom "0";  filename 'ut04.ps';file-properties "XNPEU";}}}%
%BeginExpansion
{\includegraphics[
%natheight=3.000000in,
%natwidth=3.000000in,
height=0.1773in,
width=0.1773in
]%
{ut04.ps}%
}%
%EndExpansion
&
%TCIMACRO{\FRAME{itbpF}{0.1773in}{0.1773in}{0in}{}{}{ut00.ps}%
%width 0.1773in;  height 0.1773in;  depth 0in;  original-width 3in;
%original-height 3in;  cropleft "0";  croptop "1";  cropright "1";
%cropbottom "0";  filename 'ut00.ps';file-properties "XNPEU";}}}%
%BeginExpansion
{\includegraphics[
%natheight=3.000000in,
%natwidth=3.000000in,
height=0.1773in,
width=0.1773in
]%
{ut00.ps}%
}%
%EndExpansion
\\%
%TCIMACRO{\FRAME{itbpF}{0.1773in}{0.1773in}{0in}{}{}{ut00.ps}%
%width 0.1773in;  height 0.1773in;  depth 0in;  original-width 3in;
%original-height 3in;  cropleft "0";  croptop "1";  cropright "1";
%cropbottom "0";  filename 'ut00.ps';file-properties "XNPEU";}}}%
%BeginExpansion
{\includegraphics[
%natheight=3.000000in,
%natwidth=3.000000in,
height=0.1773in,
width=0.1773in
]%
{ut00.ps}%
}%
%EndExpansion
&
%TCIMACRO{\FRAME{itbpF}{0.1773in}{0.1773in}{0in}{}{}{ut00.ps}%
%width 0.1773in;  height 0.1773in;  depth 0in;  original-width 3in;
%original-height 3in;  cropleft "0";  croptop "1";  cropright "1";
%cropbottom "0";  filename 'ut00.ps';file-properties "XNPEU";}}}%
%BeginExpansion
{\includegraphics[
%natheight=3.000000in,
%natwidth=3.000000in,
height=0.1773in,
width=0.1773in
]%
{ut00.ps}%
}%
%EndExpansion
&
%TCIMACRO{\FRAME{itbpF}{0.1773in}{0.1773in}{0in}{}{}{ut00.ps}%
%width 0.1773in;  height 0.1773in;  depth 0in;  original-width 3in;
%original-height 3in;  cropleft "0";  croptop "1";  cropright "1";
%cropbottom "0";  filename 'ut00.ps';file-properties "XNPEU";}}}%
%BeginExpansion
{\includegraphics[
%natheight=3.000000in,
%natwidth=3.000000in,
height=0.1773in,
width=0.1773in
]%
{ut00.ps}%
}%
%EndExpansion
\end{array}
\right\rangle +\left\vert
\begin{array}
[c]{ccc}%
%TCIMACRO{\FRAME{itbpF}{0.1773in}{0.1773in}{0in}{}{}{ut00.ps}%
%width 0.1773in;  height 0.1773in;  depth 0in;  original-width 3in;
%original-height 3in;  cropleft "0";  croptop "1";  cropright "1";
%cropbottom "0";  filename 'ut00.ps';file-properties "XNPEU";}}}%
%BeginExpansion
{\includegraphics[
%natheight=3.000000in,
%natwidth=3.000000in,
height=0.1773in,
width=0.1773in
]%
{ut00.ps}%
}%
%EndExpansion
&
%TCIMACRO{\FRAME{itbpF}{0.1773in}{0.1773in}{0in}{}{}{ut00.ps}%
%width 0.1773in;  height 0.1773in;  depth 0in;  original-width 3in;
%original-height 3in;  cropleft "0";  croptop "1";  cropright "1";
%cropbottom "0";  filename 'ut00.ps';file-properties "XNPEU";}}}%
%BeginExpansion
{\includegraphics[
%natheight=3.000000in,
%natwidth=3.000000in,
height=0.1773in,
width=0.1773in
]%
{ut00.ps}%
}%
%EndExpansion
&
%TCIMACRO{\FRAME{itbpF}{0.1773in}{0.1773in}{0in}{}{}{ut00.ps}%
%width 0.1773in;  height 0.1773in;  depth 0in;  original-width 3in;
%original-height 3in;  cropleft "0";  croptop "1";  cropright "1";
%cropbottom "0";  filename 'ut00.ps';file-properties "XNPEU";}}}%
%BeginExpansion
{\includegraphics[
%natheight=3.000000in,
%natwidth=3.000000in,
height=0.1773in,
width=0.1773in
]%
{ut00.ps}%
}%
%EndExpansion
\\%
%TCIMACRO{\FRAME{itbpF}{0.1773in}{0.1773in}{0in}{}{}{ut00.ps}%
%width 0.1773in;  height 0.1773in;  depth 0in;  original-width 3in;
%original-height 3in;  cropleft "0";  croptop "1";  cropright "1";
%cropbottom "0";  filename 'ut00.ps';file-properties "XNPEU";}}}%
%BeginExpansion
{\includegraphics[
%natheight=3.000000in,
%natwidth=3.000000in,
height=0.1773in,
width=0.1773in
]%
{ut00.ps}%
}%
%EndExpansion
&
%TCIMACRO{\FRAME{itbpF}{0.1773in}{0.1773in}{0in}{}{}{ut02.ps}%
%width 0.1773in;  height 0.1773in;  depth 0in;  original-width 3in;
%original-height 3in;  cropleft "0";  croptop "1";  cropright "1";
%cropbottom "0";  filename 'ut02.ps';file-properties "XNPEU";}}}%
%BeginExpansion
{\includegraphics[
%natheight=3.000000in,
%natwidth=3.000000in,
height=0.1773in,
width=0.1773in
]%
{ut02.ps}%
}%
%EndExpansion
&
%TCIMACRO{\FRAME{itbpF}{0.1773in}{0.1773in}{0in}{}{}{ut01.ps}%
%width 0.1773in;  height 0.1773in;  depth 0in;  original-width 3in;
%original-height 3in;  cropleft "0";  croptop "1";  cropright "1";
%cropbottom "0";  filename 'ut01.ps';file-properties "XNPEU";}}}%
%BeginExpansion
{\includegraphics[
%natheight=3.000000in,
%natwidth=3.000000in,
height=0.1773in,
width=0.1773in
]%
{ut01.ps}%
}%
%EndExpansion
\\%
%TCIMACRO{\FRAME{itbpF}{0.1773in}{0.1773in}{0in}{}{}{ut00.ps}%
%width 0.1773in;  height 0.1773in;  depth 0in;  original-width 3in;
%original-height 3in;  cropleft "0";  croptop "1";  cropright "1";
%cropbottom "0";  filename 'ut00.ps';file-properties "XNPEU";}}}%
%BeginExpansion
{\includegraphics[
%natheight=3.000000in,
%natwidth=3.000000in,
height=0.1773in,
width=0.1773in
]%
{ut00.ps}%
}%
%EndExpansion
&
%TCIMACRO{\FRAME{itbpF}{0.1773in}{0.1773in}{0in}{}{}{ut03.ps}%
%width 0.1773in;  height 0.1773in;  depth 0in;  original-width 3in;
%original-height 3in;  cropleft "0";  croptop "1";  cropright "1";
%cropbottom "0";  filename 'ut03.ps';file-properties "XNPEU";}}}%
%BeginExpansion
{\includegraphics[
%natheight=3.000000in,
%natwidth=3.000000in,
height=0.1773in,
width=0.1773in
]%
{ut03.ps}%
}%
%EndExpansion
&
%TCIMACRO{\FRAME{itbpF}{0.1773in}{0.1773in}{0in}{}{}{ut04.ps}%
%width 0.1773in;  height 0.1773in;  depth 0in;  original-width 3in;
%original-height 3in;  cropleft "0";  croptop "1";  cropright "1";
%cropbottom "0";  filename 'ut04.ps';file-properties "XNPEU";}}}%
%BeginExpansion
{\includegraphics[
%natheight=3.000000in,
%natwidth=3.000000in,
height=0.1773in,
width=0.1773in
]%
{ut04.ps}%
}%
%EndExpansion
\end{array}
\right\rangle \right)
\]
\bigskip Is the above conjecture true? \ If so, what is the structure of the
continuous ambient group $\widetilde{\mathbb{A}}(n)$ ? \ What are its
irreducible representations?\bigskip

\item[\textbf{11)}] Quantum knot tomography: Given repeated copies of a
quantum knot $\left\vert \psi\right\rangle $, how does one employ the method
of quantum state tomography\cite{Leonhardt1} to determine $\left\vert
\psi\right\rangle $? \ Most importantly, how can this be done with the
greatest efficiency? \ For example, given repeated copies of the unknown
quantum knot basis state
\[
\left\vert \psi\right\rangle =\left\vert
\begin{array}
[c]{cccc}%
%TCIMACRO{\FRAME{itbpF}{0.1773in}{0.1773in}{0in}{}{}{ut00.ps}%
%width 0.1773in;  height 0.1773in;  depth 0in;  original-width 3in;
%original-height 3in;  cropleft "0";  croptop "1";  cropright "1";
%cropbottom "0";  filename 'ut00.ps';file-properties "XNPEU";}}}%
%BeginExpansion
{\includegraphics[
%natheight=3.000000in,
%natwidth=3.000000in,
height=0.1773in,
width=0.1773in
]%
{ut00.ps}%
}%
%EndExpansion
&
%TCIMACRO{\FRAME{itbpF}{0.1773in}{0.1773in}{0in}{}{}{ut02.ps}%
%width 0.1773in;  height 0.1773in;  depth 0in;  original-width 3in;
%original-height 3in;  cropleft "0";  croptop "1";  cropright "1";
%cropbottom "0";  filename 'ut02.ps';file-properties "XNPEU";}}}%
%BeginExpansion
{\includegraphics[
%natheight=3.000000in,
%natwidth=3.000000in,
height=0.1773in,
width=0.1773in
]%
{ut02.ps}%
}%
%EndExpansion
&
%TCIMACRO{\FRAME{itbpF}{0.1773in}{0.1773in}{0in}{}{}{ut01.ps}%
%width 0.1773in;  height 0.1773in;  depth 0in;  original-width 3in;
%original-height 3in;  cropleft "0";  croptop "1";  cropright "1";
%cropbottom "0";  filename 'ut01.ps';file-properties "XNPEU";}}}%
%BeginExpansion
{\includegraphics[
%natheight=3.000000in,
%natwidth=3.000000in,
height=0.1773in,
width=0.1773in
]%
{ut01.ps}%
}%
%EndExpansion
&
%TCIMACRO{\FRAME{itbpF}{0.1773in}{0.1773in}{0in}{}{}{ut00.ps}%
%width 0.1773in;  height 0.1773in;  depth 0in;  original-width 3in;
%original-height 3in;  cropleft "0";  croptop "1";  cropright "1";
%cropbottom "0";  filename 'ut00.ps';file-properties "XNPEU";}}}%
%BeginExpansion
{\includegraphics[
%natheight=3.000000in,
%natwidth=3.000000in,
height=0.1773in,
width=0.1773in
]%
{ut00.ps}%
}%
%EndExpansion
\\%
%TCIMACRO{\FRAME{itbpF}{0.1773in}{0.1773in}{0in}{}{}{ut02.ps}%
%width 0.1773in;  height 0.1773in;  depth 0in;  original-width 3in;
%original-height 3in;  cropleft "0";  croptop "1";  cropright "1";
%cropbottom "0";  filename 'ut02.ps';file-properties "XNPEU";}}}%
%BeginExpansion
{\includegraphics[
%natheight=3.000000in,
%natwidth=3.000000in,
height=0.1773in,
width=0.1773in
]%
{ut02.ps}%
}%
%EndExpansion
&
%TCIMACRO{\FRAME{itbpF}{0.1773in}{0.1773in}{0in}{}{}{ut09.ps}%
%width 0.1773in;  height 0.1773in;  depth 0in;  original-width 3in;
%original-height 3in;  cropleft "0";  croptop "1";  cropright "1";
%cropbottom "0";  filename 'ut09.ps';file-properties "XNPEU";}}}%
%BeginExpansion
{\includegraphics[
%natheight=3.000000in,
%natwidth=3.000000in,
height=0.1773in,
width=0.1773in
]%
{ut09.ps}%
}%
%EndExpansion
&
%TCIMACRO{\FRAME{itbpF}{0.1773in}{0.1773in}{0in}{}{}{ut10.ps}%
%width 0.1773in;  height 0.1773in;  depth 0in;  original-width 3in;
%original-height 3in;  cropleft "0";  croptop "1";  cropright "1";
%cropbottom "0";  filename 'ut10.ps';file-properties "XNPEU";}}}%
%BeginExpansion
{\includegraphics[
%natheight=3.000000in,
%natwidth=3.000000in,
height=0.1773in,
width=0.1773in
]%
{ut10.ps}%
}%
%EndExpansion
&
%TCIMACRO{\FRAME{itbpF}{0.1773in}{0.1773in}{0in}{}{}{ut01.ps}%
%width 0.1773in;  height 0.1773in;  depth 0in;  original-width 3in;
%original-height 3in;  cropleft "0";  croptop "1";  cropright "1";
%cropbottom "0";  filename 'ut01.ps';file-properties "XNPEU";}}}%
%BeginExpansion
{\includegraphics[
%natheight=3.000000in,
%natwidth=3.000000in,
height=0.1773in,
width=0.1773in
]%
{ut01.ps}%
}%
%EndExpansion
\\%
%TCIMACRO{\FRAME{itbpF}{0.1773in}{0.1773in}{0in}{}{}{ut06.ps}%
%width 0.1773in;  height 0.1773in;  depth 0in;  original-width 3in;
%original-height 3in;  cropleft "0";  croptop "1";  cropright "1";
%cropbottom "0";  filename 'ut06.ps';file-properties "XNPEU";}}}%
%BeginExpansion
{\includegraphics[
%natheight=3.000000in,
%natwidth=3.000000in,
height=0.1773in,
width=0.1773in
]%
{ut06.ps}%
}%
%EndExpansion
&
%TCIMACRO{\FRAME{itbpF}{0.1773in}{0.1773in}{0in}{}{}{ut03.ps}%
%width 0.1773in;  height 0.1773in;  depth 0in;  original-width 3in;
%original-height 3in;  cropleft "0";  croptop "1";  cropright "1";
%cropbottom "0";  filename 'ut03.ps';file-properties "XNPEU";}}}%
%BeginExpansion
{\includegraphics[
%natheight=3.000000in,
%natwidth=3.000000in,
height=0.1773in,
width=0.1773in
]%
{ut03.ps}%
}%
%EndExpansion
&
%TCIMACRO{\FRAME{itbpF}{0.1773in}{0.1773in}{0in}{}{}{ut09.ps}%
%width 0.1773in;  height 0.1773in;  depth 0in;  original-width 3in;
%original-height 3in;  cropleft "0";  croptop "1";  cropright "1";
%cropbottom "0";  filename 'ut09.ps';file-properties "XNPEU";}}}%
%BeginExpansion
{\includegraphics[
%natheight=3.000000in,
%natwidth=3.000000in,
height=0.1773in,
width=0.1773in
]%
{ut09.ps}%
}%
%EndExpansion
&
%TCIMACRO{\FRAME{itbpF}{0.1773in}{0.1773in}{0in}{}{}{ut04.ps}%
%width 0.1773in;  height 0.1773in;  depth 0in;  original-width 3in;
%original-height 3in;  cropleft "0";  croptop "1";  cropright "1";
%cropbottom "0";  filename 'ut04.ps';file-properties "XNPEU";}}}%
%BeginExpansion
{\includegraphics[
%natheight=3.000000in,
%natwidth=3.000000in,
height=0.1773in,
width=0.1773in
]%
{ut04.ps}%
}%
%EndExpansion
\\%
%TCIMACRO{\FRAME{itbpF}{0.1773in}{0.1773in}{0in}{}{}{ut03.ps}%
%width 0.1773in;  height 0.1773in;  depth 0in;  original-width 3in;
%original-height 3in;  cropleft "0";  croptop "1";  cropright "1";
%cropbottom "0";  filename 'ut03.ps';file-properties "XNPEU";}}}%
%BeginExpansion
{\includegraphics[
%natheight=3.000000in,
%natwidth=3.000000in,
height=0.1773in,
width=0.1773in
]%
{ut03.ps}%
}%
%EndExpansion
&
%TCIMACRO{\FRAME{itbpF}{0.1773in}{0.1773in}{0in}{}{}{ut05.ps}%
%width 0.1773in;  height 0.1773in;  depth 0in;  original-width 3in;
%original-height 3in;  cropleft "0";  croptop "1";  cropright "1";
%cropbottom "0";  filename 'ut05.ps';file-properties "XNPEU";}}}%
%BeginExpansion
{\includegraphics[
%natheight=3.000000in,
%natwidth=3.000000in,
height=0.1773in,
width=0.1773in
]%
{ut05.ps}%
}%
%EndExpansion
&
%TCIMACRO{\FRAME{itbpF}{0.1773in}{0.1773in}{0in}{}{}{ut04.ps}%
%width 0.1773in;  height 0.1773in;  depth 0in;  original-width 3in;
%original-height 3in;  cropleft "0";  croptop "1";  cropright "1";
%cropbottom "0";  filename 'ut04.ps';file-properties "XNPEU";}}}%
%BeginExpansion
{\includegraphics[
%natheight=3.000000in,
%natwidth=3.000000in,
height=0.1773in,
width=0.1773in
]%
{ut04.ps}%
}%
%EndExpansion
&
%TCIMACRO{\FRAME{itbpF}{0.1773in}{0.1773in}{0in}{}{}{ut00.ps}%
%width 0.1773in;  height 0.1773in;  depth 0in;  original-width 3in;
%original-height 3in;  cropleft "0";  croptop "1";  cropright "1";
%cropbottom "0";  filename 'ut00.ps';file-properties "XNPEU";}}}%
%BeginExpansion
{\includegraphics[
%natheight=3.000000in,
%natwidth=3.000000in,
height=0.1773in,
width=0.1773in
]%
{ut00.ps}%
}%
%EndExpansion
\end{array}
\right\rangle \text{ \ ,}%
\]
one could make repeated measurements of this state with respect to the
$11\cdot n^{2}$ (for $n=4$) tile observables
\[
\left\{  \Omega_{ij}^{(p)}=1^{\otimes\left(  nj+i-1\right)  }\otimes\left(
\left\vert T_{p}\right\rangle \left\langle T_{p}\right\vert \right)
\otimes1^{\otimes\left(  n^{2}-nj-i\right)  }\ :\ 0\leq p<11,\ 0\leq
i,j<n\right\}
\]
to determine the state with a probability of error less than a chosen positive
threshold $\epsilon$. \ This is obviously not the most efficient set of
observables for the task, nor is it a universal set of observables for quantum
knot tomography . \ Given many copies of an arbitrary quantum knot, how does
find a universal set observables that is best in the sense of greatest
efficiency for a given threshold $\epsilon$? \ \bigskip

\item[\textbf{12)}] Quantum Braids: One can also use mosaics to define quantum
braids. \ How is this related to the work found in \cite{Jacak1},
\cite{Kitaev1}, \cite{Sarma1}?\bigskip

\item[\textbf{13)}] Can quantum knot systems be used to model and to predict
the behavior of

\begin{itemize}
\item[\textbf{i)}] Quantum vortices in liquid Helium II? (See \cite{Rasetti1}.)

\item[\textbf{ii)}] Quantum vortices in the Bose-Einstein condensate?

\item[\textbf{iii)}] Fractional charge quantification that is manifest in the
fractional quantum Hall effect? (See \cite{Jacak1} and \cite{Wen1}.)
\end{itemize}
\end{itemize}

\bigskip

We should mention that we have throughout this paper used a version of knot
diagrams (i.e., two dimensional knot mosaics), that is susceptible to
quantization. \ Within this context, the familiar Reidemeister moves have
become unitary transformations on an appropriate Hilbert space. \ Knots as we
know them in topology and geometry occur as embeddings in three dimensional
space, and are projected to knot diagrams for combinatorial and topological
purposes. \ Thus, we have chosen to model quantum knots through the extra
structure of knot diagrams. \ However, it is also possible to investigate
quantum knots in full three dimensional space by using three dimensional knot
mosaics. \ This will be the subject of a forthcoming paper.

\bigskip

In closing this section, we should finally also say that, in the open
literature, the phrase "quantum knot" has many different meanings, and is
sometimes a phrase that is used loosely. \ We mention only two examples. \ In
\cite{Kauffman1}, a quantum knot is essentially defined as an element of the
Hilbert space with orthonormal basis in one-one-correspondence with knot
types, rather than knot representatives. \ Within the context of the mosaic
construction, a quantum knot in \cite{Kauffman1} corresponds to an element of
the orbit Hilbert space $\mathcal{K}^{(n)}/\mathbb{A}(n)$. \ \ In
\cite{Collins1} and in \cite{Sarma1} the phrase "quantum knot" refers not to
knots, but to the use of representations of the braid group to model the
dynamic behavior of certain quantum systems. \ In this contexts, braids are
used as a tool to model topological obstructions to quantum decoherence that
are conjectured to exist within certain quantum systems. \ 

\bigskip

\bigskip

\begin{acknowledgement}
The first author would like to thank the Institute for Scientific Interchange
(ISI) in Torino, Italy, and the Mathematical Sciences Research Institute in
Berkeley, California for providing a research climate where many of the ideas
found within this paper could germinate. \ This effort was partially supported
by the Defense Advanced Research Projects Agency (DARPA) and Air Force
Research Laboratory, Air Force Materiel Command, USAF, under agreement number
F30602-01-2-0522, and by L-OO-P Fund Grant BECA2002. The second author would
like to thank the National Science Foundation for support under NSF grant DMS-0245588.
\end{acknowledgement}

\bigskip

\section{Appendix A: A list of all knot 3-mosaics}

\bigskip

A complete list of all knot 3-mosaics, listed in lexicographic (lex) order, is
given below:

\bigskip

$\hspace{-0.7in}\underset{K_{0}=\text{ 000-000-000}}{\fbox{$%

$}}$

\bigskip

\section{Appendix B: Oriented mosaics and oriented quantum knots}

\bigskip

So far, we have only discussed unoriented objects in this paper, e.g.,
unoriented mosaics, unoriented knot mosaics, unoriented quantum knots, and so
forth. \ In this appendix, we briefly discuss how all these unoriented objects
can be transformed into oriented ones.

\bigskip

Let $\mathbb{T}^{(o)}$ denote the set of the following 29 symbols \bigskip

\hspace{-0.7in}%
%TCIMACRO{\FRAME{itbpF}{0.3269in}{0.3269in}{0in}{}{}{ot00.ps}%
%width 0.3269in;  height 0.3269in;  depth 0in;  original-width 3in;
%original-height 3in;  cropleft "0";  croptop "1";  cropright "1";
%cropbottom "0";  filename 'ot00.ps';file-properties "XNPEU";}}}%
%BeginExpansion
{\includegraphics[
%natheight=3.000000in,
%natwidth=3.000000in,
height=0.3269in,
width=0.3269in
]%
{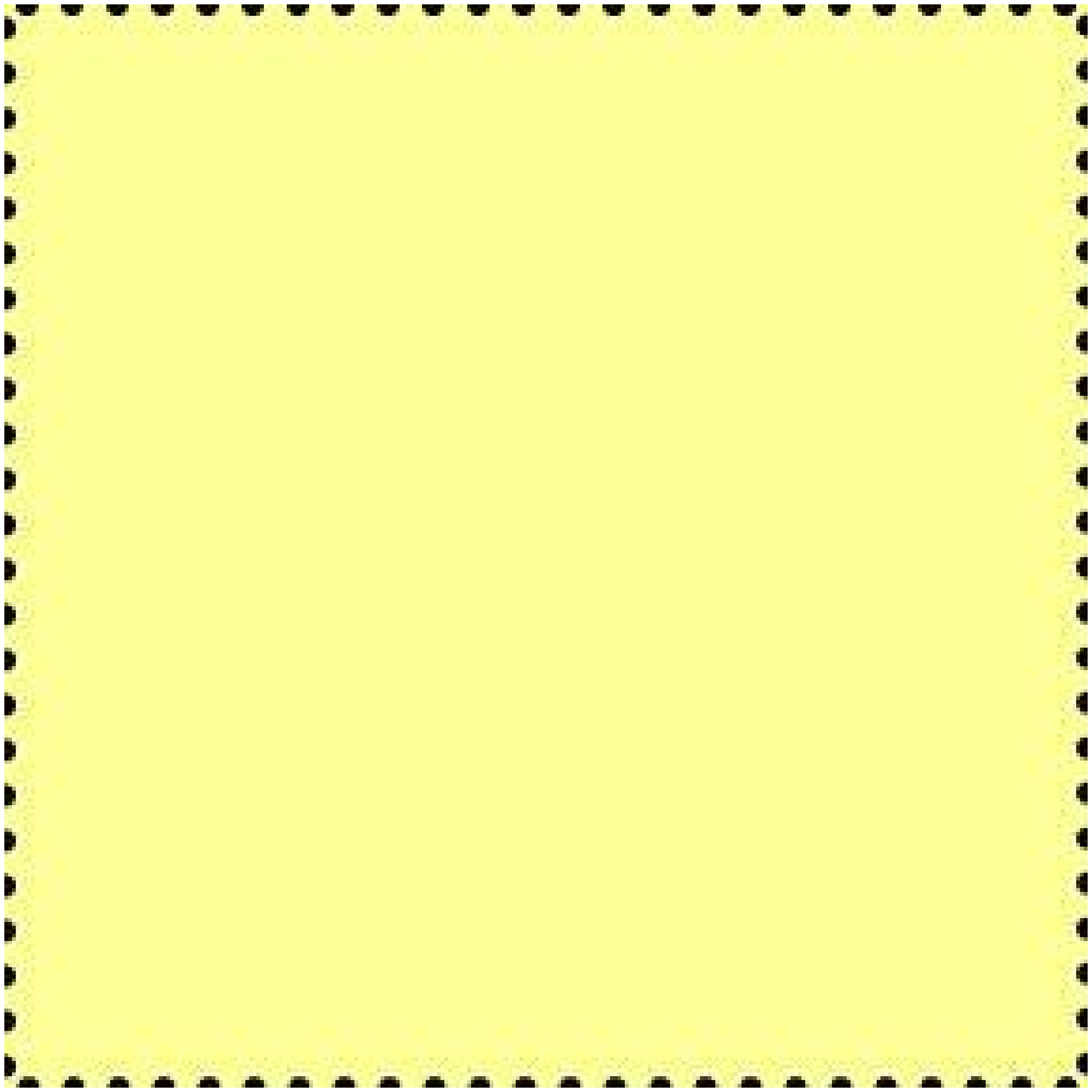}%
}%
%EndExpansion
\qquad%
%TCIMACRO{\FRAME{itbpF}{0.3269in}{0.3269in}{0in}{}{}{ot01.ps}%
%width 0.3269in;  height 0.3269in;  depth 0in;  original-width 3in;
%original-height 3in;  cropleft "0";  croptop "1";  cropright "1";
%cropbottom "0";  filename 'ot01.ps';file-properties "XNPEU";}}}%
%BeginExpansion
{\includegraphics[
%natheight=3.000000in,
%natwidth=3.000000in,
height=0.3269in,
width=0.3269in
]%
{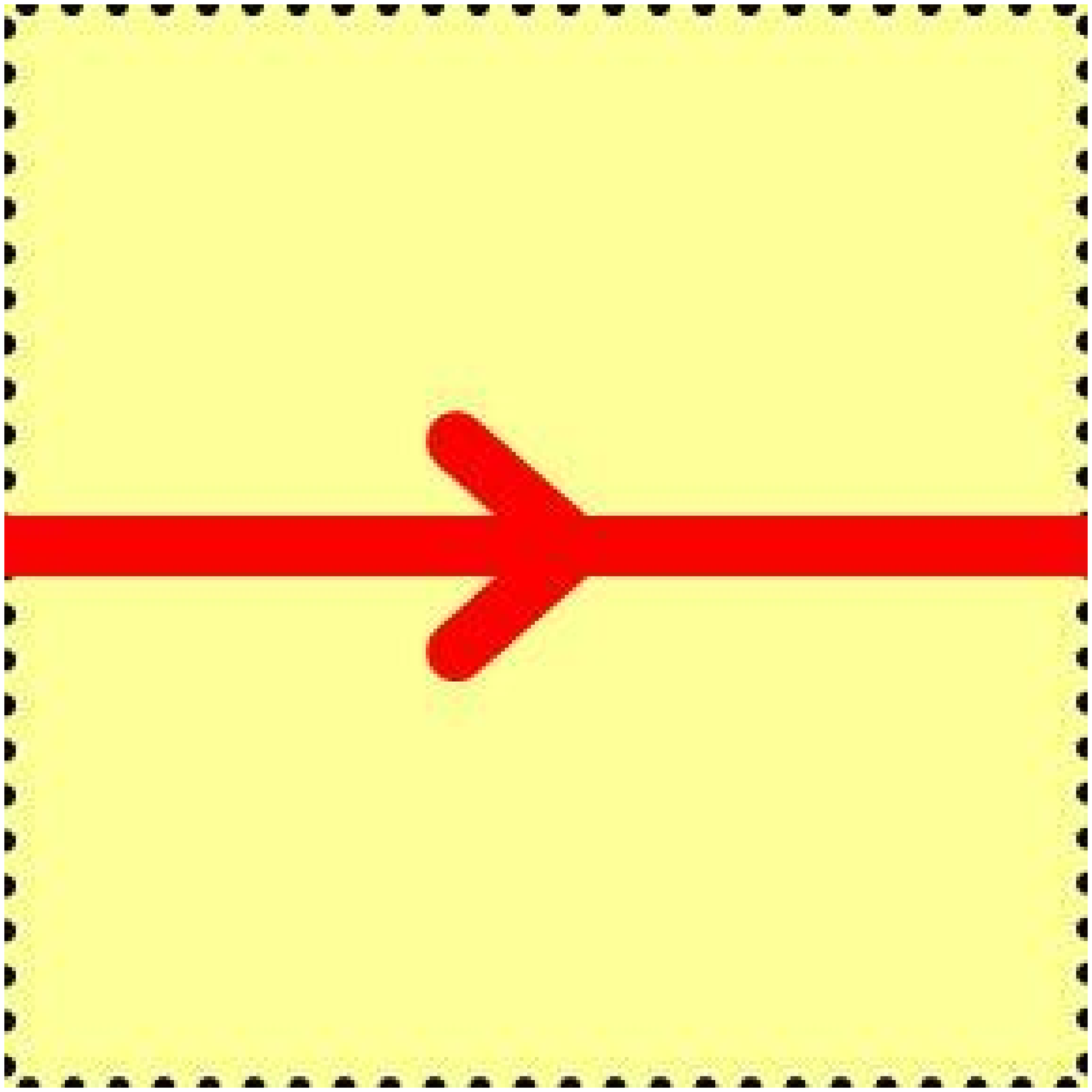}%
}%
%EndExpansion%
%TCIMACRO{\FRAME{itbpF}{0.3269in}{0.3269in}{0in}{}{}{ot02.ps}%
%width 0.3269in;  height 0.3269in;  depth 0in;  original-width 3in;
%original-height 3in;  cropleft "0";  croptop "1";  cropright "1";
%cropbottom "0";  filename 'ot02.ps';file-properties "XNPEU";}}}%
%BeginExpansion
{\includegraphics[
%natheight=3.000000in,
%natwidth=3.000000in,
height=0.3269in,
width=0.3269in
]%
{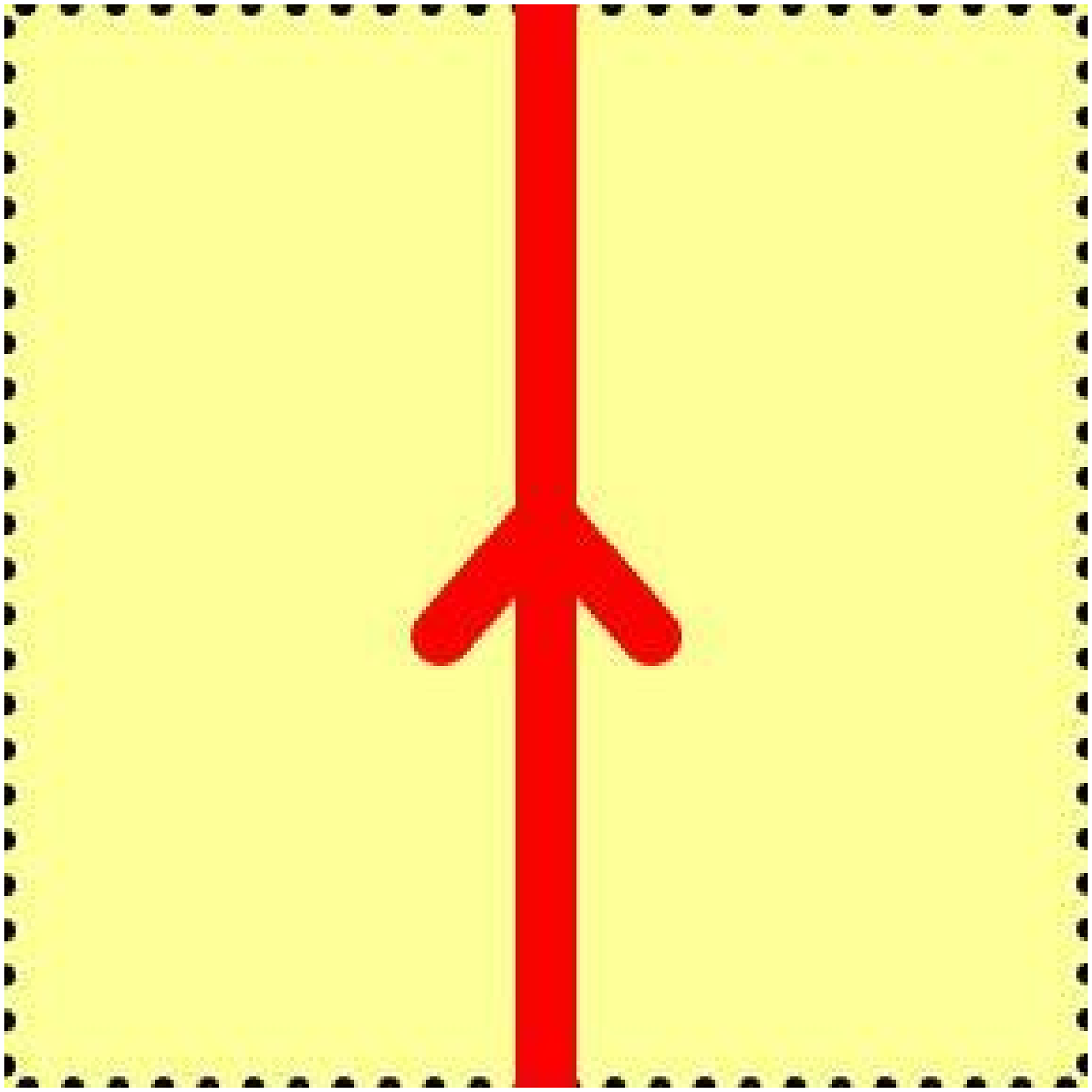}%
}%
%EndExpansion%
%TCIMACRO{\FRAME{itbpF}{0.3269in}{0.3269in}{0in}{}{}{ot03.ps}%
%width 0.3269in;  height 0.3269in;  depth 0in;  original-width 3in;
%original-height 3in;  cropleft "0";  croptop "1";  cropright "1";
%cropbottom "0";  filename 'ot03.ps';file-properties "XNPEU";}}}%
%BeginExpansion
{\includegraphics[
%natheight=3.000000in,
%natwidth=3.000000in,
height=0.3269in,
width=0.3269in
]%
{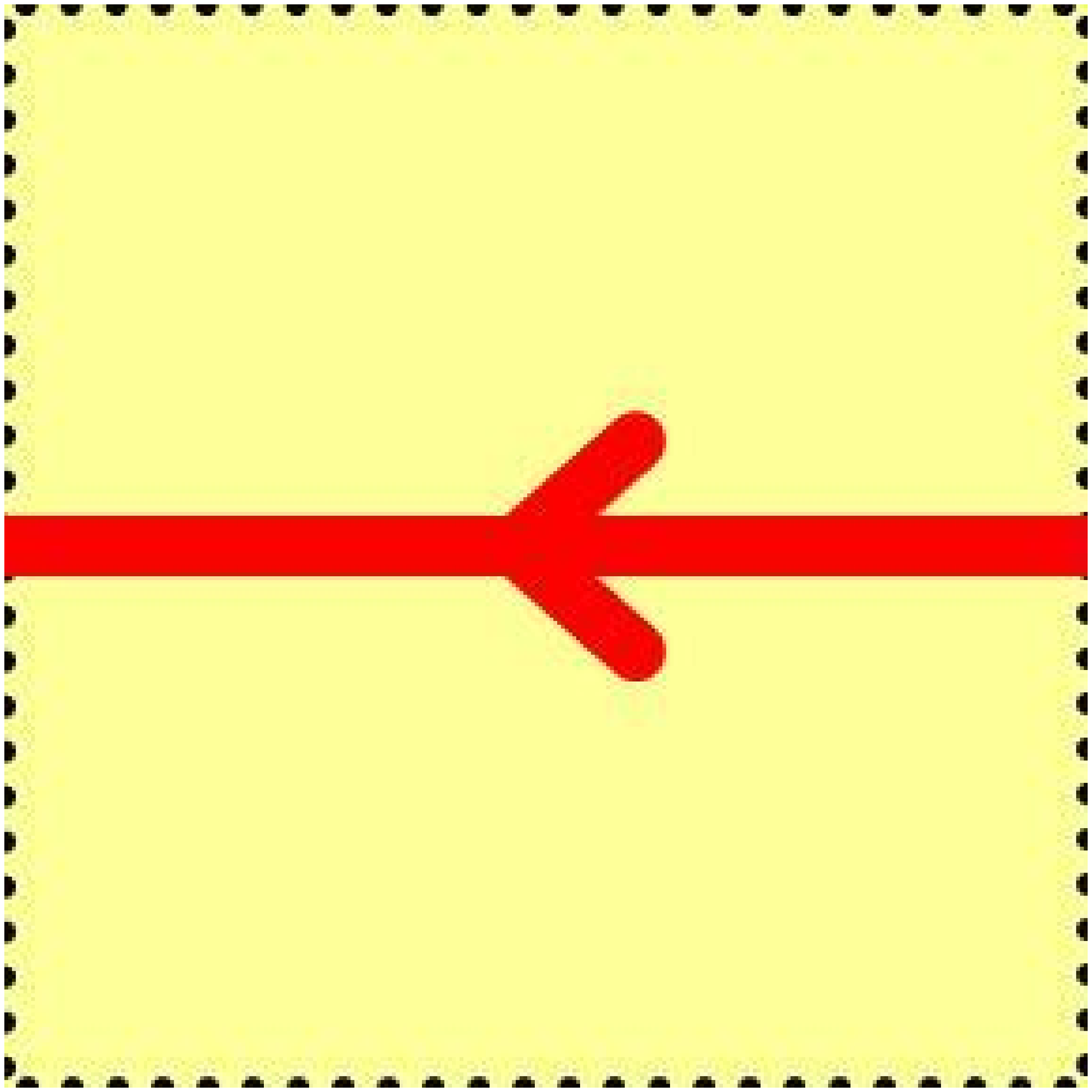}%
}%
%EndExpansion%
%TCIMACRO{\FRAME{itbpF}{0.3269in}{0.3269in}{0in}{}{}{ot04.ps}%
%width 0.3269in;  height 0.3269in;  depth 0in;  original-width 3in;
%original-height 3in;  cropleft "0";  croptop "1";  cropright "1";
%cropbottom "0";  filename 'ot04.ps';file-properties "XNPEU";}}}%
%BeginExpansion
{\includegraphics[
%natheight=3.000000in,
%natwidth=3.000000in,
height=0.3269in,
width=0.3269in
]%
{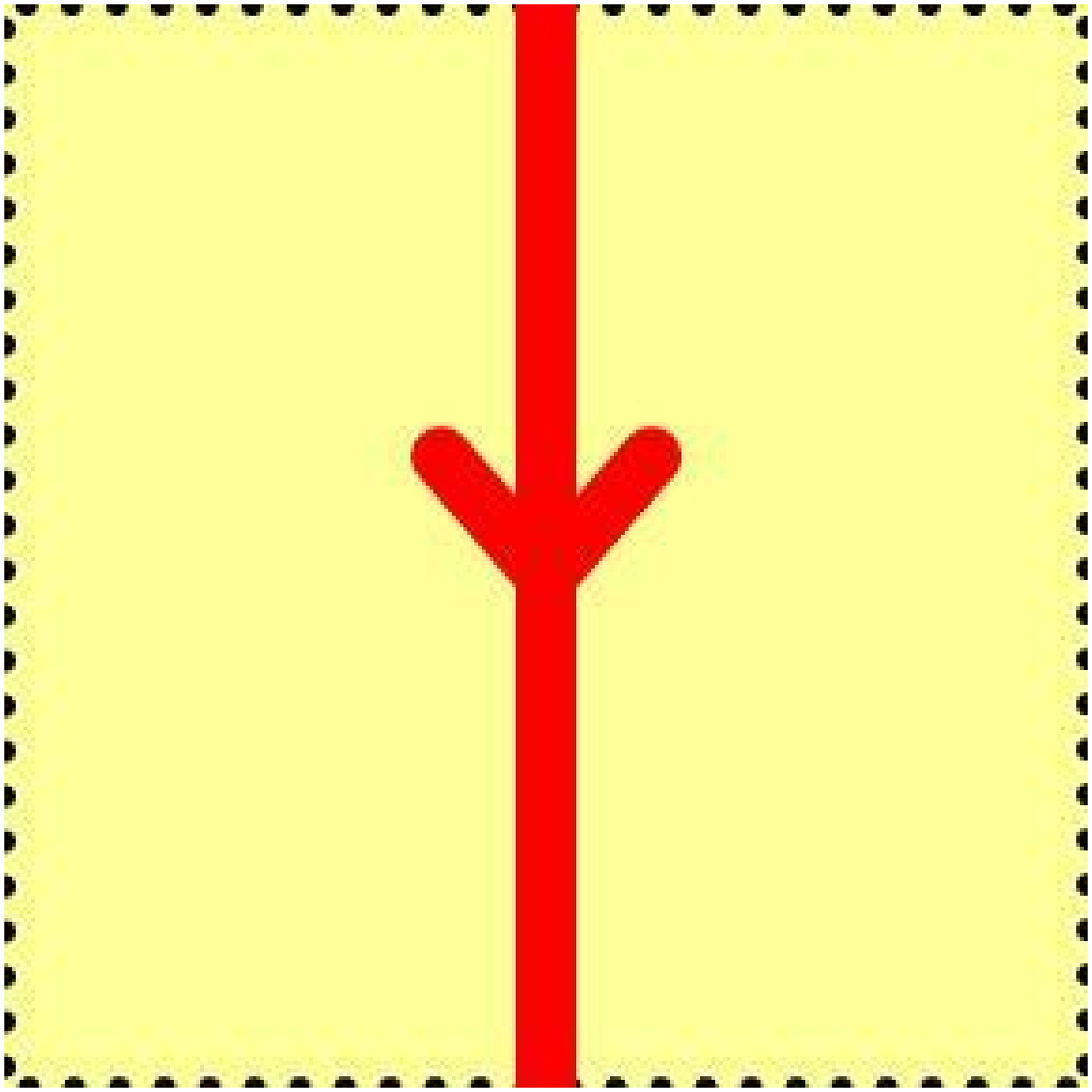}%
}%
%EndExpansion
\qquad%
%TCIMACRO{\FRAME{itbpF}{0.3269in}{0.3269in}{0in}{}{}{ot05.ps}%
%width 0.3269in;  height 0.3269in;  depth 0in;  original-width 3in;
%original-height 3in;  cropleft "0";  croptop "1";  cropright "1";
%cropbottom "0";  filename 'ot05.ps';file-properties "XNPEU";}}}%
%BeginExpansion
{\includegraphics[
%natheight=3.000000in,
%natwidth=3.000000in,
height=0.3269in,
width=0.3269in
]%
{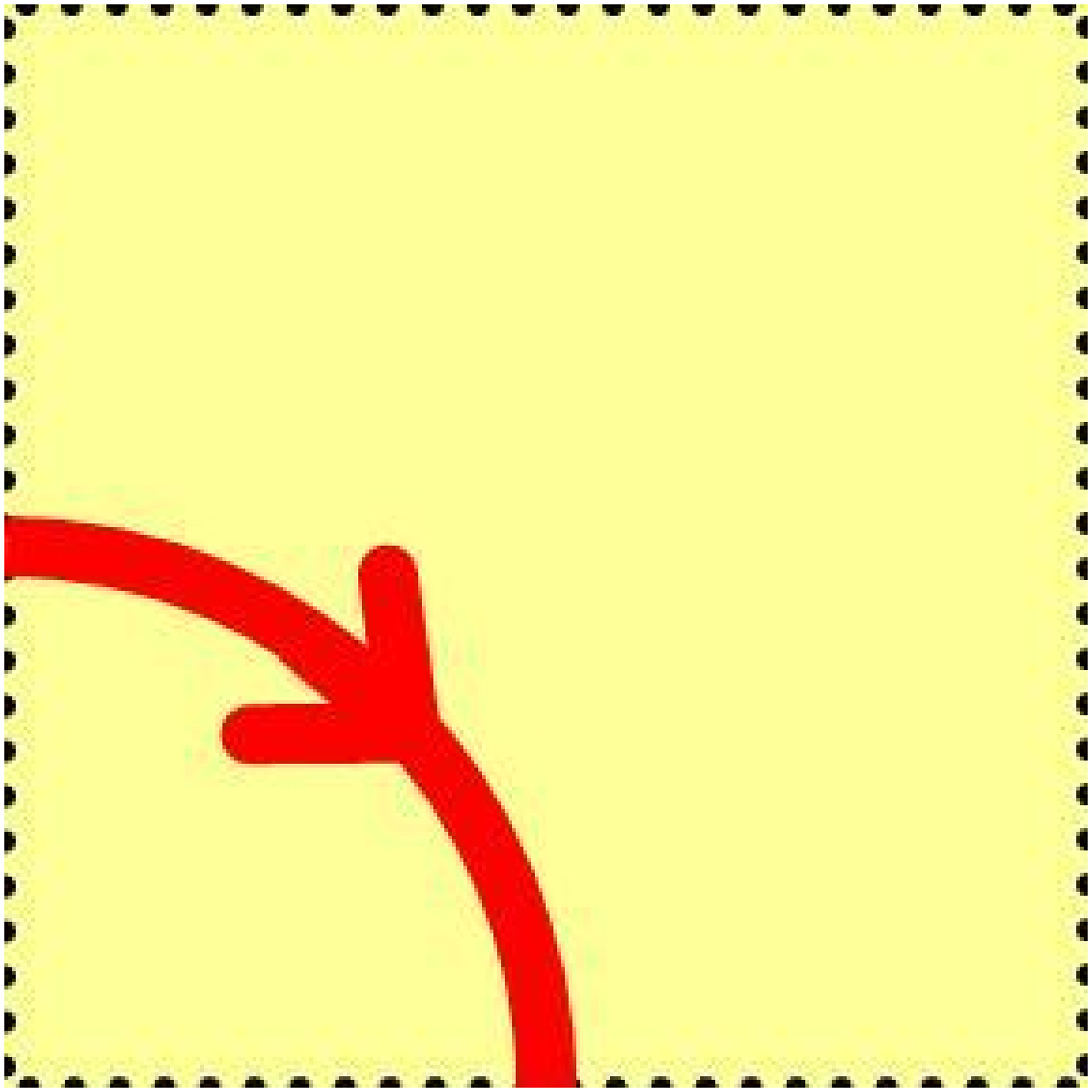}%
}%
%EndExpansion%
%TCIMACRO{\FRAME{itbpF}{0.3269in}{0.3269in}{0in}{}{}{ot06.ps}%
%width 0.3269in;  height 0.3269in;  depth 0in;  original-width 3in;
%original-height 3in;  cropleft "0";  croptop "1";  cropright "1";
%cropbottom "0";  filename 'ot06.ps';file-properties "XNPEU";}}}%
%BeginExpansion
{\includegraphics[
%natheight=3.000000in,
%natwidth=3.000000in,
height=0.3269in,
width=0.3269in
]%
{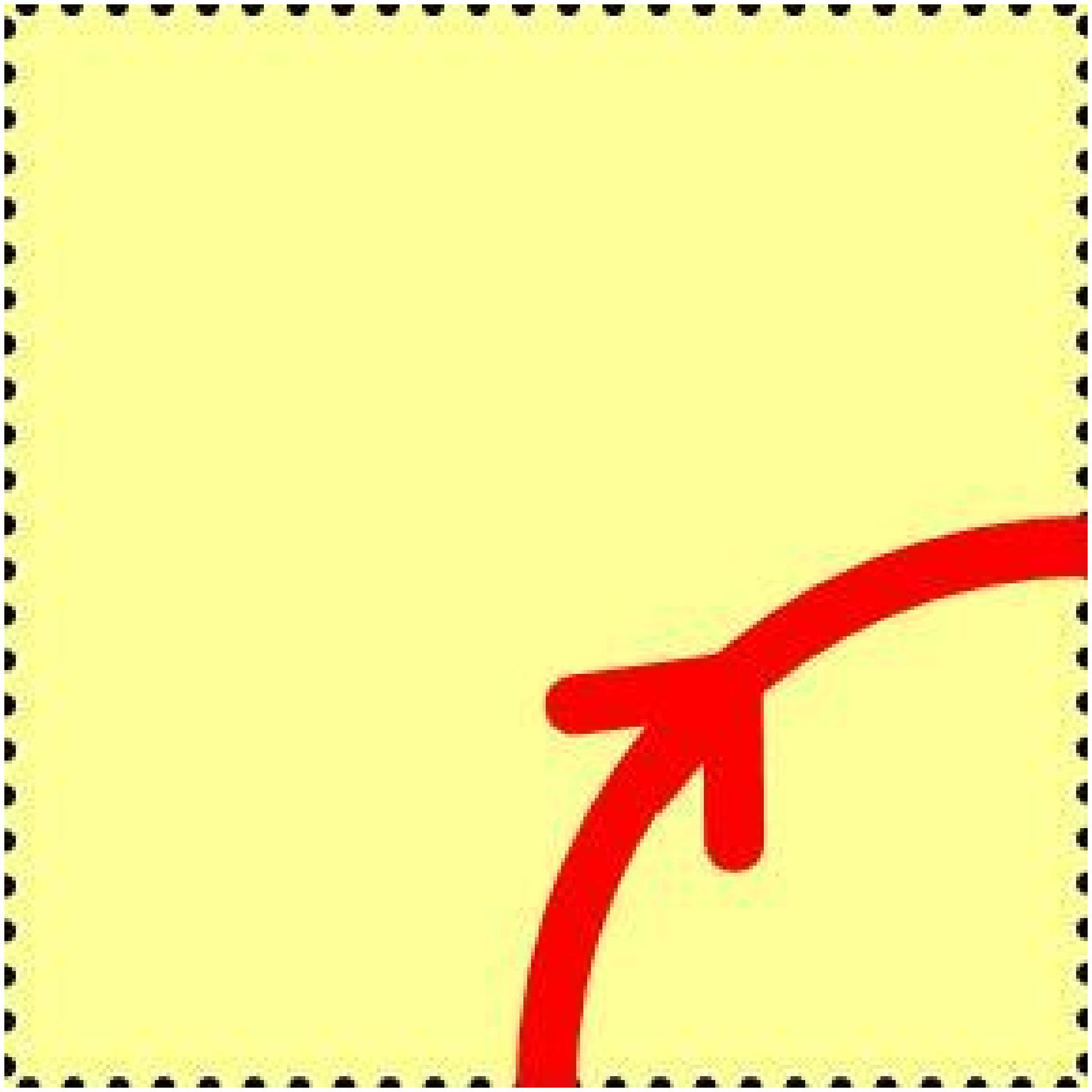}%
}%
%EndExpansion%
%TCIMACRO{\FRAME{itbpF}{0.3269in}{0.3269in}{0in}{}{}{ot07.ps}%
%width 0.3269in;  height 0.3269in;  depth 0in;  original-width 3in;
%original-height 3in;  cropleft "0";  croptop "1";  cropright "1";
%cropbottom "0";  filename 'ot07.ps';file-properties "XNPEU";}}}%
%BeginExpansion
{\includegraphics[
%natheight=3.000000in,
%natwidth=3.000000in,
height=0.3269in,
width=0.3269in
]%
{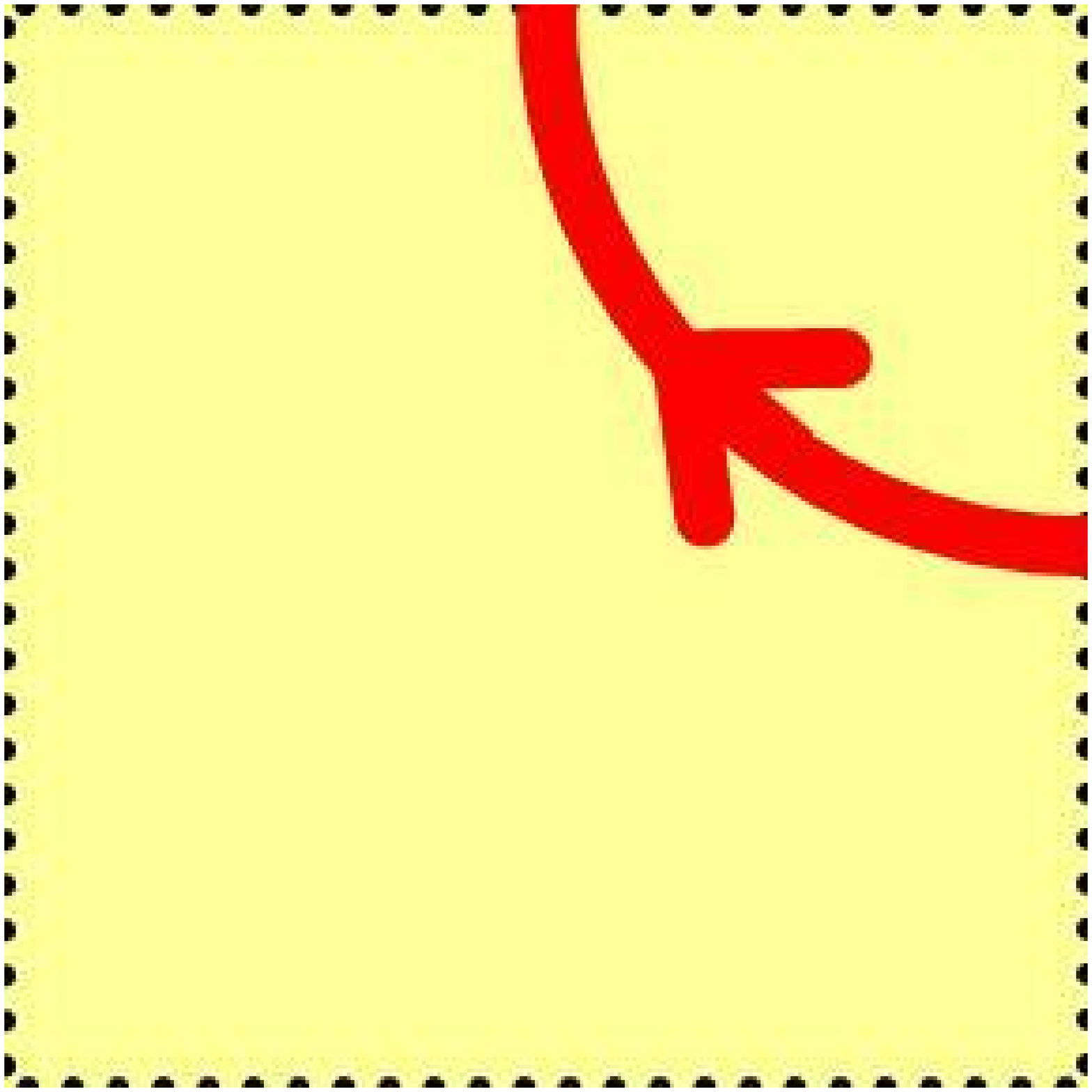}%
}%
%EndExpansion%
%TCIMACRO{\FRAME{itbpF}{0.3269in}{0.3269in}{0in}{}{}{ot08.ps}%
%width 0.3269in;  height 0.3269in;  depth 0in;  original-width 3in;
%original-height 3in;  cropleft "0";  croptop "1";  cropright "1";
%cropbottom "0";  filename 'ot08.ps';file-properties "XNPEU";}}}%
%BeginExpansion
{\includegraphics[
%natheight=3.000000in,
%natwidth=3.000000in,
height=0.3269in,
width=0.3269in
]%
{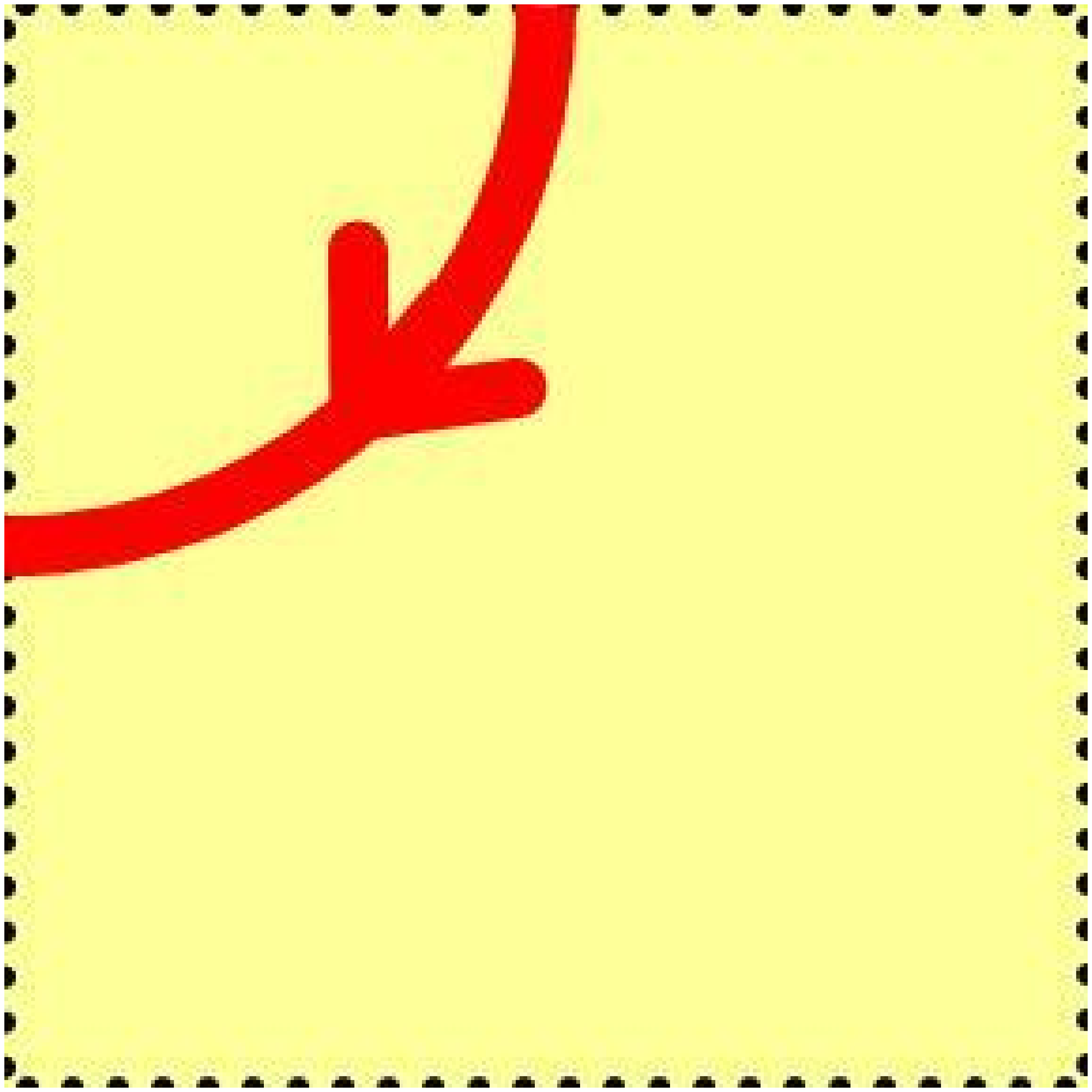}%
}%
%EndExpansion
\qquad%
%TCIMACRO{\FRAME{itbpF}{0.3269in}{0.3269in}{0in}{}{}{ot09.ps}%
%width 0.3269in;  height 0.3269in;  depth 0in;  original-width 3in;
%original-height 3in;  cropleft "0";  croptop "1";  cropright "1";
%cropbottom "0";  filename 'ot09.ps';file-properties "XNPEU";}}}%
%BeginExpansion
{\includegraphics[
%natheight=3.000000in,
%natwidth=3.000000in,
height=0.3269in,
width=0.3269in
]%
{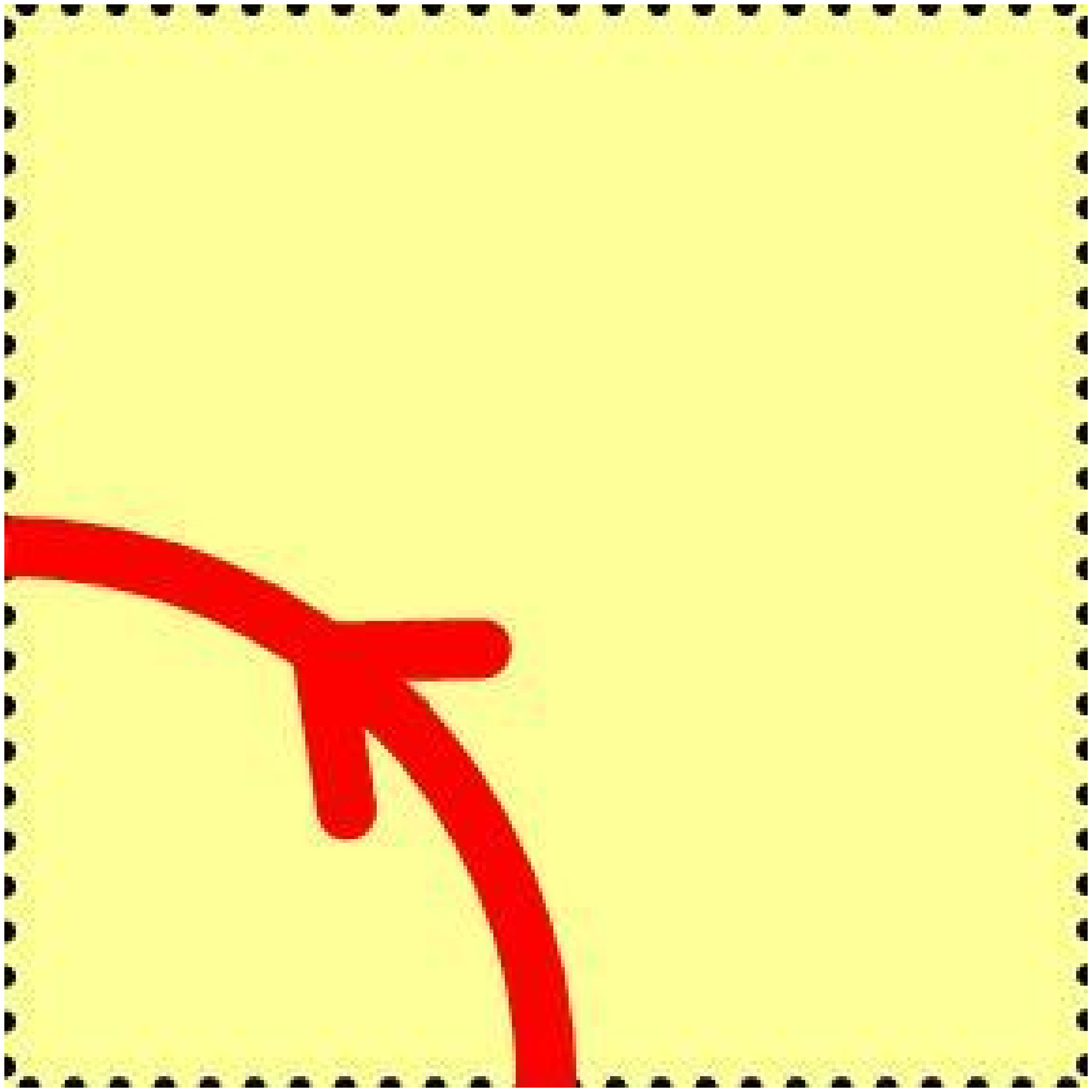}%
}%
%EndExpansion%
%TCIMACRO{\FRAME{itbpF}{0.3269in}{0.3269in}{0in}{}{}{ot10.ps}%
%width 0.3269in;  height 0.3269in;  depth 0in;  original-width 3in;
%original-height 3in;  cropleft "0";  croptop "1";  cropright "1";
%cropbottom "0";  filename 'ot10.ps';file-properties "XNPEU";}}}%
%BeginExpansion
{\includegraphics[
%natheight=3.000000in,
%natwidth=3.000000in,
height=0.3269in,
width=0.3269in
]%
{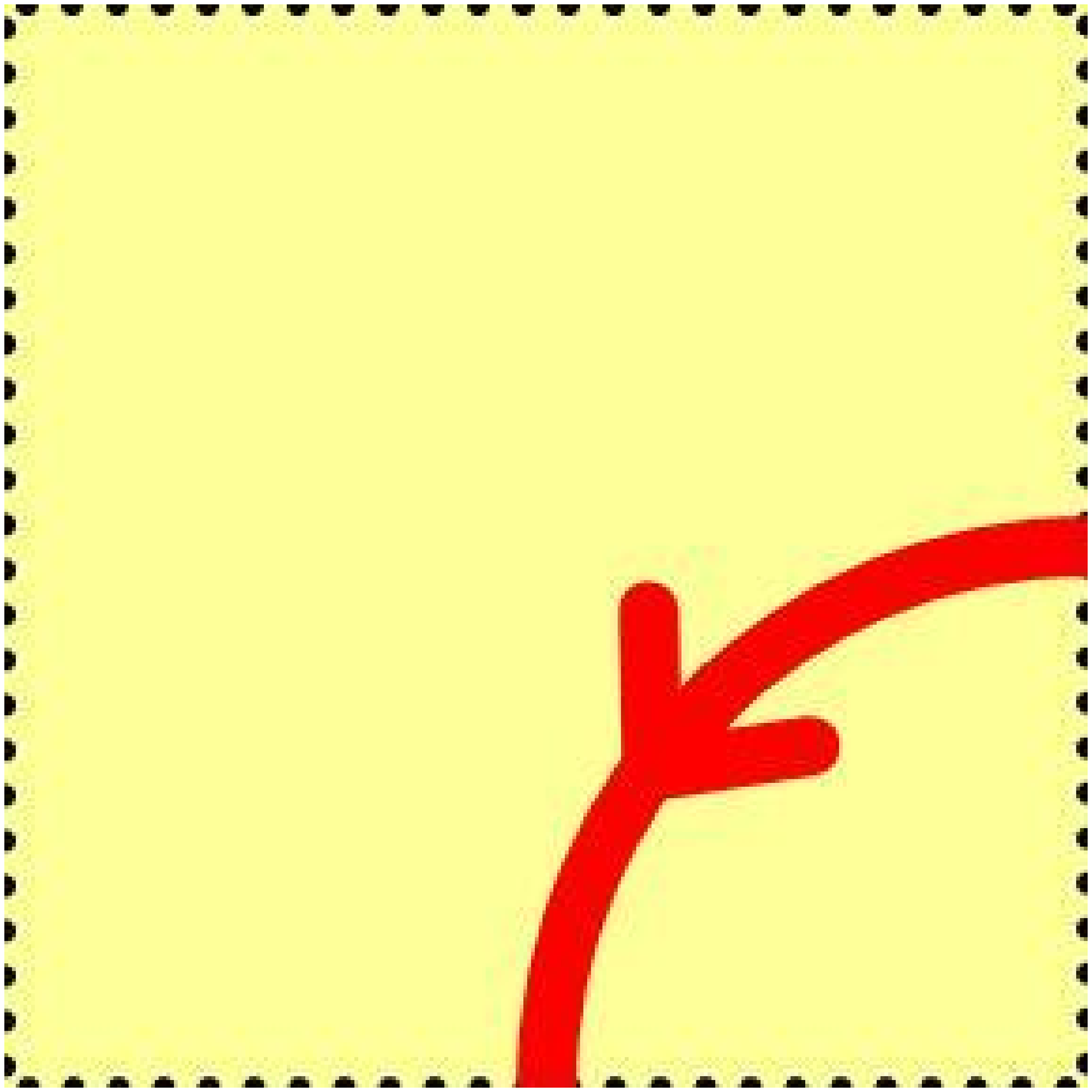}%
}%
%EndExpansion%
%TCIMACRO{\FRAME{itbpF}{0.3269in}{0.3269in}{0in}{}{}{ot11.ps}%
%width 0.3269in;  height 0.3269in;  depth 0in;  original-width 3in;
%original-height 3in;  cropleft "0";  croptop "1";  cropright "1";
%cropbottom "0";  filename 'ot11.ps';file-properties "XNPEU";}}}%
%BeginExpansion
{\includegraphics[
%natheight=3.000000in,
%natwidth=3.000000in,
height=0.3269in,
width=0.3269in
]%
{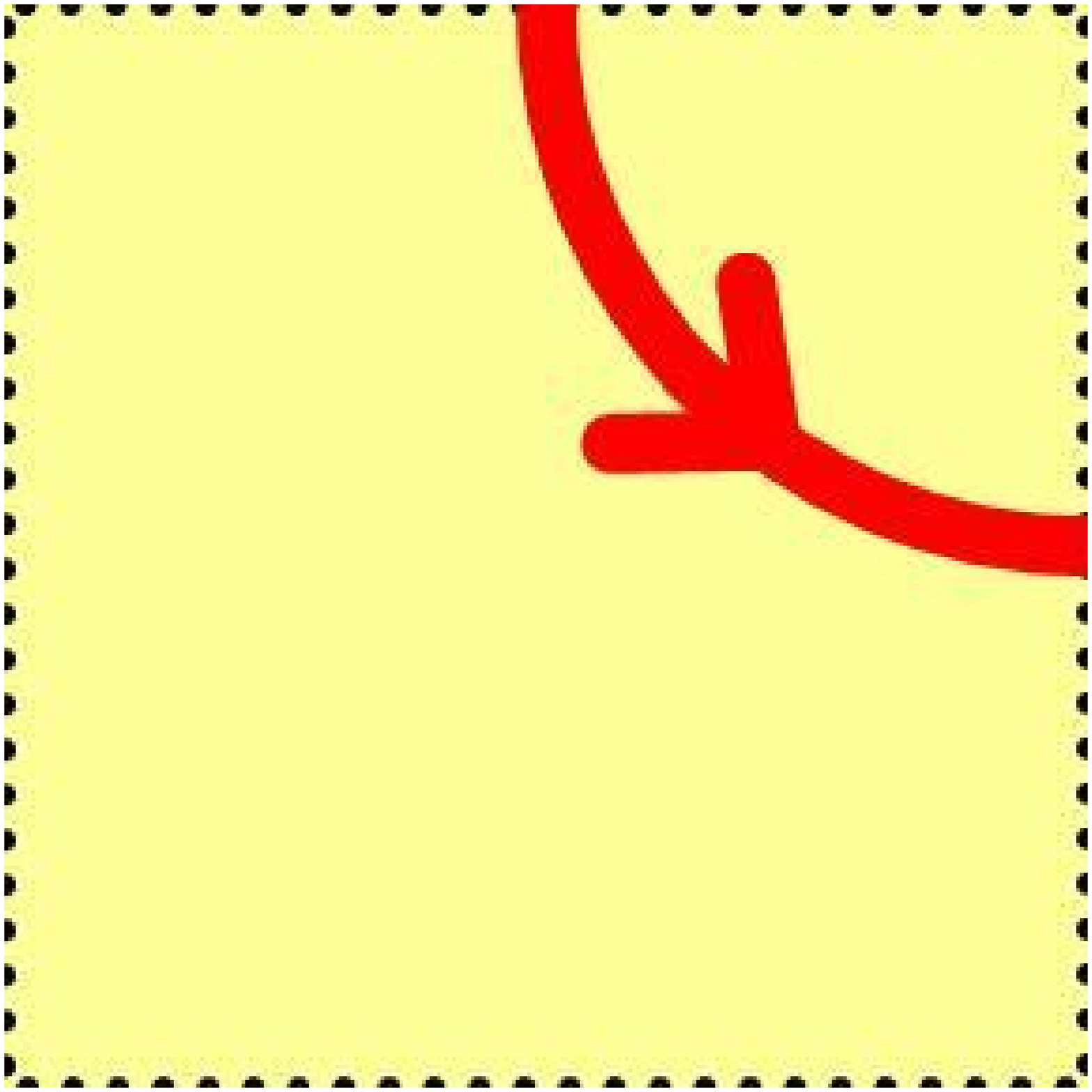}%
}%
%EndExpansion%
%TCIMACRO{\FRAME{itbpF}{0.3269in}{0.3269in}{0in}{}{}{ot12.ps}%
%width 0.3269in;  height 0.3269in;  depth 0in;  original-width 3in;
%original-height 3in;  cropleft "0";  croptop "1";  cropright "1";
%cropbottom "0";  filename 'ot12.ps';file-properties "XNPEU";}}}%
%BeginExpansion
{\includegraphics[
%natheight=3.000000in,
%natwidth=3.000000in,
height=0.3269in,
width=0.3269in
]%
{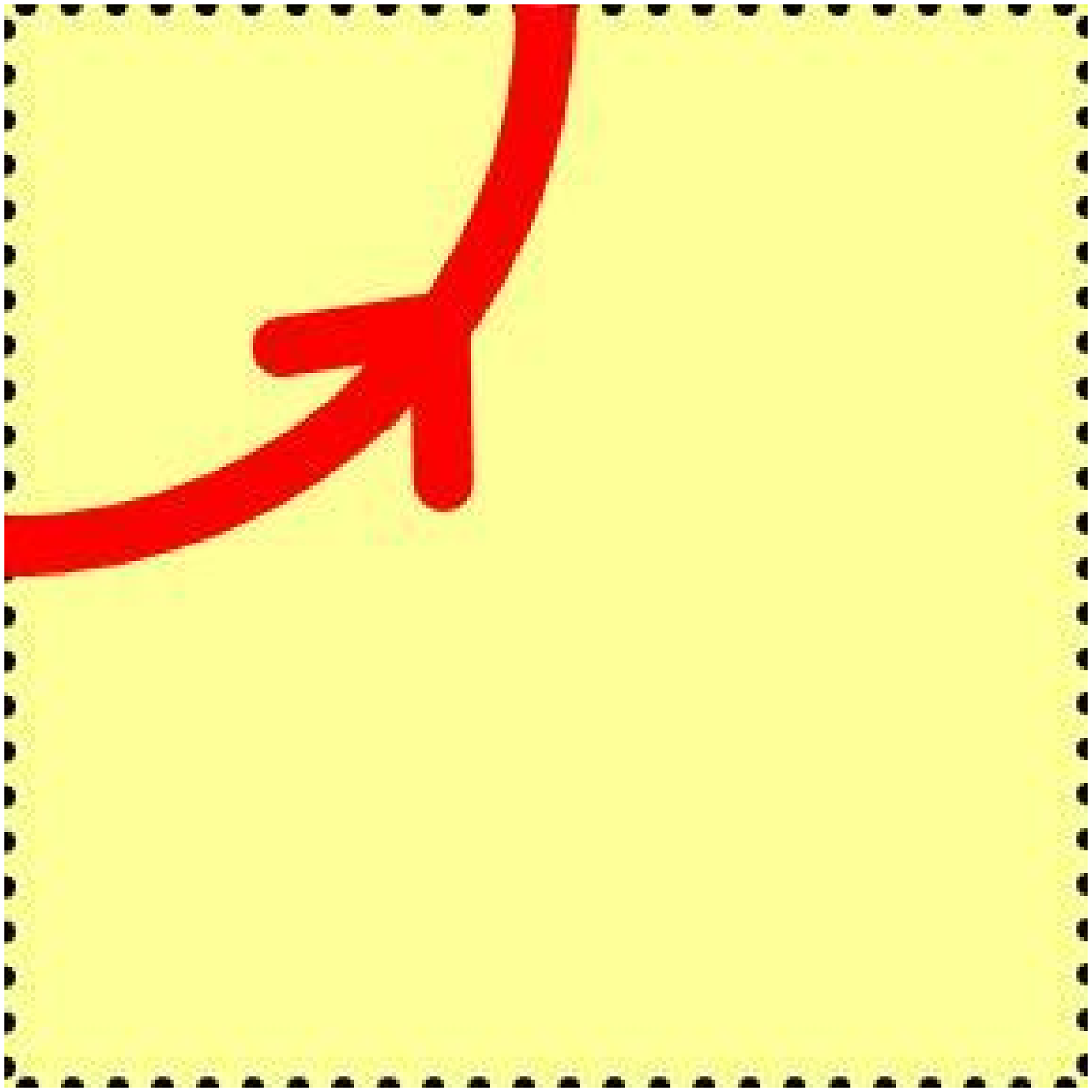}%
}%
%EndExpansion
\bigskip%

%TCIMACRO{\FRAME{itbpF}{0.3269in}{0.3269in}{0in}{}{}{ot13.ps}%
%width 0.3269in;  height 0.3269in;  depth 0in;  original-width 3in;
%original-height 3in;  cropleft "0";  croptop "1";  cropright "1";
%cropbottom "0";  filename 'ot13.ps';file-properties "XNPEU";}}}%
%BeginExpansion
{\includegraphics[
%natheight=3.000000in,
%natwidth=3.000000in,
height=0.3269in,
width=0.3269in
]%
{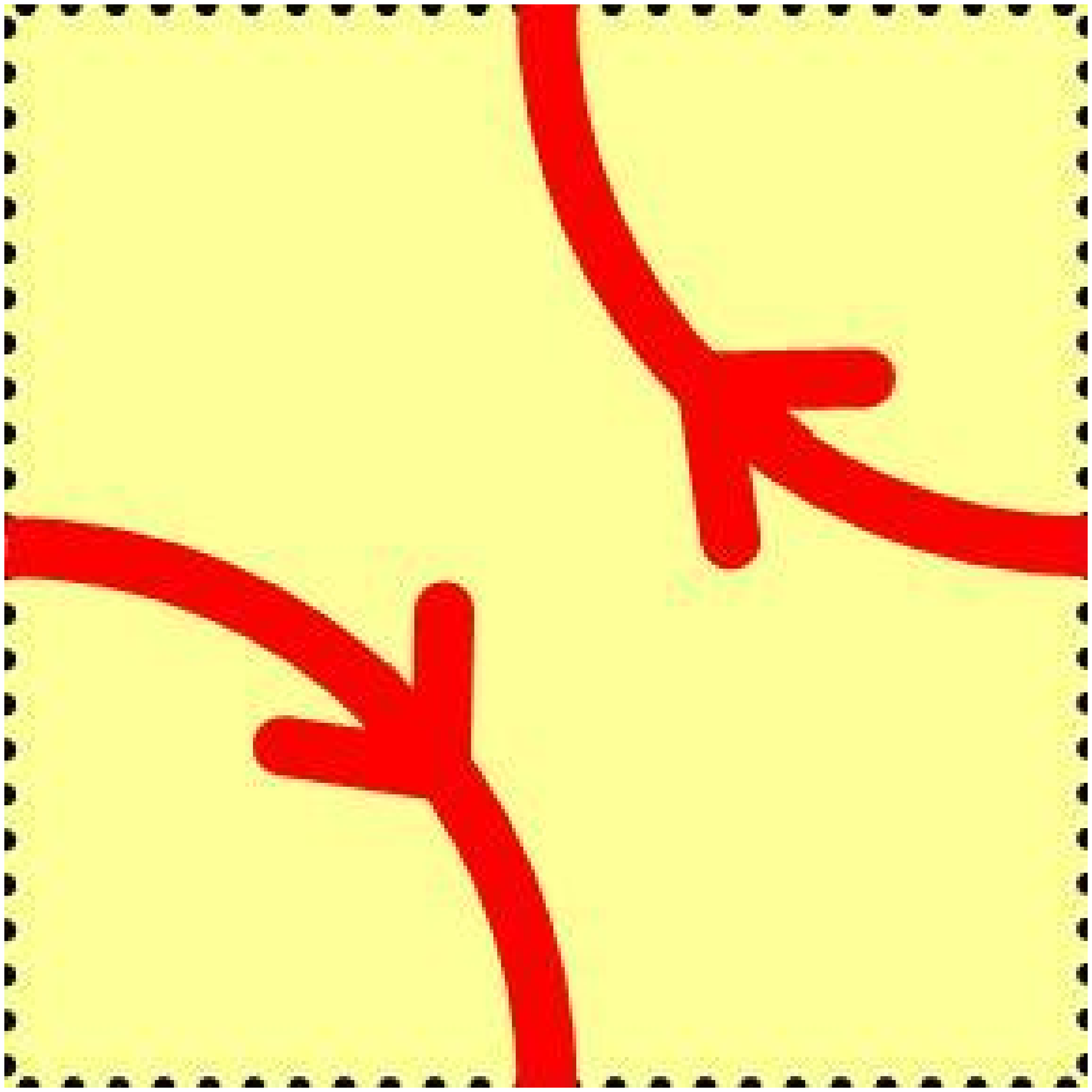}%
}%
%EndExpansion%
%TCIMACRO{\FRAME{itbpF}{0.3269in}{0.3269in}{0in}{}{}{ot14.ps}%
%width 0.3269in;  height 0.3269in;  depth 0in;  original-width 3in;
%original-height 3in;  cropleft "0";  croptop "1";  cropright "1";
%cropbottom "0";  filename 'ot14.ps';file-properties "XNPEU";}}}%
%BeginExpansion
{\includegraphics[
%natheight=3.000000in,
%natwidth=3.000000in,
height=0.3269in,
width=0.3269in
]%
{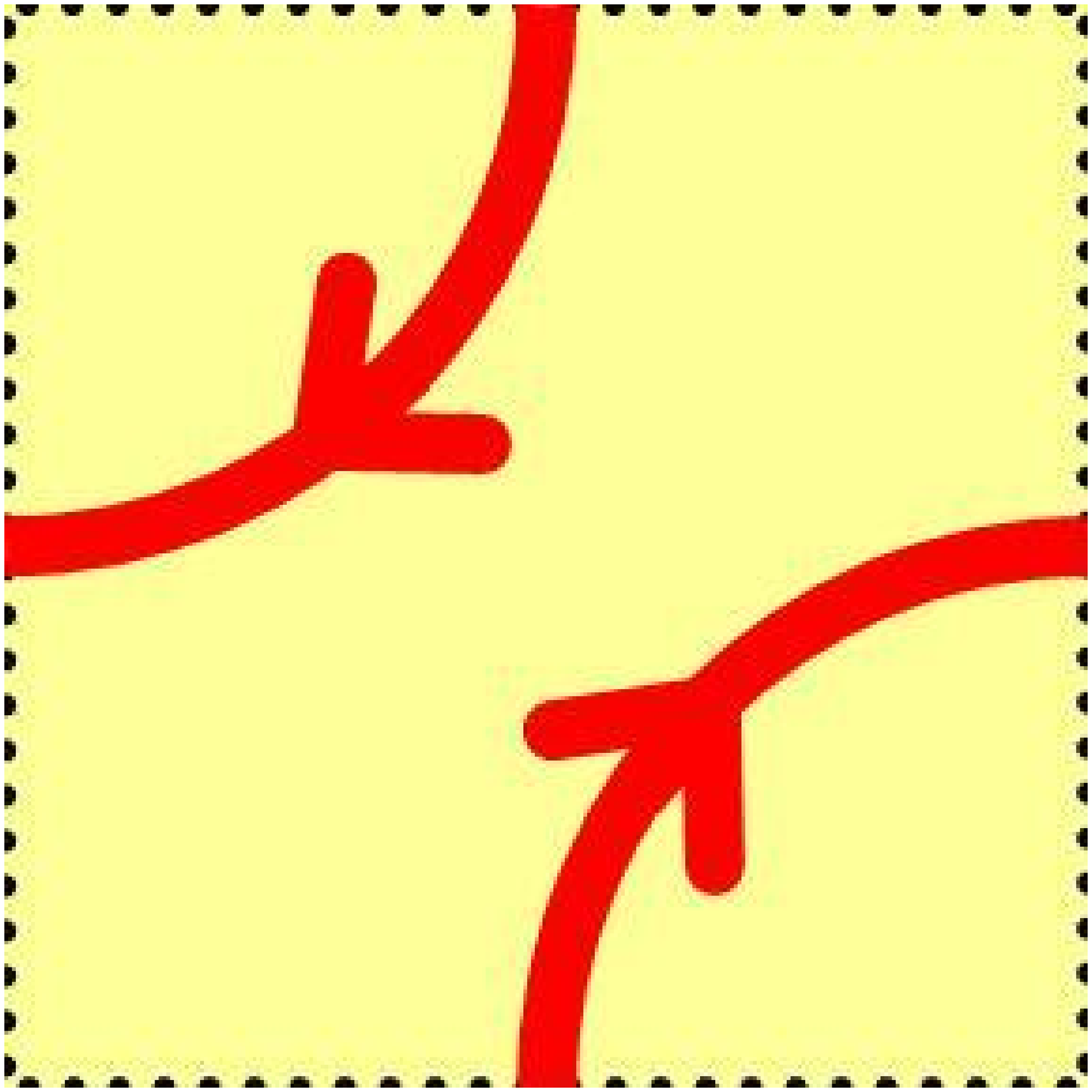}%
}%
%EndExpansion
\qquad%
%TCIMACRO{\FRAME{itbpF}{0.3269in}{0.3269in}{0in}{}{}{ot15.ps}%
%width 0.3269in;  height 0.3269in;  depth 0in;  original-width 3in;
%original-height 3in;  cropleft "0";  croptop "1";  cropright "1";
%cropbottom "0";  filename 'ot15.ps';file-properties "XNPEU";}}}%
%BeginExpansion
{\includegraphics[
%natheight=3.000000in,
%natwidth=3.000000in,
height=0.3269in,
width=0.3269in
]%
{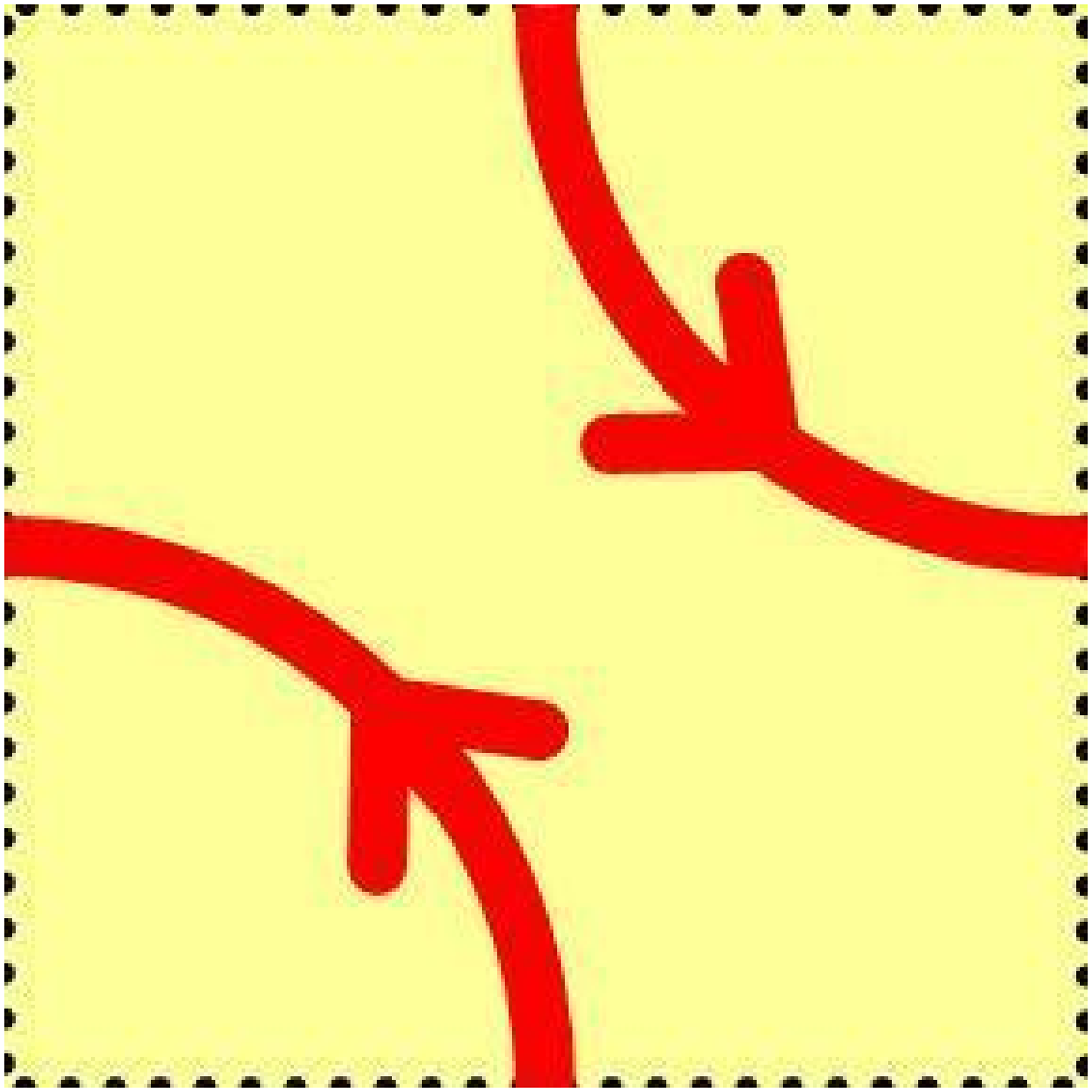}%
}%
%EndExpansion%
%TCIMACRO{\FRAME{itbpF}{0.3269in}{0.3269in}{0in}{}{}{ot16.ps}%
%width 0.3269in;  height 0.3269in;  depth 0in;  original-width 3in;
%original-height 3in;  cropleft "0";  croptop "1";  cropright "1";
%cropbottom "0";  filename 'ot16.ps';file-properties "XNPEU";}}}%
%BeginExpansion
{\includegraphics[
%natheight=3.000000in,
%natwidth=3.000000in,
height=0.3269in,
width=0.3269in
]%
{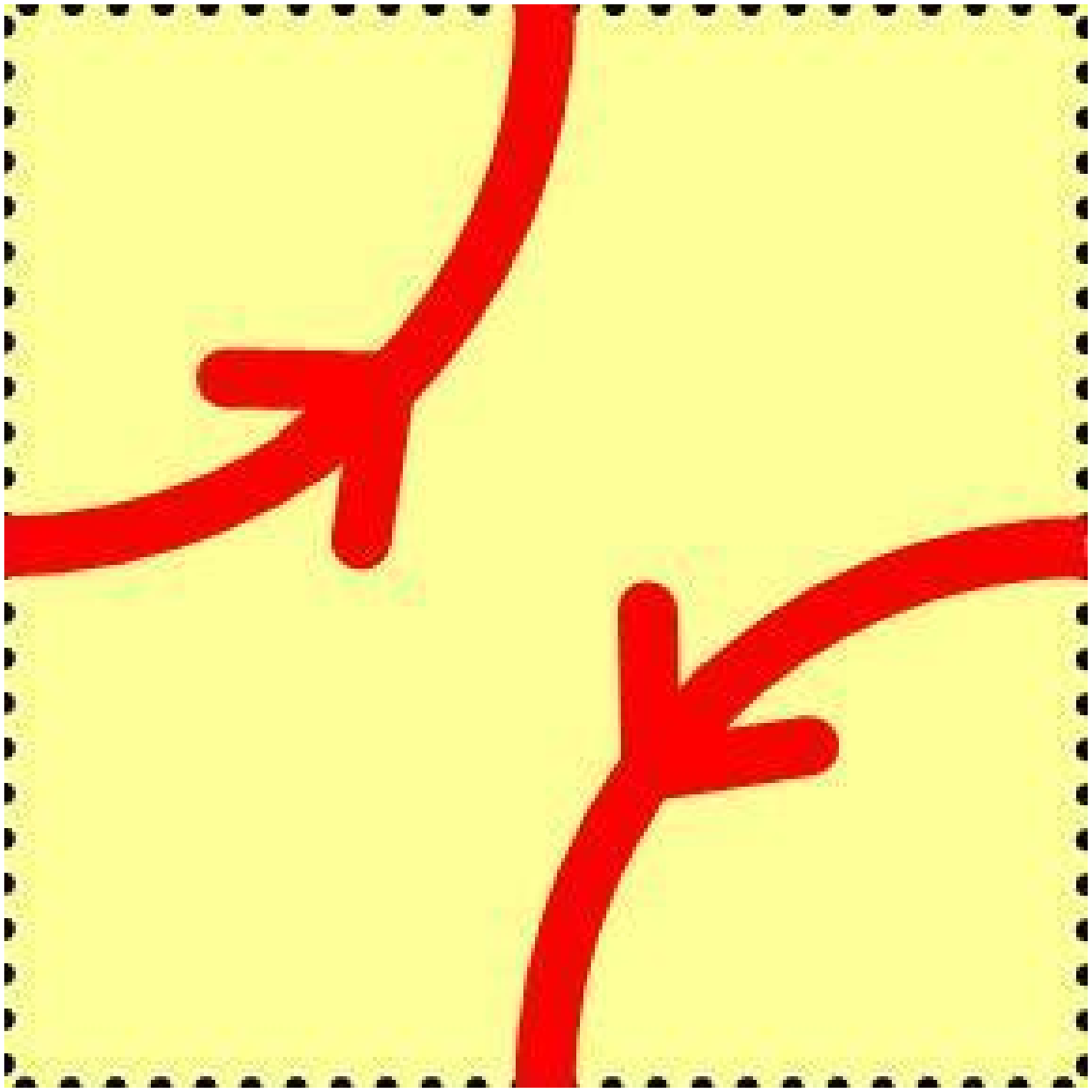}%
}%
%EndExpansion
\qquad%
%TCIMACRO{\FRAME{itbpF}{0.3269in}{0.3269in}{0in}{}{}{ot17.ps}%
%width 0.3269in;  height 0.3269in;  depth 0in;  original-width 3in;
%original-height 3in;  cropleft "0";  croptop "1";  cropright "1";
%cropbottom "0";  filename 'ot17.ps';file-properties "XNPEU";}}}%
%BeginExpansion
{\includegraphics[
%natheight=3.000000in,
%natwidth=3.000000in,
height=0.3269in,
width=0.3269in
]%
{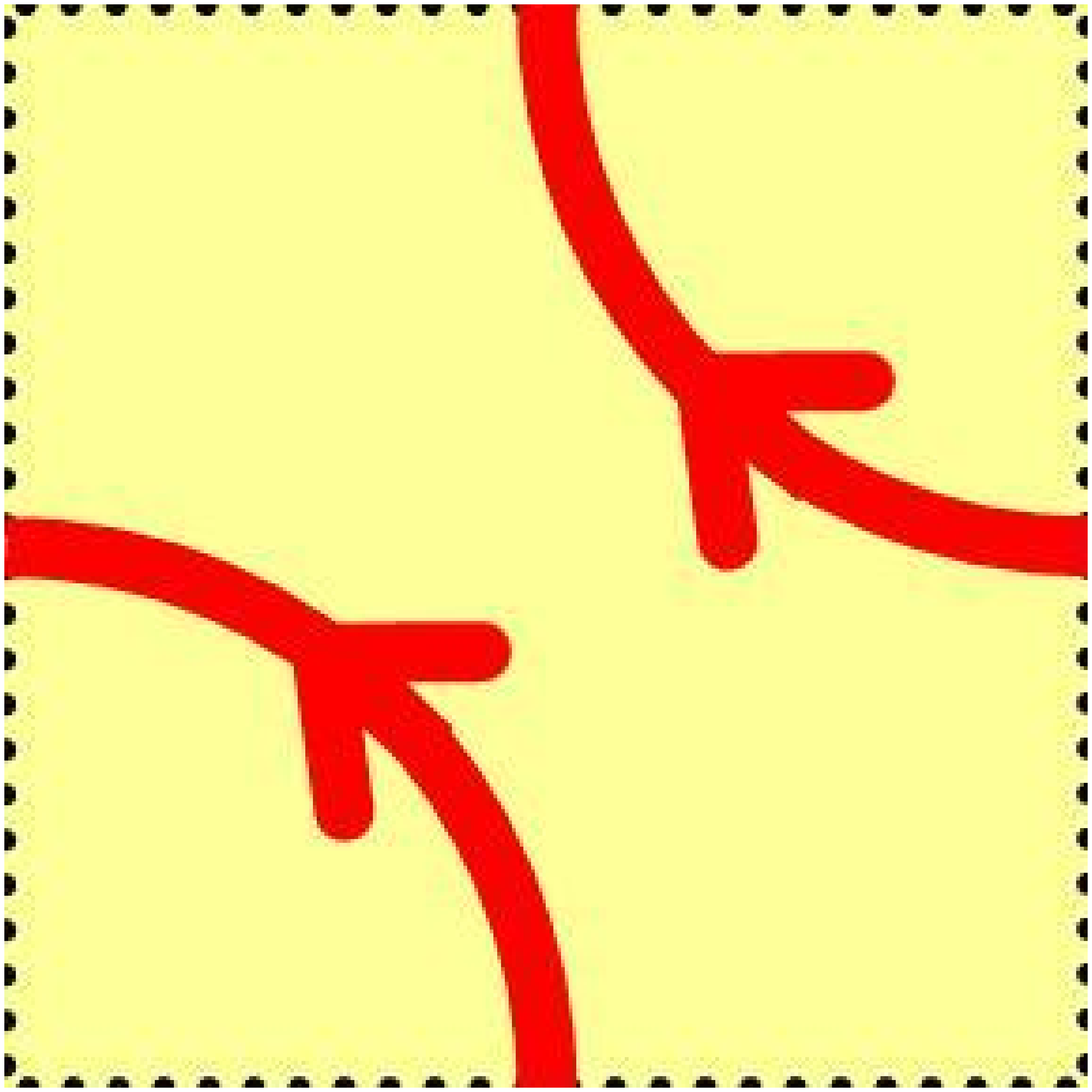}%
}%
%EndExpansion%
%TCIMACRO{\FRAME{itbpF}{0.3269in}{0.3269in}{0in}{}{}{ot18.ps}%
%width 0.3269in;  height 0.3269in;  depth 0in;  original-width 3in;
%original-height 3in;  cropleft "0";  croptop "1";  cropright "1";
%cropbottom "0";  filename 'ot18.ps';file-properties "XNPEU";}}}%
%BeginExpansion
{\includegraphics[
%natheight=3.000000in,
%natwidth=3.000000in,
height=0.3269in,
width=0.3269in
]%
{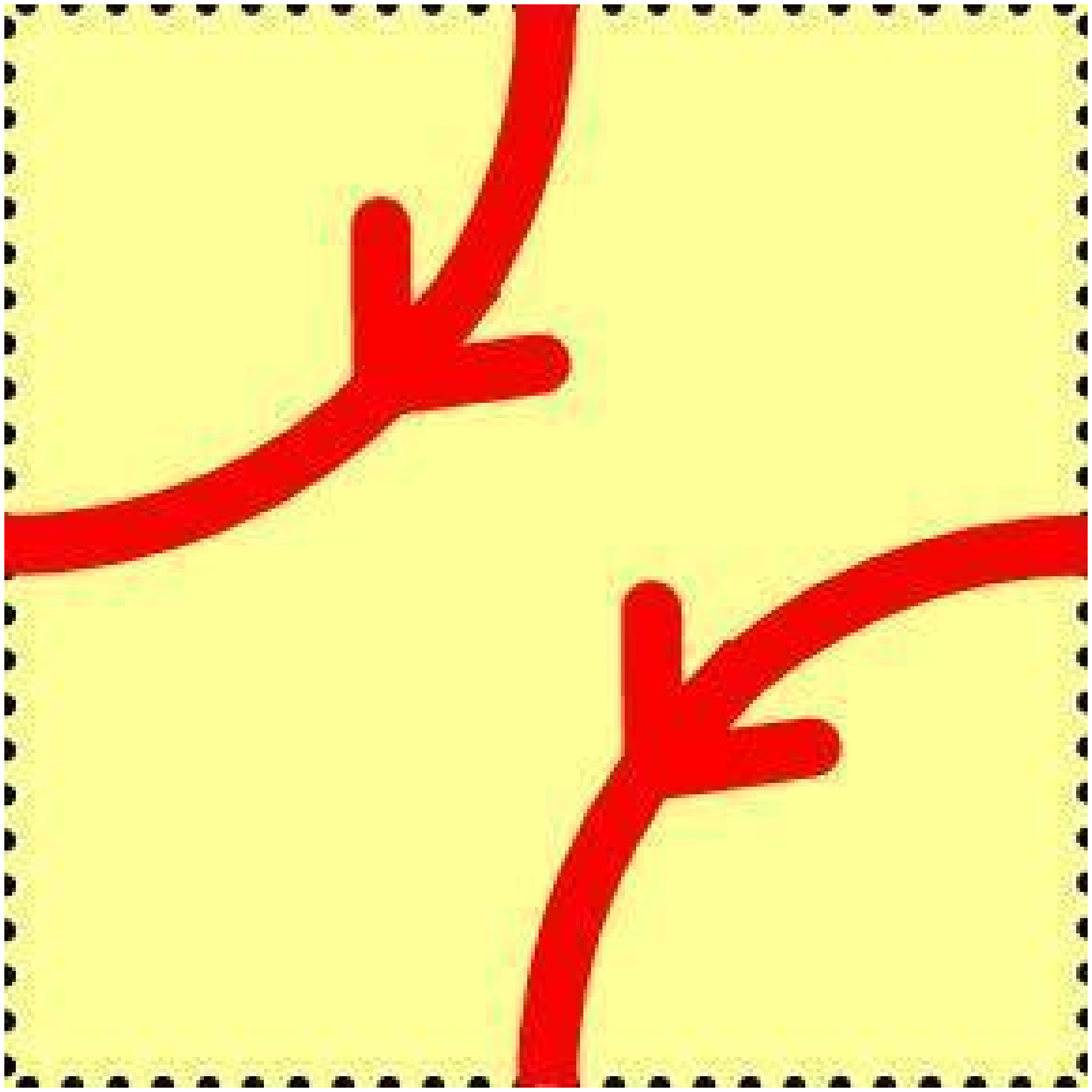}%
}%
%EndExpansion%
%TCIMACRO{\FRAME{itbpF}{0.3269in}{0.3269in}{0in}{}{}{ot19.ps}%
%width 0.3269in;  height 0.3269in;  depth 0in;  original-width 3in;
%original-height 3in;  cropleft "0";  croptop "1";  cropright "1";
%cropbottom "0";  filename 'ot19.ps';file-properties "XNPEU";}}}%
%BeginExpansion
{\includegraphics[
%natheight=3.000000in,
%natwidth=3.000000in,
height=0.3269in,
width=0.3269in
]%
{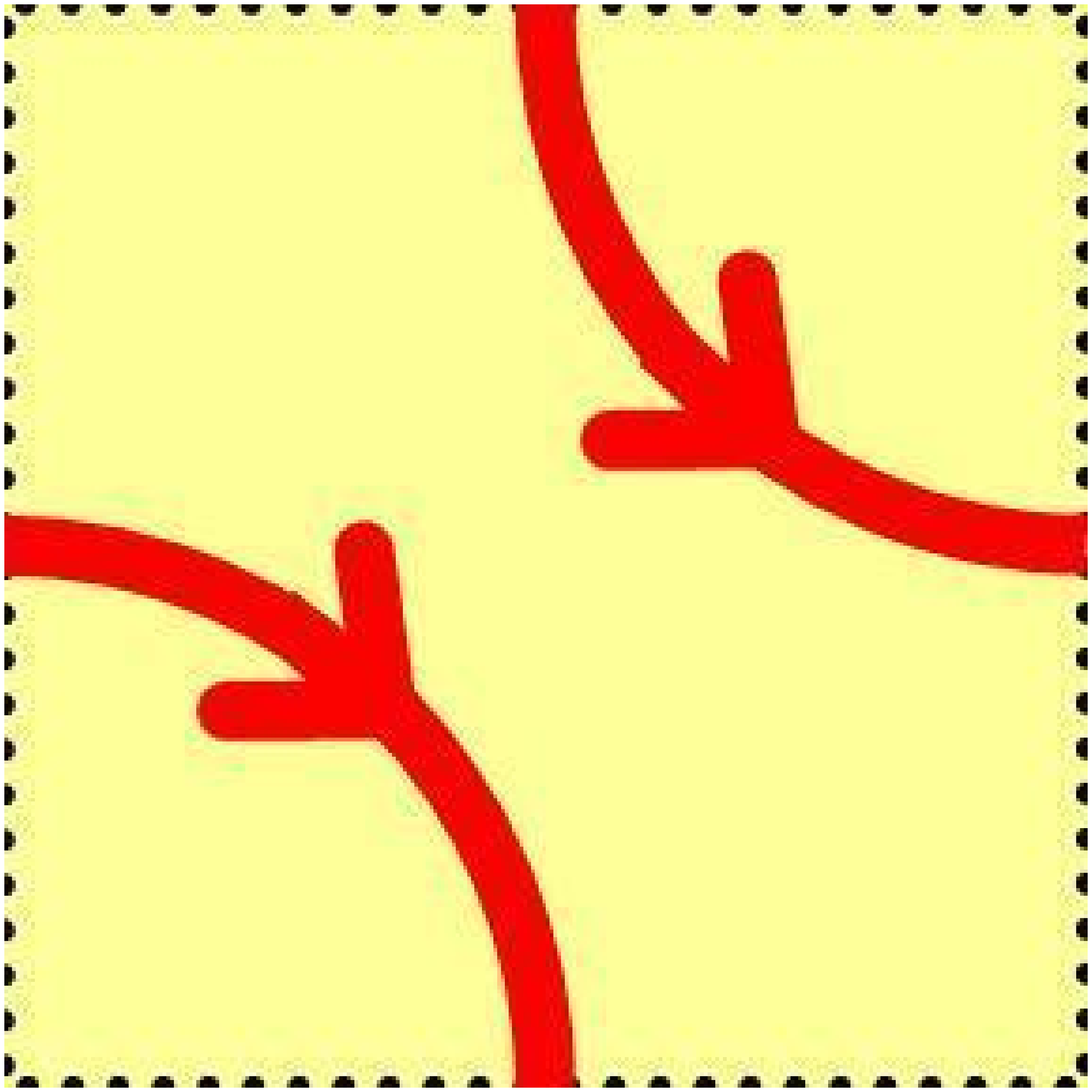}%
}%
%EndExpansion%
%TCIMACRO{\FRAME{itbpF}{0.3269in}{0.3269in}{0in}{}{}{ot20.ps}%
%width 0.3269in;  height 0.3269in;  depth 0in;  original-width 3in;
%original-height 3in;  cropleft "0";  croptop "1";  cropright "1";
%cropbottom "0";  filename 'ot20.ps';file-properties "XNPEU";}}}%
%BeginExpansion
{\includegraphics[
%natheight=3.000000in,
%natwidth=3.000000in,
height=0.3269in,
width=0.3269in
]%
{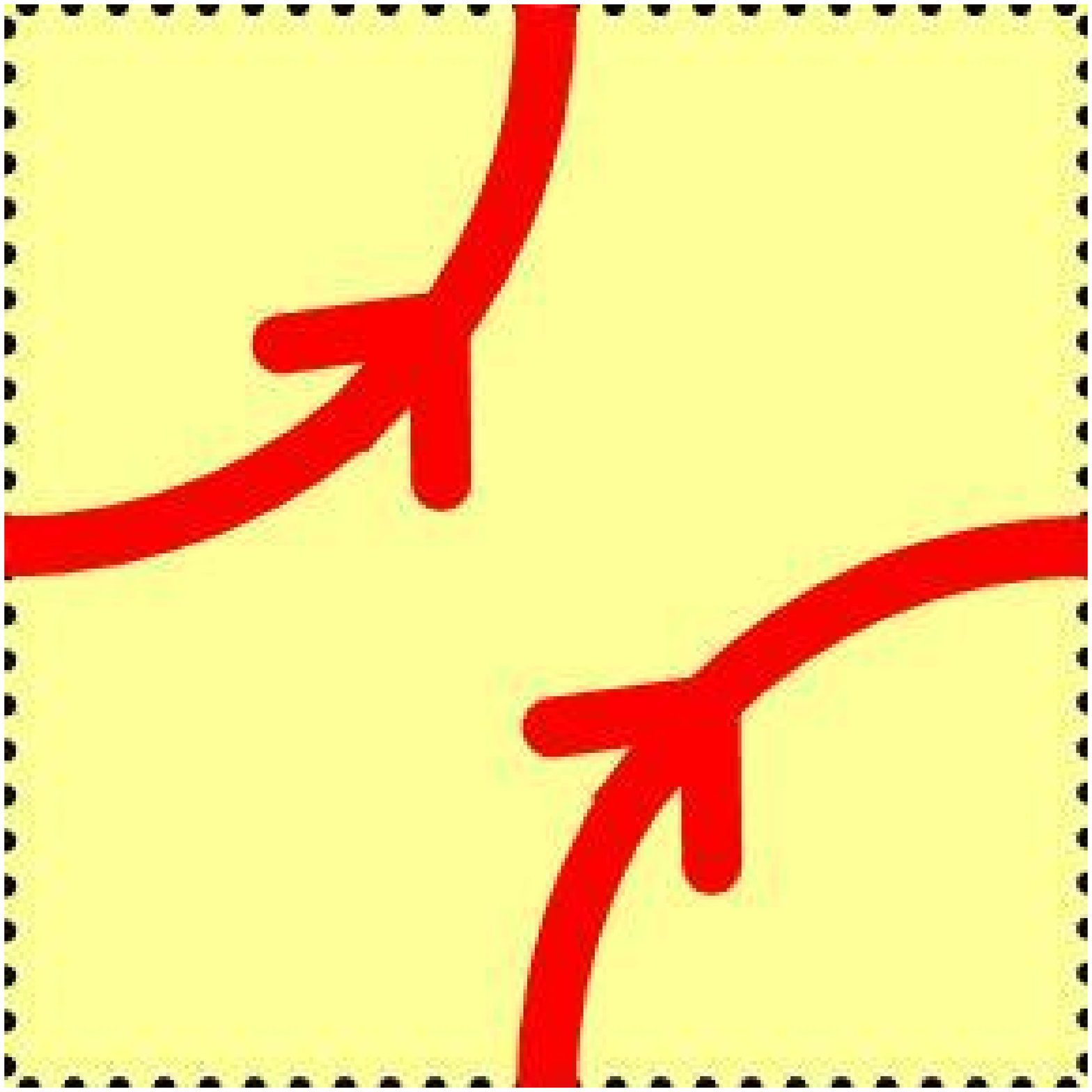}%
}%
%EndExpansion
\bigskip

\quad%
%TCIMACRO{\FRAME{itbpF}{0.3269in}{0.3269in}{0in}{}{}{ot21.ps}%
%width 0.3269in;  height 0.3269in;  depth 0in;  original-width 3in;
%original-height 3in;  cropleft "0";  croptop "1";  cropright "1";
%cropbottom "0";  filename 'ot21.ps';file-properties "XNPEU";}}}%
%BeginExpansion
{\includegraphics[
%natheight=3.000000in,
%natwidth=3.000000in,
height=0.3269in,
width=0.3269in
]%
{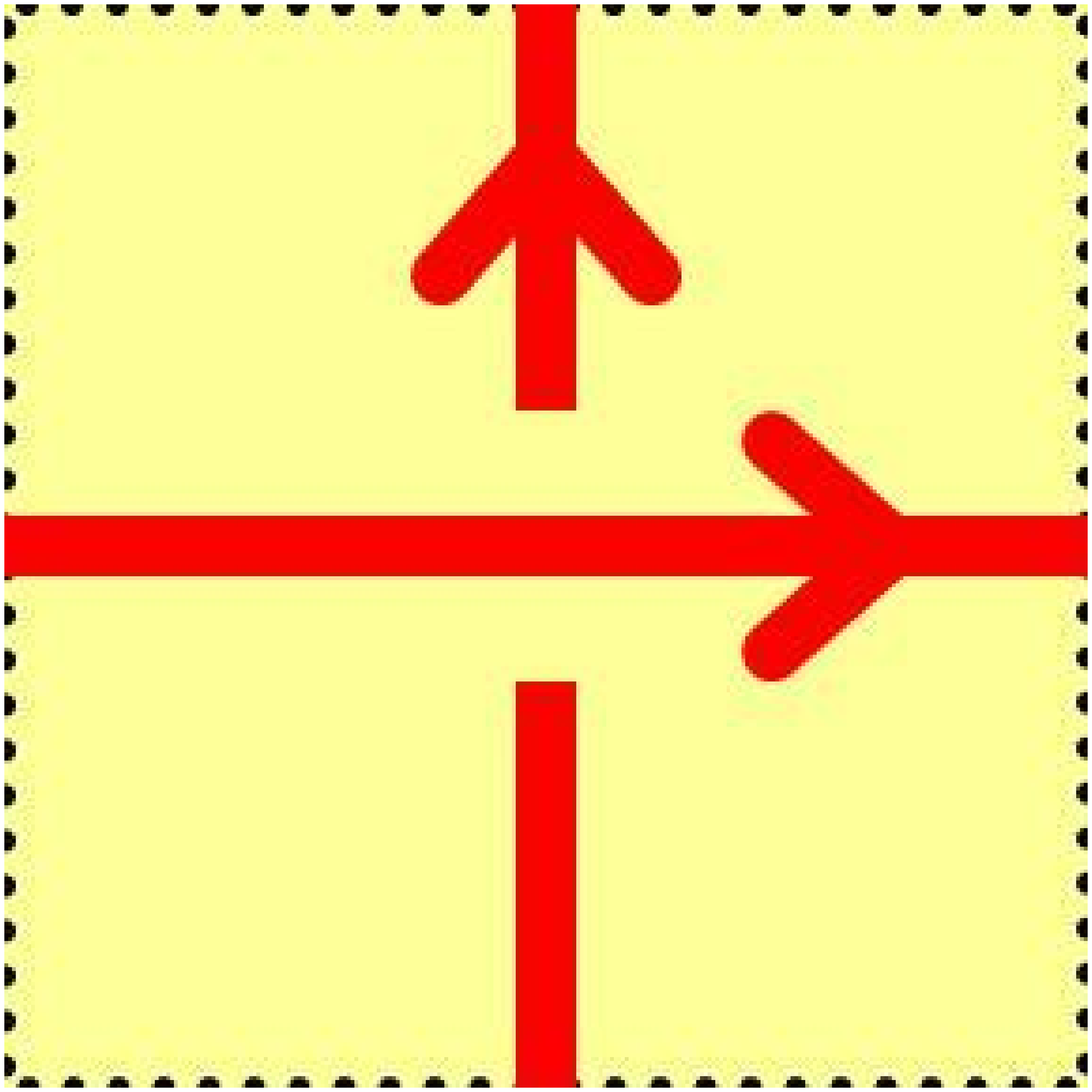}%
}%
%EndExpansion%
%TCIMACRO{\FRAME{itbpF}{0.3269in}{0.3269in}{0in}{}{}{ot22.ps}%
%width 0.3269in;  height 0.3269in;  depth 0in;  original-width 3in;
%original-height 3in;  cropleft "0";  croptop "1";  cropright "1";
%cropbottom "0";  filename 'ot22.ps';file-properties "XNPEU";}}}%
%BeginExpansion
{\includegraphics[
%natheight=3.000000in,
%natwidth=3.000000in,
height=0.3269in,
width=0.3269in
]%
{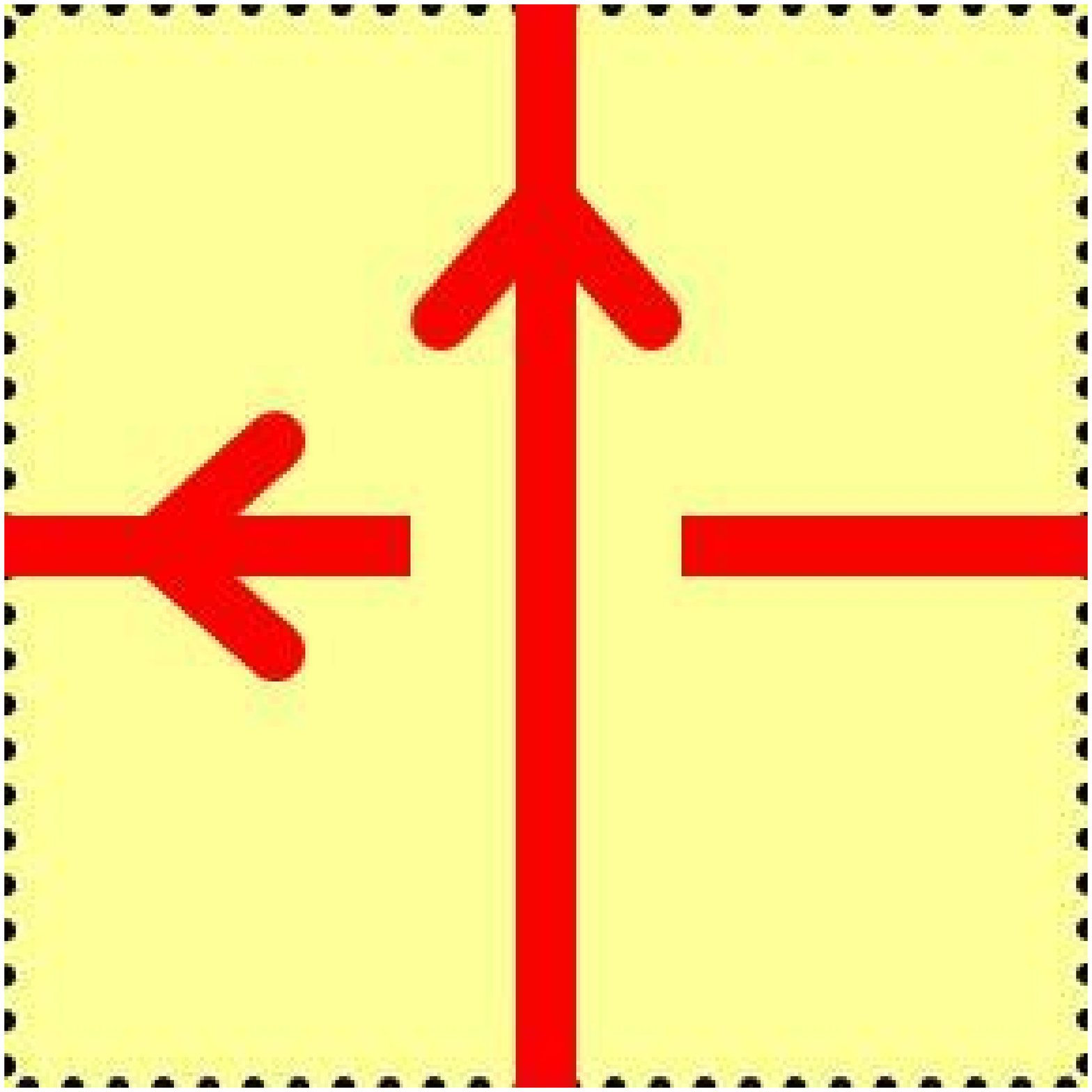}%
}%
%EndExpansion%
%TCIMACRO{\FRAME{itbpF}{0.3269in}{0.3269in}{0in}{}{}{ot23.ps}%
%width 0.3269in;  height 0.3269in;  depth 0in;  original-width 3in;
%original-height 3in;  cropleft "0";  croptop "1";  cropright "1";
%cropbottom "0";  filename 'ot23.ps';file-properties "XNPEU";}}}%
%BeginExpansion
{\includegraphics[
%natheight=3.000000in,
%natwidth=3.000000in,
height=0.3269in,
width=0.3269in
]%
{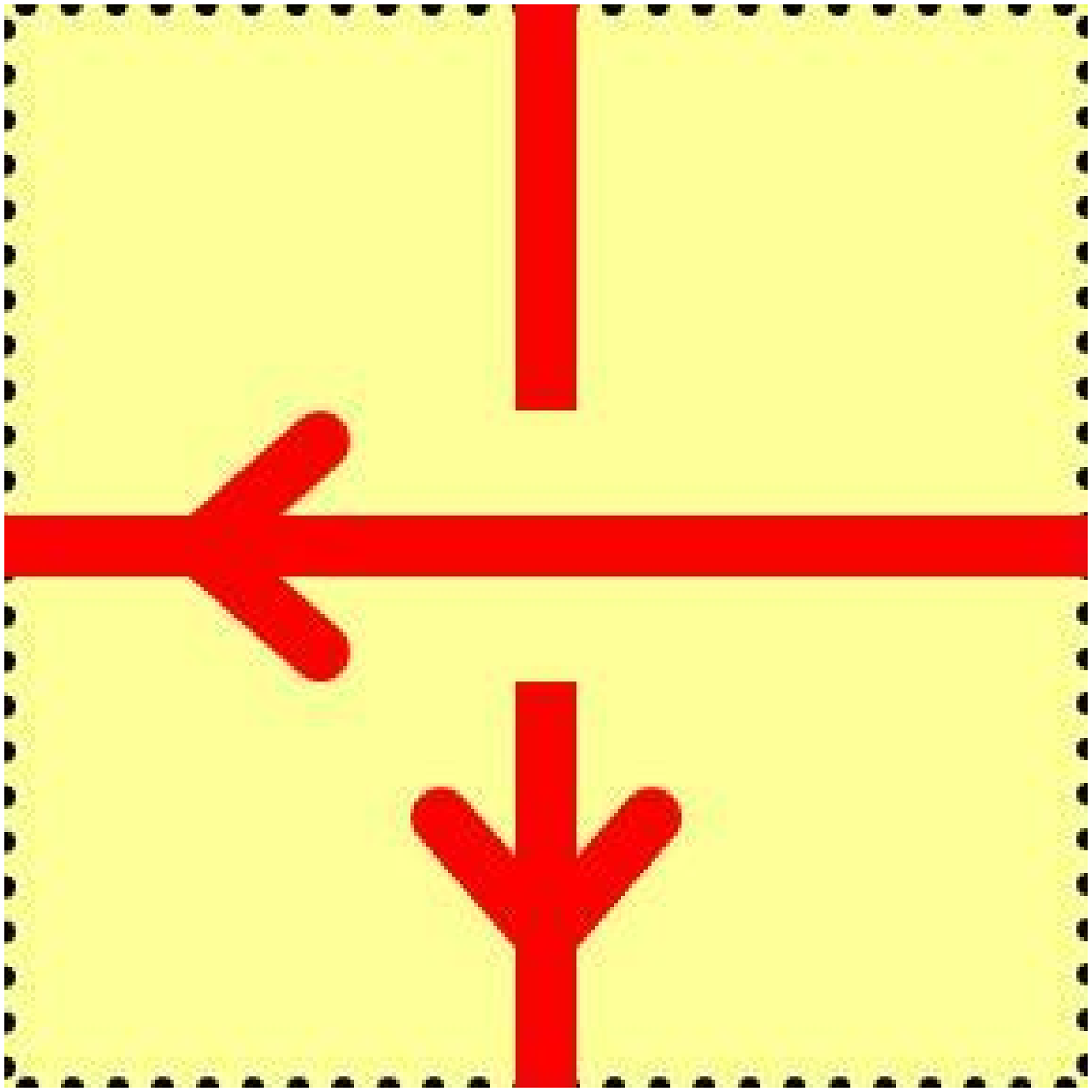}%
}%
%EndExpansion%
%TCIMACRO{\FRAME{itbpF}{0.3269in}{0.3269in}{0in}{}{}{ot24.ps}%
%width 0.3269in;  height 0.3269in;  depth 0in;  original-width 3in;
%original-height 3in;  cropleft "0";  croptop "1";  cropright "1";
%cropbottom "0";  filename 'ot24.ps';file-properties "XNPEU";}}}%
%BeginExpansion
{\includegraphics[
%natheight=3.000000in,
%natwidth=3.000000in,
height=0.3269in,
width=0.3269in
]%
{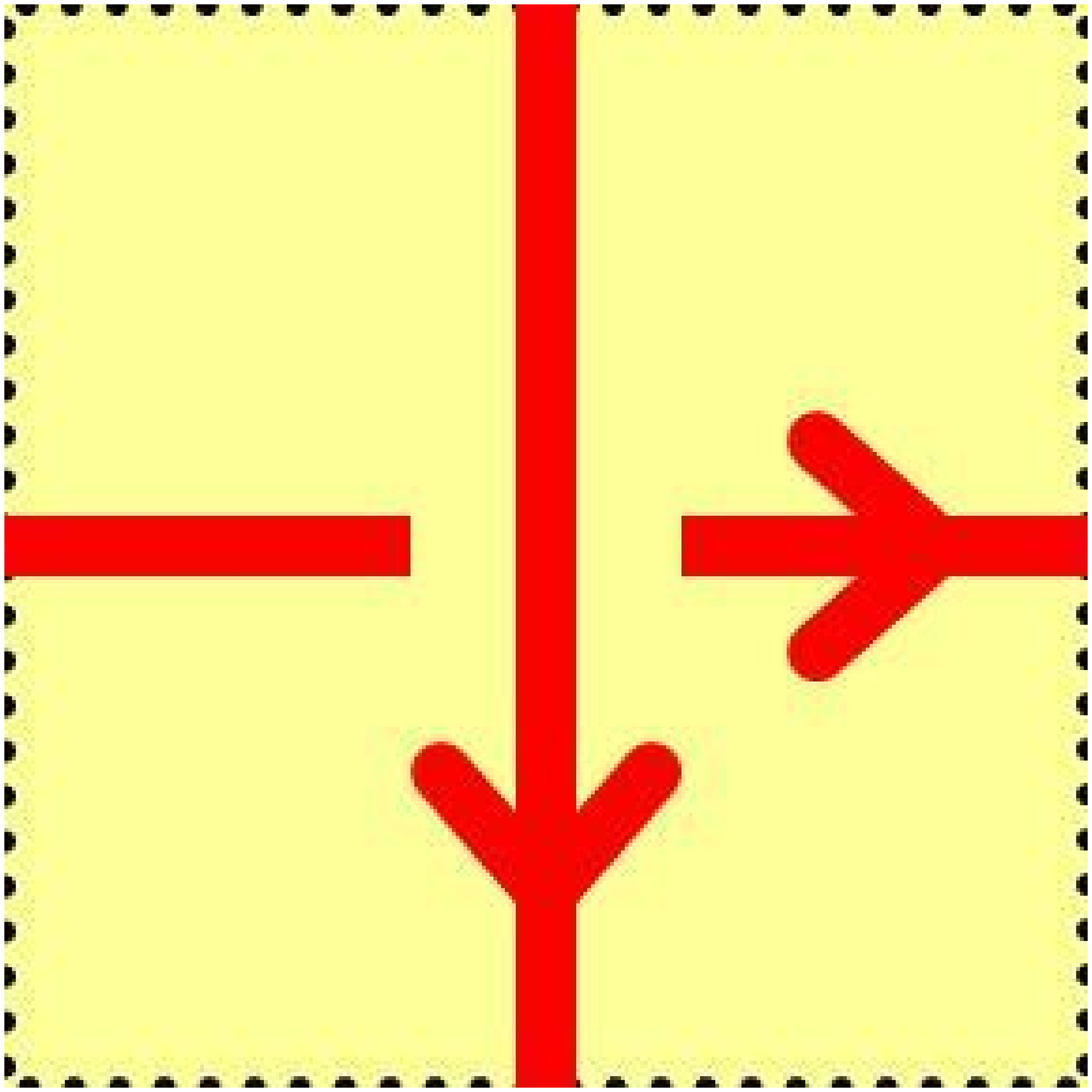}%
}%
%EndExpansion
\qquad%
%TCIMACRO{\FRAME{itbpF}{0.3269in}{0.3269in}{0in}{}{}{ot25.ps}%
%width 0.3269in;  height 0.3269in;  depth 0in;  original-width 3in;
%original-height 3in;  cropleft "0";  croptop "1";  cropright "1";
%cropbottom "0";  filename 'ot25.ps';file-properties "XNPEU";}}}%
%BeginExpansion
{\includegraphics[
%natheight=3.000000in,
%natwidth=3.000000in,
height=0.3269in,
width=0.3269in
]%
{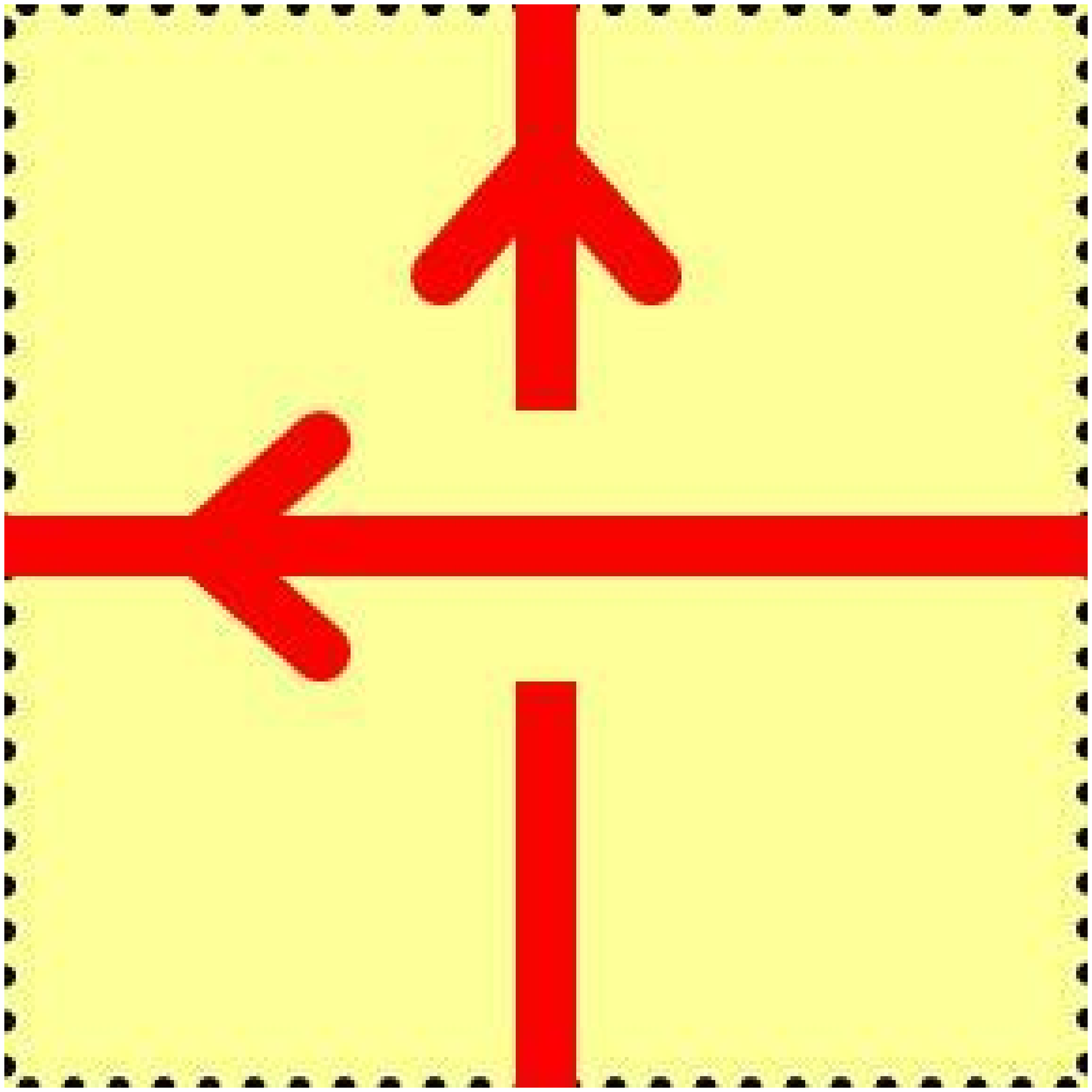}%
}%
%EndExpansion%
%TCIMACRO{\FRAME{itbpF}{0.3269in}{0.3269in}{0in}{}{}{ot26.ps}%
%width 0.3269in;  height 0.3269in;  depth 0in;  original-width 3in;
%original-height 3in;  cropleft "0";  croptop "1";  cropright "1";
%cropbottom "0";  filename 'ot26.ps';file-properties "XNPEU";}}}%
%BeginExpansion
{\includegraphics[
%natheight=3.000000in,
%natwidth=3.000000in,
height=0.3269in,
width=0.3269in
]%
{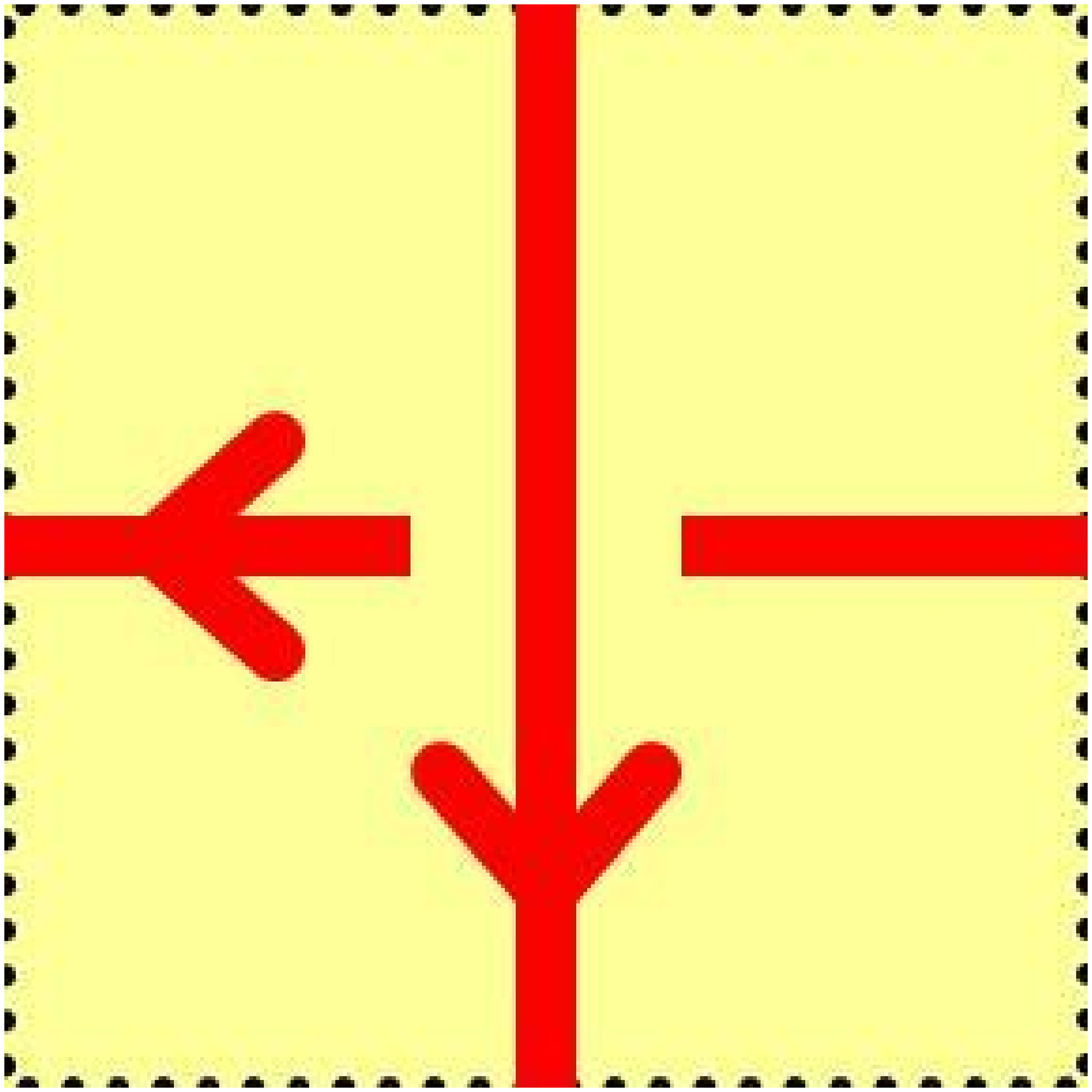}%
}%
%EndExpansion%
%TCIMACRO{\FRAME{itbpF}{0.3269in}{0.3269in}{0in}{}{}{ot27.ps}%
%width 0.3269in;  height 0.3269in;  depth 0in;  original-width 3in;
%original-height 3in;  cropleft "0";  croptop "1";  cropright "1";
%cropbottom "0";  filename 'ot27.ps';file-properties "XNPEU";}}}%
%BeginExpansion
{\includegraphics[
%natheight=3.000000in,
%natwidth=3.000000in,
height=0.3269in,
width=0.3269in
]%
{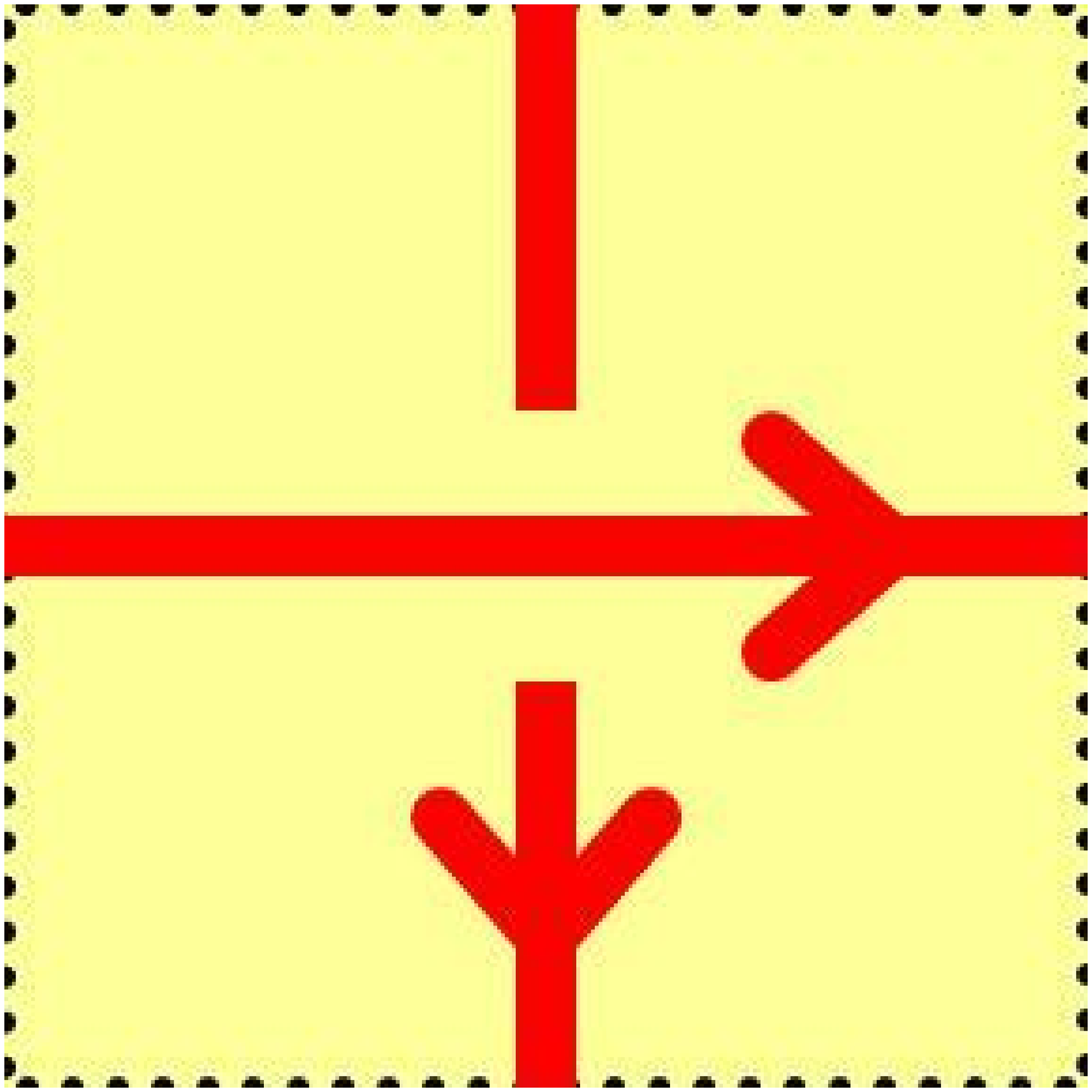}%
}%
%EndExpansion%
%TCIMACRO{\FRAME{itbpF}{0.3269in}{0.3269in}{0in}{}{}{ot28.ps}%
%width 0.3269in;  height 0.3269in;  depth 0in;  original-width 3in;
%original-height 3in;  cropleft "0";  croptop "1";  cropright "1";
%cropbottom "0";  filename 'ot28.ps';file-properties "XNPEU";}}}%
%BeginExpansion
{\includegraphics[
%natheight=3.000000in,
%natwidth=3.000000in,
height=0.3269in,
width=0.3269in
]%
{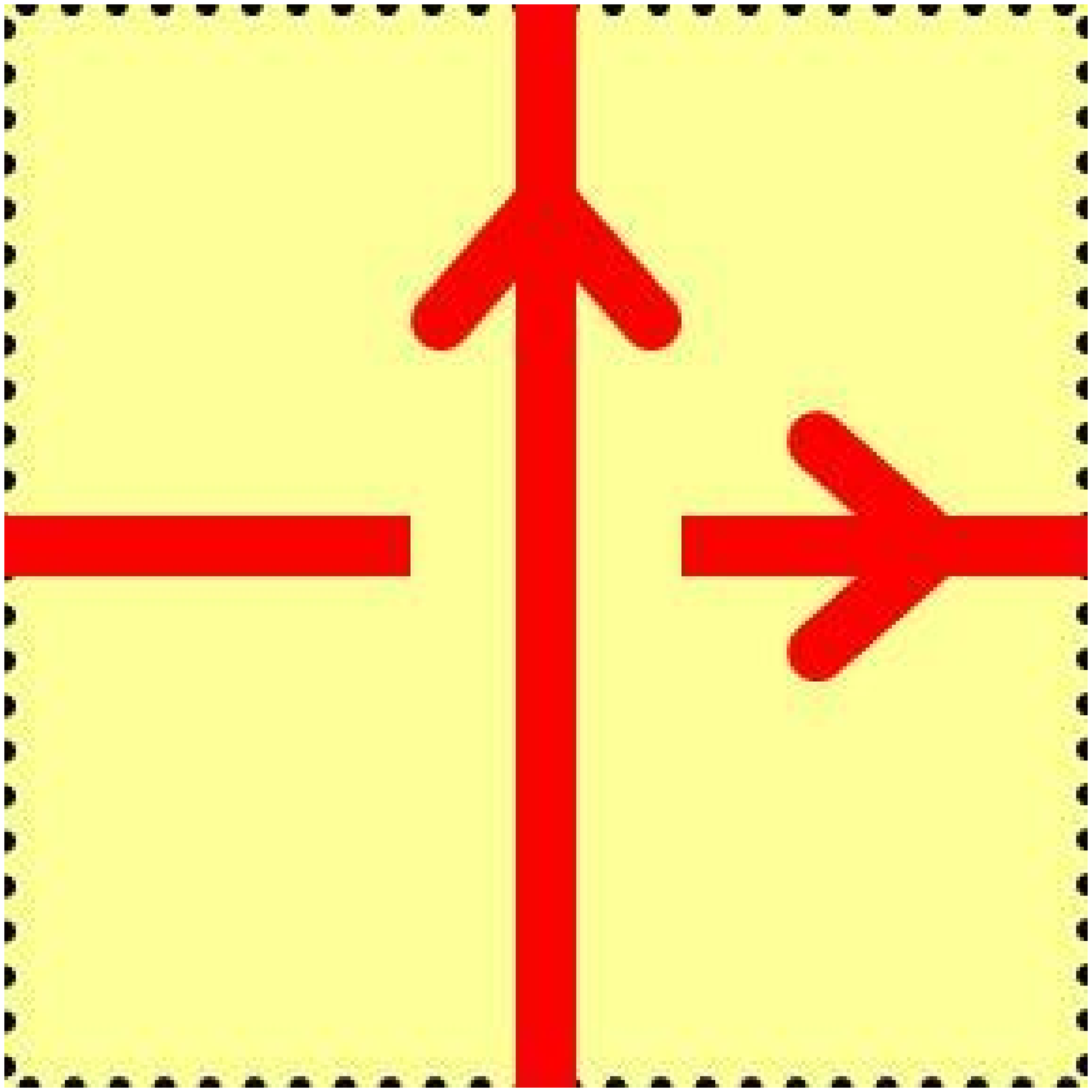}%
}%
%EndExpansion

\bigskip

\noindent called (\textbf{oriented}) \textbf{tiles}. \ We often will also
denote these tiles respectively by the symbols
\[
T_{0}^{(o)},T_{1}^{(o)},T_{2}^{(o)},\ldots,T_{28}^{(o)},
\]
Moreover, we will frequently omit the superscript `$(o)$' (standing for
`oriented') when it can be understood from context.

\bigskip

\begin{remark}
Please note that up to rotation there are exactly 9 oriented tiles. \ The
above oriented tiles are grouped according to rotational equivalence.
\end{remark}

\bigskip

\begin{definition}
Let $n$ be a positive integer. \ We define an \textbf{(oriented)} $\mathbf{n}%
$\textbf{-mosaic} as an $n\times n$ matrix $M=\left(  M_{ij}\right)  =\left(
T_{k\left(  i,j\right)  }\right)  $ of (oriented) tiles with rows and columns
indexed from $0$ to $n-1$. \ We let $\mathbb{M}^{(n)}$ also denote the
\textbf{set of oriented }$\mathbf{n}$\textbf{-mosaics}. \ 
\end{definition}

\bigskip

Two examples of oriented $4$-mosaics are shown are shown below:%
\[%
\begin{array}
[c]{cccc}%
%TCIMACRO{\FRAME{itbpF}{0.3269in}{0.3269in}{0in}{}{}{ut00.ps}%
%width 0.3269in;  height 0.3269in;  depth 0in;  original-width 3in;
%original-height 3in;  cropleft "0";  croptop "1";  cropright "1";
%cropbottom "0";  filename 'ut00.ps';file-properties "XNPEU";}}}%
%BeginExpansion
{\includegraphics[
%natheight=3.000000in,
%natwidth=3.000000in,
height=0.3269in,
width=0.3269in
]%
{ut00.ps}%
}%
%EndExpansion
&
%TCIMACRO{\FRAME{itbpF}{0.3269in}{0.3269in}{0in}{}{}{ot01.ps}%
%width 0.3269in;  height 0.3269in;  depth 0in;  original-width 3in;
%original-height 3in;  cropleft "0";  croptop "1";  cropright "1";
%cropbottom "0";  filename 'ot01.ps';file-properties "XNPEU";}}}%
%BeginExpansion
{\includegraphics[
%natheight=3.000000in,
%natwidth=3.000000in,
height=0.3269in,
width=0.3269in
]%
{ot01.ps}%
}%
%EndExpansion
&
%TCIMACRO{\FRAME{itbpF}{0.3269in}{0.3269in}{0in}{}{}{ot08.ps}%
%width 0.3269in;  height 0.3269in;  depth 0in;  original-width 3in;
%original-height 3in;  cropleft "0";  croptop "1";  cropright "1";
%cropbottom "0";  filename 'ot08.ps';file-properties "XNPEU";}}}%
%BeginExpansion
{\includegraphics[
%natheight=3.000000in,
%natwidth=3.000000in,
height=0.3269in,
width=0.3269in
]%
{ot08.ps}%
}%
%EndExpansion
&
%TCIMACRO{\FRAME{itbpF}{0.3269in}{0.3269in}{0in}{}{}{ut00.ps}%
%width 0.3269in;  height 0.3269in;  depth 0in;  original-width 3in;
%original-height 3in;  cropleft "0";  croptop "1";  cropright "1";
%cropbottom "0";  filename 'ut00.ps';file-properties "XNPEU";}}}%
%BeginExpansion
{\includegraphics[
%natheight=3.000000in,
%natwidth=3.000000in,
height=0.3269in,
width=0.3269in
]%
{ut00.ps}%
}%
%EndExpansion
\\%
%TCIMACRO{\FRAME{itbpF}{0.3269in}{0.3269in}{0in}{}{}{ot12.ps}%
%width 0.3269in;  height 0.3269in;  depth 0in;  original-width 3in;
%original-height 3in;  cropleft "0";  croptop "1";  cropright "1";
%cropbottom "0";  filename 'ot12.ps';file-properties "XNPEU";}}}%
%BeginExpansion
{\includegraphics[
%natheight=3.000000in,
%natwidth=3.000000in,
height=0.3269in,
width=0.3269in
]%
{ot12.ps}%
}%
%EndExpansion
&
%TCIMACRO{\FRAME{itbpF}{0.3269in}{0.3269in}{0in}{}{}{ot22.ps}%
%width 0.3269in;  height 0.3269in;  depth 0in;  original-width 3in;
%original-height 3in;  cropleft "0";  croptop "1";  cropright "1";
%cropbottom "0";  filename 'ot22.ps';file-properties "XNPEU";}}}%
%BeginExpansion
{\includegraphics[
%natheight=3.000000in,
%natwidth=3.000000in,
height=0.3269in,
width=0.3269in
]%
{ot22.ps}%
}%
%EndExpansion
&
%TCIMACRO{\FRAME{itbpF}{0.3269in}{0.3269in}{0in}{}{}{ot08.ps}%
%width 0.3269in;  height 0.3269in;  depth 0in;  original-width 3in;
%original-height 3in;  cropleft "0";  croptop "1";  cropright "1";
%cropbottom "0";  filename 'ot08.ps';file-properties "XNPEU";}}}%
%BeginExpansion
{\includegraphics[
%natheight=3.000000in,
%natwidth=3.000000in,
height=0.3269in,
width=0.3269in
]%
{ot08.ps}%
}%
%EndExpansion
&
%TCIMACRO{\FRAME{itbpF}{0.3269in}{0.3269in}{0in}{}{}{ot09.ps}%
%width 0.3269in;  height 0.3269in;  depth 0in;  original-width 3in;
%original-height 3in;  cropleft "0";  croptop "1";  cropright "1";
%cropbottom "0";  filename 'ot09.ps';file-properties "XNPEU";}}}%
%BeginExpansion
{\includegraphics[
%natheight=3.000000in,
%natwidth=3.000000in,
height=0.3269in,
width=0.3269in
]%
{ot09.ps}%
}%
%EndExpansion
\\%
%TCIMACRO{\FRAME{itbpF}{0.3269in}{0.3269in}{0in}{}{}{ot27.ps}%
%width 0.3269in;  height 0.3269in;  depth 0in;  original-width 3in;
%original-height 3in;  cropleft "0";  croptop "1";  cropright "1";
%cropbottom "0";  filename 'ot27.ps';file-properties "XNPEU";}}}%
%BeginExpansion
{\includegraphics[
%natheight=3.000000in,
%natwidth=3.000000in,
height=0.3269in,
width=0.3269in
]%
{ot27.ps}%
}%
%EndExpansion
&
%TCIMACRO{\FRAME{itbpF}{0.3269in}{0.3269in}{0in}{}{}{ot11.ps}%
%width 0.3269in;  height 0.3269in;  depth 0in;  original-width 3in;
%original-height 3in;  cropleft "0";  croptop "1";  cropright "1";
%cropbottom "0";  filename 'ot11.ps';file-properties "XNPEU";}}}%
%BeginExpansion
{\includegraphics[
%natheight=3.000000in,
%natwidth=3.000000in,
height=0.3269in,
width=0.3269in
]%
{ot11.ps}%
}%
%EndExpansion
&
%TCIMACRO{\FRAME{itbpF}{0.3269in}{0.3269in}{0in}{}{}{ot15.ps}%
%width 0.3269in;  height 0.3269in;  depth 0in;  original-width 3in;
%original-height 3in;  cropleft "0";  croptop "1";  cropright "1";
%cropbottom "0";  filename 'ot15.ps';file-properties "XNPEU";}}}%
%BeginExpansion
{\includegraphics[
%natheight=3.000000in,
%natwidth=3.000000in,
height=0.3269in,
width=0.3269in
]%
{ot15.ps}%
}%
%EndExpansion
&
%TCIMACRO{\FRAME{itbpF}{0.3269in}{0.3269in}{0in}{}{}{ot06.ps}%
%width 0.3269in;  height 0.3269in;  depth 0in;  original-width 3in;
%original-height 3in;  cropleft "0";  croptop "1";  cropright "1";
%cropbottom "0";  filename 'ot06.ps';file-properties "XNPEU";}}}%
%BeginExpansion
{\includegraphics[
%natheight=3.000000in,
%natwidth=3.000000in,
height=0.3269in,
width=0.3269in
]%
{ot06.ps}%
}%
%EndExpansion
\\%
%TCIMACRO{\FRAME{itbpF}{0.3269in}{0.3269in}{0in}{}{}{ot07.ps}%
%width 0.3269in;  height 0.3269in;  depth 0in;  original-width 3in;
%original-height 3in;  cropleft "0";  croptop "1";  cropright "1";
%cropbottom "0";  filename 'ot07.ps';file-properties "XNPEU";}}}%
%BeginExpansion
{\includegraphics[
%natheight=3.000000in,
%natwidth=3.000000in,
height=0.3269in,
width=0.3269in
]%
{ot07.ps}%
}%
%EndExpansion
&
%TCIMACRO{\FRAME{itbpF}{0.3269in}{0.3269in}{0in}{}{}{ot04.ps}%
%width 0.3269in;  height 0.3269in;  depth 0in;  original-width 3in;
%original-height 3in;  cropleft "0";  croptop "1";  cropright "1";
%cropbottom "0";  filename 'ot04.ps';file-properties "XNPEU";}}}%
%BeginExpansion
{\includegraphics[
%natheight=3.000000in,
%natwidth=3.000000in,
height=0.3269in,
width=0.3269in
]%
{ot04.ps}%
}%
%EndExpansion
&
%TCIMACRO{\FRAME{itbpF}{0.3269in}{0.3269in}{0in}{}{}{ot12.ps}%
%width 0.3269in;  height 0.3269in;  depth 0in;  original-width 3in;
%original-height 3in;  cropleft "0";  croptop "1";  cropright "1";
%cropbottom "0";  filename 'ot12.ps';file-properties "XNPEU";}}}%
%BeginExpansion
{\includegraphics[
%natheight=3.000000in,
%natwidth=3.000000in,
height=0.3269in,
width=0.3269in
]%
{ot12.ps}%
}%
%EndExpansion
&
%TCIMACRO{\FRAME{itbpF}{0.3269in}{0.3269in}{0in}{}{}{ut00.ps}%
%width 0.3269in;  height 0.3269in;  depth 0in;  original-width 3in;
%original-height 3in;  cropleft "0";  croptop "1";  cropright "1";
%cropbottom "0";  filename 'ut00.ps';file-properties "XNPEU";}}}%
%BeginExpansion
{\includegraphics[
%natheight=3.000000in,
%natwidth=3.000000in,
height=0.3269in,
width=0.3269in
]%
{ut00.ps}%
}%
%EndExpansion
\end{array}
\qquad\qquad\qquad%
\begin{array}
[c]{cccc}%
%TCIMACRO{\FRAME{itbpF}{0.3269in}{0.3269in}{0in}{}{}{ut00.ps}%
%width 0.3269in;  height 0.3269in;  depth 0in;  original-width 3in;
%original-height 3in;  cropleft "0";  croptop "1";  cropright "1";
%cropbottom "0";  filename 'ut00.ps';file-properties "XNPEU";}}}%
%BeginExpansion
{\includegraphics[
%natheight=3.000000in,
%natwidth=3.000000in,
height=0.3269in,
width=0.3269in
]%
{ut00.ps}%
}%
%EndExpansion
&
%TCIMACRO{\FRAME{itbpF}{0.3269in}{0.3269in}{0in}{}{}{ot06.ps}%
%width 0.3269in;  height 0.3269in;  depth 0in;  original-width 3in;
%original-height 3in;  cropleft "0";  croptop "1";  cropright "1";
%cropbottom "0";  filename 'ot06.ps';file-properties "XNPEU";}}}%
%BeginExpansion
{\includegraphics[
%natheight=3.000000in,
%natwidth=3.000000in,
height=0.3269in,
width=0.3269in
]%
{ot06.ps}%
}%
%EndExpansion
&
%TCIMACRO{\FRAME{itbpF}{0.3269in}{0.3269in}{0in}{}{}{ot05.ps}%
%width 0.3269in;  height 0.3269in;  depth 0in;  original-width 3in;
%original-height 3in;  cropleft "0";  croptop "1";  cropright "1";
%cropbottom "0";  filename 'ot05.ps';file-properties "XNPEU";}}}%
%BeginExpansion
{\includegraphics[
%natheight=3.000000in,
%natwidth=3.000000in,
height=0.3269in,
width=0.3269in
]%
{ot05.ps}%
}%
%EndExpansion
&
%TCIMACRO{\FRAME{itbpF}{0.3269in}{0.3269in}{0in}{}{}{ut00.ps}%
%width 0.3269in;  height 0.3269in;  depth 0in;  original-width 3in;
%original-height 3in;  cropleft "0";  croptop "1";  cropright "1";
%cropbottom "0";  filename 'ut00.ps';file-properties "XNPEU";}}}%
%BeginExpansion
{\includegraphics[
%natheight=3.000000in,
%natwidth=3.000000in,
height=0.3269in,
width=0.3269in
]%
{ut00.ps}%
}%
%EndExpansion
\\%
%TCIMACRO{\FRAME{itbpF}{0.3269in}{0.3269in}{0in}{}{}{ot06.ps}%
%width 0.3269in;  height 0.3269in;  depth 0in;  original-width 3in;
%original-height 3in;  cropleft "0";  croptop "1";  cropright "1";
%cropbottom "0";  filename 'ot06.ps';file-properties "XNPEU";}}}%
%BeginExpansion
{\includegraphics[
%natheight=3.000000in,
%natwidth=3.000000in,
height=0.3269in,
width=0.3269in
]%
{ot06.ps}%
}%
%EndExpansion
&
%TCIMACRO{\FRAME{itbpF}{0.3269in}{0.3269in}{0in}{}{}{ot21.ps}%
%width 0.3269in;  height 0.3269in;  depth 0in;  original-width 3in;
%original-height 3in;  cropleft "0";  croptop "1";  cropright "1";
%cropbottom "0";  filename 'ot21.ps';file-properties "XNPEU";}}}%
%BeginExpansion
{\includegraphics[
%natheight=3.000000in,
%natwidth=3.000000in,
height=0.3269in,
width=0.3269in
]%
{ot21.ps}%
}%
%EndExpansion
&
%TCIMACRO{\FRAME{itbpF}{0.3269in}{0.3269in}{0in}{}{}{ot24.ps}%
%width 0.3269in;  height 0.3269in;  depth 0in;  original-width 3in;
%original-height 3in;  cropleft "0";  croptop "1";  cropright "1";
%cropbottom "0";  filename 'ot24.ps';file-properties "XNPEU";}}}%
%BeginExpansion
{\includegraphics[
%natheight=3.000000in,
%natwidth=3.000000in,
height=0.3269in,
width=0.3269in
]%
{ot24.ps}%
}%
%EndExpansion
&
%TCIMACRO{\FRAME{itbpF}{0.3269in}{0.3269in}{0in}{}{}{ot05.ps}%
%width 0.3269in;  height 0.3269in;  depth 0in;  original-width 3in;
%original-height 3in;  cropleft "0";  croptop "1";  cropright "1";
%cropbottom "0";  filename 'ot05.ps';file-properties "XNPEU";}}}%
%BeginExpansion
{\includegraphics[
%natheight=3.000000in,
%natwidth=3.000000in,
height=0.3269in,
width=0.3269in
]%
{ot05.ps}%
}%
%EndExpansion
\\%
%TCIMACRO{\FRAME{itbpF}{0.3269in}{0.3269in}{0in}{}{}{ot02.ps}%
%width 0.3269in;  height 0.3269in;  depth 0in;  original-width 3in;
%original-height 3in;  cropleft "0";  croptop "1";  cropright "1";
%cropbottom "0";  filename 'ot02.ps';file-properties "XNPEU";}}}%
%BeginExpansion
{\includegraphics[
%natheight=3.000000in,
%natwidth=3.000000in,
height=0.3269in,
width=0.3269in
]%
{ot02.ps}%
}%
%EndExpansion
&
%TCIMACRO{\FRAME{itbpF}{0.3269in}{0.3269in}{0in}{}{}{ot07.ps}%
%width 0.3269in;  height 0.3269in;  depth 0in;  original-width 3in;
%original-height 3in;  cropleft "0";  croptop "1";  cropright "1";
%cropbottom "0";  filename 'ot07.ps';file-properties "XNPEU";}}}%
%BeginExpansion
{\includegraphics[
%natheight=3.000000in,
%natwidth=3.000000in,
height=0.3269in,
width=0.3269in
]%
{ot07.ps}%
}%
%EndExpansion
&
%TCIMACRO{\FRAME{itbpF}{0.3269in}{0.3269in}{0in}{}{}{ot23.ps}%
%width 0.3269in;  height 0.3269in;  depth 0in;  original-width 3in;
%original-height 3in;  cropleft "0";  croptop "1";  cropright "1";
%cropbottom "0";  filename 'ot23.ps';file-properties "XNPEU";}}}%
%BeginExpansion
{\includegraphics[
%natheight=3.000000in,
%natwidth=3.000000in,
height=0.3269in,
width=0.3269in
]%
{ot23.ps}%
}%
%EndExpansion
&
%TCIMACRO{\FRAME{itbpF}{0.3269in}{0.3269in}{0in}{}{}{ot08.ps}%
%width 0.3269in;  height 0.3269in;  depth 0in;  original-width 3in;
%original-height 3in;  cropleft "0";  croptop "1";  cropright "1";
%cropbottom "0";  filename 'ot08.ps';file-properties "XNPEU";}}}%
%BeginExpansion
{\includegraphics[
%natheight=3.000000in,
%natwidth=3.000000in,
height=0.3269in,
width=0.3269in
]%
{ot08.ps}%
}%
%EndExpansion
\\%
%TCIMACRO{\FRAME{itbpF}{0.3269in}{0.3269in}{0in}{}{}{ot07.ps}%
%width 0.3269in;  height 0.3269in;  depth 0in;  original-width 3in;
%original-height 3in;  cropleft "0";  croptop "1";  cropright "1";
%cropbottom "0";  filename 'ot07.ps';file-properties "XNPEU";}}}%
%BeginExpansion
{\includegraphics[
%natheight=3.000000in,
%natwidth=3.000000in,
height=0.3269in,
width=0.3269in
]%
{ot07.ps}%
}%
%EndExpansion
&
%TCIMACRO{\FRAME{itbpF}{0.3269in}{0.3269in}{0in}{}{}{ot03.ps}%
%width 0.3269in;  height 0.3269in;  depth 0in;  original-width 3in;
%original-height 3in;  cropleft "0";  croptop "1";  cropright "1";
%cropbottom "0";  filename 'ot03.ps';file-properties "XNPEU";}}}%
%BeginExpansion
{\includegraphics[
%natheight=3.000000in,
%natwidth=3.000000in,
height=0.3269in,
width=0.3269in
]%
{ot03.ps}%
}%
%EndExpansion
&
%TCIMACRO{\FRAME{itbpF}{0.3269in}{0.3269in}{0in}{}{}{ot08.ps}%
%width 0.3269in;  height 0.3269in;  depth 0in;  original-width 3in;
%original-height 3in;  cropleft "0";  croptop "1";  cropright "1";
%cropbottom "0";  filename 'ot08.ps';file-properties "XNPEU";}}}%
%BeginExpansion
{\includegraphics[
%natheight=3.000000in,
%natwidth=3.000000in,
height=0.3269in,
width=0.3269in
]%
{ot08.ps}%
}%
%EndExpansion
&
%TCIMACRO{\FRAME{itbpF}{0.3269in}{0.3269in}{0in}{}{}{ut00.ps}%
%width 0.3269in;  height 0.3269in;  depth 0in;  original-width 3in;
%original-height 3in;  cropleft "0";  croptop "1";  cropright "1";
%cropbottom "0";  filename 'ut00.ps';file-properties "XNPEU";}}}%
%BeginExpansion
{\includegraphics[
%natheight=3.000000in,
%natwidth=3.000000in,
height=0.3269in,
width=0.3269in
]%
{ut00.ps}%
}%
%EndExpansion
\end{array}
\]
\bigskip

A \textbf{connection point} of an oriented tile is defined as the midpoint of
a tile edge which is either the beginning or ending point of an oriented curve
drawn on the tile. \ We define the \textbf{sign} of a connection point as
\textbf{minus} $(-)$ or \textbf{plus} $(+)$ accordingly as it is either the
beginning point or the ending point of oriented tile curve.

\bigskip

We say that two tiles in an oriented mosaic are \textbf{contiguous} if they
lie immediately next to each other in either the same row or the same column.
\ An oriented tile within a oriented mosaic is said to be \textbf{suitably
connected} if each its connection points touches a connection point of
opposite sign of a contiguous tile.

\bigskip

We are now in a postiion to define what is meant by an oriented mosaic knot:

\begin{definition}
An \textbf{oriented knot }$\mathbf{n}$-\textbf{mosaic} is an oriented mosaic
in which all tile connection points are suitably connected. \ We also let
$\mathbb{K}^{(n)}$ denote the subset of $\mathbb{M}^{(n)}$ of all oriented
$\mathbf{n}$-mosaic knots.
\end{definition}

\bigskip

The remaining definitions are straight forward, and left to the reader.

\bigskip
\end{document}